\documentclass[11pt]{article}
\ifdefined\pdfminorversion
\fi

\usepackage[utf8]{inputenc}
\usepackage{amsmath,amssymb,mathtools,amsthm}
\usepackage{graphicx}
\usepackage{thmtools}
\usepackage{float}
\usepackage{xcolor}
\usepackage{stmaryrd}

\definecolor{linkcolour}{rgb}{0.15,0.15,0.55}
\definecolor{citecolour}{rgb}{0.15,0.15,0.55}

\usepackage[linktoc=all,backref=page]{hyperref}
\hypersetup{
  colorlinks=true,
  linkcolor=linkcolour,
  urlcolor=citecolour,
  citecolor=citecolour,
  linktoc=all,
  hypertexnames=false,
  unicode=true,
  bookmarksnumbered=false,
  pdfmenubar=true,
  pdftoolbar=true,
  pdfauthor={Mohamad Alameddine, Olivier Marchal,
             Nikita Nikolaev, Nicolas Orantin},
  pdftitle={The WKB Riemann--Hilbert Problem
            for the Painlev\'e I System}
}

\usepackage[
  top=2cm,
  bottom=2cm,
  left=2cm,
  right=2cm,
  a4paper
]{geometry}
\usepackage[margin=2cm]{caption}
\DeclareGraphicsExtensions{.pdf,.png,.jpg,.eps}

\makeatletter
\@addtoreset{equation}{section}
\makeatother

\newtheorem{theorem}{Theorem}[section]
\newtheorem{proposition}{Proposition}[section]
\newtheorem{lemma}{Lemma}[section]
\newtheorem{corollary}{Corollary}[section]

\theoremstyle{definition}
\newtheorem{definition}{Definition}[section]
\newtheorem{remark}{Remark}[section]
\newtheorem{problem}{Problem}[section]

\renewcommand{\sectionautorefname}{Section}
\renewcommand{\subsectionautorefname}{Subsection}
\renewcommand{\subsubsectionautorefname}{Subsection}

\renewcommand{\appendixautorefname}{Appendix}

\ExplSyntaxOn
\NewCommandCopy{\RHPOriginalAutoref}{\autoref}
\RenewDocumentCommand{\autoref}{s m}
  {
    \group_begin:
    \tl_set:Nx \l_tmpa_tl { \getrefnumber{#2} }
    \exp_args:NnV \regex_match:nnT
      { \A [A-Z] ( \. [0-9]+ )* \Z }
      \l_tmpa_tl
      {
        \cs_set_eq:NN
          \sectionautorefname \appendixautorefname
        \cs_set_eq:NN
          \subsectionautorefname \appendixautorefname
        \cs_set_eq:NN
          \subsubsectionautorefname \appendixautorefname
      }
    \IfBooleanTF{#1}
      { \RHPOriginalAutoref*{#2} }
      { \RHPOriginalAutoref{#2} }
    \group_end:
  }
\ExplSyntaxOff

\begin{document}

\title{The WKB Riemann--Hilbert Problem \\ for the Painlevé I System}

\author{Mohamad Alameddine, Olivier Marchal, Nikita Nikolaev, and Nicolas Orantin}

\date{}

\maketitle
\thispagestyle{empty}

\noindent
\begin{abstract}
In this paper, we solve the direct monodromy problem for the
Borel-summed $0$-parameter solution of Painlev\'e~I.
Using the auxiliary time equation, we first construct formal
WKB wave matrices whose amplitudes are regular at the double
turning point.
Rigorous Borel summation then yields exact WKB matrices in
each Stokes region.
Uniform estimates on moving action strips provide a common
exact wave matrix in a full neighbourhood of the double
point, proving that all four incident jumps are identities
for consistent column normalisations.
We formulate the resulting normalised Riemann--Hilbert
problem and prove its existence and uniqueness on the
Borel domains of the construction.
Finally, after Painlev\'e scaling, the five Borel-summed
solutions coincide with the five rotated tritronqu\'ee
solutions of Painlev\'e~I.
\end{abstract}

{\footnotesize

}

{\small
\setcounter{tocdepth}{2}
\tableofcontents
 
}

\setcounter{section}{0}
\section{Introduction}
The first Painlevé equation 
\begin{equation} \label{PI}
\hbar^2\partial_t^2 q=6q^2+t
\end{equation}
occupies a particularly distinguished position among the six Painlevé equations. Despite its deceptively simple form, it provides an exceptionally explicit setting in which several central themes of modern asymptotic analysis and mathematical physics meet, including divergent asymptotic expansions, isomonodromic deformation theory, Stokes phenomena, exact WKB analysis, and Riemann--Hilbert (RH) problems; see, for instance, \cite{delabaere:hal-01067086}. Apart from the scaling parameter $\hbar$, Painlev\'e~I contains no fixed equation parameters. Its general solution nevertheless depends on two integration constants, represented by the nontrivial Stokes data of the associated linear system. This makes Painlev\'e~I a natural setting for studying how asymptotic conditions select particular monodromy data. Of particular importance are the distinguished \emph{tritronqu\'ee} solutions appearing in Boutroux's analysis of truncated solutions \cite{Boutroux1913}. These solutions are characterised by prescribed algebraic asymptotics in sectors of opening $\frac{8\pi}{5}$ and are
pole-free for sufficiently large $|t|$ in every closed subsector. Their asymptotic properties, including the precise description of their large-$t$ behaviour and the distribution of their poles, were subsequently investigated by a number of authors; see, in particular, the work of Joshi and Kitaev \cite{JoshiKitaev2001}.

\medskip

From the isomonodromic point of view, Painlev\'e I arises as the compatibility condition of a rank-two linear differential system with an irregular singularity at infinity. The associated Stokes phenomenon provides the monodromy data that remain invariant under the nonlinear deformation. In this framework, the distinguished Painlev\'{e}~I solutions are singled out by special choices of Stokes data, and the connection between their asymptotic behaviour in different sectors can be expressed in terms of the corresponding Stokes multipliers. The general isomonodromic formalism is developed in \cite{JimboMiwaUeno}, while the monodromy and connection formulae for Painlev\'{e}~I were described explicitly by Kapaev \cite{Kapaev_2004}. A systematic RH formulation can be found in \cite{FIKN2006}. More recently, Painlev\'e I has also become a fundamental model for the study of nonlinear asymptotics, resurgence, and exact WKB analysis. Its formal solutions are generically divergent but exhibit factorial growth of their coefficients, reflecting the presence of neighbouring transseries sectors and the associated Stokes phenomenon. The Painlev\'e I equation therefore provides a particularly tractable laboratory in which one can study, in a concrete nonlinear setting, the relation between formal asymptotic expansions, Borel singularities, Stokes automorphisms, and isomonodromic monodromy data. These descriptions encode related aspects of the same isomonodromic problem, but the relation between nonlinear Borel summation and the Stokes data of the spatial linear
system requires a separate analysis. Establishing this relation for the perturbative solution of Painlev\'e~I is the purpose of the present paper. Building on the analysis of the resurgent structure of the
$0$-parameter solutions in \cite{250408130737}, we determine
the monodromy data directly
from the Borel-summed perturbative solution of \eqref{PI}. For a chosen branch $\tau$ satisfying $t=-6\tau^2$, the equation has a uniquely determined formal power-series
solution with leading term $\tau$. We start from its directional Borel sums constructed in \cite{250408130737} and insert them into the Painlev\'e I Lax pair. The problem is then to construct exact WKB wave matrices, determine their connections across the Stokes graph, and identify their normalisation at spatial infinity.

\medskip

The main geometric difficulty is visible in the leading
spatial characteristic equation
\begin{equation*}
    y^2=4(x-\tau)^2(x+2\tau).
\end{equation*}
As described in \autoref{sec:stokes-geometry}, its Stokes
graph has a trivalent vertex at the simple turning point
$x=-2\tau$ and a quadrivalent vertex at the double turning
point $x=\tau$.
The normalised spectral cover is described by $w^2=x+2\tau$.
Although the spatial eigenvalues coincide at $x=\tau$,
the leading auxiliary time matrix has eigenvalues $\pm w$,
which remain distinct there when $\tau\neq0$.
We use this auxiliary equation to construct the formal wave
matrix without dividing by the vanishing spatial eigenvalue.
The formal construction is stated in
\autoref{241113190220} and \autoref{250716102028}.
A time Riccati equation determines each column ratio.
Flatness implies compatibility with the spatial equation,
and a closed one-form reconstructs the scalar amplitude
with its prescribed normalisation.
For normalisation series of factorial type, Nagumo estimates
\cite{Nagumo1941} give the factorial bounds for the amplitudes
in \autoref{250708195116}.
An independent order-by-order matrix construction is retained
in \autoref{app:matrix-formal-construction}, together with
its extension under stated hypotheses to several deformation
variables and higher-rank systems.

\medskip

The passage from formal regularity to a statement about exact
jumps requires a separate analytic argument.
Borel summation connects divergent formal expansions with
analytic solutions; its relation to nonlinear Stokes phenomena
is developed, for example, in \cite{Costin1998}.
Resurgent methods for semiclassical equations are treated in
\cite{DelabaerePham1999}, while the Borel-plane approach to
singularly perturbed differential equations is developed in
\cite{nikolaev_riccati_2023,nikolaev_2022_3}.
A broader geometric perspective on resurgence and
wall-crossing is presented in \cite{Kontsevich-Soibelman}.
Here we prove a uniform strip-summation theorem
(\autoref{lem:uniform-strip-summation}) and its application
to component ratios in an eigenbasis
(\autoref{lem:projective-strip}).
These results are formulated for parameter-dependent
differential equations under explicit analyticity, spectral
and Gevrey assumptions, independently of Painlev\'e~I.
Their conclusions include uniform Borel continuation,
exponential bounds and factorial remainder estimates
throughout an unbounded strip, locally uniformly in the
remaining parameters.
This uniformity is essential when applying the construction
to families of spatial and moving action strips.
The lemmas therefore provide analytic tools for other
isomonodromic systems, provided that the required action
domains and coefficient estimates can be established.

\medskip

More precisely, we first prove uniform Borel summation on full strips in an action coordinate. Spatial strips give the exact projective solutions away from the Stokes graph, and amplitude reconstruction then yields the exact WKB matrices of \autoref{thm:exact-WKB-existence}. Near the quadrivalent vertex, we instead vary $x$ and $t$ together while keeping a local Morse coordinate fixed. The resulting moving strips give common Borel sums throughout a full neighbourhood of the vertex. Consequently, \autoref{thm:no-jump-quadrivalent} identifies the four adjacent WKB matrices when their formal data and normalisations agree. We next compare the WKB columns with the canonical recessive solution lines at spatial infinity. This comparison determines which diagonal right normalisations put the matrices in the standard RH convention. In the representative sector
\begin{equation*}
    \alpha=0,\qquad -\frac{2\pi}{5}<\arg\tau<0,
\end{equation*}
the resulting canonical Stokes multipliers are
\begin{equation*}
    s_{-1}=s_0=0,\qquad s_{-2}=s_1=s_2=i.
\end{equation*}
The vanishing multipliers follow from the common wave matrix near the double turning point. The remaining values follow from the cyclic relations and the normalisation at infinity. The parameter appearing before this final normalisation records the relative scale of the two columns.

\medskip

We formulate the corresponding WKB RH problem and prove its existence and uniqueness in \autoref{thm:RH-WKB-existence-uniqueness}. An explicit constant gauge transformation and the scaling 
\begin{equation*}
    X=\hbar^{-\frac{2}{5}}x,\qquad
    T=\hbar^{-\frac{4}{5}}t,\qquad
    q(t,\hbar)=\hbar^{\frac{2}{5}}u(T)
\end{equation*}
reduce the problem to the classical tritronqu\'ee RHP.
Uniqueness then gives the identification in
\autoref{thm:canonical-RHP}.
This identification holds first on the pole-free domains
supplied by the Borel construction and subsequently extends
by meromorphic continuation.
A related direct monodromy problem was studied by Iwaki
\cite{Iwaki-P1} for a two-parameter construction based on
topological recursion on elliptic spectral curves.
His calculation of the Stokes multipliers assumes conjectures
on Borel summability and connection formulas.
The present argument treats the perturbative power-series
branch and its nodal spectral curve, using the auxiliary
deformation equation to establish the analytic matching at
the double turning point.

\medskip

The paper is organised as follows.
\autoref{250408124523} introduces the equation, domains and
scalar reduction.
\autoref{250408124547} states the formal results, proved in
\autoref{250109110415} using the conventions of
\autoref{250410111125}.
\autoref{sec:exact-WKB-wave-matrix} constructs the exact WKB
matrices in the Stokes regions, and \autoref{Section5}
establishes regular summation at the double turning point.
\autoref{sec:WKB-RHP} determines the normalisation and
monodromy data and identifies the resulting RH problem.
The analytic proofs are given in
\autoref{app:auxiliary-Riccati-estimates} and
\autoref{app:uniform-Morse-estimates}.
Finally, \autoref{app:matrix-formal-construction} presents
the independent matrix construction and its extension,
under stated hypotheses, to several deformation variables
and higher-rank systems.

\paragraph{Acknowledgements.}
{\sloppy M. A. was partially funded by the SNSF under the grant agreement PZ00P2 223297 and by the DFG -- project ID 551478549. O. M. was supported by the fundamental junior IUF grant G752IUFMAR. N. N. was supported by the European Union's Horizon 2020 Research and Innovation Programme under the Marie Skłodowska-Curie Grant Agreement No.\,101026083 (\href{https://cordis.europa.eu/project/id/101026083}{\texttt{AbQuantumSpec}}), as well as the Leverhulme Trust through the Leverhulme Research Project grant \textit{Extended Riemann-Hilbert Correspondence, Quantum Curves and Mirror Symmetry}.\par}

\paragraph{Statements.}The authors declare no conflict of interest. AI was used to improve the presentation in accordance with the Leiden Declaration on Artificial Intelligence and Mathematics.

\section{The \texorpdfstring{$\hbar$}{hbar}-deformed Painlev\'e I System}
\label{250408124523}

\subsection{Domains and asymptotic conventions}
\label{250419183135}

In this article, the variables $x$, $t$ and $\hbar$ denote respectively the spatial coordinate, time and formal deformation parameter. All branches are specified on open domains and uniform estimates on a time sector mean estimates on its
relatively compact subsets, unless a larger domain is explicitly stated.
For a direction $\alpha\in\mathbb R/(2\pi\mathbb Z)$
and a constant $R>0$, we define
\begin{equation}
\label{defArc}
 \operatorname{\mathsf{Arc}}_\pi(\alpha)
       =\left(\alpha-\frac{\pi}{2},\alpha+\frac{\pi}{2}\right),\qquad
 S_{\alpha,R}=\left\{\hbar\neq0 \,\text{ such that }\, \operatorname{Re}\left(\frac{e^{i\alpha}}{\hbar}\right)>R^{-1}\right\}.
\end{equation}
The latter is the open disk of centre $\frac{R}{2}e^{i\alpha}$ and radius
$\frac{R}{2}$, tangent to the origin, which we call a Borel disk of radius $\frac{R}{2}$ in direction $\alpha$. Throughout the paper, scalar quantities are measured by their absolute value. For vectors and matrices, we use the maximum norm and its induced operator norm:
\begin{equation*}
    \|v\|=\max_{1\leq i\leq N}|v_i|,
    \qquad
    \|M\|=\max_{1\leq i\leq r}
                 \sum_{j=1}^{s}|M_{ij}|,
    \qquad
    v\in\mathbb C^N,\quad M\in\mathbb C^{r\times s}.
\end{equation*}
These norms satisfy the triangle inequality and absolute
homogeneity. Moreover, for all compatible matrix and vector
products,
\begin{equation*}
    \|MN\|\leq\|M\|\,\|N\|,
    \qquad
    \|Mv\|\leq\|M\|\,\|v\|,
    \qquad
    \|I_N\|=1.
\end{equation*}
For a square matrix, the projections onto its diagonal and
off-diagonal parts are contractions:
\begin{equation*}
    \|\operatorname{diag}M\|\leq\|M\|,
    \qquad
    \|\operatorname{off}M\|\leq\|M\|.
\end{equation*}
For a function $F$ on a parameter domain and a compact subset $K$,
we write
\begin{equation*}
    \|F\|_K=\sup_{z\in K}\|F(z)\|.
\end{equation*}
All finite-dimensional norms are equivalent. Consequently, the factorial-type and uniform asymptotic properties defined below are independent of this choice of norm, although their bounding constants may change. 

\medskip

A formal series $\hat f(\hbar)=\underset{n\geq 0}{\sum}f_n\hbar^n$ is locally uniformly Gevrey-$1$, or of factorial type, on an open set $U$ if its coefficients satisfy
$\|f_n\|_K\leq C_K (A_K)^n n!$ on every compact set $K$ of $U$ (we use the standard notation $K\Subset U$ for compact sets of $U$) for some constants $(C_K,A_K)\in (\mathbb{R}_+)^2$. A holomorphic function $f$ has $\hat f$ as its locally uniform Gevrey-$1$ asymptotic expansion on Borel disks if, for every compact subset $K\Subset U$, there exist constants $C_K,A_K,R_K>0$ such that
\begin{equation}
\label{eq:uniform-Gevrey-convention}
\left\|f(\hbar)-\sum_{n=0}^{N-1}f_n\hbar^n\right\|_K
       \leq C_K (A_K)^N N!|\hbar|^N
 \quad \, \forall \, N\geq0, \, \hbar\in S_{\alpha,R_K}.
\end{equation}
Here and below $R_K$ may decrease when $K$ increases. This uniform condition on a Borel disk, rather than estimates only on smaller angular sectors, is the condition used for uniqueness of Borel resummation. Indeed, the Nevanlinna-Sokal theorem identifies such germs with directional Borel sums (see \cite{Sokal1980} and \cite[Appendix~A.2, Theorem~2]{nikolaev_2022_3}). In particular, our Borel-Laplace convention is specified in \autoref{sec:sectorial-Borel-summation}. 

\subsection{The perturbative Painlev\'e~I solution}
\label{250408124753}

We study the Painlev\'e~I Hamiltonian system
\begin{equation}
\label{250408131147}
 \hbar^2\partial_t^2q=6q^2+t,\qquad p=\hbar\partial_tq.
\end{equation}
We set $t=-6\tau^2$ and choose a nonzero branch
$\tau=\tau(t)$. In particular, $\partial_t\tau=-\frac{1}{12\tau}$.

\begin{proposition}[Formal perturbative solution of the Painlev\'{e}~I equation]
\label{250408135328}
\eqref{250408131147} admits a unique formal solution
\begin{equation}
\label{250205091941}
 \hat q=\sum_{n\geq0}q_n(t)\hbar^n,\qquad q_0=\tau,
 \qquad \hat p=\hbar\partial_t\hat q.
\end{equation}
Its coefficients are holomorphic away from $\tau=0$ and satisfy
\begin{equation}
\label{eq:q-formal-recursion}
 q_n=\frac{1}{12\tau}
  \left(\partial_t^2q_{n-2}
             -6\sum_{j=1}^{n-1}q_jq_{n-j}\right),\qquad \forall\, n\geq1,
 \qquad q_{-1}=q_{-2}=0.
\end{equation}
All odd coefficients vanish, and
\begin{equation}\label{ccoeff}
    q_{2n}=c_n\tau^{1-5n},\qquad n\geq0,
\end{equation}
where the constants $(c_n)_{n\geq 0}$ are uniquely determined by the recursion
\begin{equation}
\label{eq:c-formal-recursion}
    c_0=1,\qquad
    c_n=
    \frac{(5n-6)(5n-4)}{1728}\,c_{n-1}
    -\frac12\sum_{j=1}^{n-1}c_jc_{n-j},
    \qquad \forall\, n\geq1.
\end{equation}
In particular, $c_1=-\frac{1}{1728}$, hence $q_2=-\frac{1}{1728\tau^4}$. Moreover,
\begin{equation}
\label{eq:c-factorial-bound}
    |c_n|\leq (n!)^2,\qquad \forall\, n\geq0.
\end{equation}
Consequently, $\hat q$ is locally uniformly of factorial type: for every compact subset $K$ of the domain of the chosen nonzero branch $\tau(t)$, there exist constants $C_K,A_K>0$ such that
\begin{equation}
\label{eq:q-factorial-bound}
    \|q_n\|_K\leq C_K (A_K)^n n!,
    \qquad \forall\, n\geq0.
\end{equation}
\end{proposition}

\begin{proof}
Inserting the formal series into \eqref{250408131147} gives \eqref{eq:q-formal-recursion}. Induction gives the parity and homogeneity. Details can be found in \cite{250408130737}. To prove \eqref{eq:c-factorial-bound}, proceed by induction.
The estimate holds for $c_0=1$. Assuming it for all smaller
indices, \eqref{eq:c-formal-recursion} together with $j!(n-j)!\leq (n-1)!$ gives
\begin{equation*}
    \frac{|c_n|}{(n!)^2}
       \leq
       \frac{|(5n-6)(5n-4)|}{1728n^2}
       +\frac12\sum_{j=1}^{n-1}\binom{n}{j}^{-2}
       \leq
       \frac{25}{1728}+\frac{n-1}{2n^2}
       \leq \frac{25}{1728}+\frac18<1.
\end{equation*}
Here the sum is empty for $n=1$, and $\binom{n}{j}\geq n$ for $1\leq j\leq n-1$. For a nonempty compact set $K$ as in the statement, put $\delta_K=\underset{K}{\min}|\tau|>0$ and $T_K=\underset{K}{\max}|\tau|$. The homogeneity formula \eqref{ccoeff} and $(n!)^2\leq(2n)!$ imply
\begin{equation*}
    \|q_{2n}\|_K
       \leq T_K(\delta_K)^{-5n}(n!)^2
       \leq T_K
               \bigl((\delta_K)^{-\frac{5}{2}}\bigr)^{2n}(2n)!.
\end{equation*} 
Since all odd coefficients vanish, this proves \eqref{eq:q-factorial-bound} with
$C_K=T_K$ and $A_K=(\delta_K)^{-\frac{5}{2}}$.
\end{proof}

\begin{definition}[Time sectors]
\label{250224171850}
For a given direction $\alpha$, there are five time sectors on the $\tau$-plane for Borel-resummation of $\hat{q}$ (see \cite{250408130737}). They are the connected components of the complement of the Stokes rays $5\arg\tau-2\alpha\in2\pi\mathbb Z$. Thus, we define the time sectors in the $\tau$-plane as
\begin{equation}
 \widetilde V_k=\left\{\frac{2(k\pi+\alpha)}5<\arg\tau<\frac{2((k+1)\pi+\alpha)}5\right\}, \qquad k\in\mathbb Z/5\mathbb Z.
\end{equation}
Their images $V_k$ under $t=-6\tau^2$ are sectors of opening $\frac{4\pi}{5}$, each equipped with the inverse branch $\tau(t)$. These images overlap: they are not the connected components of a partition of the punctured $t$-plane.
\end{definition}

\subsection{Borel-resummed solutions of the Painlev\'{e}~I equation}
We state in this section the Borel resummation result of \cite{250408130737}, adapted to any direction and justify this adaptation below. We have 
\begin{proposition}[Borel-resummed solutions of the Painlev\'{e}~I equation]
\label{250408141643}
Fix a direction $\alpha$ and a sector $\widetilde V$ in the $\tau$-plane as defined in \autoref{250224171850}. Let $V$ be its image under $t=-6\tau^2$. The formal series $\hat q$ selected by this branch in \autoref{250408135328} is locally uniformly Borel summable in direction $\alpha$ on $V$. More precisely, for every relatively compact open subset $V_0\Subset V$, there exists $R_{V_0}>0$ such that the Borel sums $q$ and $p$ are jointly holomorphic on $V_0\times S_{\alpha,R_{V_0}}$. They satisfy the Painlev\'e~I system \eqref{250408131147} and admit $\hat q$ and $\hat p$, respectively, as their uniform Gevrey-1 asymptotic expansions on this product, in the sense of
\eqref{eq:uniform-Gevrey-convention}. The constants in these estimates, as well as $R_{V_0}$, may depend on $V_0$. These properties determine the germs of $(q,p)$ uniquely: the sums constructed on two such domains agree on their intersection after restriction to a common smaller Borel disk.
\end{proposition}

For $\alpha=0$, this follows from
\cite[Corollary~2.15]{250408130737}.
For a general direction, choose a real representative $\alpha\in\mathbb R$ and set
\begin{equation*}
    \tilde\tau=e^{-\frac{2i\alpha}{5}}\tau,\qquad
    \tilde t=e^{-\frac{4i\alpha}{5}}t,\qquad
    \tilde\hbar=e^{-i\alpha}\hbar.
\end{equation*}
The homogeneity of the formal coefficients gives
\begin{equation*}
    \hat q(t,\hbar)
      =e^{\frac{2i\alpha}{5}}\hat q(\tilde t,\tilde\hbar),
\end{equation*}
where the formal series on the right is taken with the
corresponding branch $\tilde\tau$.
The substitution $\xi=e^{i\alpha}\eta$ in the Laplace integral
therefore transports Borel summability in direction $0$ to
Borel summability in direction $\alpha$.
The resulting sectors are precisely those of
\autoref{250224171850}. Differentiation in $t$ is justified by Cauchy estimates on a slightly larger compact subset of $V$. \autoref{app:uniform-Morse-estimates} also constructs the required sums on unbounded action strips; their agreement with this proposition uses only uniqueness of the Borel sum. We call these the five \emph{Borel-resummed solutions} of the Painlev\'{e}~I equation. Their identification with the classical tritronqu\'ees is a conclusion of \autoref{sec:WKB-RHP}.

\subsection{The flat Painlev\'{e}~I Lax pair}
\label{250408124810}

For any holomorphic solution $q$ of \eqref{250408131147}, set
\begin{equation}
\label{241115131339}
 L(x,t)=\begin{pmatrix}p&4(x-q)\\
      x^2+xq+q^2+\frac{t}{2}&-p\end{pmatrix},\qquad
 A(x,t)=\begin{pmatrix}0&2\\\frac{x}{2}+q&0\end{pmatrix}.
\end{equation}
The $\hbar$-deformed Painlev\'e I system is
\begin{equation}
\label{250408155610}
 \hbar\partial_x \mathbf{v}=L\mathbf{v},\qquad \hbar\partial_t \mathbf{v}=A\mathbf{v}.
\end{equation}
or equivalently
\begin{equation}
\label{250408155610_bis}
 \hbar\partial_x \Psi=L\Psi,\qquad \hbar\partial_t \Psi=A\Psi.
\end{equation}
We use $\mathbf{v}$ for a column solution and $\Psi$ for a fundamental matrix, whose two columns are independent solutions. The relation with isomonodromic deformations is the usual one \cite{JimboMiwaUeno} and we recall that the $\hbar$-deformed version is obtained through a rescaling of all quantities \cite{MarchalOrantinAlameddine2022,MarchalAlameddineP1Hierarchy2023}.

\begin{lemma}[Flatness and compatibility condition]
\label{250408155701}
The compatibility condition of the Painlev\'{e}~I Lax pair is
\begin{equation}
\label{250408171411}
 \hbar\partial_tL-\hbar\partial_xA+[L,A]=0.
\end{equation}
For matrices \eqref{241115131339} with arbitrary $(q,p)$, it is equivalent to the Painlev\'{e}~I system \eqref{250408131147}.
\end{lemma}

\begin{proof}The proof is standard and straightforward.\end{proof}

The Lax matrices are polynomial in $x$. Thus the classical turning points introduced below are ordinary points of the exact system since only $\{\infty\}$ is an irregular point of the connection. Also, $\det\Psi$ is independent of $(x,t)$, since both matrices are traceless. For fixed nonzero $\hbar$, every local fundamental matrix $\Psi$ continues to an entire fundamental matrix in $x$ as long as $(q,p)$ are finite.

\subsection{The scalar quantum curve}
\label{250717105422}

On a domain avoiding $x=q(t,\hbar)$, the first component $\psi$ of $\mathbf{v}$ determines the second component by
\begin{equation}
\label{eq:scalar-reconstruction}
 \mathbf{v}=\begin{pmatrix}\psi\\
        \frac{\hbar\partial_x\psi-p\psi}{4(x-q)}\end{pmatrix}.
\end{equation}
Eliminating that component gives the quantum curve
\begin{equation}
\label{250408193221}
 \mathcal P\psi=\hbar^2\partial_x^2\psi
   -\frac{\hbar^2}{x-q}\partial_x\psi
   -\left(4x^3+2tx+p^2-4q^3-2tq
                       -\frac{\hbar p}{x-q}\right)\psi=0.
\end{equation}
To correspond to the \emph{full} Lax pair, $\psi$ must also satisfy
\begin{equation}
\label{eq:scalar-time-equation}
 \hbar\partial_t\psi
       =\frac{\hbar\partial_x\psi-p\psi}{2(x-q)}.
\end{equation}
Conversely, suppose that $\psi$ satisfies both the scalar spatial
equation \eqref{250408193221} and the scalar time equation
\eqref{eq:scalar-time-equation}, and assume the flatness identity
\eqref{250408171411}. On a domain where $x\neq q(t,\hbar)$, set
\begin{equation*}  
\chi\coloneqq
        \frac{\hbar\partial_x\psi-p\psi}{4(x-q)},
    \qquad
    \mathbf{v}\coloneqq
        \begin{pmatrix}\psi\\ \chi\end{pmatrix}.
\end{equation*}
We claim that $\mathbf{v}$ satisfies both equations in \eqref{250408155610}. By definition of $\chi$, $\hbar\partial_x\psi=p\psi+4(x-q)\chi$ so the first component of $\hbar\partial_x\mathbf{v}=L\mathbf{v}$ is verified. Since $q$ and $p$ are independent of $x$, differentiation with respect to $x$ gives
\begin{equation*}
   \hbar\partial_x\chi
      -\left(x^2+xq+q^2+\frac{t}{2}\right)\psi
      +p\chi
    =
    \frac{\mathcal P\psi}{4(x-q)}
    =0. 
\end{equation*}
Thus $\hbar\partial_x\mathbf{v}=L\mathbf{v}$. The scalar time equation \eqref{eq:scalar-time-equation} gives $\hbar\partial_t\psi=2\chi$, which is the first component of $\hbar\partial_t\mathbf{v}=A\mathbf{v}$. To prove the second component, define the residual $\mathbf{r}\coloneqq\hbar\partial_t\mathbf{v}-A\mathbf{v}$. Using the spatial equation already proved and the flatness
identity \eqref{250408171411}, we obtain
\begin{equation*}
    \hbar\partial_x\mathbf{r}=\hbar(\partial_xA) \mathbf{v} -[L,A]\mathbf{v}+ \hbar L \partial_t \mathbf{v}-\hbar (\partial_x A)\mathbf{v}-AL\mathbf{v}= L(-A\mathbf{v} +\hbar \partial_t \mathbf{v})=L\mathbf{r}.
\end{equation*}
The first component of $\mathbf{r}$ vanishes, so we may write $\mathbf{r}=(0,\rho)^{\mathsf T}$. The first component of this last differential equation is therefore $0=4(x-q)\rho$. Since $x\neq q$, we conclude that $\rho=0$. Hence
$\hbar\partial_t\mathbf{v}=A\mathbf{v}$, completing the reconstruction of the full Lax pair.
\medskip

We call scalar solutions satisfying both \eqref{250408193221} and \eqref{eq:scalar-time-equation} \emph{time-compatible}. The spatial equation \eqref{250408193221} alone allows arbitrary $t$-dependent linear combinations of a local basis of its
solutions, because it differentiates only with respect to $x$. The scalar time equation imposes the additional compatibility in $t$. The poles at $x=q$ belong to this scalar representation. They do not occur in \eqref{241115131339}: our construction at the double turning point will therefore use the matrix system directly.

\section{The formal WKB wave matrix}
\label{250408124547}

This section states the formal results and the construction behind them. Their proofs are collected in \autoref{250109110415}. The key choice is to construct the projective ratio from the time equation, whose leading eigenvalues remain distinct at the double turning point.

\subsection{The formal system and its spectral cover}
\label{250408124856}
\label{250408211930}

Substitution of $(\hat q,\hat p)$ from \eqref{250205091941} into \eqref{241115131339} defines $\hat L$ and $\hat A$. They obey the formal zero-curvature identity and the formal system is
\begin{equation}
\label{250408174515}
 \hbar\partial_x\hat{\mathbf{v}}=\hat L\hat{\mathbf{v}},\qquad
 \hbar\partial_t\hat{\mathbf{v}}=\hat A\hat{\mathbf{v}}.
\end{equation}
Formal expressions and differentiation are understood coefficientwise, as specified in \autoref{250410111125}. The leading matrices satisfy
\begin{equation}
 L_0=2(x-\tau)A_0,\qquad
 A_0=\begin{pmatrix}0&2\\\frac{x}{2}+\tau&0\end{pmatrix}.
\end{equation}
Their eigenvalues are respectively $\pm y$ and $\pm w$, with
\begin{equation}
\label{250706182441}
 w^2=x+2\tau,\qquad y=2(x-\tau)w.
\end{equation}
For fixed $\tau\neq0$, the spatial characteristic cubic $y^2=4(x-\tau)^2(x+2\tau)$ is nodal at $(x,y)=(\tau,0)$. At $\tau=0$ it is cuspidal ($y^2=4x^3$). Its normalisation is $w^2=x+2\tau$, with the map given by \eqref{250706182441}. The map is finite and generically one-to-one, and the source is smooth. The two branches over the node are $w=\pm\sqrt{3\tau}$. The node and cusp assertions follow from the quadratic terms and, at $\tau=0$, the equation $y^2=4x^3$. The common normalised spectral surface and the domain used below are
\begin{equation}
\label{250409162453}
 \mathsf\Sigma\coloneqq\{(x,\tau,y,w):w^2=x+2\tau,\ y=2(x-\tau)w\},
 \qquad
 \mathring{\mathsf\Sigma}\coloneqq\{\tau w\neq0\}\subset\mathsf\Sigma.
\end{equation}
The surface is smooth with global coordinates $(\tau,w)$, $x=w^2-2\tau$ and $y=2(x-\tau)w$. Its projection to $(x,\tau)$ is ramified only at $w=0$. In particular, we retain both points over $x=\tau$.
We use $\mathrm d$ for exterior differentiation on the spectral
surface $\mathsf\Sigma$. Its global coordinates are $(\tau,w)$,
while $x$ and $t$ are the functions $x=w^2-2\tau$ and $t=-6\tau^2$. Consequently,
\begin{equation*}
 \mathrm dx=2w\,\mathrm dw-2\,\mathrm d\tau,\qquad \mathrm dt=-12\tau\,\mathrm d\tau.   
\end{equation*}
The eigenvalue one-form $\lambda$  has an explicit global primitive $Z$ on $\mathsf\Sigma$:
\begin{equation}
\label{250716160458lambda}
    \lambda\coloneqq y\,\mathrm dx+w\,\mathrm dt
        =\mathrm dZ,
    \qquad
    Z\coloneqq\frac45(x-3\tau)w^3
        =\frac45w^5-4\tau w^3.
\end{equation}
Indeed, since $y=2(w^2-3\tau)w$, differentiation in the independent coordinates $(\tau,w)$ gives
\begin{align*}
    \mathrm dZ
      &=4(w^2-3\tau)w^2\,\mathrm dw-4w^3\,\mathrm d\tau=2(w^2-3\tau)w
          \bigl(2w\,\mathrm dw-2\,\mathrm d\tau\bigr)
          -12\tau w\,\mathrm d\tau\\
      &=y\,\mathrm dx+w\,\mathrm dt.
\end{align*}
The primitive is global on $\mathsf\Sigma$ because $Z=\frac45w^5-4\tau w^3$ is a polynomial in its global coordinates $(\tau,w)$. On $\mathring{\mathsf\Sigma}$, where $\tau w\neq0$, the pair $(x,t)$ also gives local coordinates. After fixing local branches $\tau=\tau(t)$ and $w=w(x,t)$, the preceding identity becomes
\begin{equation*}
    \mathrm dZ=\left.\frac{\partial Z}{\partial x}\right|_t\mathrm dx+\left.\frac{\partial Z}{\partial t}\right|_x\mathrm dt,
    \qquad
    \left.\frac{\partial Z}{\partial x}\right|_t=y,
    \qquad
    \left.\frac{\partial Z}{\partial t}\right|_x=w.
\end{equation*}
In particular, differentiation with respect to $x$ at fixed $t$
keeps $\tau$ fixed but changes $w$, whereas differentiation
with respect to $t$ at fixed $x$ changes both $\tau$ and $w$:
\begin{equation*}
    \left.\frac{\partial\tau}{\partial x}\right|_t=0,
    \qquad
    \frac{\mathrm d\tau}{\mathrm dt}=-\frac{1}{12\tau},
    \qquad
    \left.\frac{\partial w}{\partial x}\right|_t=\frac{1}{2w},
    \qquad
    \left.\frac{\partial w}{\partial t}\right|_x =-\frac{1}{12\tau w}.
\end{equation*}

\subsection{The formal time Riccati equation}
To construct formal WKB columns, we first consider the ratio
of their second component to their first component.
Indeed, any column solution
$\hat{\mathbf v}=(\hat v_1,\hat v_2)^{\mathsf T}$
of \eqref{250408174515} satisfies
\[
    \hbar\partial_t\hat v_1=2\hat v_2,
    \qquad
    \hbar\partial_t\hat v_2
       =\left(\frac{x}{2}+\hat q\right)\hat v_1.
\]
Whenever the quotient $\hat r=\frac{\hat v_2}{\hat v_1}$ is defined,
differentiating at fixed $x$ gives
\begin{equation}\label{diffrt}
    \hbar\partial_t\hat r
      =
      \frac{
        (\hbar\partial_t\hat v_2)\hat v_1
        -\hat v_2(\hbar\partial_t\hat v_1)
      }{\hat v_1^2}
      =\frac{x}{2}+\hat q-2\hat r^2.
\end{equation}

We therefore first construct formal series
\[
    \hat r_\varepsilon
       =\sum_{n=0}^{\infty}
           r_{\varepsilon,n}(\tau,w)\hbar^n,
    \qquad \varepsilon\in\{+1,-1\},
\]
by imposing the time Riccati equation
\begin{equation}
\label{eq:formal-time-Riccati}
    \hbar\partial_t\hat r_\varepsilon
       =\frac{x}{2}+\hat q-2\hat r_\varepsilon^2,
    \qquad
    r_{\varepsilon,0}=\frac{\varepsilon w}{2}.
\end{equation}
The two possible leading terms follow from
$q_0=\tau$ and $w^2=x+2\tau$.
Writing $\hat q=\underset{n\geq 0}{\sum}q_n\hbar^n$, we have
\begin{equation}
\label{eq:formal-r-recursion}
    r_{\varepsilon,n}
      =
      \frac{1}{4r_{\varepsilon,0}}\left(
        q_n-\partial_t r_{\varepsilon,n-1}
        -2\sum_{j=1}^{n-1}
             r_{\varepsilon,j}r_{\varepsilon,n-j}\right)
      ,
    \qquad n\geq1.
\end{equation}
Since $4r_{\varepsilon,0}=2\varepsilon w$ does not vanish
on $\mathring{\mathsf\Sigma}$, this recursion together with \autoref{250408135328} determines a unique formal series with holomorphic coefficients there.
In particular, it remains regular over $x=\tau\neq0$.


\subsection{Existence, normalisation and continuation}
\label{250408124908}
In this section we state the main results regarding existence of formal WKB wave matrix solving the Painlev\'{e} Lax system. 

\begin{theorem}[Formal WKB columns]
\label{241113190220}
Fix a point $b_0=(x_0,t_0)$ and branches $(\tau,w)$ near $b_0$ such that $t_0=-6\tau_0^2$ and $\tau_0w_0\neq0$. For each $\varepsilon\in\{+1,-1\}$ and nonzero transmonomial
\begin{equation}
 \hat c_\varepsilon(\hbar)\coloneqq\hbar^{m_\varepsilon}\hat b_\varepsilon(\hbar)e^{-\frac{\kappa_\varepsilon}{\hbar}},\qquad
 m_\varepsilon \in \mathbb{Z}_{\geq 0},\quad \hat b_\varepsilon\in\mathbb C[[\hbar]]^\times, \kappa_\varepsilon\in\mathbb C
\end{equation}
there is a unique germ at $(\tau_0,w_0)$ of a formal WKB column of \eqref{250408174515} of the form
\begin{equation}
\label{eq:formal-column-construction}
 \hat{\mathbf{v}}_\varepsilon
   =\hbar^{m_\varepsilon}\hat a_\varepsilon(\tau,w;\hbar)
       \binom{1}{\hat r_\varepsilon}e^{\frac{S_\varepsilon}{\hbar}},
 \qquad
 S_\varepsilon=\varepsilon\bigl(Z-Z(\tau_0,w_0)\bigr)-\kappa_\varepsilon,
\end{equation}
where $\hat r_\varepsilon$ is the formal series determined by
\eqref{eq:formal-time-Riccati} and
\eqref{eq:formal-r-recursion}. More explicitly,
\begin{equation*}
    \hat r_\varepsilon(\tau,w;\hbar)
       =\frac{\varepsilon w}{2}
        +\sum_{n=1}^{\infty}
            r_{\varepsilon,n}(\tau,w)\hbar^n,
    \qquad
    \hat a_\varepsilon(\tau,w;\hbar)
       =\sum_{n=0}^{\infty}
            a_{\varepsilon,n}(\tau,w)\hbar^n
\end{equation*}
are formal power series with holomorphic coefficients on a sufficiently small connected neighbourhood $U$ of $(\tau_0,w_0)$ in $\mathring{\mathsf\Sigma}$, and
$a_{\varepsilon,0}$ is nowhere zero on $U$. Its normalisation is
$\hat a_\varepsilon(\tau_0,w_0;\hbar)=\hat b_\varepsilon(\hbar)$. Thus the first component of $\hat{\mathbf v}_\varepsilon$ at $(\tau_0,w_0)$ is exactly $\hat c_\varepsilon(\hbar)$. The phase $S_\varepsilon$ is equivalently characterised on $U$ by $\mathrm dS_\varepsilon=\varepsilon\lambda$ and $S_\varepsilon(\tau_0,w_0)=-\kappa_\varepsilon$.
\end{theorem}

\begin{corollary}[Formal WKB wave matrices]
\label{250716102028}
For the prescribed normalisations $\hat c_+(\hbar)$ and $\hat c_-(\hbar)$,
the two columns $(\hat{\mathbf{v}}_+,\hat{\mathbf{v}}_-)$ of \autoref{241113190220} give the uniquely normalised formal WKB wave matrix
\begin{equation}
\label{250421155248}
    \hat\Psi
      =\left[\hat{\mathbf v}_+\ \hat{\mathbf v}_-\right]
      =\hat E(\tau,w;\hbar)
        \exp\left(
          \frac{1}{\hbar}\operatorname{diag}(S_+,S_-)
        \right),
\end{equation}
where
\begin{equation*}
   \hat E
      =
      \begin{pmatrix}
        \hat a_+ & \hat a_-\\
        \hat a_+\hat r_+ & \hat a_-\hat r_-
      \end{pmatrix}
      \operatorname{diag}
        \left(\hbar^{m_+},\hbar^{m_-}\right)
      =\sum_{n=0}^{\infty}E_n(\tau,w)\hbar^n. 
\end{equation*}
After restricting to a common neighbourhood $U$, all coefficients $E_n$ are holomorphic on $U$. If $m_+=m_-=0$, the matrix $\hat E$ is invertible as a matrix
of formal power series with holomorphic coefficients, and its
leading term is
\begin{equation}
\label{eq:formal-leading-amplitude}
    E_0(\tau,w)
      =\left(\frac{w_0}{w}\right)^{\frac{1}{2}}
        \begin{pmatrix}
          1&1\\
          \frac{w}{2}&-\frac{w}{2}
        \end{pmatrix}
        \operatorname{diag}(b_{+,0},b_{-,0}),
    \qquad
    \det E_0=-w_0b_{+,0}b_{-,0}\neq0.
\end{equation}
Here $b_{\varepsilon,0}$ is the constant coefficient of
$\hat b_\varepsilon$, and, after shrinking $U$ if necessary,
the square root is the holomorphic branch equal to $1$
at $(\tau_0,w_0)$.
For arbitrary $m_\pm\in\mathbb Z_{\geq0}$, the matrix
$\hat E$ is invertible over formal Laurent series with
holomorphic coefficients. If $x_0\neq\tau_0$, then, after possibly shrinking $U$,
the first-row components are the corresponding uniquely
normalised time-compatible formal WKB solutions of the
scalar quantum curve.
\end{corollary}

\begin{proof}The proof of \autoref{241113190220} and \autoref{250716102028} is done in \autoref{250408124924}. It has three steps: solve the time Riccati recursion, use flatness to deduce the spatial Riccati equation, and integrate a closed logarithmic-amplitude form. The construction does not divide by the spatial eigenvalue $y$. Consequently, the formal amplitude coefficients remain
holomorphic at the double turning point $x=\tau\neq0$. An independent proof by matrix coefficients is given in \autoref{app:matrix-formal-construction}. It retains the classical eigenbasis gauge and the separate off-diagonal and diagonal compatibility calculations, for use in more general Lax systems.
\end{proof}

In particular, taking $\hat c_+(\hbar)=\hat c_-(\hbar)=1$ specifies the \emph{well-normalised} matrix at $b_0$. This means that the first component of each column equals $1$ at $b_0$, coefficientwise in $\hbar$. Indeed, this choice gives $m_\pm=0$, $\kappa_\pm=0$, $\hat b_\pm=1$, and hence $S_\pm(b_0)=0$ and $\hat a_\pm(b_0;\hbar)=1$.
Consequently, since $S_\pm(b_0)=0$,
\[
    \hat\Psi(b_0;\hbar)
      =
      \begin{pmatrix}
        1&1\\
        \hat r_+(b_0;\hbar)&\hat r_-(b_0;\hbar)
      \end{pmatrix}
      =G_0(b_0)+O(\hbar),
    \qquad
    G_0(b_0)\coloneqq
      \begin{pmatrix}
        1&1\\
        \frac{w_0}{2}&-\frac{w_0}{2}
      \end{pmatrix}.
\]
Thus the leading matrix at $b_0$ is the common eigenvector
matrix of $L_0(b_0)$ and $A_0(b_0)$, with each eigenvector
normalised to have $1$ as a first component.

\medskip

The local results obtained in \autoref{241113190220} and \autoref{250716102028} can be used to extend the formal quantities as explained in the following proposition.

\begin{proposition}[Formal fundamental solutions and continuation]
\label{250109214318}
\label{250708114601}
\label{250716170516}
Assume $m_+=m_-=0$, so that the prescribed normalisations $\hat c_\pm$ are units of the constant finite-exponential algebra $\mathcal T$ defined in \autoref{250410111125}. Then the columns of $\hat\Psi$ of \autoref{250716102028} form a $\mathcal T$-basis of the formal solutions in the algebra of germs defined there. This includes, but does not require, the well-normalised choice $\hat c_+=\hat c_-=1$.
With the ordered phases fixed, two such WKB matrices differ by a
unique diagonal matrix of units of $\mathbb C[[\hbar]]$, independent
of $(x,t)$. Allowing additive constants in the phases introduces the
corresponding constant exponential factors: interchanging the two
eigenlines interchanges the columns.

Each formal WKB column continues coefficientwise along every path
in $\mathring{\mathsf\Sigma}$, hence to its universal cover. Its
$\hbar$-adic order and the nonvanishing of its leading amplitude
are preserved. The same assertions hold for the wave matrix.
\end{proposition}

\begin{proof}The proof is done in \autoref{AppendixB2}.
\end{proof}

The phase $Z-Z(b_0)$ itself is already single-valued on $\mathring{\mathsf\Sigma}$. The universal cover is needed for the amplitude primitives and their square roots. This distinguishes phase continuation from amplitude continuation. The following proposition gives the explicit reconstruction of the amplitudes and relates the matrix construction to the scalar quantum curve.

\begin{proposition}[Amplitude reconstruction and scalar formulation]
\label{prop:formal-amplitude-reconstruction}
Use the notation and normalisations of \autoref{241113190220}.
For each $\varepsilon\in\{+1,-1\}$, set
\begin{equation}
\label{eq:formal-log-derivatives}
    \hat P_\varepsilon
       \coloneqq\hat p+4(x-\hat q)\hat r_\varepsilon,
    \qquad
    \hat T_\varepsilon\coloneqq2\hat r_\varepsilon.
\end{equation}
Then the ordinary formal one-form
\begin{equation}
\label{eq:formal-amplitude-form}
    \hat\beta_\varepsilon
       \coloneqq\frac{\hat P_\varepsilon-\varepsilon y}{\hbar}
            \,\mathrm dx
        +\frac{\hat T_\varepsilon-\varepsilon w}{\hbar}
            \,\mathrm dt
\end{equation}
has holomorphic coefficients on $\mathring{\mathsf\Sigma}$
and is closed. The amplitude in
\eqref{eq:formal-column-construction} is given by
\begin{equation}
\label{eq:formal-amplitude-primitive}
    \hat a_\varepsilon
       =\hat b_\varepsilon
          \exp\left(
             \int_{b_0}^{(x,t)}\hat\beta_\varepsilon
          \right).
\end{equation}
The integral is taken coefficientwise on a sufficiently small
simply connected neighbourhood of $b_0$, or along a prescribed
path on the universal cover of $\mathring{\mathsf\Sigma}$. At every point where $x\neq\tau$, the series
$\hat P_\varepsilon$ is the unique formal solution with
leading term $\varepsilon y$ of the spatial Riccati equation
\begin{equation}
\label{250714195543}
    \hbar\partial_x\hat P_\varepsilon+\hat P_\varepsilon^2
    -\frac{\hbar\hat P_\varepsilon}{x-\hat q}
    -\left(
        4x^3+2tx+\hat p^2-4\hat q^3-2t\hat q
        -\frac{\hbar\hat p}{x-\hat q}
     \right)=0.
\end{equation}
Moreover, if $m_+=m_-=0$, the first-row components of $\hat\Psi$ form a $\mathcal T$-basis of the germs of time-compatible formal scalar solutions at such a point.
Here time compatibility means satisfying both \eqref{250408193221} and \eqref{eq:scalar-time-equation}, interpreted formally.
\end{proposition}

\begin{proof}
The compatibility identity $\partial_t\hat P_\varepsilon=\partial_x\hat T_\varepsilon$ is proved in \autoref{250408124924}. Together with $P_{\varepsilon,0}=\varepsilon y$, $T_{\varepsilon,0}=\varepsilon w$, and $\mathrm dZ=y\,\mathrm dx+w\,\mathrm dt$, it shows that $\hat\beta_\varepsilon$ is an ordinary
formal closed one-form (See \eqref{eq:formal-amplitude-form_bis} and the discussion above it). The amplitude formula is precisely the construction in \eqref{eq:formal-amplitude-primitive_bis}.

Substituting
$\hbar\partial_x\hat\psi=\hat P\hat\psi$ (i.e. the first-row equation of \eqref{250408174515}, using \eqref{eq:formal-log-derivatives})  into the formal scalar equation \eqref{250408193221} gives \eqref{250714195543}. Away from $x=\tau$, the formal inverse of $x-\hat q$
exists and $y\neq0$. Once the leading term $P_0=\varepsilon y$ is fixed, each higher coefficient is uniquely determined by division by $2\varepsilon y$.

Finally, the scalar reconstruction gives a bijection between time-compatible formal scalar solutions and formal solutions of the Lax system wherever $x-\tau\neq0$.
The scalar basis assertion therefore follows from \autoref{250109214318}.
\end{proof}

\subsection{Factorial type}
\label{250408124938}

\begin{theorem}[Factorial type of the formal amplitudes]
\label{250708195116}
Suppose that the ordinary formal series $\hat b_\varepsilon(\hbar)$ in the prescribed normalisations are of factorial type. Then the scalar amplitudes $\hat a_\varepsilon(\tau,w;\hbar)$ and the ratios $\hat r_\varepsilon(\tau,w;\hbar)$ in \autoref{241113190220} are locally uniformly of factorial type. Consequently, both the reduced column prefactors $\hbar^{-m_\varepsilon}e^{-\frac{S_\varepsilon}{\hbar}}\hat{\mathbf v}_\varepsilon=\hat a_\varepsilon \begin{pmatrix} 1\\ \hat r_\varepsilon\end{pmatrix}$ and the full column prefactors $e^{-\frac{S_\varepsilon}{\hbar}}\hat{\mathbf v}_\varepsilon$ are locally uniformly of factorial type. In particular, the matrix $\hat E(\tau,w;\hbar)$ in \eqref{250421155248}, which includes the prescribed powers $\hbar^{m_\pm}$, also has this property. More explicitly, for every compact subset $K$ of the
universal cover of $\mathring{\mathsf\Sigma}$, there exist
constants $C_K,M_K>0$ such that
\begin{equation}
\label{250708201455}
    \|E_n\|_K
       \coloneqq \sup_{b\in K}\|E_n(b)\|
       \leq C_K M_K^n n!,
    \qquad \forall\, n\geq0.
\end{equation}
The constants may depend on $K$ and on the prescribed
normalisations, but are independent of $n$. In particular, these estimates hold on a full neighbourhood of each point over $x=\tau$ when $\tau\neq0$.
\end{theorem}

\begin{proof}
The proof in \autoref{app:factorial-type-proof} applies Nagumo norms directly to \eqref{eq:formal-r-recursion}, then uses \autoref{prop:formal-amplitude-reconstruction}.    
\end{proof}

 Factorial coefficient bounds alone do not prove summability: the continuation and growth of the Borel transform are established separately in \autoref{sec:exact-WKB-wave-matrix}.

\section{The Exact WKB Wave Matrix}
\label{sec:exact-WKB-wave-matrix}

Fix a phase $\alpha\in\mathbb R/(2\pi\mathbb Z)$ and one of the Stokes sectors
$\widetilde V\subset\mathbb{C}_\tau$ of \autoref{250224171850}.  Let
$V\subset\mathbb{C}_t$ be its projection under $t=-6\tau^2$, and let
$q(t,\hbar)$ be the corresponding Borel-summed perturbative solution supplied
by \autoref{250408141643}.  All uniform asymptotic assertions below are understood on Borel
disks as in \eqref{eq:uniform-Gevrey-convention}.

\subsection{Stokes geometry}
\label{sec:stokes-geometry}

For fixed $\tau\in\widetilde V$, the normalised classical spectral curve is
\begin{equation}
    \mathsf{\Sigma}_\tau
      =\bigl\{(x,w)\in\mathbb{C}^2\mid w^2=x+2\tau\bigr\},
    \qquad
    y=2(x-\tau)w.
\end{equation}
On a chosen local sheet of the projection
$\mathsf\Sigma_\tau\to\mathbb C_x$, write $w=w(x)$,
where $w(x)^2=x+2\tau$. A primitive of the eigenvalue
one-form $y\,\mathrm dx$ is
\begin{equation}
\label{eq:space-action}
    Z_\tau(x)
       =\frac45(x-3\tau)w(x)^3.
\end{equation}
To distinguish this local expression from the action
written in the global coordinate $w$, define
\begin{equation}\label{DefZtilde}
      \widetilde Z_\tau(w)
       \coloneqq\frac45w^5-4\tau w^3.  
\end{equation}
Thus
\[
    Z_\tau(x)=\widetilde Z_\tau(w(x))=\frac45(x-3\tau)(x+2\tau)^{\frac{3}{2}}.
\]
The polynomial $\widetilde Z_\tau$ is single-valued on
the normalised curve and satisfies $\widetilde Z_\tau(0)=0$,
whereas $Z_\tau$ denotes its expression on a chosen local
sheet over the $x$-plane. The finite critical points of the action on
$\mathsf\Sigma_\tau$ are determined by
\[
    \frac{\mathrm d\widetilde Z_\tau}{\mathrm dw}
       =4w^2(w^2-3\tau)=0.
\]
Choose a holomorphic branch of $\tau^{\frac{1}{2}}$ on
$\widetilde V$ and put
\[
    w_\pm(\tau)\coloneqq\pm\sqrt3\,\tau^{\frac{1}{2}}.
\]
Since $\tau\neq0$, there are exactly three distinct
finite critical points:
\begin{equation}
    p_0(\tau)=(-2\tau,0),
    \qquad
    p_\pm(\tau)=\bigl(\tau,w_\pm(\tau)\bigr).
\end{equation}
Their projections to the $x$-plane are
\begin{equation}
\label{eq:spatial-turning-points}
    x_{\mathrm s}(\tau)\coloneqq-2\tau,
    \qquad
    x_{\mathrm d}(\tau)\coloneqq\tau.
\end{equation}
We call these the \emph{simple turning point} and the
\emph{double turning point}, respectively, according
to their multiplicities as zeros of
$y^2=4(x-\tau)^2(x+2\tau)$.
As shown below, they are respectively the trivalent
and quadrivalent vertices of the Stokes graph (See \autoref{fig:final-Stokes-matrices}). The Stokes curves of phase $\alpha$ are the images,
under $x=w^2-2\tau$, of the level-set branches
\begin{equation}
\label{eq:Stokes-curves}
    \operatorname{Im}\!\left(
        e^{-i\alpha}
        \bigl(
            \widetilde Z_\tau(w)
            -\widetilde Z_\tau(w_\ast)
        \bigr)
    \right)=0
\end{equation}
issuing from
$w_\ast\in\{0,w_+(\tau),w_-(\tau)\}$. The two projected critical locations $x_{\mathrm s}(\tau)$ and $x_{\mathrm d}(\tau)$ have different local geometries. At $p_0(\tau)$, the projection $\mathsf\Sigma_\tau\to\mathbb C_x$ is simply ramified,
and
\begin{equation}
    \widetilde Z_\tau(w)
       =-4\tau w^3+O(w^5).
\end{equation}
There are therefore six local level-set branches in
the $w$-plane. The involution $w\mapsto-w$ pairs these
branches and identifies their projections, giving three
Stokes curves issuing from $x=-2\tau$.

At each of $p_\pm(\tau)$, the projection is unramified.
Choose the local sheet with $w(\tau)=w_\pm(\tau)$ and
put $u=x-\tau$. On the chosen sheet,
$Z_\tau(\tau)=\widetilde Z_\tau(w_\pm(\tau))$.
Subtracting this critical value expresses the Stokes
level set through $p_\pm(\tau)$ as the vanishing of
the rotated imaginary part of the action difference. Then we have
\begin{equation}
\label{eq:nodal-action-local}
    Z_\tau(x)
       -\widetilde Z_\tau\bigl(w_\pm(\tau)\bigr)
       =w_\pm(\tau)u^2+O(u^3).
\end{equation}
Thus four local Stokes branches issue from each of
$p_+(\tau)$ and $p_-(\tau)$.
The involution $w\mapsto-w$ exchanges these two families
and leaves their projections unchanged. Hence four
projected Stokes curves meet at $x=\tau$.
This quadrivalent vertex comes from the quadratic
critical points of the action: the projection is
unramified at both points above it.

Let us describe the geometry of the Stokes graph more closely. There are no connections between distinct critical
points in this time sector. Indeed, the polynomial
$\widetilde Z_\tau$ has critical values
\[
    \widetilde Z_\tau(0)=0,
    \qquad
    \widetilde Z_\tau\bigl(w_\pm(\tau)\bigr)
       =-\frac{24}{5}\tau^2w_\pm(\tau).
\]
Their rotated imaginary parts are distinct inside
$\widetilde V$. Each critical trajectory therefore
goes to infinity, by properness of the polynomial.
Distinct trajectories cannot cross at a regular point,
and a compact cycle of a level set is excluded by the
maximum principle for the imaginary part of the action.
The projection identifies the two opposite lifts of
each trajectory; it creates no further crossings.
The projected graph thus consists of disjoint
three-armed and four-armed trees. Its complement has
$1+(3-1)+(4-1)=6$ connected components.
Their labelling, illustrated in
\autoref{fig:final-Stokes-matrices}, is specified
at the end of this subsection.

\medskip

Since the opening of $\widetilde V$ is smaller than $\pi$, the map
\begin{equation*}
    \widetilde V\longrightarrow\mathbb{C}_t,
    \qquad
    \tau\longmapsto t=-6\tau^2,
\end{equation*}
is injective. We recall that its image is denoted by $V
        \coloneqq
        \left\{
            -6\tau^2
            \,\middle|\,
            \tau\in\widetilde V
        \right\}
        \subset\mathbb{C}_t
$
and its holomorphic inverse by $\tau=\tau(t)$, for $t\in V$. For every $\tau\in\widetilde V$, let
\begin{equation*}
    R_{\mathrm A}(\tau),\ldots,R_{\mathrm F}(\tau)
        \subset\mathbb{C}_x
\end{equation*}
be the six connected components of the complement of the Stokes graph
in the $x$-plane. The corresponding time-dependent Stokes regions are
the open subsets
\begin{equation}
\label{eq:time-dependent-Stokes-regions}
    \mathcal U_j
        \coloneqq
        \left\{
            (x,t)\in\mathbb{C}_x\times V
            \,\middle|\,
            x\in R_j\bigl(\tau(t)\bigr)
        \right\},
    \qquad
    \forall\, j\in\{\mathrm A,\ldots,\mathrm F\}.
\end{equation}

The labelling of these regions is fixed as follows.
We denote by $R_{\mathrm D}(\tau)$ the unique region
whose boundary contains both the trivalent vertex
$x_{\mathrm s}(\tau)=-2\tau$ and the quadrivalent vertex
$x_{\mathrm d}(\tau)=\tau$.
The remaining labels are determined by requiring that
the regions incident with the quadrivalent vertex occur
in the counterclockwise cyclic order $R_{\mathrm A}(\tau), R_{\mathrm B}(\tau), R_{\mathrm C}(\tau), R_{\mathrm D}(\tau)$, whereas those incident with the trivalent vertex occur in the counterclockwise cyclic order $R_{\mathrm D}(\tau), R_{\mathrm E}(\tau), R_{\mathrm F}(\tau)$. This convention is illustrated in \autoref{fig:final-Stokes-matrices}. For the fixed direction $\alpha$ and the chosen branch
$\tau(t)$, the restriction $t\in V$ prevents wall
crossings of the spatial Stokes graph.
Indeed, a Stokes connection between the simple and
double turning points occurs when
$\operatorname{Im}\!\left(e^{-i\alpha}\tau^{\frac52}\right)=0$, or equivalently for $5\arg\tau-2\alpha\in2\pi\mathbb Z$. These are precisely the boundary rays of the sectors in \autoref{250224171850}. The corresponding wall crossing, through a saddle
connection between the two vertices, is illustrated
in \cite[Section~3.2.1 and Figure~2]{vanSpaendonckVonk2022}. Inside $\widetilde V$, the vertices and regular arcs
vary smoothly, while the incidence relations, cyclic
orders and attachment to the asymptotic directions
at infinity remain unchanged.
Thus the labels $R_{\mathrm A}(\tau),\ldots,
R_{\mathrm F}(\tau)$ remain consistent throughout
$\widetilde V$, and hence throughout $V$ via the
chosen branch $\tau(t)$.

\subsection{Borel summation on adapted strips}
\label{sec:sectorial-Borel-summation}

The analytic input is a uniform strip version of the Borel-plane
construction in \cite{nikolaev_2022_3}, proved in \autoref{app:auxiliary-Riccati-estimates}. We use spatial trajectories away
from the Stokes graph. At the quadrivalent vertex we instead vary
$x$ and $t$ together, keeping a local Morse coordinate fixed\footnote{Near the quadrivalent vertex, a local Morse coordinate is a holomorphic coordinate $\zeta$, vanishing at $x=\tau$, such that $Z_\tau(x)-Z_\tau(\tau)=\zeta^2$ on the chosen sheet. Keeping $\zeta$ fixed therefore means varying $x$ and $t$ together while preserving this action difference.}. This
choice avoids any assertion that fixed-$x$ auxiliary trajectories
remain inside the selected time sector. 
\medskip

Let us introduce the notation for Borel resummation. For $R,\delta>0$, define
\begin{equation}
\label{eq:Borel-disk-and-tube}
 S_{\alpha,R}
   \coloneqq
   \left\{
     \hbar\in\mathbb C^\ast
     \,\middle|\,
     \operatorname{Re}\left(\frac{e^{i\alpha}}{\hbar}\right)
       >R^{-1}
   \right\},
 \qquad
 \Xi_{\alpha,\delta}
   \coloneqq
   \left\{
     \xi\in\mathbb C
     \,\middle|\,
     \operatorname{dist}
       \bigl(\xi,e^{i\alpha}\mathbb R_{\geq0}\bigr)<\delta
   \right\}.
\end{equation}
Here $\operatorname{dist}$ is the Euclidean distance in
$\mathbb C$; explicitly, $\operatorname{dist} \bigl(\xi,e^{i\alpha}\mathbb R_{\geq0}\bigr)=\underset{s\geq0}{\inf}\left|\xi-se^{i\alpha}\right|$. Thus, $\Xi_{\alpha,\delta}$ is the open tubular neighbourhood of radius $\delta$ around the ray
$e^{i\alpha}\mathbb R_{\geq0}$. 

\begin{figure}[htbp]
    \centering
    \includegraphics[width=0.48\textwidth]{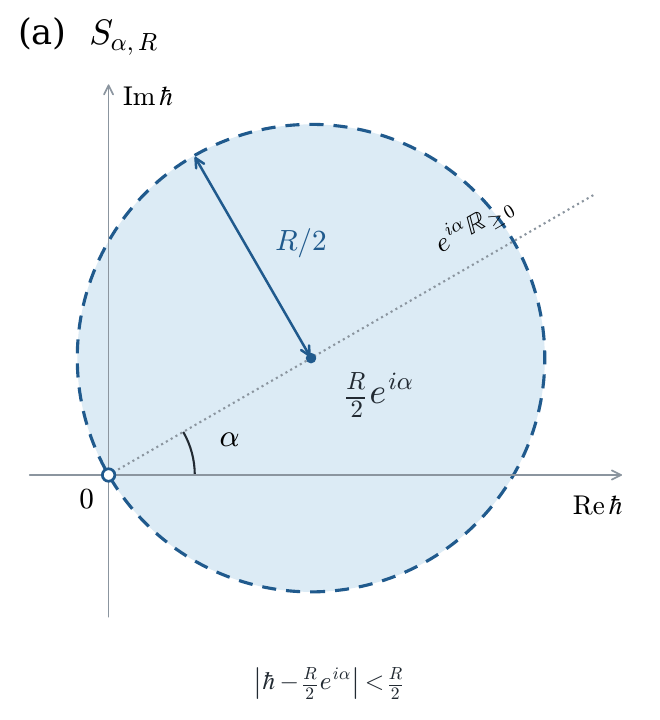}
    \hfill
    \includegraphics[width=0.48\textwidth]{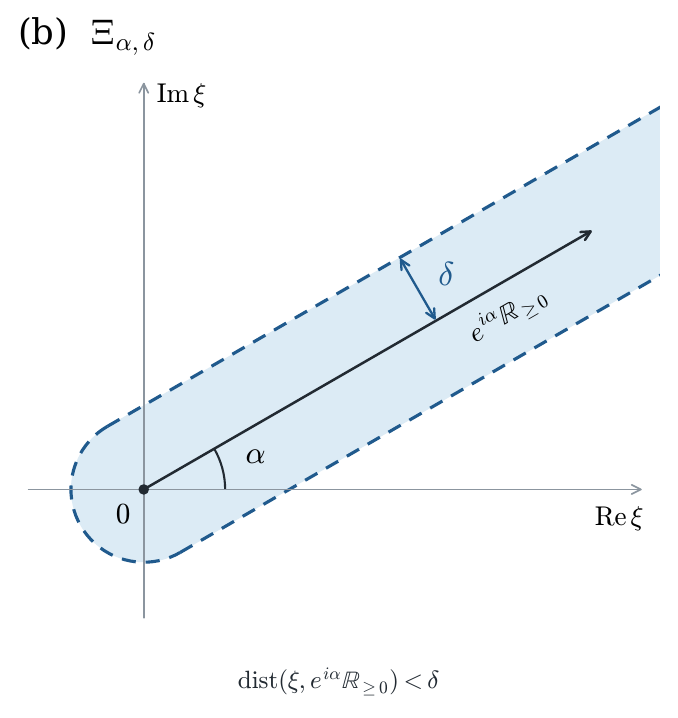}
    \caption{The Borel disk $S_{\alpha,R}$ and the Borel tube
    $\Xi_{\alpha,\delta}$, illustrated for $\alpha=\frac{\pi}{6}$.
    The disk has centre $\frac{R}{2}e^{i\alpha}$ and radius
    $\frac{R}{2}$; its boundary point $0$ is excluded.
    The tube is the open Euclidean neighbourhood of radius
    $\delta$ around the nonnegative ray $e^{i\alpha}\mathbb R_{\geq0}$;
    it includes a neighbourhood of $0$.
    Dashed boundaries are excluded.}
    \label{fig:Borel-disk-and-tube}
\end{figure}

For an ordinary formal series $\hat f(\hbar)=\underset{n\geq 0}{\sum}f_n\hbar^n$, our
Borel and Laplace conventions are
\begin{equation}
\label{eq:parametric-Borel-sum}
 \mathcal B\hat f(\xi)\coloneqq\sum_{n\geq0}\frac{f_{n+1}}{n!}\xi^n,
 \qquad
 s_\alpha(\hat f)(\hbar)\coloneqq f_0+
       \int_{e^{i\alpha}\mathbb R_{\geq0}}
          e^{-\frac{\xi}{\hbar}}\mathcal B\hat f(\xi)\,\mathrm d\xi.
\end{equation}
In the Laplace integral, $\mathcal B\hat f(\xi)$ denotes
the holomorphic continuation of the germ defined by
the Borel series, whenever this continuation exists.

\begin{definition}[Locally uniform Borel summability]
\label{def:locally-uniform-Borel-summability}
Suppose that the coefficients of a formal series
$\hat f(k,\hbar)=\underset{n\geq 0}{\sum}f_n(k)\hbar^n$
depend holomorphically on parameters $k$ in a
complex domain $P$.
We say that $\hat f$ is \emph{locally uniformly
Borel summable in direction $\alpha$ on $P$} if,
for every relatively compact open subset
$P_0\Subset P$, there exist constants
$\delta,C,c>0$ such that the Borel series
$\mathcal B\hat f(k,\xi)$ defined by
\eqref{eq:parametric-Borel-sum} converges locally
uniformly near $\xi=0$ and admits a joint
holomorphic continuation to
$P_0\times\Xi_{\alpha,\delta}$, satisfying
\[
    \sup_{k\in P_0}
    \left\|\mathcal B\hat f(k,\xi)\right\|
       \leq C e^{c|\xi|},
    \qquad
    \forall\,\xi\in\Xi_{\alpha,\delta}.
\]
\end{definition}
For any $R>0$ such that $R^{-1}>c$, the Laplace integral
converges locally normally and defines a holomorphic
function on $P_0\times S_{\alpha,R}$.
Moreover, this function admits $\hat f(\hbar)$ as its uniform
Gevrey-$1$ asymptotic expansion on this product, by the
Nevanlinna--Sokal theorem.
The factorial remainder estimate on the whole Borel
disk uses the tubular continuation, in addition to
convergence of the integral along the summation ray.
\medskip

We use the following result on unbounded strips. For $d>0$, set
\begin{equation}
\label{eq:parallel-strip}
 H_d(z_0)=\{z:|\operatorname{Im}(e^{-i\alpha}(z-z_0))|<d\}.
\end{equation}
$H_d(z_0)$ is the open strip of width $2d$ centred on the straight line passing through $z_0$ with direction $\alpha$. Every translation $z\mapsto z+s e^{i\alpha}$, with $s\in\mathbb R$, preserves this strip.

\begin{figure}[htbp]
    \centering
    \includegraphics[width=0.58\textwidth]{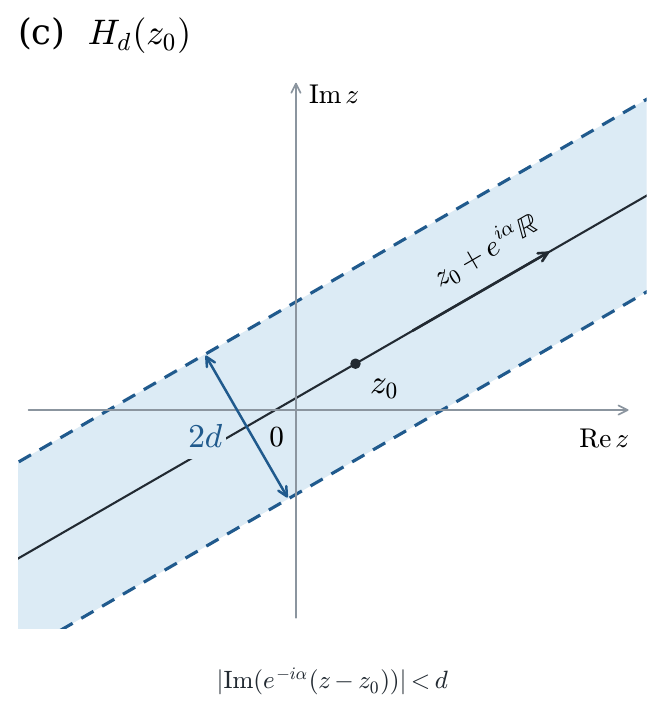}
    \caption{The open strip $H_d(z_0)$ of width $2d$ around
    $z_0+e^{i\alpha}\mathbb R$, illustrated for $\alpha=\frac{\pi}{6}$.
    It extends infinitely in both directions and is preserved
    by translations $z\mapsto z+s e^{i\alpha}$, $s\in\mathbb R$.
    Its dashed boundary lines are excluded.}
    \label{fig:action-strip}
\end{figure}

\begin{lemma}[Uniform summation on strips]
\label{lem:uniform-strip-summation}
Let $U\subset\mathbb C^m$ be an open parameter domain. Fix a summation direction
$\alpha\in\mathbb R/(2\pi\mathbb Z)$, a point $z_0\in\mathbb C$, and constants $d,R_0>0$. Recall that $S_{\alpha,R_0}$ is the Borel disk in direction $\alpha$ defined in \eqref{eq:Borel-disk-and-tube}. Let
\begin{equation*}
  \begin{aligned}
F:\quad U\times H_d(z_0)\times S_{\alpha,R_0}\times\mathbb C^r&\longrightarrow \mathbb C^r,\\(k,z,\hbar,v)&\longmapsto F(k,z,\hbar,v)
\end{aligned}  
\end{equation*}
be jointly holomorphic and polynomial in $v$. Consider the differential equation
\begin{equation}\label{deff}
\hbar\partial_z f(k,z,\hbar)=F\bigl(k,z,\hbar,f(k,z,\hbar)\bigr),
\end{equation}
where the unknown $f$ takes values in $\mathbb C^r$.

Assume that $F$ admits a Gevrey-$1$ asymptotic expansion in $\hbar$,
\begin{equation*}
     F(k,z,\hbar,v)\sim\sum_{n=0}^{\infty}F_n(k,z,v)\hbar^n,
\end{equation*}
whose coefficients are holomorphic on
$U\times H_d(z_0)\times\mathbb C^r$.
The expansion is assumed uniform on the entire strip $H_d(z_0)$ and locally uniform in the remaining variables: for every compact set $K\Subset U$ and every bounded
set $B\subset\mathbb C^r$, there exist constants $C,A>0$ and $0<R_{K,B}\leq R_0$ such that
\begin{equation}
\label{eq:strip-input-remainders}
\sup_{\substack{k\in K,\ z\in H_d(z_0)\\v\in B}}\left\|F(k,z,\hbar,v)-\sum_{n=0}^{N-1}F_n(k,z,v)\hbar^n\right\|\leq C A^N N!|\hbar|^N
\end{equation}
for every integer $N\geq0$ and every $\hbar\in S_{\alpha,R_{K,B}}$, with the sum understood to be empty when $N=0$. Let $F_0$ denote the leading coefficient of this expansion. Assume that, for every $(k,z)\in U\times H_d(z_0)$,
\begin{equation*}
   F_0(k,z,0)=0,\qquad \partial_vF_0(k,z,0)=\Lambda, 
\end{equation*} 
where $\partial_vF_0$ denotes the Jacobian matrix with respect to $v$, and $\Lambda$ is a constant diagonal $r\times r$ matrix with entries in $\{+1,-1\}$.

Then the differential equation \eqref{deff} admits a unique
formal solution $\hat f(k,z,\hbar)=\underset{n=1}{\overset{\infty}{\sum}}f_n(k,z)\hbar^n$, with coefficients holomorphic on $U\times H_d(z_0)$ and vanishing constant term $f_0(k,z)\equiv0$. Moreover, for every relatively compact open subset $U_0\Subset U$ and every $0<d'<d$, this series is uniformly Borel summable in direction $\alpha$ on $U_0\times H_{d'}(z_0)$. Its Borel sum $f(k,z,\hbar)$ is jointly holomorphic on $U_0\times H_{d'}(z_0)\times S_{\alpha,R}$ for some $0<R\leq R_0$, and satisfies $\hbar\partial_z f(k,z,\hbar)=F\bigl(k,z,\hbar,f(k,z,\hbar)\bigr)$. The Borel tube, exponential growth constants, Borel disk and factorial remainder constants may depend on $U_0$ and $d'$, but are uniform throughout $U_0\times H_{d'}(z_0)$.
\end{lemma}

\begin{proof}
The complete proof is in \autoref{app:uniform-strip-proof}. It
uses the Borel-transformed equation, characteristic translations
$z\mapsto z\pm\xi$, and a contraction estimate uniform in the
unbounded direction of the strip. In particular, compact-parameter
estimates are not being extended to an unbounded strip without proof.
\end{proof}

We now apply \autoref{lem:uniform-strip-summation} to
linear systems written in an eigenbasis of their leading
matrix. The direction of each column is described by a
ratio of its two components. These ratios satisfy Riccati
equations, to which the preceding lemma applies.
This gives Borel summability of the column directions,
independently of their scalar normalisations.
We will use this construction both on spatial strips
and on the moving family through the quadrivalent vertex. 
The scalar amplitudes will then be reconstructed in
\autoref{lem:projective-reconstruction} and the wave matrix in \autoref{thm:exact-WKB-existence}.

\begin{lemma}[Summation of the component ratios in an eigenbasis]
\label{lem:projective-strip}
Let $U\subset\mathbb C^m$ be an open parameter domain.
Fix a direction $\alpha\in\mathbb R/(2\pi\mathbb Z)$,
a point $z_0\in\mathbb C$, and constants $d>0$, $R_0>0$.
Let $H_d(z_0)$ and $S_{\alpha,R_0}$ be the strip and
Borel disk defined in \eqref{eq:parallel-strip} and
\eqref{eq:Borel-disk-and-tube}, respectively. Let
\[
    B=(B_{ij})_{1\leq i,j\leq2}:
    U\times H_d(z_0)\times S_{\alpha,R_0}
    \longrightarrow \operatorname{Mat}_2(\mathbb C)
\]
be jointly holomorphic. Assume that $B$ admits a
Gevrey-$1$ expansion
\[
    B(k,z,\hbar)\sim
       \hat B(k,z,\hbar)
       =\sum_{n=0}^{\infty}B_n(k,z)\hbar^n,
    \qquad B_0\equiv0,
\]
with holomorphic coefficients $B_n(k,z)$ on $U\times H_d(z_0)$.
More precisely, for every compact set $K\Subset U$,
there exist constants $C_K>0$, $A_K>0$ and
$0<R_K\leq R_0$ such that
\[
    \sup_{\substack{k\in K\\z\in H_d(z_0)}}
    \left\|
       B(k,z,\hbar)
       -\sum_{n=0}^{N-1}B_n(k,z)\hbar^n
    \right\|
    \leq C_K (A_K)^N N!|\hbar|^N
\]
for every integer $N\geq0$ and every
$\hbar\in S_{\alpha,R_K}$, with the sum understood
to be empty when $N=0$. Consider the linear system
\begin{equation}
\label{eq:strip-linear-system}
    \hbar\partial_z\chi=(D+B)\chi,
    \qquad
    D=\operatorname{diag}\left(\frac12,-\frac12\right),
    \qquad
    \chi\coloneqq\binom{\chi_1}{\chi_2}.
\end{equation}
To describe a column direction near the first or second
coordinate eigenline of $D$, we use the respective
representatives $\begin{pmatrix} 1\\ g_+\end{pmatrix}$ and $\begin{pmatrix} g_-\\ 1\end{pmatrix}$. The scalar coordinates $g_+$ and $g_-$ satisfy, in their respective charts, the Riccati equations
\begin{align}
    \hbar\partial_zg_+
       &=B_{21}+(-1+B_{22}-B_{11})g_+
                   -B_{12}g_+^2,
       \label{eq:gplus-Riccati}\\
    \hbar\partial_zg_-
       &=B_{12}+(1+B_{11}-B_{22})g_-
                   -B_{21}g_-^2.
       \label{eq:gminus-Riccati}
\end{align}
We consider these two equations independently: the solutions selected below describe two different column directions.
Each of these Riccati equations, with $B$ replaced by
its formal expansion $\hat B$, has a unique formal
solution with zero constant coefficient, $\hat g_\pm(k,z,\hbar) =\underset{n=1}{\overset{\infty}{\sum}}g_{\pm,n}(k,z)\hbar^n$, whose coefficients are holomorphic on $U\times H_d(z_0)$. The vectors $\begin{pmatrix} 1\\ \hat g_+\end{pmatrix}$, $\begin{pmatrix} \hat g_-\\ 1\end{pmatrix}$ represent the corresponding formal column directions: their leading terms are the two coordinate eigenvectors of $D$.

Then, for every relatively compact open subset $U_0\Subset U$
and every $0<d'<d$, both series $\hat g_\pm$ are
uniformly Borel summable in direction $\alpha$ on
$U_0\times H_{d'}(z_0)$.
With the Borel--Laplace convention
\eqref{eq:parametric-Borel-sum}, their sums are
\[
    g_\pm(k,z,\hbar)
       \coloneqq s_\alpha(\hat g_\pm)(k,z,\hbar)
       =\int_{e^{i\alpha}\mathbb R_{\geq0}}
          e^{-\frac{\xi}{\hbar}}
          \varphi_\pm(k,z,\xi)\,\mathrm d\xi,
\]
where $\varphi_\pm$ denotes the holomorphic continuation
to a common Borel tube of the convergent germ $(\mathcal B\hat g_\pm)(k,z,\xi)=\underset{n=0}{\overset{\infty}{\sum}}\frac{g_{\pm,n+1}(k,z)}{n!}\xi^n$. The functions $g_\pm$ are jointly holomorphic on
$U_0\times H_{d'}(z_0)\times S_{\alpha,R}$
for a common $0<R\leq R_0$, and satisfy
\eqref{eq:gplus-Riccati} and
\eqref{eq:gminus-Riccati}, respectively. They admit $\hat g_\pm$ as their uniform Gevrey-$1$ asymptotic expansions on this product $U_0\times H_{d'}(z_0)\times S_{\alpha,R}$.  In particular, $g_\pm=O(\hbar)$ uniformly there. The Borel tubes, growth constants, remainder constants and $R$ may depend on $U_0$ and $d'$, but are uniform throughout the entire strip $H_{d'}(z_0)$.
\end{lemma}

\begin{proof}
The two rows of \eqref{eq:strip-linear-system} are
\[
    \hbar\partial_z\chi_1
       =\left(\frac12+B_{11}\right)\chi_1+B_{12}\chi_2,
    \qquad
    \hbar\partial_z\chi_2
       =B_{21}\chi_1+\left(-\frac12+B_{22}\right)\chi_2.
\]
For any local solution column $\chi$, the quotient rule
shows that $\frac{\chi_2}{\chi_1}$, wherever $\chi_1\neq0$,
satisfies \eqref{eq:gplus-Riccati}, while
$\frac{\chi_1}{\chi_2}$, wherever $\chi_2\neq0$,
satisfies \eqref{eq:gminus-Riccati}.
We now regard these as independent differential
equations for $g_+$ and $g_-$, without assuming that
the corresponding linear-system columns have already
been constructed. Together, these two Riccati equations form a system
for the column vector $f=\begin{pmatrix}g_+\\g_-\end{pmatrix}\in\mathbb C^2$. This system has the form $\hbar\partial_z f(k,z,\hbar)=F\bigl(k,z,\hbar,f(k,z,\hbar)\bigr)$, where, for $v=\begin{pmatrix}v_1\\v_2\end{pmatrix}
\in\mathbb C^2$, we define
\begin{equation*}
    F(k,z,\hbar,v)
       \coloneqq
       \begin{pmatrix}
          B_{21}+(-1+B_{22}-B_{11})v_1-B_{12}v_1^2
          \\[3pt]
          B_{12}+(1+B_{11}-B_{22})v_2-B_{21}v_2^2
       \end{pmatrix}.
\end{equation*} 
Here every entry $(B_{ij})_{1\leq i,j\leq 2}$ is evaluated at $(k,z,\hbar)$.
The function $F$ is jointly holomorphic and polynomial
in $v$. The assumed remainder estimates for $B$ give
the remainder estimates required in
\autoref{lem:uniform-strip-summation}, uniformly on
bounded sets of $v$. Since the leading coefficient $B_0$ of the expansion
in $\hbar$ of $B$ vanishes, the coefficient of $\hbar^0$ in the
expansion of $F$ is $F_0(k,z,v)=\begin{pmatrix}-v_1\\v_2\end{pmatrix}$. Consequently,
\[
    F_0(k,z,0)=0,
    \qquad
    \partial_vF_0(k,z,0)
       =\begin{pmatrix}-1&0\\0&1\end{pmatrix}.
\]
Thus all hypotheses of
\autoref{lem:uniform-strip-summation} hold.
That lemma supplies the unique formal solutions with
zero constant coefficients, their uniform Borel
summability, and the stated properties of their sums. The integral formula in the statement follows from the Borel--Laplace convention
\eqref{eq:parametric-Borel-sum}: since the formal
solutions $\hat g_\pm$ have zero constant coefficient,
$g_{\pm,0}=0$, there is no constant term outside
the Laplace integral.
\end{proof}

\subsection{Exact WKB wave matrices in the Stokes regions}
\label{sec:exact-WKB-reconstruction}
In this section, we construct exact WKB wave matrices in each
time-dependent Stokes region in three steps. First,
\autoref{prop:sectorial-projective-summability} establishes Borel
summability of the formal component ratios $\hat r_\pm$, using
spatial action strips. After specifying the required exact wave
matrix in \autoref{def:exact-WKB-wave-matrix},
\autoref{lem:projective-reconstruction} reconstructs the scalar
amplitudes with prescribed Borel-summable normalisations; its proof
is deferred to \autoref{app:projective-reconstruction-proof}.
Finally, \autoref{thm:exact-WKB-existence} patches the local
constructions and establishes existence and uniqueness throughout
the Stokes region.

\paragraph{Borel summability of the formal component ratios $\hat r_\pm$.}
We begin by constructing the spatial action strips needed to apply \autoref{lem:projective-strip} to the spatial equation written in an eigenbasis.

\begin{lemma}[Spatial action strips]
\label{lem:spatial-adapted-strips}
Let $\mathcal U_j$, $j\in\{\mathrm A,\ldots,\mathrm F\}$,
be one of the time-dependent Stokes regions defined in
\eqref{eq:time-dependent-Stokes-regions}, and fix
$(x_0,t_0)\in\mathcal U_j$.
Put $\tau_0=\tau(t_0)$ and choose either value $w_0$
satisfying $w_0^2=x_0+2\tau_0$. This choice specifies a local sheet of the normalised spectral curve described in \autoref{sec:stokes-geometry}, and hence a local branch of the action $Z_{\tau(t)}(x)$ in \eqref{eq:space-action}.
Set $s_0\coloneqq 2Z_{\tau_0}(x_0)=2\widetilde Z_{\tau_0}(w_0)$, where $\widetilde Z_\tau$ is defined in \eqref{DefZtilde}. There exist a relatively compact open neighbourhood $V_0\Subset V$ of $t_0$, a constant $d>0$, and jointly
holomorphic functions $x,w:H_d(s_0)\times V_0\longrightarrow\mathbb C$, where $H_d(s_0)
       =\left\{
          s\in\mathbb C:
          \left|\operatorname{Im}
             \bigl(e^{-i\alpha}(s-s_0)\bigr)\right|<d
        \right\}$,
such that
\[
    2\widetilde Z_{\tau(t)}(w(s,t))=s,
    \qquad
    x(s,t)=w(s,t)^2-2\tau(t), \qquad  w(s_0,t_0)=w_0,
    \qquad
    x(s_0,t_0)=x_0.
\]
Thus $x(s,t)$ is the inverse branch of $s=2Z_{\tau(t)}(x)$ obtained by continuation from the chosen local sheet.

Throughout $H_d(s_0)\times V_0$, we have
\[
    w(s,t)y(s,t)\neq0
    \quad \text{ with }\quad
    y(s,t)\coloneqq
       2\bigl(x(s,t)-\tau(t)\bigr)w(s,t).
\]
At both ends of the strip, that is, as
$\operatorname{Re}(e^{-i\alpha}(s-s_0))\to+\infty$
or $-\infty$, we have\footnote{For nonnegative quantities $F(s,t)$ and
$G(s,t)$, the notation $F\asymp G$ means that there
exist constants $c,C>0$ such that
$cG(s,t)\leq F(s,t)\leq CG(s,t)$ for all sufficiently
large $|s|$ in the stated domain. Here the constants
and the threshold for $|s|$ are independent of $t\in V_0$.}, uniformly for $t\in V_0$,
\[
    |w(s,t)|\asymp |s|^{\frac{1}{5}},
    \qquad
    |x(s,t)|\asymp |s|^{\frac{2}{5}},
    \qquad
    |y(s,t)|\asymp |s|^{\frac{3}{5}}.
\]
The choices and bounds can be made locally uniformly
as the base point $(x_0,t_0)$ varies in $\mathcal U_j$,
with the chosen sheet continued holomorphically
in each such neighbourhood.
\end{lemma}

\begin{proof}The proof is done in \autoref{AppendixProofLemma43}
\end{proof} 

We now use these spatial action strips to apply
\autoref{lem:projective-strip}, with $s$ as the
differential variable and $t$ as the holomorphic
parameter.
We write the spatial Lax equation in an eigenbasis
of its leading matrix and verify the uniform
Gevrey-$1$ estimates required for its remainder.
The resulting summability of the component ratios
in this eigenbasis then gives, after transforming
back to the original basis, the following result
for the formal ratios $\hat r_\pm$ constructed in
\autoref{241113190220}.

\begin{proposition}[Summability of the formal projective solutions]
\label{prop:sectorial-projective-summability}
Let $\mathcal U_j$, $j\in\{\mathrm A,\ldots,\mathrm F\}$,
be one of the time-dependent Stokes regions defined
in \eqref{eq:time-dependent-Stokes-regions}.
Let $\Omega\subset\mathcal U_j$ be an open set carrying
a chosen holomorphic branch $w=w(x,t)$ satisfying $w(x,t)^2=x+2\tau(t)$. For each $\varepsilon\in\{+1,-1\}$, consider the formal ratio
\[
    \hat r_\varepsilon(x,t;\hbar)
       =\frac{\hat v_{\varepsilon,2}}
              {\hat v_{\varepsilon,1}}
       =\frac{\varepsilon w(x,t)}2
        +\sum_{n=1}^{\infty}
           r_{\varepsilon,n}\bigl(\tau(t),w(x,t)\bigr)
           \hbar^n
\]
of the formal WKB column in \autoref{241113190220}.
This is the formal series determined by
\eqref{eq:formal-time-Riccati} and the recursion
\eqref{eq:formal-r-recursion}, evaluated on the
chosen branch.

Then both series $\hat r_\varepsilon$ are locally
uniformly Borel summable in direction $\alpha$ on
$\Omega$, in the sense of
\autoref{def:locally-uniform-Borel-summability}. More precisely, for every relatively compact open
subset $\Omega_0\Subset\Omega$, there exists
$R_{\Omega_0}>0$, common to both signs, such that
their Borel--Laplace sums
\[
    r_\varepsilon(x,t,\hbar)
       \coloneqq
       s_\alpha(\hat r_\varepsilon)(x,t,\hbar)
\]
are jointly holomorphic on
$\Omega_0\times S_{\alpha,R_{\Omega_0}}$.
With the convention \eqref{eq:parametric-Borel-sum},
these sums are given by
\[
    r_\varepsilon(x,t,\hbar)
       =\frac{\varepsilon w(x,t)}2
        +\int_{e^{i\alpha}\mathbb R_{\geq0}}
           e^{-\frac{\xi}{\hbar}}
           \varphi_\varepsilon(x,t,\xi)\,
           \mathrm d\xi,
\]
where $\varphi_\varepsilon$ is the holomorphic
continuation of $\mathcal B\hat r_\varepsilon$
to the Borel tube supplied by local uniform
summability. The functions $r_\varepsilon$ admit
$\hat r_\varepsilon$ as their uniform Gevrey-$1$
asymptotic expansions on $\Omega_0\times S_{\alpha,R_{\Omega_0}}$, in the sense
of \eqref{eq:uniform-Gevrey-convention}.
The Borel-disk parameter $R_{\Omega_0}$ and the
remainder constants may depend on $\Omega_0$.
The Borel sums $r_\varepsilon$ constructed on overlapping subdomains
of $\Omega$ agree after restriction to a common smaller Borel disk.
\end{proposition}

\begin{proof}
Fix $(x_0,t_0)\in\Omega$ and use the branch of $w$
chosen in the statement of the proposition.
By \autoref{lem:spatial-adapted-strips}, there exist
a strip $H_d(s_0)$, a time neighbourhood
$V_0\Subset V$ of $t_0$, and jointly holomorphic
functions $x=x(s,t)$ and $w=w(s,t)$ on
$H_d(s_0)\times V_0$, such that
\[
    s=2\widetilde Z_{\tau(t)}(w(s,t)),
    \qquad
    x(s,t)=w(s,t)^2-2\tau(t),
    \qquad
    w(s,t)y(s,t)\neq0.
\]
After shrinking the time neighbourhood, we may
assume that $V_0$ is an open disk centred at $t_0$.

We start from the spatial equation of the Lax system $\hbar\partial_x\mathbf v=L\mathbf v$ given in 
\eqref{250408155610}, where $L$ is the matrix in \eqref{241115131339}.
In this part of the proof, $t$ is a parameter.
Since $s=2Z_{\tau(t)}(x)$ and
$\partial_xZ_{\tau(t)}(x)=y$, differentiation at
fixed $t$ gives $\partial_xs=2y, 
    \partial_sx=\frac1{2y}$.
Thus $s$ replaces the spatial variable $x$ but it does not replace the time variable $t$. The product $H_d(s_0)\times V_0$ is simply connected,
and $w$ does not vanish there.
We may therefore choose a holomorphic square root
$w^{\frac{1}{2}}$ on this product and define
\begin{equation}
\label{eq:regular-eigenbasis}
    E=w^{-\frac{1}{2}}
       \begin{pmatrix}
          1&1\\
          \frac{w}{2}&-\frac{w}{2}
       \end{pmatrix}
    \quad \Rightarrow\quad
    \det E=-1.
\end{equation}
Its columns are eigenvectors of the leading
spatial matrix $L_0$: $E^{-1}L_0E=\operatorname{diag}(y,-y)$. For a local solution column $\mathbf v$, set
\[
    \chi(s,t,\hbar)
       =E(s,t)^{-1}
          \mathbf v(x(s,t),t,\hbar).
\]
The substitution $\mathbf v=E\chi$ in the spatial
equation \eqref{250408155610} gives $\hbar\partial_x\chi
       =\left(E^{-1}LE-\hbar E^{-1}\partial_xE\right)
          \chi$.
Using $\partial_x=2y\,\partial_s$ at fixed $t$,
we obtain
\begin{equation}\label{EqdiffMs}\hbar\partial_s\chi=M_x\chi,
    \qquad
    M_x=
       \frac{E^{-1}LE-\hbar E^{-1}\partial_xE}{2y}.
\end{equation}
Here and below, all expressions are evaluated at
$x=x(s,t)$.

Put $\Delta=q-\tau$. Since $t=-6\tau^2$ and $w^2=x+2\tau$, the matrix
$L$ in \eqref{241115131339} decomposes as
\[
    L=L_0+
       \begin{pmatrix}
          p&-4\Delta\\
          \Delta w^2+\Delta^2&-p
       \end{pmatrix},
    \qquad
    L_0=
       \begin{pmatrix}
          0&4(x-\tau)\\
          (x-\tau)w^2&0
       \end{pmatrix}.
\]
Indeed, $x^2+xq+q^2+\frac t2=(x-\tau)w^2+\Delta w^2+\Delta^2$. From \eqref{EqdiffMs} we obtain
\begin{equation}
\label{eq:spatial-eigenbasis-remainder}
    M_x=D+\frac1{2y}
       \begin{pmatrix}
          \frac{\Delta^2}{w}
          &
          p+\frac{\Delta(2x+q+3\tau)}{w}+\frac{\hbar}{4w^2}
          \\[2pt]
          p-\frac{\Delta(2x+q+3\tau)}{w}+\frac{\hbar}{4w^2}
          &
          -\frac{\Delta^2}{w}
       \end{pmatrix},
    \qquad
    D=\operatorname{diag}\left(\frac12,-\frac12\right).
\end{equation}

Choose $0<d_1<d$ and $t_0\in V_1\Subset V_0$, where $V_1$ is an open disk.
By \autoref{250408141643}, there exists $R_1>0$
such that $q$ and $p$ are jointly holomorphic on $V_1\times S_{\alpha,R_1}$ and admit their uniform Gevrey-$1$ expansions there.
In particular, they are uniformly bounded on this
product, and their leading coefficients are
$q_0=\tau$ and $p_0=0$.

We claim that the four geometric factors $\left(\frac1{yw},
    \frac1y,
    \frac wy,
    \frac1{yw^2}\right)
$
are holomorphic and uniformly bounded on
$H_{d_1}(s_0)\times V_1$. Their holomorphy follows from $wy\neq0$.
To prove boundedness, first restrict to a finite
middle part of the strip.
Its closure, together with $\overline{V_1}$,
is a compact subset of $H_d(s_0)\times V_0$.
Since $w$ and $y$ do not vanish there, all four
factors are bounded on this middle part. At either end, the estimates in
\autoref{lem:spatial-adapted-strips} give
\[
    \frac1{yw}=O(|s|^{-\frac{4}{5}}),
    \qquad
    \frac1y=O(|s|^{-\frac{3}{5}}),
    \qquad
    \frac wy=O(|s|^{-\frac{2}{5}}),
    \qquad
    \frac1{yw^2}=O(|s|^{-1}),
\]
uniformly for $t\in V_1$.
This proves boundedness on the entire product
$H_{d_1}(s_0)\times V_1$. Set $B=M_x-D$. Using
\[
    \frac{\Delta(2x+q+3\tau)}{yw}
       =2\Delta\frac wy+\Delta^2\frac1{yw},
\]
we see that every entry of $B$ is a finite sum
of products of the four bounded geometric factors
with $p$, $\Delta$, $\Delta^2$ or $\hbar$.
Consequently, $B$ is jointly holomorphic on $H_{d_1}(s_0)\times V_1\times S_{\alpha,R_1}$. The uniform Gevrey-$1$ estimates for $(q,p)$,
together with stability of these estimates under
finite sums and products, give constants $C>0$, $A>0$
such that, after decreasing $R_1$ if necessary,
\[
    \sup_{\substack{s\in H_{d_1}(s_0)\\t\in V_1}}
    \left\|
       B(s,t,\hbar)
       -\sum_{n=0}^{N-1}B_n(s,t)\hbar^n
    \right\|
       \leq C A^N N!|\hbar|^N
\]
for every $N\geq0$ and
$\hbar\in S_{\alpha,R_1}$.
Since $\Delta_0=0$ and $p_0=0$, we have $B_0=0$.

We can therefore apply \autoref{lem:projective-strip},
with $s$ as the differential variable and $t\in V_1$
as the parameter.
Choose
$0<d_2<d_1$ and $t_0\in V_2\Subset V_1$, with $V_2$ an open disk.
\autoref{lem:projective-strip} gives uniform Borel summability of the
two formal eigenbasis ratios $\hat g_\varepsilon$
on $H_{d_2}(s_0)\times V_2$.
Their sums $g_\varepsilon$ are jointly holomorphic
on $H_{d_2}(s_0)\times V_2\times S_{\alpha,R_2}$ for a common $0<R_2\leq R_1$.
The formal eigenbasis directions are $\begin{pmatrix}1\\ \hat g_+ \end{pmatrix}$, $\begin{pmatrix}\hat g_-\\ 1 \end{pmatrix}$. Multiplication by $E$ gives their component ratios
in the original basis. We first identify these
ratios with the formal solutions appearing in the
statement. The gauge transformation and the change of spatial
variable transform the corresponding projective
equations into one another. Hence each transformed
ratio satisfies the formal spatial Riccati equation
\[
    \hbar\partial_x\hat r
       =x^2+x\hat q+\hat q^2+\frac t2
        -2\hat p\hat r
        -4(x-\hat q)\hat r^2.
\]
Its leading term is $\frac{\varepsilon w}{2}$. The leading
right-hand side of this equation is
$(x-\tau)(w^2-4r^2)$, and
\[
    \left.
    \partial_r\bigl((x-\tau)(w^2-4r^2)\bigr)
    \right|_{r=\frac{\varepsilon w}{2}}
       =-4\varepsilon(x-\tau)w
       =-2\varepsilon y
       \neq0.
\]
Consequently, comparison of successive powers of
$\hbar$ determines every higher coefficient uniquely:
at each order, the unknown coefficient is multiplied
by $-2\varepsilon y$, while all remaining terms
depend only on previously determined coefficients.
The formal ratio supplied by
\autoref{241113190220} satisfies this same equation
and has the same leading term. It therefore agrees
with the transformed ratio. Explicitly,
\begin{equation}
\label{eq:projective-back-transform}
    \hat r_\varepsilon
       =\varepsilon\frac w2\,
          \frac{1-\hat g_\varepsilon}
               {1+\hat g_\varepsilon}.
\end{equation}

Locally uniform Borel summability is stable under
finite sums, products, and inversion of a series
whose constant coefficient is nonzero.
Since $\hat g_\varepsilon$ has zero constant
coefficient, \eqref{eq:projective-back-transform}
therefore proves local uniform Borel summability
of $\hat r_\varepsilon$ on the strip product.
Here multiplication by $w$ preserves the estimates
on relatively compact subsets; no uniform bound
for $w$ on the entire unbounded strip is required. Moreover, \autoref{lem:projective-strip} gives
$g_\varepsilon=O(\hbar)$ uniformly on
$H_{d_2}(s_0)\times V_2$.
After decreasing $R_2$ if necessary, we therefore
have $|g_\varepsilon|\leq\frac12$ for both signs
throughout the corresponding product with
$S_{\alpha,R_2}$. Thus $1+g_\varepsilon$ does not
vanish there, and compatibility of Borel summation
with these algebraic operations gives
\[
    s_\alpha(\hat r_\varepsilon)
       =\varepsilon\frac w2\,
          \frac{1-g_\varepsilon}
               {1+g_\varepsilon}.
\]
These functions are jointly holomorphic. Since $\partial_sx=\frac{1}{2y}\neq0$, the map $(s,t)\longmapsto(x(s,t),t)$ is locally biholomorphic. Restricting near the
chosen point therefore gives local uniform Borel
summability in the original variables $(x,t)$,
on a neighbourhood contained in $\Omega$. Finally, let $\Omega_0\Subset\Omega$.
Cover $\overline{\Omega_0}$ by finitely many such
neighbourhoods, chosen sufficiently small that the
local Borel estimates are uniform on each of them.
On their overlaps, the local Borel transforms have
the same Taylor series at $\xi=0$, namely
$\mathcal B\hat r_\varepsilon$.
The identity theorem therefore identifies their
continuations on a common smaller Borel tube.
Taking the minimum of the finitely many tube widths
and Borel-disk parameters, and enlarging the
finitely many bounds, gives common estimates for
both signs throughout $\Omega_0$. The resulting Laplace integrals consequently agree
on overlaps and define jointly holomorphic
functions on
$\Omega_0\times S_{\alpha,R_{\Omega_0}}$.
Their integral representations and uniform
Gevrey-$1$ remainder estimates are those asserted
in the proposition.
\end{proof}

\paragraph{From projective solutions to exact WKB wave matrices.}
\autoref{prop:sectorial-projective-summability} supplies the Borel
sums of the formal ratios $\hat r_\pm$, which determine the directions
of the two columns. It remains to reconstruct their scalar amplitudes
with the prescribed normalisations. We first define the exact WKB
wave matrix that we seek. Then \autoref{lem:projective-reconstruction}
gives the local reconstruction, and
\autoref{thm:exact-WKB-existence} combines it with
\autoref{prop:sectorial-projective-summability} to obtain existence
and uniqueness throughout a Stokes region.

\begin{definition}[Exact WKB wave matrix]
\label{def:exact-WKB-wave-matrix}
Let $\Omega\subset\mathbb C_x\times V$ be a connected open set
carrying fixed holomorphic branches of $\tau$ and $w$, with
$t=-6\tau^2$, $w^2=x+2\tau$ and $\tau w\neq0$.
Fix a formal fundamental WKB wave matrix
\[
    \hat\Psi=\hat E e^{\frac{S}{\hbar}},
    \qquad
    \hat E=\sum_{n=0}^{\infty}E_n\hbar^n,
    \qquad
    S=\operatorname{diag}(S_+,S_-),
\]
where the ordered phases and the coefficients $E_n$ are holomorphic
on $\Omega$, and $E_0$ is invertible there.
An exact WKB wave matrix with these formal data, in direction
$\alpha$, is a holomorphic fundamental solution $\Psi$ of both
equations of \eqref{250408155610}, defined on an open set
$\mathcal D\subset\Omega\times\mathbb C^\ast$ with the following
property: for every relatively compact open subset
$\Omega_0\Subset\Omega$, there are constants $A,B,R>0$ such that
$\Omega_0\times S_{\alpha,R}\subset\mathcal D$ and its amplitude
$E=\Psi e^{-\frac{S}{\hbar}}$ satisfies
\begin{equation}
\label{eq:exact-amplitude-Gevrey}
    \sup_{(x,t)\in\Omega_0}
    \left\|E(x,t,\hbar)-\sum_{n=0}^{N-1}E_n(x,t)\hbar^n\right\|
       \leq A B^N N!|\hbar|^N,
    \qquad N\geq0,\quad \hbar\in S_{\alpha,R}.
\end{equation}
The sum is empty when $N=0$. The constants may depend on
$\Omega_0$, but are independent of $N$ and $\hbar$.
Here $S_{\alpha,R}$ is the Borel disk defined in
\eqref{eq:Borel-disk-and-tube}. Uniqueness is understood after
restriction to common such products.
\end{definition}

\paragraph{Local reconstruction of the scalar amplitudes.}
Let $b=(x_b,t_b)$ and let $\Omega\subset\mathbb C_x\times V$ be a
polydisc centred at $b$, carrying holomorphic branches with
\[
    \tau(t)\neq0,\qquad w(x,t)^2=x+2\tau(t)\neq0,
    \qquad t=-6\tau(t)^2.
\]
Put $y=2(x-\tau)w$. Since $y\,\mathrm dx+w\,\mathrm dt=\mathrm dZ$
by \eqref{250716160458lambda}, choose the phases
\begin{equation}
\label{eq:reconstruction-phases}
    S_\varepsilon=\varepsilon\bigl(Z-Z(b)\bigr)-\kappa_\varepsilon,
    \qquad
    \mathrm dS_\varepsilon
       =\varepsilon(y\,\mathrm dx+w\,\mathrm dt),
    \qquad \varepsilon\in\{+1,-1\},
\end{equation}
where $\kappa_\varepsilon\in\mathbb C$ are prescribed constants.
Consider the formal columns of \autoref{241113190220}, with
$m_+=m_-=0$, written as
\begin{equation}
\label{eq:formal-projective-columns}
    \hat{\mathbf v}_\varepsilon
       =\hat a_\varepsilon
         \begin{pmatrix}1\\ \hat r_\varepsilon\end{pmatrix}
         e^{\frac{S_\varepsilon}{\hbar}},
    \qquad
    \hat a_\varepsilon(b;\hbar)=\hat b_\varepsilon(\hbar)
       \in\mathbb C[[\hbar]]^\times.
\end{equation}
After shrinking $\Omega$ if necessary, the leading coefficients of
$\hat a_\varepsilon$ are nonvanishing. The formal amplitude matrix is
therefore
\begin{equation}
\label{eq:projective-amplitude-matrix}
    \hat E=
       \begin{pmatrix}
          \hat a_+&\hat a_-\\
          \hat r_+\hat a_+&\hat r_-\hat a_-
       \end{pmatrix}.
\end{equation}
The ratios $\hat r_\varepsilon$ determine the column directions,
but not the scalar amplitudes $\hat a_\varepsilon$. The latter are
reconstructed by integrating a closed logarithmic one-form, with
initial values fixed by $\hat b_\varepsilon$. The next lemma shows
that this reconstruction preserves Borel summability and gives
solutions of both equations of the Lax system. Its detailed proof
is given in \autoref{app:projective-reconstruction-proof}.

\begin{lemma}[Reconstruction from the projective solutions]
\label{lem:projective-reconstruction}
With the preceding notation, assume that $\hat r_+$ and $\hat r_-$
are locally uniformly Borel summable in direction $\alpha$ on
$\Omega$, and that $\hat b_+$ and $\hat b_-$ are Borel summable in
the same direction. Recall that $s_\alpha$ denotes Borel--Laplace summation in
direction $\alpha$, with the convention
\eqref{eq:parametric-Borel-sum}.
Let $q=s_\alpha(\hat q)$ be the selected
Painlev\'e~I solution and $p=\hbar\partial_tq$.
The constants $\kappa_\varepsilon\in\mathbb C$
are the prescribed phase constants introduced in
\eqref{eq:reconstruction-phases} and we have
$S_\varepsilon(b)=-\kappa_\varepsilon$.
Then, on every smaller polydisc $\Omega_0\Subset\Omega$ centred at
$b$, the scalar amplitudes $\hat a_\pm$ and the matrix $\hat E$ are
uniformly Borel summable in direction $\alpha$.
Writing
\[
    r_\varepsilon=s_\alpha(\hat r_\varepsilon),
    \qquad a_\varepsilon=s_\alpha(\hat a_\varepsilon),
    \qquad
    E=s_\alpha(\hat E)=
       \begin{pmatrix}
          a_+&a_-\\ r_+a_+&r_-a_-
       \end{pmatrix},
\]
there is an $R>0$ such that $\Psi=E e^{\frac{S}{\hbar}}$ is holomorphic and invertible on
$\Omega_0\times S_{\alpha,R}$ and solves both equations of
\eqref{250408155610}. Its first-row values at $b$ are
$s_\alpha(\hat b_\varepsilon)e^{-\frac{\kappa_\varepsilon}{\hbar}}$.
It is the unique exact WKB wave matrix with the prescribed formal
data, in the sense of \autoref{def:exact-WKB-wave-matrix}.
\end{lemma}

\begin{proof}
The proof is done in \autoref{app:projective-reconstruction-proof}.
\end{proof}

\paragraph{Existence and uniqueness on a Stokes region.}
\autoref{prop:sectorial-projective-summability} verifies the
summability hypothesis on the ratios in every Stokes region.
Applying \autoref{lem:projective-reconstruction} locally and
continuing the prescribed formal amplitudes then gives the following
result. Indeed, the local Borel sums agree on overlaps by uniqueness.

\begin{theorem}[Existence and uniqueness of exact WKB wave matrices]
\label{thm:exact-WKB-existence} 
Fix a summation direction
$\alpha\in\mathbb R/(2\pi\mathbb Z)$ and a time sector
$\widetilde V$ as in \autoref{250224171850}.
Let $V$ be its image under $t=-6\tau^2$, equipped
with the corresponding holomorphic inverse branch
$\tau=\tau(t)$. Let $(q,p)$ be the Borel-summed Painlev\'e~I solution
associated with this branch in \autoref{250408141643}. Fix $j\in\{\mathrm A,\ldots,\mathrm F\}$ and let
$\mathcal U_j\subset\mathbb C_x\times V$ be the
time-dependent Stokes region defined in
\eqref{eq:time-dependent-Stokes-regions}.
Choose a holomorphic branch $w=w(x,t)$ on
$\mathcal U_j$ satisfying $w(x,t)^2=x+2\tau(t)$. Together with $\tau(t)$, this specifies a lift to the spectral surface \eqref{250409162453}. Let $Z$ denote the primitive \eqref{250716160458lambda}, evaluated on this lift. Fix a basepoint $b=(x_b,t_b)\in\mathcal U_j$. For each $\varepsilon\in\{+1,-1\}$, prescribe
$\kappa_\varepsilon\in\mathbb C$ and a formal unit $\hat b_\varepsilon(\hbar)=\underset{n=0}{\overset{\infty}{\sum}}b_{\varepsilon,n}\hbar^n$, $b_{\varepsilon,0}\neq0$,
which is Borel summable in direction $\alpha$.
Use the first-component normalisations $\hat c_\varepsilon(\hbar)
       =\hat b_\varepsilon(\hbar)
          e^{-\frac{\kappa_\varepsilon}{\hbar}}$
in \autoref{241113190220} (thus $m_+=m_-=0$).
Set $S_\varepsilon
       =\varepsilon\bigl(Z-Z(b)\bigr)
          -\kappa_\varepsilon$, $S=\operatorname{diag}(S_+,S_-)$.
Let
\[
    \hat\Psi=\hat E e^{\frac{S}{\hbar}},
    \qquad
    \hat E=\sum_{n=0}^{\infty}E_n\hbar^n,
\]
be the formal fundamental WKB wave matrix obtained
locally from \autoref{250716102028} and continued
coefficientwise as in \autoref{250109214318}.
Then the coefficientwise continuation of $\hat E$
is single-valued on $\mathcal U_j$, and $\hat E$
is locally uniformly Borel summable in direction
$\alpha$ on $\mathcal U_j$ in the
sense of \autoref{def:locally-uniform-Borel-summability}.
With the Borel--Laplace convention
\eqref{eq:parametric-Borel-sum}, define
\[
    E=s_\alpha(\hat E),
    \qquad
    \Psi_{(j)}=E e^{\frac{S}{\hbar}}.
\]
Then $\Psi_{(j)}$ is the unique exact WKB wave matrix
with these prescribed formal data, in the sense of
\autoref{def:exact-WKB-wave-matrix}.
More precisely, for every relatively compact open
subset $\Omega_0\Subset\mathcal U_j$, there exists
$R_{\Omega_0}>0$ such that
\[
    \Psi_{(j)}:
    \Omega_0\times S_{\alpha,R_{\Omega_0}}
       \longrightarrow \operatorname{GL}_2(\mathbb C)
\]
is jointly holomorphic and satisfies $\hbar\partial_x\Psi_{(j)}=L\Psi_{(j)}$ and $\hbar\partial_t\Psi_{(j)}=A\Psi_{(j)}$ throughout $\Omega_0\times S_{\alpha,R_{\Omega_0}}$, where $L$ and $A$ are the Lax matrices in \eqref{241115131339}.
The amplitude $E$ admits $\hat E$ as its uniform
Gevrey-$1$ asymptotic expansion on
$\Omega_0\times S_{\alpha,R_{\Omega_0}}$,
in the sense of \eqref{eq:exact-amplitude-Gevrey}. 
The Borel-disk parameter and the remainder constants
may depend on $\Omega_0$. Finally, on a sufficiently small Borel disk, the first row of $\Psi_{(j)}$
at the basepoint is
\[
    \left([\Psi_{(j)}]_{1,1}(b,\hbar),
          [\Psi_{(j)}]_{1,2}(b,\hbar)\right)
       =
       \left(
          s_\alpha(\hat b_+)(\hbar)
             e^{-\frac{\kappa_+}{\hbar}},
          s_\alpha(\hat b_-)(\hbar)
             e^{-\frac{\kappa_-}{\hbar}}
       \right).
\]
The local constructions agree on overlaps after
restriction to a common smaller Borel disk.
\end{theorem}

\begin{proof}
By \autoref{prop:sectorial-projective-summability},
the formal component ratios $\hat r_\pm$ are locally uniformly
Borel summable in direction $\alpha$ on
$\mathcal U_j$. We denote their Borel--Laplace sums by
$r_\pm=s_\alpha(\hat r_\pm)$. Apply
\autoref{lem:projective-reconstruction} near the basepoint.
Its closed logarithmic one-forms can then be integrated along any
fixed finite path in the region $\mathcal{U}_j$. Cover the path by finitely many
polydiscs: on a common Borel disk, integration, multiplication and
exponentiation preserve Borel summability, as proved in that lemma.
This propagates Borel summability of the prescribed formal amplitudes to
every endpoint. We next prove that this continuation is single-valued.
Let $\gamma$ be a loop in $\mathcal U_j$ based at $b$.
Choose a simply connected relatively compact open set
$V_0\Subset V$ containing its time projection.
Such a choice is possible because $V$ is a sector
of opening $\frac{4\pi}{5}<\pi$ and is therefore convex.
By \autoref{250408141643}, after decreasing the
common Borel-disk parameter, $q$ and $p$ are
holomorphic on $V_0\times S_{\alpha,R}$.

For each fixed $\hbar\in S_{\alpha,R}$, the exact
Lax connection is holomorphic on
$\mathbb C_x\times V_0$, since its matrices
\eqref{241115131339} are polynomial in $x$.
It is flat by \autoref{250408155701}.
On this simply connected domain, continuation of
a local fundamental matrix is therefore independent
of the path.

Along $\gamma$, the local matrices supplied by
\autoref{lem:projective-reconstruction} realise
this analytic continuation: successive constructions
agree by their initial values and uniqueness for
the exact linear system.
Consequently, continuation around $\gamma$ returns
the original exact matrix.
Since the prescribed phases are single-valued on
$\mathcal U_j$, the same is true of its amplitude.
Uniqueness of asymptotic coefficients then shows
that the continued formal amplitude also returns
to its original germ.
Thus $\hat E$ is single-valued on $\mathcal U_j$.

On overlaps the local sums coincide by uniqueness
of the Borel sum of the given single-valued formal amplitude. Thus
they patch, independently of the path, and solve the complete system.
A finite cover of the closure of any relatively compact subdomain
gives a common Borel disk and common remainder constants.
Uniqueness follows by applying the Nevanlinna uniqueness
theorem entrywise to the difference of two amplitudes
with the prescribed formal expansion, as in the proof of
\autoref{lem:projective-reconstruction}.
\end{proof}

\begin{remark}
Normalisations of the more general form
$\hbar^{m_\varepsilon}\hat b_\varepsilon
e^{-\frac{\kappa_\varepsilon}{\hbar}}$,
with $m_\varepsilon\in\mathbb Z_{\geq0}$, are obtained
by replacing $\Psi_{(j)}$ with $\Psi_{(j)}
       \operatorname{diag}
          \bigl(\hbar^{m_+},\hbar^{m_-}\bigr)$.
This matrix still solves both equations of the Lax
system and remains invertible for $\hbar\neq0$.
The amplitude condition in
\autoref{def:exact-WKB-wave-matrix} and the uniqueness
statement apply after removing these prescribed
column factors.
\end{remark}

\section{Regularity at the quadrivalent vertex for the exact WKB matrix}\label{Section5}
This section proves that the exact WKB wave matrix constructed from \autoref{thm:exact-WKB-existence}
is regular at the double turning point
$x=\tau(t)$, which is the quadrivalent vertex of
the Stokes graph. More precisely, we construct
a common Borel sum on a full neighbourhood of this
vertex and show that the matrices in the four
incident Stokes regions coincide there when their
ordered phases, branches and normalisations agree.

The argument proceeds in three steps.
First, \autoref{prop:formal-regularity-quadrivalent}
establishes regularity and factorial-type bounds
for the formal amplitude at the vertex.
Second, since the spatial action coordinate
degenerates there, \autoref{lem:Morse-family}
constructs a family in which $x$ and $t$ vary
together while a local Morse coordinate is kept
fixed. The uniform estimates proved in
\autoref{app:uniform-Morse-estimates} allow us
to apply \autoref{lem:projective-strip} along
this family, yielding common Borel summability
of the formal ratios in
\autoref{prop:vertex-projective-summability}.
Finally, \autoref{lem:projective-reconstruction}
reconstructs the amplitudes and gives the regular
exact wave matrix. Uniqueness of Borel summation
then identifies its restrictions with the
consistently normalised sectorial matrices,
proving the absence of a jump in
\autoref{thm:no-jump-quadrivalent}. The effect of changing the scalar normalisations
is described in \autoref{sec:exact-normalisation}.


\begin{proposition}[Formal regularity]
\label{prop:formal-regularity-quadrivalent}
Let $\tau_\ast\in\widetilde V$ and set $v_\ast
       \coloneqq
       \bigl(\tau_\ast,-6\tau_\ast^2\bigr)
       \in\mathbb C_x\times V$.
Thus $v_\ast$ corresponds to the double turning point
$x=x_{\mathrm d}(\tau_\ast)=\tau_\ast$
at time $t=-6\tau_\ast^2$, as defined in
\eqref{eq:spatial-turning-points}; it is the
quadrivalent vertex of the spatial Stokes graph. On a chosen local sheet of the spectral cover, there is a formal fundamental WKB wave matrix of the form
\[
    \hat\Psi^{\mathrm v}
       =\hat E^{\mathrm v}e^{\frac{S}{\hbar}},
    \qquad
    \hat E^{\mathrm v}
       =\sum_{n=0}^{\infty}E_n^{\mathrm v}\hbar^n,
    \qquad
    S=\operatorname{diag}(S_+,S_-).
\]
The phases are prescribed by
\eqref{eq:reconstruction-phases} with basepoint
$b=v_\ast$ and $\kappa_+=\kappa_-=0$:
\[
    S_\varepsilon
       =\varepsilon\bigl(Z-Z(v_\ast)\bigr),
    \qquad \varepsilon\in\{+1,-1\},
\]
where $Z$ is the primitive
\eqref{250716160458lambda} on the chosen sheet.
In particular, $S(v_\ast)=0$.
The first-row normalisations are
\[
    [\hat\Psi^{\mathrm v}]_{1,1}(v_\ast;\hbar)
       =[\hat\Psi^{\mathrm v}]_{1,2}(v_\ast;\hbar)
       =1
\]
coefficientwise in $\hbar$.
All amplitude coefficients $(E_n^{\mathrm v})_{n\geq 0}$ are
holomorphic on a common full neighbourhood of
$v_\ast$, and $\hat E^{\mathrm v}$ is locally uniformly
of factorial type there. Its leading coefficient
$E_0^{\mathrm v}$ is invertible throughout a
sufficiently small such neighbourhood.
\end{proposition}

\begin{proof}
At the vertex $w^2=3\tau_\ast\neq0$. Choose a full polydisc on which
$\tau w\neq0$. With $w_\ast=w(v_\ast)$, the leading amplitude is
\begin{equation*}
 E^{\mathrm v}_0=\left(\frac{w_\ast}{w}\right)^{\frac{1}{2}}
       \begin{pmatrix}1&1\\\frac{w}{2}&-\frac{w}{2}\end{pmatrix},
 \qquad \det E^{\mathrm v}_0=-w_\ast\neq0.
\end{equation*}
Here the square root is chosen to equal $1$ at $v_\ast$. 
The construction of \autoref{241113190220} applies on this entire
polydisc. Its ratio recursion \eqref{eq:formal-r-recursion} divides
by $w$, and its amplitude is the exponential of the closed primitive
\eqref{eq:formal-amplitude-primitive}. Take $\kappa_+=\kappa_-=0$ and
$\hat b_+=\hat b_-=1$.
Then $S(v_\ast)=0$ and the prescribed first-row
normalisations are both equal to $1$.
Holomorphy and the displayed leading determinant follow from
\autoref{250716102028} and factorial bounds follow from
\autoref{250708195116}. The key point is that the formal ratios $\hat{r}_\pm$ are constructed using the auxiliary time equation through the recursion
\eqref{eq:formal-r-recursion}. Its divisor is $4r_{\varepsilon,0}=2\varepsilon w$,
which remains nonzero at $v_\ast$ because
$w(v_\ast)^2=3\tau_\ast\neq0$.
Thus the recursion remains regular at the double
turning point, where the spatial eigenvalue $y$ vanishes.
The holomorphic amplitude reconstruction
\eqref{eq:formal-amplitude-primitive} then preserves
this regularity.
\end{proof}

The main task is now to prove the analytic assertion on the exact WKB matrix. Using the chosen branch
of $\tau^{\frac{1}{2}}$ on $\widetilde V$, set
\begin{equation}
\label{eq:PI-action-coordinate}
    z=z(\tau)
       \coloneqq
       2\widetilde Z_\tau\bigl(w_+(\tau)\bigr)
       =-\frac{48\sqrt3}{5}\tau^{\frac{5}{2}}.
\end{equation}
Changing the square-root branch reverses the sign of $z$.
In this subsection, $Z_\tau(x)$ denotes the local branch
corresponding to $w(\tau)=w_+(\tau)$, so that
$Z_\tau(\tau)=\frac{z}{2}$. The sector $\widetilde V$ has opening $\frac{2\pi}{5}$, and its boundary rays satisfy
$5\arg\tau-2\alpha\in2\pi\mathbb Z$.
Consequently, the map $\tau\mapsto z(\tau)$ is
biholomorphic from $\widetilde V$ onto one of the two
open half-planes bounded by $e^{i\alpha}\mathbb R$.
Put $z_\ast\coloneqq z(\tau_\ast)$ and denote the
image of the chosen time sector by
\[
    \mathbb H_\ast
       \coloneqq z(\widetilde V)
       =
       \left\{
          z\in\mathbb C
          \,\middle|\,
          \operatorname{Im}(e^{-i\alpha}z)\,
          \operatorname{Im}(e^{-i\alpha}z_\ast)>0
       \right\}.
\]
Thus $\mathbb H_\ast$ is the open half-plane bounded
by the straight line $e^{i\alpha}\mathbb R$ and containing $z_\ast$.
Let  $\tau:\mathbb H_\ast\longrightarrow\widetilde V$
be the holomorphic inverse of
\eqref{eq:PI-action-coordinate}.
Choose $0<d<
       \left|\operatorname{Im}
                    (e^{-i\alpha}z_\ast)\right|$,
so that the strip $\overline{H_d(z_\ast)}\subset\mathbb H_\ast$.

\medskip

The following lemma constructs curves along which
$x$ and $t$ vary together while the action difference
from the moving double turning point remains fixed.
Along these curves, the derivative of the action
with respect to $z$ is $\frac12$, even at the double
turning point. This provides the geometric setup
for applying \autoref{lem:projective-strip} in the
proof of \autoref{prop:vertex-projective-summability}.

\begin{lemma}[Coordinates near the moving double turning point]
\label{lem:Morse-family}
With the preceding notation, there exist $\rho>0$
and a jointly holomorphic function
\[
    x:\Delta_\rho\times H_d(z_\ast)
         \longrightarrow\mathbb C,
    \qquad
    \Delta_\rho\coloneqq
       \{\zeta\in\mathbb C:|\zeta|<\rho\},
\]
such that, throughout the product $\Delta_\rho\times H_d(z_\ast)$,
\begin{equation}
\label{eq:Morse-action-coordinate}
    Z_{\tau(z)}(x(\zeta,z))
       =\frac z2+\zeta^2,
    \qquad
    x(0,z)=\tau(z),\quad\forall \, (\zeta,z)\in \Delta_\rho\times H_d(z_\ast) 
\end{equation}
The corresponding holomorphic branch $w(\zeta,z)$ satisfies
\[
    w(\zeta,z)^2=x(\zeta,z)+2\tau(z),
    \qquad
    w(0,z)=w_+(\tau(z)).
\]
It is nowhere zero on
$\Delta_\rho\times H_d(z_\ast)$. Define
\[
    t(z)\coloneqq-6\tau(z)^2,
    \qquad
    \Phi(\zeta,z)\coloneqq
       \bigl(x(\zeta,z),t(z)\bigr).
\]
Then $\Phi$ is a holomorphic map from
$\Delta_\rho\times H_d(z_\ast)$ to
$\mathbb C_x\times V$. It is locally biholomorphic
near $(0,z_\ast)$, and satisfies
$\Phi(0,z_\ast)=v_\ast$.

For each fixed $z$, the value $\zeta=0$ corresponds
to the double turning point $x=\tau(z)$.
Small nonzero values of $\zeta$ give nearby points
distinct from that turning point.
For each fixed $\zeta$, the curve
$z\mapsto\Phi(\zeta,z)$ remains in
$\mathbb C_x\times V$ and satisfies
\[
    Z_{\tau(z)}(x(\zeta,z))
       -Z_{\tau(z)}(\tau(z))=\zeta^2,
    \qquad
    \left.
    \frac{\partial}{\partial z}
       Z_{\tau(z)}(x(\zeta,z))
    \right|_{\zeta\ {\rm fixed}}
       =\frac12.
\]

Finally,
$s(\zeta,z)\coloneqq \frac{x(\zeta,z)}{\tau(z)}$
stays in a fixed small neighbourhood of $1$ for $(\zeta,z)\in \Delta_\rho\times H_d(z_\ast)$. For every $0<\rho'<\rho$, it tends to $1$
uniformly for $|\zeta|\leq\rho'$ at both ends
of the strip $z\in H_d(z_\ast)$, namely as $\operatorname{Re}
       \bigl(e^{-i\alpha}(z-z_\ast)\bigr)
       \longrightarrow\pm\infty$.
\end{lemma}

\begin{proof}
We first construct a local Morse coordinate for the
spatial action. We then choose a common $\zeta$-disk
valid along the entire strip $H_d(z_\ast)$ and verify the geometric
properties of the resulting parametrisation.

\smallskip
\noindent\emph{Step 1. Construction of the Morse coordinate.}
Writing $x=\tau s$, the action formula
\eqref{eq:space-action} gives
\begin{equation}\label{Ztaus} Z_\tau(\tau s)=\tau^{\frac{5}{2}}A(s),
    \qquad
    A(s)\coloneqq\frac45(s-3)(s+2)^{\frac{3}{2}}.
\end{equation}
Here $(s+2)^{\frac{1}{2}}$ is the holomorphic branch near
$s=1$ whose value there is $\sqrt3$ so that $A(1)=-\frac{24}{5}\sqrt{3}$.
Since $A'(s)=2(s-1)(s+2)^{\frac{1}{2}}$, we obtain $A(s)-A(1)=\sqrt3\,(s-1)^2+O((s-1)^3)$.
 Consequently, the quotient
$\frac{A(s)-A(1)}{(s-1)^2}$ extends holomorphically
and is nonzero near $s=1$.
Choose the holomorphic square root of this extended
quotient whose value at $s=1$ is $3^{\frac{1}{4}}$, and set $f(s)\coloneqq
       (s-1)
       \left(
          \frac{A(s)-A(1)}{(s-1)^2}
       \right)^{\frac{1}{2}}$,
where the quotient is understood through its
holomorphic extension at $s=1$.
Then
\begin{equation}
\label{rrrr}
    f(s)^2=A(s)-A(1),
    \qquad
    f(1)=0,
    \qquad
    f'(1)=3^{\frac{1}{4}}.
\end{equation}
The holomorphic inverse-function theorem therefore
gives a local inverse $s=F(\eta)$, defined for
$|\eta|<\eta_0$, with $F(0)=1$ and $F'(0)=3^{-\frac{1}{4}}$. Shrink $\eta_0$ so that $F(\eta)$ stays in a small neighbourhood of $1$ on which the chosen branches
are defined.

\smallskip
\noindent\emph{Step 2. A common domain along the full strip.}
Set $\delta_\ast
       \coloneqq
       \left|\operatorname{Im}
                    (e^{-i\alpha}z_\ast)\right|-d>0$.
For $z\in H_d(z_\ast)$, we have $|z|\geq\delta_\ast$.
By \eqref{eq:PI-action-coordinate}, this implies 
\begin{equation}|\tau(z)|\geq
    \left(\frac{5\delta_\ast}{48\sqrt3}\right)^{\frac{2}{5}}
       \eqqcolon\mu>0.
\end{equation}
Choose the powers of $\tau$ consistently with the
fixed branch of $\tau^{\frac{1}{2}}$, so that
$\left(\tau^{\frac{5}{4}}\right)^2=\tau^{\frac{5}{2}}$.
Choose $\rho>0$ such that
$\rho\mu^{-\frac{5}{4}}<\frac{\eta_0}{2}$, and define
\begin{equation}
\label{eq:explicit-Morse-family}
    s(\zeta,z)
       =F\bigl(\zeta\tau(z)^{-\frac{5}{4}}\bigr),
    \qquad
    x(\zeta,z)=\tau(z)s(\zeta,z).
\end{equation}
These functions are jointly holomorphic on
$\Delta_\rho\times H_d(z_\ast)$, because the
argument of $F$ remains in $|\eta|<\frac{\eta_0}{2}$
throughout this product. Applying \eqref{Ztaus}
gives the first equality below. The remaining equalities
follow from \eqref{rrrr} and
$f(s(\zeta,z))=\zeta\tau(z)^{-\frac54}$,
since $f\circ F=\mathrm{id}$:
\begin{equation*}
    Z_{\tau(z)}(x(\zeta,z))
       -Z_{\tau(z)}(\tau(z))
    =\tau(z)^{\frac{5}{2}}
         \bigl(A(s(\zeta,z))-A(1)\bigr)
    =\tau(z)^{\frac{5}{2}}
         f(s(\zeta,z))^2
      =\zeta^2.
\end{equation*}
Since $Z_{\tau(z)}(\tau(z))=\frac{z}{2}$ and $F(0)=1$,
this proves \eqref{eq:Morse-action-coordinate}.
Differentiating with respect to $z$ at fixed
$\zeta$ gives the stated derivative $\frac12$.

Moreover, the relation $w^2=x+2\tau$ in
\eqref{250706182441}, together with
$x(\zeta,z)=\tau(z)s(\zeta,z)$, gives
$w(\zeta,z)^2=\tau(z)(s(\zeta,z)+2)$.
With the square-root branches fixed above, we therefore have
\[
    w(\zeta,z)
       =\tau(z)^{\frac12}
          \bigl(s(\zeta,z)+2\bigr)^{\frac12}.
\]
This function is holomorphic and nowhere zero on
$\Delta_\rho\times H_d(z_\ast)$, since
$\tau(z)\neq0$ and $s(\zeta,z)+2\neq0$ throughout
this product. Moreover,
$w(0,z)=\sqrt3\,\tau(z)^{\frac{1}{2}}=w_+(\tau(z))$
for every $z\in H_d(z_\ast)$. 
By construction, $\tau(z)\in\widetilde V$, so
$t(z)=-6\tau(z)^2$ belongs to $V$ throughout
the strip.

\smallskip
\noindent\emph{Step 3. Local invertibility and behaviour at the ends.}
At $\zeta=0$, differentiation of
\eqref{eq:explicit-Morse-family} gives
\[
    \partial_\zeta x(0,z)
       =3^{-\frac{1}{4}}\tau(z)^{-\frac{1}{4}}\neq0.
\]
Also, differentiating \eqref{eq:PI-action-coordinate}
and $t=-6\tau^2$ yields
\[
    \partial_z t(z)
       =\frac{-12\tau(z)}
               {-24\sqrt3\,\tau(z)^{\frac{3}{2}}}
       =\frac{1}{2\sqrt3\,\tau(z)^{\frac{1}{2}}}\neq0.
\]
Since $t$ is independent of $\zeta$, the Jacobian
determinant of $\Phi$ at $(0,z_\ast)$ is
\[
    \left[\det
    \begin{pmatrix}
        \partial_\zeta x&\partial_z x\\
        0&\partial_z t
    \end{pmatrix}\right]_{|(\zeta,z)=(0,z_\ast)}
       =
       \partial_\zeta x(0,z_\ast)\,
       \partial_z t(z_\ast)\neq0.
\]
Thus $\Phi$ is locally biholomorphic there, and
$\Phi(0,z_\ast)=v_\ast$.
The action identity also shows that
$x(\zeta,z)=\tau(z)$ holds precisely when $\zeta=0$. Finally, holomorphy of $F$ and $F(0)=1$ give
$|F(\eta)-1|\leq C|\eta|$ for $|\eta|\leq\frac{\eta_0}{2}$.
Hence, for every $0<\rho'<\rho$,
\[
    \sup_{|\zeta|\leq\rho'}
       |s(\zeta,z)-1|
       \leq C\rho'|\tau(z)|^{-\frac{5}{4}}
       =O\left(|z|^{-\frac{1}{2}}\right).
\]
At either end of the strip, $|z|\to\infty$.
This proves the required uniform convergence to $1$.
\end{proof}

\begin{remark}[Role of the moving coordinates $(\zeta,z)$]
Inside a Stokes region, \autoref{lem:spatial-adapted-strips}
uses the spatial action $s=2Z_{\tau(t)}(x)$ as the
differential variable, with $t$ fixed.
Since $\partial_xs=2y\neq0$, this change has a local
holomorphic inverse, which
\autoref{lem:spatial-adapted-strips} extends to a full
spatial strip as required by
\autoref{lem:projective-strip}.
At the double turning point, however, $y=0$:
the spatial action has a critical point, so this
inverse is no longer holomorphic there. The moving coordinates $(\zeta,z)$ of \autoref{lem:Morse-family}
replace these spatial trajectories by curves along
which $x$ and $t$ vary together, with
$\partial_zZ|_\zeta=\frac12$.
The leading matrix of the pulled-back $z$-equation
therefore has the distinct eigenvalues
$\frac12$ and $-\frac12$, including at the double point.
This construction also applies to nearby points
inside the incident Stokes regions.
Its present domain, however, requires $\frac{x}{\tau}$
to remain near $1$ and does not cover an arbitrary
Stokes region.

In both constructions, control on a full unbounded
strip is used to continue the Borel transforms along
the entire summation ray: the characteristic formulas
in the proof of \autoref{lem:uniform-strip-summation}
involve arbitrarily large translations parallel to
that ray. A local holomorphic change of coordinates
alone does not provide this continuation or its
exponential bounds.
\end{remark}

We now pull back the full Lax system along the curves
of \autoref{lem:Morse-family}, using $z$ as the
differential variable and $\zeta$ as a parameter.
The constant action derivative $\frac12$ gives
leading eigenvalues $\pm\frac12$, even at $\zeta=0$.
Together with the uniform coefficient estimates
of \autoref{lem:Morse-uniform-coefficients}, this
allows us to apply \autoref{lem:projective-strip}
and obtain common Borel sums of the formal ratios
$\hat r_\pm$ on a full neighbourhood of the double
turning point.

\begin{proposition}[Common projective summability at the vertex]
\label{prop:vertex-projective-summability}
Fix a summation direction
$\alpha\in\mathbb R/(2\pi\mathbb Z)$ and retain
the time sector, branches, and double turning point $v_\ast=(\tau_\ast,-6\tau_\ast^2)$ used in \autoref{lem:Morse-family}.
For $\varepsilon\in\{+1,-1\}$, let
$\hat r_\varepsilon$ be the formal ratio constructed
in \autoref{241113190220}, with leading coefficient
$\frac{\varepsilon w}{2}$ on the chosen sheet.

There exist a full polydisc
$\Omega\subset\mathbb C_x\times V$ centred at
$v_\ast$ and constants $\delta>0$, $R>0$, common to
both signs, such that $\hat r_\pm$ are uniformly
Borel summable in direction $\alpha$ on $\Omega$,
in the sense of
\autoref{def:locally-uniform-Borel-summability}.
With the convention \eqref{eq:parametric-Borel-sum},
their sums are
\[
    r_\varepsilon(x,t,\hbar)
       \coloneqq
       s_\alpha(\hat r_\varepsilon)(x,t,\hbar)
       =
       \frac{\varepsilon w(x,t)}2
       +\int_{e^{i\alpha}\mathbb R_{\geq0}}
          e^{-\frac{\xi}{\hbar}}
          \varphi_\varepsilon(x,t,\xi)\,
          \mathrm d\xi,
\]
where $\varphi_\varepsilon$ is the jointly holomorphic
continuation of $\mathcal B\hat r_\varepsilon$ to
$\Omega\times\Xi_{\alpha,\delta}$, with at most
exponential growth in $\xi$, uniformly for
$(x,t)\in\Omega$. These sums are jointly holomorphic on
$\Omega\times S_{\alpha,R}$ and admit
$\hat r_\varepsilon$ as their uniform Gevrey-$1$
asymptotic expansions there, in the sense of
\eqref{eq:uniform-Gevrey-convention}.
In particular, the same Borel--Laplace construction
is valid across all four Stokes curves incident
with $v_\ast$.
\end{proposition}

\begin{proof}The proof is done in \autoref{app:uniform-Morse-estimates}.\end{proof}

We now apply \autoref{lem:projective-reconstruction}
to the common Borel sums $r_\pm$ supplied by
\autoref{prop:vertex-projective-summability} to
reconstruct the scalar amplitudes and hence an exact
WKB wave matrix on a full neighbourhood of $v_\ast$.
Uniqueness of Borel summation then identifies this
matrix with the sectorial matrices having the same
formal data, proving the absence of jumps across
the four incident Stokes curves.

\begin{theorem}[Regular summation and absence of a nodal jump]
\label{thm:no-jump-quadrivalent}
Fix a summation direction
$\alpha\in\mathbb R/(2\pi\mathbb Z)$ and a corresponding
time sector $\widetilde V$ as in
\autoref{250224171850}.
Let $V$, $\tau(t)$ and the Borel-summed Painlev\'e
solution $(q,p)$ be as in \autoref{250408141643}.
Choose $\tau_\ast\in\widetilde V$ and set $v_\ast=(\tau_\ast,-6\tau_\ast^2)$. Choose a local holomorphic branch $w$ near $v_\ast$
with $w^2=x+2\tau(t)$, and let $Z$ be the primitive
\eqref{250716160458lambda} on this sheet. Let
\[
    \hat\Psi^{\mathrm v}
       =\hat E^{\mathrm v}e^{\frac{S}{\hbar}},
    \qquad
    S=\operatorname{diag}(S_+,S_-),
    \qquad
    S_\varepsilon
       =\varepsilon\bigl(Z-Z(v_\ast)\bigr),
\]
be the formal matrix of
\autoref{prop:formal-regularity-quadrivalent}.
Write its amplitude as
\[
    \hat E^{\mathrm v}
       =
       \begin{pmatrix}
          \hat a_+^{\mathrm v}&\hat a_-^{\mathrm v}\\
          \hat r_+\hat a_+^{\mathrm v}
             &\hat r_-\hat a_-^{\mathrm v}
       \end{pmatrix},
    \qquad
    \hat a_\pm^{\mathrm v}(v_\ast;\hbar)=1.
\]

There exist a full polydisc
$\Omega\subset\mathbb C_x\times V$ centred at
$v_\ast$ and $R>0$ such that
$\hat E^{\mathrm v}$ is uniformly Borel summable
in direction $\alpha$ on $\Omega$.
With the convention \eqref{eq:parametric-Borel-sum},
set
\[
    E^{\mathrm v}\coloneqq
       s_\alpha(\hat E^{\mathrm v}),
    \qquad
    \Psi^{\mathrm v}\coloneqq E^{\mathrm v}e^{\frac{S}{\hbar}}.
\]
Then
\[
    \Psi^{\mathrm v}:
       \Omega\times S_{\alpha,R}
          \longrightarrow\operatorname{GL}_2(\mathbb C)
\]
is jointly holomorphic and satisfies both equations
of the Lax system \eqref{250408155610_bis}.
Its amplitude $E^{\mathrm v}$ admits
$\hat E^{\mathrm v}$ as its uniform Gevrey-$1$
asymptotic expansion on this product.
It is the unique exact WKB wave matrix with these
formal data, in the sense of
\autoref{def:exact-WKB-wave-matrix}, and its value
at the vertex is
\[
    \Psi^{\mathrm v}(v_\ast,\hbar)
       =
       \begin{pmatrix}
          1&1\\
          r_+(v_\ast,\hbar)&r_-(v_\ast,\hbar)
       \end{pmatrix},
    \qquad
    r_\varepsilon=s_\alpha(\hat r_\varepsilon).
\]

Let $\left(\mathcal U_j\right)_{j\in\{\mathrm A,\mathrm B,
\mathrm C,\mathrm D\}}$, be the four Stokes regions
incident with $v_\ast$, as defined in
\eqref{eq:time-dependent-Stokes-regions}.
Use branches agreeing with the chosen local sheet
and the same ordering of the two columns.
For each $j\in \{A,B,C,D\}$, choose
$b_j\in\Omega\cap\mathcal U_j$ and prescribe in
\autoref{thm:exact-WKB-existence} the normalisation
data
\[
    \kappa_{\varepsilon,j}
       =-S_\varepsilon(b_j),
    \qquad
    \hat b_{\varepsilon,j}(\hbar)
       =\hat a_\varepsilon^{\mathrm v}(b_j;\hbar),
    \qquad \forall\, \varepsilon\in\{+1,-1\}.
\]
These normalisation series are Borel-summable
formal units. Then the resulting four sectorial wave matrices
$(\Psi_{(j)})_{j\in \{A,B,C,D\}}$ provided by \autoref{thm:exact-WKB-existence} admit holomorphic continuations to
$\Omega\times S_{\alpha,R}$, after decreasing $R$
if necessary, and all four continuations equal
$\Psi^{\mathrm v}$.
Consequently, all four local connection matrices
across the Stokes curves incident with $v_\ast$
are equal to $I_2$.
\end{theorem}

\begin{proof}
\noindent\emph{Proof of the first part.}
By \autoref{prop:vertex-projective-summability},
the formal ratios $\hat r_\pm$ are uniformly
Borel summable in direction $\alpha$ on a full
neighbourhood of $v_\ast$.
Apply \autoref{lem:projective-reconstruction}
with basepoint $v_\ast$ and normalisation data
$\hat b_+=\hat b_-=1$ and 
    $\kappa_+=\kappa_-=0$.
The reconstructed formal matrix is
$\hat\Psi^{\mathrm v}$ by uniqueness in
\autoref{241113190220}.
After shrinking to a full polydisc $\Omega$
centred at $v_\ast$, the reconstruction lemma
therefore gives uniform Borel summability of
$\hat E^{\mathrm v}$ and a common Borel disk
$S_{\alpha,R}$ on which $\Psi^{\mathrm v}
       =s_\alpha(\hat E^{\mathrm v})e^{\frac{S}{\hbar}}$
is jointly holomorphic, invertible, and satisfies
both equations of \eqref{250408155610_bis}.
It also gives the asserted uniform Gevrey-$1$
expansion of its amplitude and uniqueness with
the prescribed formal data.

Since evaluation at $v_\ast$ commutes with
Borel--Laplace summation,
\[
    a_\varepsilon^{\mathrm v}(v_\ast,\hbar)
       =s_\alpha\bigl(
          \hat a_\varepsilon^{\mathrm v}
                       (v_\ast;\hbar)\bigr)
       =s_\alpha(1)=1.
\]
Together with $S_\varepsilon(v_\ast)=0$,
this proves that both first-row entries of
$\Psi^{\mathrm v}$ at $v_\ast$ equal $1$.

\smallskip
\noindent\emph{Proof of the second part.}
Fix $j\in\{\mathrm A,\mathrm B,\mathrm C,\mathrm D\}$
and choose $b_j\in\Omega\cap\mathcal U_j$.
By the first part,
$\hat a_\varepsilon^{\mathrm v}(b_j;\hbar)$
is a Borel-summable formal unit, so the prescribed
normalisation data are admissible in
\autoref{thm:exact-WKB-existence}. The phase assigned by that theorem is
\begin{equation*}
    \varepsilon\bigl(Z-Z(b_j)\bigr)
       -\kappa_{\varepsilon,j}
    =\varepsilon\bigl(Z-Z(b_j)\bigr)
       +S_\varepsilon(b_j)
    =\varepsilon\bigl(Z-Z(v_\ast)\bigr)
      =S_\varepsilon.
\end{equation*}
Moreover, the formal first-component normalisation
at $b_j$ is $\hat b_{\varepsilon,j}(\hbar)
       e^{-\frac{\kappa_{\varepsilon,j}}{\hbar}}
       =
       \hat a_\varepsilon^{\mathrm v}(b_j;\hbar)
       e^{\frac{S_\varepsilon(b_j)}{\hbar}}$,
which is precisely the first component of the
corresponding column of $\hat\Psi^{\mathrm v}$
at $b_j$.

The branches and column ordering also agree.
Uniqueness in \autoref{241113190220} therefore
identifies the formal matrix underlying
$\Psi_{(j)}$ with $\hat\Psi^{\mathrm v}$
on a neighbourhood of $b_j$.
Their exact amplitudes are Borel sums of the
same formal amplitude in direction $\alpha$;
hence $\Psi_{(j)}=\Psi^{\mathrm v}$
on that neighbourhood and a common sufficiently
small Borel disk. The matrix $\Psi^{\mathrm v}$ thus provides a
holomorphic continuation of the local restriction
of $\Psi_{(j)}$ to the full product
$\Omega\times S_{\alpha,R}$.
There are only four regions, so decreasing $R$
once makes this conclusion valid for all four
matrices simultaneously.
Their continuations consequently coincide. For adjacent regions, the connection matrix on
the common continuation domain is therefore
\[
    C_{jk}
       =(\Psi_{(j)})^{-1}\Psi_{(k)}
       =(\Psi^{\mathrm v})^{-1}\Psi^{\mathrm v}
       =I_2.
\]
More generally, the two Lax equations imply
$\partial_xC_{jk}=\partial_tC_{jk}=0$.
Thus this identity persists throughout any
connected overlap reached by consistent analytic
continuation.
\end{proof}

\begin{remark}
Formal regularity alone would not exclude exponentially small
differences. Here the common Borel sum is proved using the moving family
\eqref{eq:Morse-action-coordinate}, the strip-summation theorem
proved in \autoref{app:uniform-strip-proof}, and the uniform
estimates established in \autoref{app:uniform-Morse-estimates}. No fixed-$x$ auxiliary trajectory is assumed
to stay in $V$, and no characteristic end is assumed to remain in a
closed subsector of $\widetilde V$.
\end{remark}

\section{The WKB Riemann--Hilbert problem}
\label{sec:WKB-RHP}

We first describe the connection and jump matrices between the exact WKB
wave matrices and the effect of changing their normalisations in \autoref{sec:connection-matrices}.
In \autoref{sec:exact-normalisation}, we use the recessive behaviour of their columns at spatial infinity to establish the triangular form of the jump matrices. We then impose the cyclic relations
at the two vertices of the Stokes graph to reduce the number of Stokes constants in the jump matrices to two. The diagrams in \autoref{fig:RH-cut-folded-general} and 
\autoref{fig:RH-branch-cut-positions} display the two remaining coefficients and the effect of moving the branch cut. They also distinguish these coefficients from the conventional relative column scale $a$, which may be fixed to $i$ to recover the convention of \cite{Kapaev_2004}. Finally, we use \autoref{thm:no-jump-quadrivalent} to prove that these Stokes constants vanish. After fixing the normalisation at infinity, we state the resulting
Riemann--Hilbert problem in the original Lax gauge in \autoref{sec:WKB-RHP-problem}. The comparison with the standard
Painlev\'e~I Riemann--Hilbert problem is made in \autoref{sec:canonical-RHP}.

\subsection{Connection and jump matrices}
\label{sec:connection-matrices}

Fix the summation direction $\alpha$ and a time sector
$\widetilde V$ as in \autoref{250224171850}. Let $V$, $\tau(t)$
and $(q,p)$ be as in \autoref{250408141643}.
Let $\mathcal U_j\subset\mathbb C_x\times V$,
$j\in\{\mathrm A,\ldots,\mathrm F\}$, be the six
time-dependent Stokes regions defined in
\eqref{eq:time-dependent-Stokes-regions}.
Choose branches,
ordered phases and Borel-summable scalar normalisations in each
Stokes region. By \autoref{thm:exact-WKB-existence}, these data define
\begin{equation}
\label{eq:RH-sectorial-matrices}
    \Psi_{(j)}^{\mathrm{WKB}}
       =s_\alpha(\hat E_{(j)})e^{\frac{S_{(j)}}{\hbar}},
    \qquad
    S_{(j)}=\operatorname{diag}(S_{+,j},S_{-,j}).
\end{equation}
Here $s_\alpha$ uses the convention
\eqref{eq:parametric-Borel-sum}. Initially these are the exact WKB matrices of
\autoref{def:exact-WKB-wave-matrix}. We shall also
rescale their columns by writing
\begin{equation}\Psi_{(j)}
       =\Psi_{(j)}^{\mathrm{WKB}}C_j(\hbar),
    \qquad
    C_j(\hbar)
       =\operatorname{diag}
          \bigl(c_{+,j}(\hbar),c_{-,j}(\hbar)\bigr),
\end{equation}
where the scalar factors $\left(c_{+,j}(\hbar),c_{-,j}(\hbar)\right)$ are holomorphic and nowhere
zero on the Borel disk under consideration and are
independent of $(x,t)$. Such a rescaling preserves both
Lax equations, invertibility, and the ordered column
directions. At this stage, the diagonal factors $C_j(\hbar)$
are not fixed. We use $\Psi_{(j)}$ to denote a
fundamental matrix obtained from
$\Psi_{(j)}^{\mathrm{WKB}}$ by such a diagonal
rescaling. The choices of these factors will be
specified below, first through normalisation along
a chain of regions and then through the canonical
normalisation at infinity. The additional factors
are assumed holomorphic and nowhere zero on the
common Borel disk, but are not assumed to be Borel
sums of formal units.

All constructions are local in the time variable.
Fix a simply connected relatively compact open
subset $V_0\Subset V$. Choose $R_{V_0}>0$ sufficiently
small that $(q,p)$ are holomorphic on
$V_0\times S_{\alpha,R_{V_0}}$ and the finitely many
local exact WKB matrices under consideration are
holomorphic on their respective $(x,t)$-neighbourhoods
in $\mathbb C_x\times V_0$, times the same Borel disk.
Set
\[
    \mathcal P_{V_0}
       \coloneqq V_0\times S_{\alpha,R_{V_0}}.
\]
For each $(t,\hbar)\in\mathcal P_{V_0}$, every
$\Psi_{(j)}$ continues to an entire fundamental
matrix in $x$, because the spatial Lax matrix
is polynomial in $x$. Flatness and the simple
connectedness of $V_0$ give compatible single-valued
continuation in $t$. These continuations are jointly
holomorphic on $\mathbb C_x\times\mathcal P_{V_0}$
and satisfy both Lax equations.
They permit comparisons even when the open Stokes
regions themselves are disjoint.
The parameter $R_{V_0}$ may depend on $V_0$ and
may be decreased in subsequent constructions.

For each fixed $(t,\hbar)\in\mathcal P_{V_0}$,
the entire continuations of the matrices
$\Psi_{(j)}$ are fundamental matrices of the same
spatial linear system on $\mathbb C_x$. Consequently, for every pair
$(i,j)\in\{\mathrm A,\ldots,\mathrm F\}^2$, there is a
unique invertible \emph{connection matrix}, defined by
\begin{equation}
\label{eq:jump-convention}
    \Psi_{(i)}=\Psi_{(j)}S_{ij}(\hbar),
    \qquad
    S_{ij}(\hbar)=(\Psi_{(j)})^{-1}\Psi_{(i)}.
\end{equation}
The spatial equation implies that $S_{ij}$ is
independent of $x$; see
\cite[Chapter~3, Theorem~2.3, p.~70]
{CoddingtonLevinson1955}.
The time equation also implies independence of $t$,
as verified in \autoref{lem:connection-matrices-constant}. These connection matrices satisfy
\begin{equation}
\label{eq:connection-cocycle}
    S_{ii}=I,
    \qquad
    S_{ji}=(S_{ij})^{-1},
    \qquad
    S_{ik}=S_{jk}S_{ij}.
\end{equation}
Indeed,
$\Psi_{(i)}=\Psi_{(j)}S_{ij}
=\Psi_{(k)}S_{jk}S_{ij}$.
Thus the ordered product along any closed cycle is
the identity: for example, $S_{ki}S_{jk}S_{ij}=I$. For adjacent regions, these connection matrices
define the \emph{jump matrices} of the piecewise
wave matrix across their common boundary, with
the directed convention \eqref{eq:jump-convention}.
Since the adjacency graph is connected, it suffices
to specify these jump matrices: along any chain
of adjacent regions $i=i_0,i_1,\ldots,i_n=j$, $S_{ij}
       =S_{i_{n-1}i_n}\cdots S_{i_1i_2}S_{i_0i_1}$.
The cycle relations ensure independence of the
chosen chain. Henceforth, we restrict our attention
to the connection matrices between adjacent regions, i.e. to jump matrices only.

\begin{lemma}[Connection matrices and changes of normalisation]
\label{lem:connection-matrices-constant}
The matrices $\left(S_{ij}\right)_{(i,j)\in \{A,\dots,F\}^2}$ are holomorphic in $\hbar$ and independent
of $(x,t)$. Replacing $\Psi_{(j)}$ by $\Psi_{(j)}C_j(\hbar)$, where
$C_j$ is holomorphic in $\hbar$, invertible on the common
Borel disk $S_{\alpha,R_{V_0}}$ and independent of $(x,t)$, gives
\begin{equation}
\label{eq:RH-change-normalisation}
    S_{ij}\longmapsto C_j^{-1}S_{ij}C_i.
\end{equation}
In particular, a common diagonal change conjugates every jump matrix.
At a basepoint $b$, prescribing  $\Psi_{(j)}(b,\hbar)=M(\hbar)$, with $M(\hbar)$ a holomorphic, invertible matrix on the common
Borel disk $S_{\alpha,R_{V_0}}$ and independent of $(x,t)$, amounts to the right change of basis
$C_j(\hbar)=\Psi_{(j)}(b,\hbar)^{-1}M(\hbar)$.
\end{lemma}

\begin{proof}
Differentiate $(\Psi_{(j)})^{-1}\Psi_{(i)}$ using the two Lax
equations. For example,
\[
    \hbar\partial_xS_{ij}\overset{\eqref{eq:jump-convention}}{=}\hbar\partial_x((\Psi_{(j)})^{-1}\Psi_{(i)})
      =-(\Psi_{(j)})^{-1}L\Psi_{(i)}
        +(\Psi_{(j)})^{-1}L\Psi_{(i)}=0.
\]
The time derivative vanishes in the same way. The other assertions
follow by substitution.
\end{proof}

A general right change of basis need not preserve the two ordered
WKB column directions. Changes between the exact WKB matrices
with fixed ordered phases and fixed column orders are diagonal.
More precisely, let $\mathcal A_\alpha^\times$ be the multiplicative
group of Borel sums in direction $\alpha$ of formal units in
$\mathbb C[[\hbar]]$. For two such matrices in the same region,
formal uniqueness and multiplicativity of Borel summation give
\begin{equation}
\label{eq:exact-change-normalisation}
    \hat\Psi_1=\hat\Psi_2\hat C,
    \qquad
    \Psi_1=\Psi_2C,
    \qquad
    C=s_\alpha(\hat C)
       \in\operatorname{diag}
         (\mathcal A_\alpha^\times,\mathcal A_\alpha^\times).
\end{equation}
Indeed, the corresponding formal ratios coincide, and the
quotients of the scalar amplitudes are formal constants, obtained
by evaluation at a common basepoint. Constant exponential factors
are included in the phases. The preceding discussion describes the connection
matrices and their dependence on the choice of bases.
To formulate a normalised Riemann--Hilbert problem,
we must also prescribe the asymptotic behaviour of
the wave matrix at spatial infinity, i.e. $t$ and
$\hbar$ are fixed, while $x\to\infty$.
This normalisation at infinity fixes the column scales by
requiring the prescribed sectorial asymptotic
expansion, given below in
\eqref{eq:canonical-normalisation-infinity}.

\subsection{Triangular jump matrices and normalisation at \texorpdfstring{$x\to\infty$}{infinity}}
\label{sec:exact-normalisation}
Before imposing the normalisation at $x\to \infty$, we identify the
solution directions selected by decay at the ends
of the spatial action strips. This will explain
the triangular form of the jump matrices and subsequently allow comparison with
the canonical wave matrices normalised at infinity.
Matching the column scales may involve diagonal
factors holomorphic in $\hbar$ on the common Borel
disk $S_{\alpha,R_{V_0}}$. We do not assume that these additional
factors are themselves Borel sums of formal units. We first work in the representative case
\begin{equation}
\label{eq:RH-representative-sector}
    \alpha=0,
    \qquad -\frac{2\pi}{5}<\arg\tau<0,
    \qquad \frac{\pi}{5}<\arg t<\pi,
    \qquad \tau=\sqrt{e^{-i\pi}\frac{t}{6}}.
\end{equation}
Decay arguments are first made for sufficiently small positive
$\hbar$. Identities between the resulting connection matrices
then extend to their common Borel disk by holomorphy. The weighted
rotations at the end of the section recover the other cases.

\begin{definition}[Recessive solutions]\label{DefRecessivity}For fixed $t$ and sufficiently small positive $\hbar$, a nonzero solution column $\mathbf v$ of the Lax system \eqref{250408155610} is called \emph{recessive} at a specified end at $x\to\infty$ if for every solution column $\mathbf u$ of the Lax system \eqref{250408155610} that is linearly independent of $\mathbf v$ we have: 
$\frac{\|\mathbf v(x,t,\hbar)\|}{\|\mathbf u(x,t,\hbar)\|}\overset{x\to\infty}{\longrightarrow} 0$ along the trajectory approaching that end. Thus a recessive solution, when it exists, determines a unique one-dimensional solution
space, called the \emph{recessive line}. Its scalar normalisation remains free.
\end{definition}

On the action strips considered below, recessivity will follow from exponential decay towards the chosen end, whereas every independent solution
has an exponentially growing component. \autoref{lem:RH-recessive-comparison} identifies these recessive lines with the WKB column directions. When two adjacent matrices share such a line in the same column position, their jump matrix is triangular. This determines the triangular type before any choice of column scales.

\begin{lemma}[Recessive columns and triangular connections]
\label{lem:RH-recessive-comparison}
Fix $(x_0,t_0)\in\mathcal U_j$ with $t_0\in V_0$,
and use the chosen sheet of the action
$Z_{\tau(t)}(x)$ defined in \eqref{eq:space-action}.
Set $s_0\coloneqq 2Z_{\tau(t_0)}(x_0)$. Let $d>0$, a time neighbourhood $V_1\Subset V_0$ of $t_0$, and the jointly holomorphic inverse branch
 $x:H_d(s_0)\times V_1\longrightarrow\mathbb C$ be supplied by \autoref{lem:spatial-adapted-strips}. Since $\alpha=0$, the action strip is
\[
    H_d(s_0)
       \coloneqq
       \left\{
          s\in\mathbb C:
          |\operatorname{Im}(s-s_0)|<d
       \right\},
    \qquad
    2Z_{\tau(t)}(x(s,t))=s.
\]
Write the exact WKB matrix in the chosen column order as
\[
    \Psi_{(j)}^{\mathrm{WKB}}
       =\bigl[\mathbf v_+\ \mathbf v_-\bigr],
    \qquad
    \mathbf v_+
       \coloneqq
       \begin{pmatrix}
          [\Psi_{(j)}^{\mathrm{WKB}}]_{1,1}\\
          [\Psi_{(j)}^{\mathrm{WKB}}]_{2,1}
       \end{pmatrix},
    \qquad
    \mathbf v_-
       \coloneqq
       \begin{pmatrix}
          [\Psi_{(j)}^{\mathrm{WKB}}]_{1,2}\\
          [\Psi_{(j)}^{\mathrm{WKB}}]_{2,2}
       \end{pmatrix}.
\]
We call $\mathbf v_+$ and $\mathbf v_-$ the plus
and minus WKB columns, respectively, and use the same
notation for their analytic continuations.

After decreasing $d$ and shrinking $V_1$ if necessary,
for each fixed $t\in V_1$ and sufficiently small
positive $\hbar$, the column
$\mathbf v_+(x(s_0-u,t),t,\hbar)$ is recessive
as $u\to+\infty$, whereas
$\mathbf v_-(x(s_0+u,t),t,\hbar)$ is recessive
as $u\to+\infty$.
The same conclusions hold with $s_0$ replaced by
any fixed point of $H_d(s_0)$.
Each column spans the unique recessive solution line
at its respective end. Consequently, adjacent
sectorial wave matrices have triangular jump matrices
when their common recessive column occupies the same
column position.
\end{lemma}

\begin{proof}
Let $r_\varepsilon
       \coloneqq s_\alpha(\hat r_\varepsilon)$,  $\varepsilon\in\{+1,-1\}$,
be the Borel-summed projective solutions constructed in
\autoref{prop:sectorial-projective-summability}.
In the eigenbasis $E$ defined in
\eqref{eq:regular-eigenbasis}, the corresponding
component ratios are denoted by $g_\varepsilon$.
The exact counterpart of
\eqref{eq:projective-back-transform} gives
\[
    r_\varepsilon(x(s,t),t,\hbar)
       =\frac{\varepsilon w(s,t)}{2}
          \frac{1-g_\varepsilon(s,t,\hbar)}
               {1+g_\varepsilon(s,t,\hbar)}.
\]
By the estimates established in the proof of
\autoref{prop:sectorial-projective-summability},
the transformed spatial system \eqref{EqdiffMs}
and its component ratios $g_\pm$ satisfy,
uniformly on $H_d(s_0)\times V_1$ after the
reductions allowed in the statement,
\[
    M_x=D+B,\qquad
    D=\operatorname{diag}\left(\frac{1}{2},-\frac{1}{2}\right),\qquad
        \|B\|\leq C|\hbar|, \qquad |g_+|+|g_-|\leq C|\hbar|.
\]
Using the analytic continuation in $x$ established
in \autoref{sec:connection-matrices}, define, for
each $\varepsilon\in\{+1,-1\}$,
\[
    \chi_\varepsilon(s,t,\hbar)
       \coloneqq
       E(s,t)^{-1}
       \mathbf v_\varepsilon(x(s,t),t,\hbar),
    \qquad s\in H_d(s_0).
\]
For fixed $(t,\hbar)$, this is holomorphic in $s$
throughout $H_d(s_0)$, since $\mathbf v_\varepsilon$
is entire in $x$ and both $x(s,t)$ and $E(s,t)^{-1}$
are holomorphic on the strip. Near $s=s_0$, these columns have the component
ratios $g_\pm$ constructed in
\autoref{prop:sectorial-projective-summability}.
The identities $\chi_{+,2}=g_+\chi_{+,1}$ and $\chi_{-,1}=g_-\chi_{-,2}$
therefore extend throughout the connected strip.
Setting $a_+\coloneqq\chi_{+,1}$ and
$a_-\coloneqq\chi_{-,2}$, and omitting the
dependence on $(t,\hbar)$, the two columns have
the forms
\[
    \chi_+=a_+(s)\binom{1}{g_+(s)},\qquad
    \chi_-=a_-(s)\binom{g_-(s)}{1}.
\]
The factors $a_\pm$ are nowhere zero: otherwise
the corresponding column $\chi_\pm$, and hence
$\mathbf v_\pm$, would vanish at that point,
contradicting the invertibility of the fundamental
matrix.
Taking the first row of \eqref{EqdiffMs} for
$\chi_+=a_+(1,g_+)^{\mathsf T}$ and the second row
for $\chi_-=a_-(g_-,1)^{\mathsf T}$, with $M_x=D+B$,
and dividing by $\hbar a_+$ and $\hbar a_-$,
respectively, gives
\[
    \partial_s\log a_+
       =\frac{\frac{1}{2}+B_{11}+B_{12}g_+}{\hbar}
       =\frac1{2\hbar}+O(1),
    \qquad
    \partial_s\log a_-
       =\frac{-\frac{1}{2}+B_{22}+B_{21}g_-}{\hbar}
       =-\frac1{2\hbar}+O(1).
\]
Let $C>0$ uniformly bound the two $O(1)$ terms above.
For sufficiently small $0<\hbar\leq(4C)^{-1}$:
\[
    \frac{\mathrm d}{\mathrm du}
       \log|a_\pm(s_0\mp u)|
       \leq-\frac1{2\hbar}+C
       \leq-\frac1{4\hbar}\,\qquad \forall\, u\geq 0.
\]
Integrating from $0$ to $u$ and exponentiating gives
\[
    |a_\pm(s_0\mp u)|
       \leq C_{t,\hbar}e^{-\frac{u}{4\hbar}},
    \qquad
    C_{t,\hbar}\coloneqq
       \max_{\varepsilon\in\{+,-\}}
          |a_\varepsilon(s_0,t,\hbar)|.
\]
Here $t$ and $\hbar$ are fixed, and their appearance
as arguments of $a_\pm$ has been omitted.
By \autoref{lem:spatial-adapted-strips} and
\eqref{eq:regular-eigenbasis}, the eigenbasis $E$
grows at most polynomially as $u\to+\infty$
along the two trajectories $s=s_0\mp u$.
The original columns therefore also decay
exponentially. The images of these trajectories in the $x$-plane
approach non-oscillatory directions at infinity.
On a sector containing the tail of either
trajectory, the canonical asymptotic theory
provides a fundamental pair with one exponentially
decaying column and one exponentially growing
column; see \cite[Section~2]{Kapaev_2004}.
Expressing the continued WKB column in this pair
shows that its growing coefficient is zero.
Every linearly independent solution has a
nonzero growing coefficient. This proves
recessivity in the sense of
\autoref{DefRecessivity} and uniqueness of the
recessive line. These conclusions also apply to the diagonally
rescaled matrices. Indeed,
\[
    \Psi_{(j)}
       =\Psi_{(j)}^{\mathrm{WKB}}C_j
       =\bigl[
           c_{+,j}(\hbar)\mathbf v_+
           \quad
           c_{-,j}(\hbar)\mathbf v_-
         \bigr].
\]
For fixed $\hbar$, the factors $c_{\pm,j}(\hbar)$
are nonzero and independent of $x$. Multiplication
by such a factor preserves recessivity in the sense
of \autoref{DefRecessivity}. Thus the corresponding
columns of $\Psi_{(j)}$ span the same recessive lines
as those of $\Psi_{(j)}^{\mathrm{WKB}}$.

Fix $t$ and sufficiently small positive $\hbar$.
On the two sides of a separating Stokes curve,
nearby spatial trajectories approach the same
non-oscillatory direction at infinity, with the
sheet choices matched consistently. This follows
by continuing the inverse branches of the action
near infinity, where there are no critical points.
Their tails lie in a common closed subsector of
a canonical asymptotic sector. The same canonical fundamental pair, consisting
of an exponentially decaying column and an
exponentially growing column, is therefore valid
along both trajectories. Each continued WKB
column is a linear combination of this pair,
with coefficients independent of $x$.
Its exponential decay forces the coefficient
of the growing column to vanish. The two WKB
columns consequently span the same recessive
line and are proportional. If these are the first columns of the two
fundamental matrices, their connection matrix
is upper triangular; if they are the second
columns, it is lower triangular.
The vanishing entry is holomorphic in $\hbar$.
Since it vanishes for all sufficiently small
positive $\hbar$, the identity theorem extends
its vanishing to the common Borel disk.
\end{proof}

The same argument shows that two WKB matrices in one region,
with the same ordered column directions, differ by a diagonal
matrix. It also explains the qualification concerning the branch
cut: going around the simple turning point exchanges the two
sheets, so the column positions must then be compared after
the corresponding permutation. Our labelling of the Stokes regions gives the following limiting directions at infinity for the Stokes curves.

\begin{proposition}[Labeling at $x\to\infty$]\label{PropLabelling}
We may take the absolute labeling of ends at $x\to \infty$:    
\begin{equation}
\label{eq:RH-end-labelling}
\begin{array}{c|ccccc}
 \text{Stokes curves}&\mathrm{FD},\mathrm{DA}&\mathrm{AB}&\mathrm{BC}
                 &\mathrm{CD},\mathrm{DE}&\mathrm{EF}\\ \hline
 \arg x\text{ at infinity}&-\frac{4\pi}{5}&-\frac{2\pi}{5}&0&\frac{2\pi}{5}&\frac{4\pi}{5}.
\end{array}
\end{equation}
\end{proposition}
\begin{proof}
    We shall verify that the indexing \eqref{eq:RH-end-labelling} is correct. At $\tau=r e^{-\frac{i\pi}{5}}$, put
$w=\tau^{\frac{1}{2}}\eta$. Then
\[
    s=-i r^{\frac{5}{2}}F(\eta),\qquad
    F(\eta)=\frac85\eta^5-8\eta^3,
    \qquad c=\frac{48\sqrt3}{5}.
\]
The two roots issuing from $\sqrt3$ along
$F(\eta)=-c+iu$, $u>0$, remain in the right half-plane, since
$F(i\mathbb R)\subset i\mathbb R$. The other three roots start
in the left half-plane: indeed,
\[
 F(\eta)+c=\frac85(\eta-\sqrt3)^2
       (\eta^3+2\sqrt3\eta^2+4\eta+2\sqrt3),
\]
and the cubic satisfies the strict Hurwitz inequalities.
For $u>0$, let $R=R(u)>0$ be determined by
$\frac85R^5+8R^3=u$, so that $F(iR)=iu$.
The root approaching the positive imaginary axis
satisfies
\[
    \eta=iR-\frac{c}{F'(iR)}+O(R^{-9}),
    \qquad
    F'(iR)=8R^2(R^2+3)>0.
\]
It therefore lies in the left half-plane for
sufficiently large $u$. Consequently, the two
right-half-plane roots have limiting arguments
$\frac{\pi}{10}$ and $-\frac{3\pi}{10}$.
Complex conjugation treats $u<0$.
Projection by $x=\tau(\eta^2-2)$ gives the four nodal directions
$0,-\frac{4\pi}{5},-\frac{2\pi}{5},\frac{2\pi}{5}$. At the simple vertex, the imaginary
$\eta$-axis projects to $\frac{4\pi}{5}$. The branch of $F(\eta)=iu$
issuing into the fourth quadrant remains there and approaches
$-\frac{3\pi}{10}$; its conjugate gives the remaining branch. Their
projections have arguments $-\frac{4\pi}{5}$ and $\frac{2\pi}{5}$.
The region labeling in \eqref{eq:time-dependent-Stokes-regions}
now gives \eqref{eq:RH-end-labelling}. It persists throughout
the time sector, since no critical trajectories connect distinct
vertices there.
\end{proof}

The region labels and the cyclic ordering of these
ends are illustrated in \autoref{fig:RH-cut-folded-general} and 
\autoref{fig:RH-branch-cut-positions}. Note that the values of $\arg x$
in \eqref{eq:RH-end-labelling} specify the limiting
directions of the curves as $|x|\to\infty$, rather
than their angles at the finite vertices. The position of the branch cut joining the simple
turning point $x=-2\tau$ to infinity is a convention.
It may be moved, provided the sheet choices and
column orderings are adjusted consistently.
\autoref{fig:RH-branch-cut-positions} illustrates
two such choices, with the cut passing through
$\mathrm D$ or $\mathrm E$.
They give different presentations of the same
connection data; the corresponding transformations
of the jump matrices are described below.
Unless otherwise stated, we use the convention in
which the branch cut is folded onto $\mathrm{FD}$. 

\bigskip

We can use the last results to propagate normalisation along a chain of Stokes regions.

\begin{lemma}[Normalisation along a chain of Stokes regions]
\label{lem:RH-chain-normalisation}
Consider two adjacent Stokes regions whose common recessive
line occupies the same column position. Their
connection matrix is triangular by
\autoref{lem:RH-recessive-comparison}.
Fix the normalisation of one sectorial wave matrix.
Then, there is a unique diagonal rescaling of the wave
matrix in the other Stokes region which makes their
jump matrix unipotent (i.e. with $1$ on the diagonal). This normalisation can be
propagated along any chain of such adjacent Stokes regions
without closing a cycle.
\end{lemma}

\begin{proof}
If $\Psi_{(i)}=\Psi_{(j)}T$ with
$T=\left(\begin{smallmatrix}a&b\\0&d\end{smallmatrix}\right)$,
then $a,d\neq0$. Replace $\Psi_{(j)}$ by
$\Psi_{(j)}\operatorname{diag}(a,d)$. The new connection is
$\left(\begin{smallmatrix}1&\frac{b}{a}\\0&1\end{smallmatrix}\right)$.
The diagonal entries determine the rescaling uniquely. The
lower-triangular case is identical. Iteration proves the chain
statement. Closing a cycle requires the additional relations
examined next.
\end{proof}

\begin{remark}On a closed cycle, the cyclic identity determines
the closing jump matrix. The unipotent normalisation extends around the cycle precisely when the diagonal entries of this closing matrix
are both equal to $1$ but this is not implied by the cyclic identity alone.
\end{remark}

\subsection{Cyclic relations and determination of the jump matrices}
\label{sec:RH-cyclic-relations}

We now combine the cyclic identities at the two finite
vertices with the normalisation along chains established
in \autoref{lem:RH-chain-normalisation}. This gives a
standard algebraic form of the jump matrices, with two
undetermined nodal coefficients and a free relative conventional
column scale. We then apply
\autoref{thm:no-jump-quadrivalent} to determine the
nodal coefficients and impose the normalisation at infinity.

\paragraph{Cyclic identities, normalisations and triangular jump matrices.} For any choice of normalisations, the directed convention
\eqref{eq:jump-convention} gives the cyclic identities
\begin{equation}
\label{eq:RH-vertex-cycles}
 S_{\mathrm{DA}}S_{\mathrm{CD}}S_{\mathrm{BC}}
          S_{\mathrm{AB}}=I,
 \qquad
 J_{\mathrm{FD}}S_{\mathrm{EF}}S_{\mathrm{DE}}=I.
\end{equation}
Here $J_{\mathrm{FD}}$ denotes the actual connection matrix
between the chosen matrices in $\mathrm F$ and $\mathrm D$,
including the branch-cut factor:
$\Psi_{(\mathrm F)}=\Psi_{(\mathrm D)}J_{\mathrm{FD}}$.
These identities follow from
\eqref{eq:connection-cocycle} and do not require the
individual jump matrices to be unipotent.

\bigskip

From \autoref{lem:RH-recessive-comparison} and \eqref{eq:RH-end-labelling}, choose the column ordering for which the jump matrices $S_{AB}$ and $S_{CD}$ are upper triangular, whereas $S_{BC}$ and $S_{DA}$ are lower triangular.
Their diagonal entries are not yet prescribed. Using \autoref{lem:RH-chain-normalisation} along the chain $\mathrm A,\mathrm B,\mathrm C,\mathrm D$, we may normalize them so that $S_{AB}$, $S_{BC}$ and $S_{CD}$ are unipotent. The closing jump matrix $S_{\mathrm{DA}}$ is then determined by the cyclic
identity and is not unipotent a priori. For compactness, we introduce the notations
\begin{equation}
\label{eq:ULP-matrices}
 \operatorname{Up}(b)
    \coloneqq\begin{pmatrix}1&b\\0&1\end{pmatrix},\quad
 \operatorname{Lw}(b)
    \coloneqq\begin{pmatrix}1&0\\b&1\end{pmatrix},
 \quad
 \operatorname{Ex}(a)
    \coloneqq\begin{pmatrix}0&a\\-a^{-1}&0\end{pmatrix}\,\,\Rightarrow\,\,
 \operatorname{Ex}(a)^{-1}
    =\begin{pmatrix}0&-a\\a^{-1}&0\end{pmatrix},
\end{equation}
where $b\in\mathbb C$ and $a\in\mathbb C^\ast$.
The notation $\operatorname{Ex}(a)$ denotes a scaled
sheet-exchange matrix.

\begin{proposition}[The nodal connection coefficients]
\label{prop:RH-nodal-cyclic-data}
With the column ordering just specified, there exist
diagonal normalisations of the sectorial fundamental
matrices for which
\begin{equation}
\label{eq:RH-nodal-triangular-data}
\begin{aligned}
 S_{\mathrm{AB}}
    &=\operatorname{Up}(s_+),
 &S_{\mathrm{BC}}
    &=\operatorname{Lw}(s_-),
 \\
 S_{\mathrm{CD}}
    &=\operatorname{Up}\left(-\frac{s_+}{\delta}\right),
 &S_{\mathrm{DA}}
    &=\begin{pmatrix}
        \delta&0\\
        -s_-&\delta^{-1}
      \end{pmatrix},
\end{aligned}
\qquad
\delta\coloneqq 1+s_+s_-\neq0.
\end{equation}
The coefficients $s_+$ and $s_-$ are independent of
$(x,t)$ and are holomorphic in
$\hbar$ on the common Borel disk $S_{0,R_{V_0}}$. $\delta$ is nowhere zero there. The closing connection $S_{\mathrm{DA}}$ is unipotent
if and only if $\delta=1$.
\end{proposition}

\begin{proof}
\sloppy{Apply \autoref{lem:RH-chain-normalisation} along
$\mathrm A,\mathrm B,\mathrm C,\mathrm D$.
The first three jump matrices are then unipotent.
Write $S_{\mathrm{AB}}=\operatorname{Up}(s_+)$, $S_{\mathrm{BC}}=\operatorname{Lw}(s_-)$, $S_{\mathrm{CD}}=\operatorname{Up}(c)$.
The first relation in \eqref{eq:RH-vertex-cycles} gives $S_{\mathrm{DA}}
   =\bigl(
       \operatorname{Up}(c)
       \operatorname{Lw}(s_-)
       \operatorname{Up}(s_+)
     \bigr)^{-1}.
$
By \autoref{lem:RH-recessive-comparison}, with the
chosen column ordering, the second columns of
$\Psi_{(\mathrm D)}$ and $\Psi_{(\mathrm A)}$
span the same recessive solution line along
$\mathrm{DA}$. This property is preserved by the
diagonal rescalings used to normalise the chain.
Taking the second column of
$\Psi_{(\mathrm D)}=\Psi_{(\mathrm A)}S_{\mathrm{DA}}$
therefore gives
\[
 [\Psi_{(\mathrm D)}]_{\cdot,2}
   =(S_{\mathrm{DA}})_{12}
        [\Psi_{(\mathrm A)}]_{\cdot,1}
    +(S_{\mathrm{DA}})_{22}
        [\Psi_{(\mathrm A)}]_{\cdot,2}.
\]
Linear independence of the columns of
$\Psi_{(\mathrm A)}$ implies
$(S_{\mathrm{DA}})_{12}=0$.
Thus $S_{\mathrm{DA}}$ is lower triangular.
Applying this condition to the preceding matrix
product gives $s_++c(1+s_+s_-)=0$.
If $\delta=1+s_+s_-=0$, this equality would imply
$s_+=0$, contradicting $\delta=0$.
Thus $c=-\frac{s_+}{\delta}$, and multiplication gives
\eqref{eq:RH-nodal-triangular-data}.
Holomorphy in $\hbar$ and independence of $(x,t)$
follow from the corresponding properties of the
jump matrices.}
\end{proof}

The factor $\delta$ records the return around the quadrivalent
vertex after normalising along a chain. At this stage the two nodal
coefficients $s_+$ and $s_-$ remain undetermined. We next extend the
normalisation to $\mathrm E$ and $\mathrm F$, and distinguish these
coefficients from the conventional relative scale of the columns.

\begin{lemma}[The simple vertex and the relative column scale]
\label{lem:RH-simple-vertex-normalisation}
\sloppy{Keep the matrices in $\mathrm A,\mathrm B,\mathrm C,\mathrm D$
normalised as in \autoref{prop:RH-nodal-cyclic-data}, and extend
the chain normalisation along $\mathrm D,\mathrm E,\mathrm F$.
With the cut folded onto $\mathrm{FD}$, there is a
nowhere-zero holomorphic function
$a:S_{0,R_{V_0}}\longrightarrow\mathbb C^\ast$,
independent of $(x,t)$, such that
\begin{equation}
\label{eq:RH-simple-vertex-free-a}
 S_{\mathrm{DE}}=\operatorname{Up}(a),\qquad
 S_{\mathrm{EF}}=\operatorname{Lw}(-a^{-1}),\qquad
 J_{\mathrm{FD}}=\operatorname{Ex}(a)^{-1}\operatorname{Up}(a).
\end{equation}
Consider the simultaneous right rescaling $\Psi_{(j)}\mapsto\Psi_{(j)}C$ in every Stokes region, where $C(\hbar)=\operatorname{diag}\bigl(c_+(\hbar),c_-(\hbar)\bigr)$ and $\eta(\hbar)\coloneqq \frac{c_-(\hbar)}{c_+(\hbar)}$. Here $c_+$ and $c_-$ are nowhere-zero holomorphic functions on the common Borel disk $S_{0,R_{V_{0}}}$, independent of
$(x,t)$, multiplying the first and second columns,
respectively. The induced conjugation
$S_{ij}\mapsto C^{-1}S_{ij}C$ acts on the coefficients by
\begin{equation}
\label{eq:RH-relative-scale-action}
 (a,s_+,s_-,\delta)
   \longmapsto
 (\eta a,\eta s_+,\eta^{-1}s_-,\delta).
\end{equation}
In particular, $a$ may be fixed to $i$ by a choice of relative
column normalisation. This choice leaves the Lax matrices and
the Painlev\'e~I solution $(q,p)$ unchanged.}
\end{lemma}

\begin{proof}
The shared recessive columns on $\mathrm{DE}$ and $\mathrm{EF}$
occupy, respectively, the first and second positions. Thus
\autoref{lem:RH-chain-normalisation} gives
$S_{\mathrm{DE}}=\operatorname{Up}(a)$ and $S_{\mathrm{EF}}=\operatorname{Lw}(b)$ for some
coefficients $a=a(\hbar)$ and $b=b(\hbar)$,
holomorphic on the common Borel disk. The simple-vertex
relation in \eqref{eq:RH-vertex-cycles} gives
\[
 J_{\mathrm{FD}}
   =\operatorname{Up}(-a)\operatorname{Lw}(-b)
   =\begin{pmatrix}1+ab&-a\\-b&1\end{pmatrix}.
\]
On $\mathrm{FD}$, the common recessive line is the first column
in $\mathrm F$ and the second column in $\mathrm D$. Indeed,
the determinations of $\arg x$ at this end are $\frac{6\pi}{5}$ and
$-\frac{4\pi}{5}$, respectively. For positive $\hbar$, the leading
action $\frac45x^{\frac{5}{2}}$ therefore makes the plus column
recessive on the $\mathrm F$ side and the minus column
recessive on the $\mathrm D$ side. By
\autoref{lem:RH-recessive-comparison}, the first column of
$\Psi_{(\mathrm F)}$ is proportional to the second column of
$\Psi_{(\mathrm D)}$. Hence $[J_{\mathrm{FD}}]_{11}=0$.
This gives $1+ab=0$, so $a\neq0$, $b\neq 0$ and $b=-a^{-1}$.
Substitution yields \eqref{eq:RH-simple-vertex-free-a}.
These identities extend holomorphically to the common Borel disk.
Note that the formula for $J_{FD}$ is equivalent to
\begin{equation}
\label{eq:ULU-equals-P}
 \operatorname{Up}(a)\operatorname{Lw}(-a^{-1})\operatorname{Up}(a)=\operatorname{Ex}(a).
\end{equation}
For $C=\operatorname{diag}(c_+,c_-)$, direct conjugation gives
$C^{-1}\operatorname{Up}(s)C=\operatorname{Up}(\eta s)$,
$C^{-1}\operatorname{Lw}(s)C=\operatorname{Lw}(\eta^{-1}s)$ and
$C^{-1}\operatorname{Ex}(a)C=\operatorname{Ex}(\eta a)$. This proves
\eqref{eq:RH-relative-scale-action}. Since $a$ is holomorphic
and nowhere zero, $C=\operatorname{diag}\left(1,\frac{i}{a}\right)$ is admissible
and sets $a=i$. Right multiplication by $C(\hbar)$ does not
change either Lax matrix.
\end{proof}

\begin{figure}[H]
\centering
\includegraphics[width=0.66\textwidth]
    {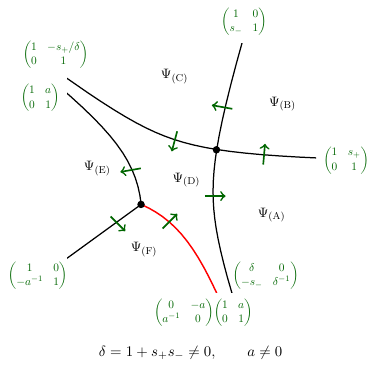}
\caption{Connection matrices from
\autoref{prop:RH-nodal-cyclic-data} and
\autoref{lem:RH-simple-vertex-normalisation},
with $\delta=1+s_+s_-\neq0$ and arbitrary $a\neq0$.
The branch cut is folded onto the red arc
$\mathrm{FD}$, whose combined jump is
$J_{\mathrm{FD}}=\operatorname{Ex}(a)^{-1}\operatorname{Up}(a)$.
Each green arrow from region $i$ to region $j$
carries the matrix $S_{ij}$ in the convention
$\Psi_{(i)}=\Psi_{(j)}S_{ij}$.
The contour is schematic.}
\label{fig:RH-cut-folded-general}
\end{figure}

Note that $a$ is conventional since it measures a relative column scale. By contrast, $s_+$ and $s_-$ are the two nodal Stokes multipliers still to be determined after that scale has been fixed. In particular, the vanishing
of either $s_+$ or $s_-$ is independent of the conventional choice of $a$. Indeed,  the combinations $\frac{s_+}{a}$ and $a s_-$ are invariant under
\eqref{eq:RH-relative-scale-action}. We keep $a$ visible in the
figures to distinguish this freedom from the analytic conclusion
$s_+=s_-=0$ obtained below in \autoref{prop:RH-nodal-vanishing}.

\paragraph{Moving the branch cut.}
The two diagrams in \autoref{fig:RH-branch-cut-positions} show
the same connection data with different positions of the cut.
An arrow from region $i$ to region $j$, labelled by $S$, means
$\Psi_{(i)}=\Psi_{(j)}S$, as in \eqref{eq:jump-convention}.
In the left diagram the cut lies inside $\mathrm D$ and creates
a narrow component carrying
\begin{equation}
\label{eq:RH-cut-in-D}
 \Psi'_{(\mathrm D)}=\Psi_{(\mathrm D)}\operatorname{Ex}(a)^{-1},\qquad
 \Psi_{(\mathrm F)}=\Psi'_{(\mathrm D)}\operatorname{Up}(a).
\end{equation}
The red cut has jump $\operatorname{Ex}(a)^{-1}$ and the neighbouring Stokes
line has jump $\operatorname{Up}(a)$. Folding the cut onto that line gives
$\Psi_{(\mathrm F)}=\Psi_{(\mathrm D)}\operatorname{Ex}(a)^{-1}\operatorname{Up}(a)$.
This explains both the combined jump and the order of its factors. For the right diagram, move the cut into $\mathrm E$ and set
\begin{equation}
\label{eq:RH-cut-in-E}
 \Psi_{(\mathrm E_1)}=\Psi_{(\mathrm E)},\qquad
 \Psi_{(\mathrm E_2)}=\Psi_{(\mathrm E)}\operatorname{Ex}(a),\qquad
 \Psi'_{(\mathrm F)}=\Psi_{(\mathrm F)}\operatorname{Ex}(a).
\end{equation}
The corresponding relations are
\begin{equation*}
 \Psi_{(\mathrm D)}=\Psi_{(\mathrm E_1)}\operatorname{Up}(a),\quad
 \Psi_{(\mathrm E_1)}=\Psi_{(\mathrm E_2)}\operatorname{Ex}(a)^{-1},\quad
 \Psi_{(\mathrm E_2)}=\Psi'_{(\mathrm F)}\operatorname{Up}(a),\quad
 \Psi'_{(\mathrm F)}=\Psi_{(\mathrm D)}\operatorname{Lw}(-a^{-1}).
 \end{equation*}
They follow from \eqref{eq:RH-simple-vertex-free-a} and
$\operatorname{Ex}(a)^{-1}\operatorname{Lw}(-a^{-1})\operatorname{Ex}(a)=\operatorname{Up}(a)$,
$\operatorname{Ex}(a)^{-1}\operatorname{Up}(a)\operatorname{Ex}(a)=\operatorname{Lw}(-a^{-1})$.
Thus moving the cut changes the column ordering and scales
on the regions it passes through, with the asymptotic branches
changed accordingly. It introduces no new monodromy parameter.
We use the folded-cut convention on $\mathrm{FD}$ in the RHP below.

\begin{figure}[H]
\centering
\begin{minipage}[t]{.485\textwidth}
 \centering
 \includegraphics[width=\linewidth]{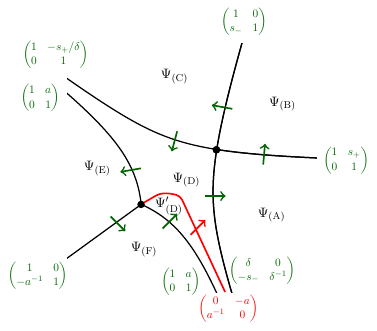}
 \par\smallskip
 {\small (a) The cut inside $\mathrm D$.}
\end{minipage}\hfill
\begin{minipage}[t]{.485\textwidth}
 \centering
 \includegraphics[width=\linewidth]{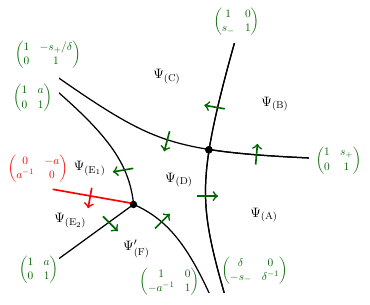}
 \par\smallskip
 {\small (b) The cut inside $\mathrm E$.}
\end{minipage}
\caption{Connection data before applying
\autoref{thm:no-jump-quadrivalent}, with
$\delta=1+s_+s_-\neq0$ and $a\neq0$. Green arrows carry
the directed connections; the red cut carries the sheet-exchange
factor $\operatorname{Ex}(a)^{-1}$. The two diagrams are related by
\eqref{eq:RH-cut-in-D} and \eqref{eq:RH-cut-in-E}.
The nodal coefficients are the same in both pictures.
The parameter $a$ is a conventional relative column scale.}
\label{fig:RH-branch-cut-positions}
\end{figure}

\paragraph{Standard algebraic normalisation and vanishing of nodal coefficients.}

\begin{definition}[Standard algebraic normalisation of the wave matrices]
\label{def:RH-standard-algebraic-normalisation}
Use the chosen branches of $\tau$, $w$ and $Z$,
with the branch cut folded onto $\mathrm{FD}$.
In each region $\mathcal U_j$, retain the ordered
pair of solution lines spanned by the first and
second columns of $\Psi_{(j)}^{\mathrm{WKB}}$.
We consider diagonal right rescalings of these
matrices, which change the scalar normalisations
of the columns while preserving their solution
lines and their ordering.
A family of sectorial fundamental matrices $\left(\Psi_{(j)}\right)_{j\in\{A,\dots,F\}}$ is in
\emph{standard algebraic normalisation} if its jump
matrices have the forms
\eqref{eq:RH-nodal-triangular-data} and
\eqref{eq:RH-simple-vertex-free-a}, for $(x,t)$-independent coefficients $(s_+(\hbar),s_-(\hbar),a(\hbar))\in \mathbb{C}^3$ holomorphic on $\hbar\in S_{0,R_{V_0}}$  satisfying $a(\hbar)\neq0$ and $\delta(\hbar)=1+s_+(\hbar)s_-(\hbar)\neq0$.
\end{definition}

The preceding proposition and lemma show that every
family with these fixed branches and ordered WKB
column directions can be brought into this form by
diagonal rescalings. However, the standard algebraic normalisation does not provide a unique family. Indeed two families that have standard algebraic normalisation differ by a common diagonal right factor
$C(\hbar)$, independent of $(x,t)$.
Indeed, write
$\Psi'_{(j)}=\Psi_{(j)}C_j$, with each $C_j$ diagonal.
Along the chain
$\mathrm A,\mathrm B,\mathrm C,\mathrm D,\mathrm E,\mathrm F$,
both families have triangular unipotent connections
of the same triangular type.
Comparing the diagonal entries of
$C_j^{-1}S_{ij}C_i$ therefore gives $C_i=C_j$
on every edge of this chain. Hence all $C_j$ coincide.
The action of this remaining common diagonal factor
on the coefficients is
\eqref{eq:RH-relative-scale-action}.
Fixing $a=i$ removes the relative column scale,
leaving only a common nonzero scalar factor, which
does not affect any jump matrix.
The normalisation at infinity below fixes this
last scalar freedom.
Once the relative column scale has been fixed,
the algebraic relations leave the two nodal
coefficients $(s_+,s_-)$ undetermined.
We now use \autoref{thm:no-jump-quadrivalent}
to show that both vanish for the Borel-summed
WKB matrices. Consequently $\delta=1$, and
all four jumps incident with the double
turning point are identities.

\begin{proposition}[Vanishing of the nodal coefficients]
\label{prop:RH-nodal-vanishing}
For the sectorial matrices obtained from
\autoref{thm:exact-WKB-existence} by the standard algebraic
normalisation of
\autoref{def:RH-standard-algebraic-normalisation}, one has
\begin{equation}
\label{eq:RH-zero-nodal-jumps}
 s_+=s_-=0,\qquad \Leftrightarrow\qquad
 S_{\mathrm{DA}}=S_{\mathrm{AB}}=S_{\mathrm{BC}}
       =S_{\mathrm{CD}}=I.
\end{equation}
\end{proposition}

\begin{proof}
By \autoref{thm:no-jump-quadrivalent}, there exist
a full neighbourhood $W$ of the double turning point
$v_\ast$ and an exact fundamental matrix $\Psi^{\mathrm v}$
holomorphic throughout $W$ on a sufficiently small
common Borel disk. Choose the branches and column
ordering consistently with those of $\Psi^{\mathrm v}$. The matching assertion of
\autoref{thm:no-jump-quadrivalent} identifies the
two columns of $\Psi^{\mathrm v}$, up to their scalar
normalisations, with the corresponding exact WKB
columns in each of the four adjacent regions.
The subsequent diagonal rescalings used in the
standard algebraic normalisation preserve this
columnwise identification. Consequently, for every
$j\in\{\mathrm A,\mathrm B,\mathrm C,\mathrm D\}$,
there is an invertible diagonal matrix $C_j(\hbar)$
such that $\Psi_{(j)}=\Psi^{\mathrm v} C_j$ on $W\cap\mathcal U_j$.
Since both matrices satisfy the same Lax pair,
$C_j$ is independent of $(x,t)$, by
\autoref{lem:connection-matrices-constant}. These identities also describe the continuations of the four regional matrices throughout $W$.
Thus the directed convention
$\Psi_{(i)}=\Psi_{(j)}S_{ij}$ gives $S_{ij}=C_j^{-1}C_i$ for every pair of adjacent regions around the
double turning point. In particular, all four
connection matrices are diagonal since $(C_j)_{j\in\{\mathrm A,\mathrm B,\mathrm C,\mathrm D\}}$ are diagonal.

By \autoref{prop:RH-nodal-cyclic-data}, the three
matrices $S_{\mathrm{AB}}$, $S_{\mathrm{BC}}$
and $S_{\mathrm{CD}}$ are also unipotent.
A diagonal unipotent matrix is the identity, so $S_{\mathrm{AB}}=S_{\mathrm{BC}}=S_{\mathrm{CD}}=I$.
The cyclic relation
\eqref{eq:RH-vertex-cycles} then gives $S_{\mathrm{DA}}=I$.

Finally, the expressions
$S_{\mathrm{AB}}=\operatorname{Up}(s_+)$ and
$S_{\mathrm{BC}}=\operatorname{Lw}(s_-)$ imply
$s_+=s_-=0$, and hence $\delta=1$.
Since the connection matrices are independent of $(x,t)$,
these identities hold along all four incident Stokes curves
on the smaller Borel disk used above. Their holomorphy in
$\hbar$ and the identity theorem extend them to the common
Borel disk $S_{0,R_{V_0}}$.
\end{proof}

Consequently, diagrams in \autoref{fig:RH-cut-folded-general} and 
\autoref{fig:RH-branch-cut-positions} reduce to identity jumps
at the quadrivalent vertex by setting $s_+=s_-=0$ and $\delta=1$.
Folding the red cut onto $\mathrm{FD}$ as in \autoref{fig:RH-cut-folded-general} then gives
\autoref{fig:final-Stokes-matrices}, with the three nontrivial
connections in \eqref{eq:RH-simple-vertex-free-a}.
The vanishing of $s_+$ and $s_-$ is the analytic conclusion supplied
by \autoref{thm:no-jump-quadrivalent}. As explained above, the choice $a=i$ is
independently available through relative column normalisation.

\paragraph{Canonical normalisation at infinity.}
We next impose the canonical normalisation at infinity. This
fixes the relative scale to $a=i$ and fixes the remaining common
scalar factor. We use the asymptotic theory of the original
linear system \eqref{250408155610_bis}. Put
\begin{equation}
\label{eq:RH-phase-Hamiltonian}
 \sigma_3=\begin{pmatrix}1&0\\0&-1\end{pmatrix},\qquad
 \sigma_1=\begin{pmatrix}0&1\\1&0\end{pmatrix},\qquad
 \theta=\frac{\frac45x^{\frac{5}{2}}+t x^{\frac{1}{2}}}{\hbar},\qquad
 \mathcal H=\frac{p^2}{2}-2q^3-tq,
\end{equation}
and
\begin{equation}
\label{eq:leading-normalisation}
    N(x)=\begin{pmatrix}
          2x^{-\frac{1}{4}}&2x^{-\frac{1}{4}}\\
          x^{\frac{1}{4}}&-x^{\frac{1}{4}}
         \end{pmatrix}\quad \Rightarrow\quad \det N(x)=-4.
\end{equation}
Fractional powers are first fixed near the positive real axis
and then continued with the chosen sheets.

The canonical-solution construction in
\cite[Section~2, equations~(2.3)--(2.4)]{Kapaev_2004}
gives canonical fundamental matrices $F_k(x,t,\hbar)$ of the
original Lax pair \eqref{250408155610_bis}.
For fixed $(t,\hbar)$ and every closed subsector
$\Omega'$ whose boundary rays lie strictly inside
$\Omega_k^\infty(\hbar)$, the following expansion
holds uniformly:
\begin{equation}
\label{eq:canonical-normalisation-infinity}
 F_k(x,t,\hbar)
 \overset{x\to\infty}{\underset{x\in\Omega'}{=}}
 N(x)\left[
 I-\frac{\mathcal H}{\hbar\sqrt{x}}\sigma_3
 +\frac1{2x}
    \left(\frac{\mathcal H^2}{\hbar^2}I-q\sigma_1\right)
 +O\left(x^{-\frac32}\right)
 \right]e^{\theta\sigma_3}.
\end{equation}
Here the canonical sectors are
\begin{equation}
\label{eq:RH-canonical-sectors}
 \Omega_k^\infty(\hbar)=
 \left\{x:
 \frac25\arg\hbar+\frac{2\pi}{5}\left(k-\frac32\right)
 <\arg x<
 \frac25\arg\hbar+\frac{2\pi}{5}\left(k+\frac12\right)
 \right\}.
\end{equation}
Five consecutive indices suffice to describe the
canonical family. For the comparison with the WKB
Stokes regions below, we use $k\in\{-1,0,1,2,3\}$.
We nevertheless retain the indexing by
$k\in\mathbb Z$ to record continuation of the
chosen determination of $\arg x$: increasing $k$
by five corresponds to one full turn.
The resulting change of the fractional powers
is accounted for by the continuation relations
below. The expansion in \eqref{eq:canonical-normalisation-infinity} continues to all orders. Its displayed coefficients
are obtained by substitution in the spatial equation of \eqref{250408155610_bis}. The time
equation is compatible with this normalisation. We use here the
standard canonical-solution theorem for this polynomial system;
see \cite[Section~2, equations~(2.1)--(2.8)]{Kapaev_2004}.
The constant change of gauge relating the conventions is written
explicitly in the last subsection. This linear-ODE result does
not assume any special Stokes data for $(q,p)$.

After decreasing $R_{V_0}$ if necessary, for each
$k\in\mathbb Z$, the canonical matrix $F_k$, continued in $x$ by the spatial Lax equation, defines a jointly holomorphic map
\[
    F_k:
    \mathbb C_x\times\mathcal P_{V_0}
       =\mathbb C_x\times V_0\times S_{0,R_{V_0}}
       \longrightarrow\operatorname{GL}_2(\mathbb C).
\]
Here the canonical sectors are followed continuously
with $\arg\hbar\in\left(-\frac{\pi}{2},\frac{\pi}{2}\right)$, keeping the index
$k$ fixed. For fixed $(t,\hbar)\in\mathcal P_{V_0}$, the
matrices $F_k$ and $F_{k+1}$ are fundamental matrices of the same spatial system \eqref{250408155610_bis}. Hence $T_k(\hbar)\coloneqq F_{k}(x,t,\hbar)^{-1}F_{k+1}(x,t,\hbar)$ is independent of $x$, and the time equation also
makes it independent of $t$, as in \autoref{lem:connection-matrices-constant}.
The matrices $(F_k(x,t,\hbar))_{k\in \mathbb{Z}}$ are related by 
\begin{equation}
\label{eq:RH-native-Stokes}
 F_{k+1}(x,t,\hbar) =F_k(x,t,\hbar)T_k(\hbar),\qquad
 T_{2m-1}(\hbar)=\operatorname{Up}(b_{2m-1}(\hbar)),\qquad T_{2m}(\hbar)=\operatorname{Lw}(b_{2m}(\hbar)), 
\end{equation}
for all $(x,t,\hbar)\in \mathbb C_x\times\mathcal P_{V_0}$. These connection matrices are called Stokes matrices
because they compare the canonical normalisations
in consecutive asymptotic sectors: their definition
does not require choosing a jump contour. Here $b_k:S_{0,R_{V_0}}\longrightarrow\mathbb C$
are holomorphic functions, called the canonical
Stokes multipliers.
The triangular type follows from the common recessive column while
the diagonal entries are one because the asymptotic
normalisations agree. Continuation of the fractional powers gives
$N(e^{2\pi i}x)=N(x)(-i\sigma_1)$
and exchanges the two exponentials.
Uniqueness of the canonical normalisation, together
with the entire continuation in $x$, therefore gives $F_{k+5}(x,t,\hbar)=F_k(x,t,\hbar)(-i\sigma_1)$.
Consequently,
\begin{equation} \forall\, k\in \mathbb{Z}\,:\,
\label{eq:RH-native-cyclic}
 T_{k+5}=\sigma_1T_k\sigma_1, \text{ and }
 T_1T_2T_3T_4T_5=-i\sigma_1.
\end{equation}

The jump matrices below already follow from
\autoref{prop:RH-nodal-vanishing} and
\autoref{lem:RH-simple-vertex-normalisation}
by choosing $a=i$. The remaining point is to show
that this algebraic normalisation is compatible
with the prescribed canonical asymptotics at infinity. The identities $F_{-1}=F_0=F_1$ asserted below are
special to the Painlev\'e solution constructed here.
They follow by transferring the vanishing result
of \autoref{prop:RH-nodal-vanishing} to the canonical
matrices through the comparison of their recessive
columns.

\begin{proposition}[Canonical normalisation of the WKB matrices]
\label{thm:WKB-RHP}
The sectorial WKB matrices admit diagonal right
normalisations, denoted by $\Psi_{\mathrm{can},(j)}$,
for which
\[
 \begin{gathered}
 \Psi_{\mathrm{can},(\mathrm A)}
   =\Psi_{\mathrm{can},(\mathrm B)}
   =\Psi_{\mathrm{can},(\mathrm C)}
   =\Psi_{\mathrm{can},(\mathrm D)}=F_0,
 \\
 \Psi_{\mathrm{can},(\mathrm E)}=F_2,\qquad
 \Psi_{\mathrm{can},(\mathrm F)}=F_3,\qquad
 F_{-1}=F_0=F_1.
 \end{gathered}
\]
Consequently, the canonically normalised fundamental matrices
$\left(\Psi_{\mathrm{can},(j)}\right)_
 {j\in\{\mathrm A,\ldots,\mathrm F\}}$
have the canonical asymptotic expansion
\eqref{eq:canonical-normalisation-infinity}
at their corresponding ends as $x\to\infty$.
These regional matrices define the piecewise matrix
$\Psi_{\mathrm{can}}$ by
$\left.\Psi_{\mathrm{can}}\right|_{\mathcal U_j}
 =\Psi_{\mathrm{can},(j)}$.
The matrix $\Psi_{\mathrm{can}}$ has the jump matrices
\begin{equation}
\label{eq:all-jumps}
 \begin{gathered}
 S_{\mathrm{DA}}=S_{\mathrm{AB}}=S_{\mathrm{BC}}
       =S_{\mathrm{CD}}=I,\qquad
 S_{\mathrm{DE}}=\operatorname{Up}(i),\qquad
 S_{\mathrm{EF}}=\operatorname{Lw}(i),\\
 J_{\mathrm{FD}}=\operatorname{Ex}(i)^{-1}\operatorname{Up}(i),
 \qquad \operatorname{Ex}(i)=i\sigma_1.
 \end{gathered}
\end{equation}
\end{proposition}

\begin{proof}
Start with the regional matrices $\Psi_{(j)}$ in
the standard algebraic normalisation of
\autoref{def:RH-standard-algebraic-normalisation}.
By \autoref{prop:RH-nodal-vanishing}, their entire
continuations satisfy $\Psi_{(\mathrm A)}
      =\Psi_{(\mathrm B)}
      =\Psi_{(\mathrm C)}
      =\Psi_{(\mathrm D)}$.
First take sufficiently small positive $\hbar$.
We compare the columns of these regional matrices
with the canonical columns by identifying their
recessive ends at infinity.
By \eqref{eq:RH-end-labelling}, the chosen column
ordering, and \autoref{lem:RH-recessive-comparison},
the corresponding limiting arguments $\arg x$ are
\[
 \begin{array}{c|c|c|c}
  \text{region}
   &\text{first column}
   &\text{second column}
   &\text{canonical matrix}\\ \hline
  \mathrm A&-\frac{2\pi}{5}&-\frac{4\pi}{5}&F_{-1}\\
  \mathrm B&-\frac{2\pi}{5}&0&F_0\\
  \mathrm C&\frac{2\pi}{5}&0&F_1\\
  \mathrm E&\frac{2\pi}{5}&\frac{4\pi}{5}&F_2\\
  \mathrm F&\frac{6\pi}{5}&\frac{4\pi}{5}&F_3
 \end{array}
\]
Indeed, the canonical expansion
\eqref{eq:canonical-normalisation-infinity}
contains $e^\theta$ in its first column and
$e^{-\theta}$ in its second column.
Since $\theta\sim \frac{4}{5\hbar}x^{\frac{5}{2}}$, these columns
are recessive at the respective ends displayed
above, which lie inside their canonical sectors
\eqref{eq:RH-canonical-sectors}.

Uniqueness of the recessive solution line implies
that the regional and canonical columns in each
row of the table are proportional.
In particular, there are invertible diagonal
matrices $C_{\mathrm A},C_{\mathrm B},C_{\mathrm C}$
such that
$F_{-1}=\Psi_{(\mathrm A)}C_{\mathrm A}$,
    $F_0=\Psi_{(\mathrm B)}C_{\mathrm B}$,
    $F_1=\Psi_{(\mathrm C)}C_{\mathrm C}$.
Both sides satisfy the same Lax pair, so
\autoref{lem:connection-matrices-constant}
shows that these comparison matrices depend
only on $\hbar$. Since the three regional matrices coincide,
\eqref{eq:RH-native-Stokes} gives
$T_{-1}=C_{\mathrm A}^{-1}C_{\mathrm B}$,
    $T_0=C_{\mathrm B}^{-1}C_{\mathrm C}$.
Thus $T_{-1}$ and $T_0$ are diagonal. Their unipotent triangular forms force
$T_{-1}=T_0=I$, so that 
    $C_{\mathrm A}=C_{\mathrm B}=C_{\mathrm C}$ and 
    $F_{-1}=F_0=F_1$.
We now apply this common diagonal normalisation
to all four regions incident with the double
turning point, defining
\[
    \Psi_{\mathrm{can},(j)}
       \coloneqq \Psi_{(j)}C_{\mathrm A}=F_0,
    \qquad
    \forall\, j\in\{\mathrm A,\mathrm B,\mathrm C,\mathrm D\}.
\]
Their jump matrices remain identities.
The periodicity in \eqref{eq:RH-native-cyclic}
gives $T_4=T_5=I$. Hence the cyclic relation becomes
\[
    \operatorname{Up}(b_1)
    \operatorname{Lw}(b_2)
    \operatorname{Up}(b_3)
       =-i\sigma_1.
\]
Its lower-left entry gives $b_2=-i$, and its diagonal
entries then give $b_1=b_3=-i$.
Define
\begin{equation*}
    \Psi_{\mathrm{can},(\mathrm E)}
       \coloneqq F_2
        =\Psi_{\mathrm{can},(\mathrm D)}
             \operatorname{Up}(-i),\qquad
    \Psi_{\mathrm{can},(\mathrm F)}
       \coloneqq F_3
        =\Psi_{\mathrm{can},(\mathrm E)}
             \operatorname{Lw}(-i).
 \end{equation*}
By the last two rows of the previous table, these matrices
are diagonal right renormalisations of
$\Psi_{(\mathrm E)}$ and $\Psi_{(\mathrm F)}$,
respectively. Inverting the preceding relations gives $S_{\mathrm{DE}}=\operatorname{Up}(i)$ and $S_{\mathrm{EF}}=\operatorname{Lw}(i)$.
Furthermore, $J_{\mathrm{FD}}
       =\operatorname{Up}(-i)\operatorname{Lw}(-i)
       =\operatorname{Ex}(i)^{-1}\operatorname{Up}(i)$.
This proves \eqref{eq:all-jumps} for sufficiently
small positive $\hbar$.
All comparison matrices are holomorphic and
invertible on the common Borel disk $S_{0,R_{V_0}}$:
they are quotients of holomorphic fundamental
matrices. Their off-diagonal entries vanish for
sufficiently small positive $\hbar$, so they vanish
throughout $S_{0,R_{V_0}}$ by the identity theorem.
The same argument extends all the displayed
connection identities to this disk, with the
canonical sectors followed continuously in
$\arg\hbar$ and their indices kept fixed.

Finally, the canonical expansions of the $F_k$
give the asserted asymptotics at the corresponding
ends of each region. For $\mathrm D$, use
$F_0=F_{-1}=F_1$ to cover both ends.
Every regional matrix is entire in $x$ and
satisfies both Lax equations.
\end{proof}

\autoref{fig:final-Stokes-matrices} displays the
resulting jump matrices with the relative column
scale $a$ left arbitrary. For the canonically
normalised matrix $\Psi_{\mathrm{can}}$, one has
$a=i$.

\begin{figure}[H]
\centering
\includegraphics[width=.62\textwidth]{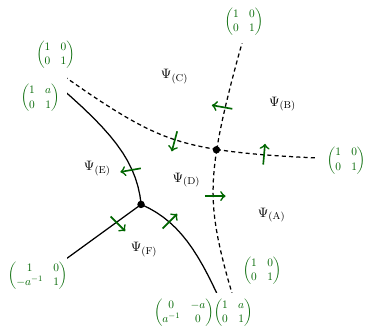}
\caption{The WKB contour after
\autoref{prop:RH-nodal-vanishing}: the four nodal jumps are
identities. The canonical prefactor $N(x)$ selects
$a=i$. The red cut of
\autoref{fig:RH-branch-cut-positions} has been folded onto
$\mathrm{FD}$, giving $J_{\mathrm{FD}}=\operatorname{Ex}(a)^{-1}\operatorname{Up}(a)$.}
\label{fig:final-Stokes-matrices}
\end{figure}

\subsection{The normalized WKB Riemann--Hilbert problem}
\label{sec:WKB-RHP-problem}

Fix $(t,\hbar)\in\mathcal P_{V_0}$, and let
$\mathcal G(t)$ be the seven-armed Stokes graph of
\eqref{eq:time-dependent-Stokes-regions}, including its two
vertices. We fold the branch cut onto the curve $\mathrm{FD}$.
Outside a sufficiently large disk, choose the
determination of $\arg x$ obtained by continuation
from the positive real axis, with the cut along
$\mathrm{FD}$. Its limiting arguments on the two
sides of this curve are $-\frac{4\pi}{5}$ on the
$\mathrm D$ side and $\frac{6\pi}{5}$ on the
$\mathrm F$ side.
The fractional powers in $N$ and $\theta$ are
defined using this determination.
In particular, beyond the negative real axis in
region $\mathrm F$, they are continued rather
than reset to their principal values.

\bigskip

We now formulate the Riemann--Hilbert problem on the
WKB contour and prove that $\Psi_{\mathrm{can}}$ is
its unique solution.
Its comparison with the standard Painlev\'e~I
Riemann--Hilbert problem
\cite[Section~2]{Kapaev_2004}
is given in \autoref{sec:canonical-RHP}.
For a related exact WKB approach to the direct
monodromy problem, see \cite[Section~5]{Iwaki-P1}.

\begin{problem}[The normalized WKB Riemann--Hilbert problem]
\label{prob:WKB-RHP}
Find a matrix $\Psi(x,t,\hbar)$ with the following properties.
\begin{enumerate}
\item For fixed $(t,\hbar)\in\mathcal P_{V_0}$, it is holomorphic on
$\mathbb C\setminus\mathcal G(t)$ and has continuous one-sided
boundary values on every open contour arc of the Stokes graph $\mathcal{G}(t)$.
\item For each $j\in\{\mathrm A,\ldots,\mathrm F\}$,
write $\Psi_{(j)}$ for the restriction of
$x\mapsto\Psi(x,t,\hbar)$ to the spatial region
$R_j(\tau(t))$ of
\eqref{eq:time-dependent-Stokes-regions}.
On each directed open arc separating regions
$i$ and $j$, the one-sided boundary values satisfy
$\Psi_{(i)}=\Psi_{(j)}S_{ij}$, with the jump
matrices in \eqref{eq:all-jumps}.
On the directed arc $\mathrm{FD}$, the jump
matrix is $J_{\mathrm{FD}}$.
\item It is bounded at both finite vertices in every adjacent
region.
\item There exists a scalar $d=d(t,\hbar)$, not prescribed
as part of the data, such that
\begin{equation}
\label{eq:RH-native-infinity-condition}
    N(x)^{-1}\Psi(x,t,\hbar)e^{-\theta\sigma_3}
       =I+\frac{d\sigma_3}{\sqrt{x}}+O(x^{-1}),
       \qquad x\longrightarrow\infty.
\end{equation}
For fixed $(t,\hbar)$, this estimate is uniform
on every end of each spatial region
$R_j(\tau(t))$, including its one-sided contour
boundaries, with the determinations of the
fractional powers specified above.
\end{enumerate}
\end{problem}

The diagonal form of the coefficient of $x^{-\frac{1}{2}}$ is part of
the normalisation. Neither $q$ nor its formal expansion is
included in the RHP data.

\begin{theorem}[Existence and uniqueness of the WKB RHP]
\label{thm:RH-WKB-existence-uniqueness}
For every $(t,\hbar)\in\mathcal P_{V_0}$,
\autoref{prob:WKB-RHP} has a unique solution. It is the matrix $\Psi_{\mathrm{can}}$ constructed
in \autoref{thm:WKB-RHP}.
Its regional entire continuations form jointly
holomorphic families on
$\mathbb C_x\times\mathcal P_{V_0}$ and satisfy
both equations of the original Lax system. It has determinant $-4$ and
$d=-\frac{1}{\hbar}\mathcal H$.
\end{theorem}

\begin{proof}
Existence follows from \autoref{thm:WKB-RHP} and
\eqref{eq:canonical-normalisation-infinity}. For clarity, the
endwise asymptotic condition holds also in the narrow ends of
$\mathrm D$: its matrix is $F_{-1}=F_0=F_1$, whose combined
canonical sectors contain both ends. The sectors for $F_2$ and
$F_3$ contain the ends of $\mathrm E$ and $\mathrm F$.
Since $|\arg\hbar|<\frac{\pi}{2}$, the angular
displacement satisfies
$\left|\frac25\arg\hbar\right|<\frac{\pi}{5}$.
Thus, for fixed $\hbar$, all limiting contour
directions remain strictly inside the relevant
canonical sectors.
For uniqueness, all jumps have determinant one. The determinant
of any solution therefore extends holomorphically across the
arcs and, by boundedness, across both vertices. Its limit at
infinity is $\det N=-4$, so it equals $-4$ by Liouville's
theorem. In particular every solution is invertible. More precisely, let $\Psi^{(1)},\Psi^{(2)}$ be two solutions of the RHP and put
$M=\Psi^{(1)}(\Psi^{(2)})^{-1}$ and
$\Delta=d_1-d_2$. The jumps of $M$ cancel, and its finite
singularities are removable.
Using \eqref{eq:RH-native-infinity-condition} and
the explicit matrix $N$ in
\eqref{eq:leading-normalisation}, we obtain
\[
 M=N\left[I+\Delta\sigma_3x^{-\frac{1}{2}}+O(x^{-1})\right]N^{-1}
   =I+\frac{\Delta}{2}E_{21}+O\left(x^{-\frac{1}{2}}\right),
 \qquad E_{21}=\begin{pmatrix}0&0\\1&0\end{pmatrix}.
\] 
Thus $M$ is bounded and equals the constant
$I+\frac\Delta2E_{21}$ by Liouville's theorem. Then, we have
$N^{-1}E_{21}N
       =\frac1{\sqrt{x}}
          \begin{pmatrix}1&1\\-1&-1\end{pmatrix}$.
The exact identity $M=I+\frac{\Delta}{2}E_{21}$
therefore gives
\[
    N^{-1}MN
       =I+\frac{\Delta}{2\sqrt{x}}
          \begin{pmatrix}1&1\\-1&-1\end{pmatrix}.
\]
On the other hand, applying
\eqref{eq:RH-native-infinity-condition}
to each of the two solutions gives
\begin{align*}
    N^{-1}MN
       &=
       \left(N^{-1}\Psi^{(1)}e^{-\theta\sigma_3}\right)
       \left(N^{-1}\Psi^{(2)}e^{-\theta\sigma_3}\right)^{-1}=
       \left(I+\frac{d_1\sigma_3}{\sqrt{x}}+O(x^{-1})\right)
       \left(I-\frac{d_2\sigma_3}{\sqrt{x}}+O(x^{-1})\right)\\
       &=I+\frac{\Delta\sigma_3}{\sqrt{x}}+O(x^{-1}),
\end{align*}
where $\Delta=d_1-d_2$.
Comparing the $(1,2)$-entries gives
$\frac{\Delta}{2\sqrt{x}}=O(x^{-1})$.
Multiplication by $\sqrt{x}$ and passage to the
limit yield $\Delta=0$. Hence $M=I$, proving uniqueness.
\end{proof}

\subsection{Comparison with the standard Painlev\'e I problem}
\label{sec:canonical-RHP}

We first change the matrix normalisation of the WKB RHP.
We then recall the standard Painlev\'e~I RHP, without
$\hbar$, and its correspondence with Painlev\'e solutions.
Finally, a rescaling identifies the two problems and
therefore identifies our Borel-summed solution with a
rescaled classical tritronqu\'ee solution.

\paragraph{A constant change of gauge and the RHP for $\Gamma$.}

Fix $(t,\hbar)\in\mathcal P_{V_0}$ and let
$\Psi_{\mathrm{can}}$ be the solution of
\autoref{prob:WKB-RHP}. Define
\begin{equation}
\label{eq:RH-constant-gauge}
    G=\begin{pmatrix}0&2\\1&0\end{pmatrix},
    \qquad
    H_0=\frac1{\sqrt2}
        \begin{pmatrix}1&1\\1&-1\end{pmatrix},
    \qquad
    \Gamma(x,t,\hbar)=\frac1{2\sqrt2}
               G\Psi_{\mathrm{can}}(x,t,\hbar)\sigma_3.
\end{equation}
The left multiplication by $G$ changes the Lax gauge.
The constant right multiplication by $\sigma_3$
changes the convention for the connection matrices.
We now state the RHP satisfied by $\Gamma$.
The contour $\mathcal G(t)$, the branch cut on
$\mathrm{FD}$ and the determinations of the
fractional powers are those of
\autoref{prob:WKB-RHP}.

\begin{problem}[The gauge-transformed WKB RHP]
\label{prob:Gamma-RHP}
For fixed $(t,\hbar)\in\mathcal P_{V_0}$, find
a matrix $\Gamma(x,t,\hbar)$ satisfying the
following conditions.
\begin{enumerate}
\item
It is holomorphic on
$\mathbb C_x\setminus\mathcal G(t)$ and has
continuous one-sided boundary values on
every open contour arc.

\item
For each $j\in\{\mathrm A,\ldots,\mathrm F\}$,
write $\Gamma_{(j)}$ for its restriction to
the spatial region $R_j(\tau(t))$.
With the directed convention of
\autoref{prob:WKB-RHP}, its boundary values
satisfy $\Gamma_{(i)}
       =\Gamma_{(j)}\widetilde S_{ij}$,
where
\begin{equation}
\label{eq:RH-gauged-jumps}
    \widetilde S_{\mathrm{DA}}
      =\widetilde S_{\mathrm{AB}}
      =\widetilde S_{\mathrm{BC}}
      =\widetilde S_{\mathrm{CD}}=I,\quad
    \widetilde S_{\mathrm{DE}}
       =\operatorname{Up}(-i),
    \quad
    \widetilde S_{\mathrm{EF}}
       =\operatorname{Lw}(-i).
\end{equation}
On the directed arc $\mathrm{FD}$, the
boundary values satisfy $\Gamma_{(\mathrm F)}
       =\Gamma_{(\mathrm D)}
          \widetilde J_{\mathrm{FD}}$ with
    $\widetilde J_{\mathrm{FD}}
       =\operatorname{Up}(i)\operatorname{Lw}(i)$.
\item
It is bounded at both finite vertices in
every adjacent region.

\item
There exists a scalar $d_\Gamma(t,\hbar)$,
not prescribed as part of the data, such that
\begin{equation}
\label{eq:RH-gauged-infinity}
    H_0x^{-\frac{\sigma_3}{4}}\Gamma(x,t,\hbar)
          e^{-\theta(x,t,\hbar)\sigma_3}
       =I+\frac{d_\Gamma(t,\hbar)\sigma_3}{\sqrt{x}}
          +O\left(x^{-1}\right),
       \qquad x\longrightarrow\infty.
\end{equation}
For fixed $(t,\hbar)$, this estimate is uniform
on every end of each spatial region
$R_j(\tau(t))$, including its one-sided contour
boundaries, with the specified determinations
of the fractional powers.
\end{enumerate}
\end{problem}

\begin{proposition}[Existence and uniqueness of the RHP for $\Gamma$]
\label{prop:Gamma-RHP}
For every $(t,\hbar)\in\mathcal P_{V_0}$,
\autoref{prob:Gamma-RHP} has a unique solution.
It is the matrix $\Gamma$ defined by
\eqref{eq:RH-constant-gauge}. Moreover, $\det\Gamma=-1$ and
$d_\Gamma(t,\hbar)=-\frac{\mathcal H(t,\hbar)}{\hbar}$,
where $\mathcal H$ is defined in
\eqref{eq:RH-phase-Hamiltonian}. Its regional entire continuations form jointly
holomorphic families on
$\mathbb C_x\times\mathcal P_{V_0}$ and satisfy
the Lax pair
\begin{equation}
\label{LaxPairGamma}
    \hbar\partial_x\Gamma
       =
       \begin{pmatrix}
          -p&2x^2+2qx+2q^2+t\\
          2(x-q)&p
       \end{pmatrix}\Gamma,
    \qquad
    \hbar\partial_t\Gamma
       =
       \begin{pmatrix}
          0&x+2q\\
          1&0
       \end{pmatrix}\Gamma.
\end{equation}
\end{proposition}

\begin{proof}
By \autoref{thm:RH-WKB-existence-uniqueness},
the matrix $\Psi_{\mathrm{can}}$ solves
\autoref{prob:WKB-RHP}.
Multiplication by the constant invertible
matrices in \eqref{eq:RH-constant-gauge}
preserves holomorphy, one-sided boundary
values and boundedness at the vertices.

On a directed arc, the identity
$\Psi_{\mathrm{can},(i)}
  =\Psi_{\mathrm{can},(j)}S_{ij}$ gives
\begin{equation*}
    \Gamma_{(i)}
       =\frac1{2\sqrt2}
          G\Psi_{\mathrm{can},(j)}S_{ij}\sigma_3
       =\Gamma_{(j)}
          \bigl(\sigma_3S_{ij}\sigma_3\bigr).
\end{equation*}
Thus the transformed jump is
$\widetilde S_{ij}=\sigma_3S_{ij}\sigma_3$.
Since $\sigma_3\operatorname{Up}(a)\sigma_3
       =\operatorname{Up}(-a)$ and $\sigma_3\operatorname{Lw}(a)\sigma_3
       =\operatorname{Lw}(-a)$, 
the jumps \eqref{eq:all-jumps} give
\eqref{eq:RH-gauged-jumps} and $\widetilde J_{\mathrm{FD}}
       =\sigma_3J_{\mathrm{FD}}\sigma_3
       =\operatorname{Up}(i)\operatorname{Lw}(i)$.
To check the normalisation at infinity, use the explicit
matrix $N$ in \eqref{eq:leading-normalisation}. Direct multiplication gives
$\frac1{2\sqrt2}GN(x)
       =x^{\frac{\sigma_3}{4}}H_0\sigma_3$.
Since $(H_0)^2=I$ and $\sigma_3$ commutes with
$e^{-\theta\sigma_3}$, it follows that
\[
    H_0x^{-\frac{\sigma_3}{4}}\Gamma e^{-\theta\sigma_3}
       =
       \sigma_3
       \left(
          N^{-1}\Psi_{\mathrm{can}}
                    e^{-\theta\sigma_3}
       \right)\sigma_3.
\]
The normalisation
\eqref{eq:RH-native-infinity-condition}
therefore gives \eqref{eq:RH-gauged-infinity},
with the same uniformity and
$d_\Gamma=-\frac{1}{\hbar}\mathcal H$.
More precisely, conjugating the expansion
\eqref{eq:canonical-normalisation-infinity}
gives
\begin{equation}
\label{eq:RH-gauged-expansion}
    H_0x^{-\frac{\sigma_3}{4}}\Gamma e^{-\theta\sigma_3}
       =
       I-\frac{\mathcal H}{\hbar\sqrt{x}}\sigma_3
       +\frac1{2x}
          \left(
             \frac{\mathcal H^2}{\hbar^2}I
             +q\sigma_1
          \right)
       +O\left(x^{-\frac{3}{2}}\right).
\end{equation}
Also, $\det\Gamma
       =\frac18\det G\,
          \det\Psi_{\mathrm{can}}\,
          \det\sigma_3
       =-1$.
Conversely, if $\widetilde\Gamma$ solves
\autoref{prob:Gamma-RHP}, then $\widetilde\Psi(x,t,\hbar)
       =2\sqrt2\,G^{-1}
          \widetilde\Gamma(x,t,\hbar)\sigma_3$
solves \autoref{prob:WKB-RHP}, by the same jump and normalisation identities.
Uniqueness in \autoref{thm:RH-WKB-existence-uniqueness}
gives $\widetilde\Psi=\Psi_{\mathrm{can}}$,
hence $\widetilde\Gamma=\Gamma$.
Finally, the regional entire continuations and
their joint holomorphy follow from those of
$\Psi_{\mathrm{can}}$.
Since the gauge matrices are constant,
\eqref{250408155610_bis} gives
\[
    \hbar\partial_x\Gamma
       =(GLG^{-1})\Gamma,
    \qquad
    \hbar\partial_t\Gamma
       =(GAG^{-1})\Gamma.
\]
Using the matrices in \eqref{241115131339}
yields \eqref{LaxPairGamma}.
\end{proof}

This change of gauge retains the parameter
$\hbar$. We next record the canonical connection
matrices before comparing with the standard
Painlev\'e~I problem. For comparison with the standard convention, set
$\Gamma_k=(2\sqrt2)^{-1}GF_k\sigma_3$, where the
$F_k$ are the entire canonical matrices introduced
in \eqref{eq:canonical-normalisation-infinity}.
Their connection relations become
\begin{equation}
\label{eq:RH-canonical-Stokes}
    \Gamma_{k+1}=\Gamma_k\mathcal S_k,
    \qquad
    \mathcal S_k=\sigma_3T_k\sigma_3,
    \qquad
    \mathcal S_{2m-1}
       =\operatorname{Up}(s_{2m-1}),
    \quad
    \mathcal S_{2m}
       =\operatorname{Lw}(s_{2m}).
\end{equation}
Thus $s_k=-b_k$, and the multipliers already
computed in \autoref{thm:WKB-RHP} become
\begin{equation}
\label{eq:RH-tritronquee-data}
    s_{-1}=s_0=0,
    \qquad
    s_{-2}=s_1=s_2=i.
\end{equation}
They satisfy the standard relations
\begin{equation}
\label{eq:RH-monodromy-relations}
    s_{k+5}=s_k,
    \qquad
    1+s_ks_{k+1}=-i s_{k+3}.
\end{equation}
These are the conventions of
\cite[Section~2]{Kapaev_2004}.
The corresponding continuation relation is
$\Gamma_{k+5}=\Gamma_k i\sigma_1$.

\paragraph{The standard Painlev\'e I equation and its RHP.}

Independently of our construction, consider the
standard Painlev\'e~I system
\begin{equation}
\label{eq:RH-standard-PI}
    \partial_Tu=v,
    \qquad
    \partial_Tv=6u^2+T,
\end{equation}
or equivalently $u_{TT}=6u^2+T$.
We use $(u,v)$ for this system, reserving $(q,p)$
for the original system containing $\hbar$.
Its standard Lax pair is
\begin{equation}\label{eq:RH-unit}
    \partial_X\Phi
       =
       \begin{pmatrix}
          -v&2X^2+2uX+2u^2+T\\
          2(X-u)&v
       \end{pmatrix}\Phi,\quad   
    \partial_T\Phi
       =
       \begin{pmatrix}
          0&X+2u\\
          1&0
       \end{pmatrix}\Phi.
\end{equation}

Introduce the fixed contour $\mathcal C_{\mathrm{RH}}$
\begin{equation}
\label{eq:RH-straight-contour}
    \gamma_k=e^{\frac{2\pi i k}{5}}\mathbb R_{>0}
       \quad(-2\leq k\leq2),
    \qquad
    \rho=\mathbb R_{<0},
    \qquad
    \mathcal C_{\mathrm{RH}}=\{0\}\cup\rho
              \cup\bigcup_{k=-2}^{2}\gamma_k.
\end{equation}
All rays are oriented away from the origin, with
the plus side on their left. Use principal powers
on $\mathbb C\setminus\mathbb R_{\leq0}$, and set
\[
    \vartheta(X,T)=\frac45X^{\frac{5}{2}}+TX^{\frac{1}{2}}.
\]
The standard RHP with the particular Stokes data
\eqref{eq:RH-tritronquee-data} is the following.

\begin{problem}[The canonical tritronqu\'ee RHP]
\label{prob:canonical-RHP}
For fixed $T\in\mathbb C$, find a matrix
$\Phi(X,T)$ satisfying the following conditions.
\begin{enumerate}
\item
It is holomorphic on $\mathbb C\setminus\mathcal C_{\mathrm{RH}}$
and has continuous one-sided boundary values
on each open contour ray.

\item
Its boundary values satisfy $\Phi_+=\Phi_-J$, where
\begin{equation}
\label{eq:RH-explicit-jumps}
    J|_{\gamma_{-1}}=J|_{\gamma_0}=I,
    \qquad
    J|_{\gamma_{\pm2}}=\operatorname{Lw}(i),
    \qquad
    J|_{\gamma_1}=\operatorname{Up}(i),
    \qquad
    J|_\rho=-i\sigma_1.
\end{equation}

\item
It is bounded at $X=0$ in every angular component
of $\mathbb C\setminus\mathcal C_{\mathrm{RH}}$.

\item
For some scalar $d_0(T)$, not prescribed as part
of the data,
\begin{equation}
\label{eq:RH-infinity-condition}
    H_0X^{-\frac{\sigma_3}{4}}\Phi(X,T)
          e^{-\vartheta(X,T)\sigma_3}
       =I+\frac{d_0(T)\sigma_3}{\sqrt X}
          +O(X^{-1}),
       \qquad X\longrightarrow\infty.
\end{equation}
For fixed $T$, this estimate is uniform in each
closed angular component up to its contour
boundaries, with the corresponding one-sided
values.
\end{enumerate}
\end{problem}

For later use, this RHP has at most one solution.
Indeed, all jump matrices have determinant one.
The determinant of any solution therefore extends
to an entire function and tends to
$\det H_0=-1$ at infinity. It consequently equals
$-1$, so every solution is invertible. Moreover, if $\Phi^{(1)}$ and $\Phi^{(2)}$ are two solutions,
put $M=\Phi^{(1)}(\Phi^{(2)})^{-1}$ and $\Delta=d_{0,1}-d_{0,2}$. Cancellation of the jumps and boundedness at the
origin make $M$ entire. The normalisation
\eqref{eq:RH-infinity-condition} gives
\begin{equation*}
      M
       =X^{\frac{\sigma_3}{4}}H_0
          \left[
             I+\frac{\Delta\sigma_3}{\sqrt X}
               +O(X^{-1})
          \right]
          H_0X^{-\frac{\sigma_3}{4}}
       =I+\Delta E_{12}+O\left(X^{-\frac{1}{2}}\right),
       \qquad
       E_{12}=\begin{pmatrix}0&1\\0&0\end{pmatrix}.  
\end{equation*}
Liouville's theorem implies $M=I+\Delta E_{12}$. Conjugating this exact identity gives
\[
    H_0X^{-\frac{\sigma_3}{4}}MX^{\frac{\sigma_3}{4}}H_0
       =I+\frac{\Delta}{\sqrt X}H_0E_{12}H_0,
    \qquad
    H_0E_{12}H_0
       =\frac12\begin{pmatrix}1&-1\\1&-1\end{pmatrix}.
\]
On the other hand, applying
\eqref{eq:RH-infinity-condition} to both solutions
gives
\[
    H_0X^{-\frac{\sigma_3}{4}}MX^{\frac{\sigma_3}{4}}H_0
       =I+\frac{\Delta\sigma_3}{\sqrt X}
          +O(X^{-1}).
\]
Comparison of the off-diagonal coefficients of
$X^{-\frac{1}{2}}$ forces $\Delta=0$, and hence $M=I$.

\bigskip

We also recall the standard inverse monodromy
correspondence for this Lax pair \eqref{eq:RH-unit}; see
\cite[Section~2, equations~(2.1)--(2.15)]{Kapaev_2004}.
On a domain of $T$ where the RHP is solvable, its
solution depends holomorphically on $T$ and
determines a holomorphic Painlev\'e pair $(u,v)$.
More precisely, writing
\[
    Y(X,T)
       \coloneqq H_0X^{-\frac{\sigma_3}{4}}\Phi(X,T)
                    e^{-\vartheta(X,T)\sigma_3},
\]
its expansion is
\begin{equation}
\label{eq:RH-scaled-expansion}
    Y
       =I-H_u\sigma_3X^{-\frac{1}{2}}
          +\frac12((H_u)^2I+u\sigma_1)X^{-1}
          +O(X^{-\frac{3}{2}}),
    \qquad
    H_u=\frac{v^2}{2}-2u^3-Tu.
\end{equation}
The Painlev\'e variables are recovered by
\begin{equation}
\label{eq:RH-recovery}
    u(T)=2\lim_{X\to\infty}X Y_{12}(X,T),
    \qquad
    v(T)=\partial_Tu(T).
\end{equation}
Conversely, a holomorphic Painlev\'e pair with
the prescribed Stokes data determines its
canonically normalised piecewise wave matrix.
Thus, with the gauge, branches and normalisation
fixed, this gives a one-to-one correspondence
between the normalised RHP solution and the
Painlev\'e pair on pole-free domains.
In particular, two pairs giving the same
normalised RHP must coincide.

\bigskip

For the Stokes data
\eqref{eq:RH-tritronquee-data}, the recovered
solution $u(T)$ is the \emph{classical tritronqu\'ee solution}
in the convention of \cite{Kapaev_2004}.
We denote it by $u_{\mathrm{tri}}$ and set
$v_{\mathrm{tri}}(T)=\partial_Tu_{\mathrm{tri}}(T)$.
Its canonical RHP matrix is denoted by $\Phi_{\mathrm{tri}}$;
wherever \autoref{prob:canonical-RHP} is solvable,
$\Phi_{\mathrm{tri}}$ is its unique solution.
This characterisation is recalled in
\cite[Theorem~2.1, equation~(2.72) and
Remark~2.3]{Kapaev_2004}.

\paragraph{Removal of $\hbar$ and identification with the tritronqu\'ee solution.}

We now return to the matrix $\Gamma(x,t,\hbar)$ constructed
from the Borel-summed solution $(q,p)$ at the beginning of this subsection.
On $S_{0,R_{V_0}}$, use the holomorphic logarithm
$\operatorname{Log}\hbar$ satisfying $-\frac{\pi}{2}
       <\operatorname{Im}\operatorname{Log}\hbar
       <\frac{\pi}{2}$,
and define
$\hbar^\beta=\exp(\beta\operatorname{Log}\hbar)$.
In particular, $\hbar^{\frac{1}{5}}>0$ for $\hbar>0$.

For fixed $\hbar$, introduce
\begin{equation}
\label{eq:RH-scaling}
    \begin{gathered}
    X=\hbar^{-\frac{2}{5}}x,
    \quad
    T=\hbar^{-\frac{4}{5}}t,\quad
    u_\hbar(T)
       =\hbar^{-\frac{2}{5}}q\left(\hbar^{\frac{4}{5}}T,\hbar\right),
    \quad
    v_\hbar(T)
       =\hbar^{-\frac{3}{5}}p\left(\hbar^{\frac{4}{5}}T,\hbar\right),\\
    \Phi_\hbar(X,T)
       =\hbar^{-\frac{\sigma_3}{10}}
          \Gamma\left(\hbar^{\frac{2}{5}}X,\hbar^{\frac{4}{5}}T,\hbar\right).
    \end{gathered}
\end{equation}
Here $T\in\hbar^{-\frac{4}{5}}V_0$, and the last formula
initially defines a piecewise matrix on the
rescaled WKB contour.
The subscripts on $u_\hbar$ and $v_\hbar$ record
their possible dependence on the fixed parameter
$\hbar$ before the RHP identification. Differentiation at fixed $\hbar$ gives
\[
    \partial_Tu_\hbar
       =\hbar^{\frac{2}{5}}\partial_tq
       =\hbar^{-\frac{3}{5}}p=v_\hbar,
    \qquad
    \partial_Tv_\hbar=6u_\hbar^2+T.
\]
Likewise, substitution of \eqref{eq:RH-scaling}
in the Lax pair \eqref{LaxPairGamma} gives
the standard Lax pair \eqref{eq:RH-unit}, with
$(u,v)=(u_\hbar,v_\hbar)$.
Thus the rescaling removes $\hbar$ from the
differential equations.

To compare the RHPs, we also replace the
rescaled WKB contour by $\mathcal C_{\mathrm{RH}}$ and move the
branch cut to the negative real axis.
This is done using the entire canonical
matrices, rather than by changing their
differential equations.
More precisely, set
\[
    \Phi_{\hbar,k}(X,T)
       =\hbar^{-\frac{\sigma_3}{10}}
          \Gamma_k(\hbar^{\frac{2}{5}}X,\hbar^{\frac{4}{5}}T,\hbar).
\]
Restrict these matrices to the angular gaps of
$\mathcal C_{\mathrm{RH}}$ with the standard sector indexing.
We retain the notation $\Phi_\hbar$ for the
resulting piecewise matrix.
The following theorem identifies it.

\begin{theorem}[Identification with the tritronqu\'ee solution]
\label{thm:canonical-RHP}
For every $(t,\hbar)\in\mathcal P_{V_0}$,
\autoref{prob:canonical-RHP} has a unique solution
at $T=\hbar^{-\frac{4}{5}}t$.
It is the matrix $\Phi_\hbar$ obtained from
$\Psi_{\mathrm{can}}$ by the gauge transformation,
rescaling and canonical continuation described
above. It coincides with $\Phi_{\mathrm{tri}}$.
Consequently,
\begin{equation}
\label{eq:RH-q-identification}
    q(t,\hbar)
       =\hbar^{\frac{2}{5}}
            u_{\mathrm{tri}}\left(\hbar^{-\frac{4}{5}}t\right),\qquad
    p(t,\hbar)
       =\hbar^{\frac{3}{5}}
            v_{\mathrm{tri}}\left(\hbar^{-\frac{4}{5}}t\right).
\end{equation}
In particular, the rescaled pair $(u_\hbar,v_\hbar)$ depends only on $T$ but not on $\hbar$.
\end{theorem}

\begin{proof}
We verify all conditions of the canonical RHP.
Each $\Phi_{\hbar,k}$ is entire in $X$, since
$\Gamma_k$ is entire in $x$.
Its restriction to an angular gap is therefore
holomorphic, has the required one-sided boundary
values, and is bounded at $X=0$.
The connection identities
\eqref{eq:RH-canonical-Stokes} are unchanged by
the common left multiplication and the rescaling.
For $-1\leq k\leq2$, put $\Phi_{\hbar,k}$ in
the gap between $\gamma_{k-1}$ and $\gamma_k$.
Put $\Phi_{\hbar,-2}$ between $\rho$ and
$\gamma_{-2}$ in the lower half-plane, and
$\Phi_{\hbar,3}$ between $\gamma_2$ and $\rho$
in the upper half-plane.
Across $\gamma_k$, the prescribed orientation
then gives $(\Phi_\hbar)_+
       =(\Phi_\hbar)_-\mathcal S_k$.
Together with \eqref{eq:RH-tritronquee-data},
these are the jumps in
\eqref{eq:RH-explicit-jumps}.
On the negative real axis, the plus side is
the lower side. The relation
$\Gamma_3=\Gamma_{-2}i\sigma_1$ therefore gives
$(\Phi_\hbar)_+
       =(\Phi_\hbar)_-(-i\sigma_1)$,
as required.
The canonical sectors
\eqref{eq:RH-canonical-sectors} become, after
rescaling,
\[
    \frac{2\pi}{5}\left(k-\frac32\right)
       <\arg X<
    \frac{2\pi}{5}\left(k+\frac12\right).
\]
Each closed angular gap assigned to
$\Phi_{\hbar,k}$ lies strictly inside this
canonical sector. Thus its asymptotic expansion
holds uniformly up to the contour boundaries. With the corresponding determinations of the
fractional powers, one has
\[
    \theta\left(\hbar^{\frac{2}{5}}X,\hbar^{\frac{4}{5}}T,\hbar\right)
       =\vartheta(X,T),
    \qquad
    \hbar^{-\frac{\sigma_3}{10}}
       \left(\hbar^{\frac{2}{5}}X\right)^{\frac{\sigma_3}{4}}
       =X^{\frac{\sigma_3}{4}}.
\]
Moreover, $H_{u_\hbar}
       \coloneqq
       \frac{v_\hbar^2}{2}-2u_\hbar^3-Tu_\hbar
       =\hbar^{-\frac{6}{5}}\mathcal H$.
Consequently, \eqref{eq:RH-gauged-expansion}
becomes
\[
    Y_\hbar
       \coloneqq
       H_0X^{-\frac{\sigma_3}{4}}\Phi_\hbar
                         e^{-\vartheta\sigma_3}
       =
       I-\frac{H_{u_\hbar}\sigma_3}{\sqrt X}
       +\frac{(H_{u_\hbar})^2I+u_\hbar\sigma_1}{2X}
       +O\left(X^{-\frac{3}{2}}\right).
\]
This proves the normalisation
\eqref{eq:RH-infinity-condition}, with
$d_0=-H_{u_\hbar}$, and hence establishes
existence of a solution of
\autoref{prob:canonical-RHP}.
Its uniqueness was proved above.
The recovery formula \eqref{eq:RH-recovery}
applied to this expansion gives exactly
$u_\hbar$, and the time equation gives
$v_\hbar=\partial_Tu_\hbar$.
The jumps, branches, finite-point condition
and normalisation are those of the canonical
tritronqu\'ee RHP.
The RHP correspondence therefore identifies
$\Phi_\hbar=\Phi_{\mathrm{tri}}$, 
    $u_\hbar=u_{\mathrm{tri}}$ and 
    $v_\hbar=v_{\mathrm{tri}}$
on the resulting pole-free domains.
Inverting the rescaling proves
\eqref{eq:RH-q-identification}.
The identification of the Painlev\'e~I solutions
then extends by meromorphic continuation.
There is also a direct explanation of the
absence of residual $\hbar$-dependence.
If two pairs $(t,\hbar)$ give the same $T$,
the corresponding matrices solve precisely
the same normalised RHP.
Uniqueness makes these matrices equal, and
the recovery formula makes their values of
$u_\hbar$ equal.
The same argument on overlapping
$T$-neighbourhoods identifies their derivatives
$v_\hbar$.

Finally, the tritronqu\'ee asymptotic sector can
also be recovered from the Borel construction.
Allowing $V_0\Subset V$ to vary, fix
$\theta_T\in\left(-\frac{\pi}{5},\frac{7\pi}{5}\right)$ and choose $t_j\in V$
such that
\[
    \theta_\hbar
       =\frac54(\arg t_j-\theta_T)
          \in\left(-\frac{\pi}{2},\frac{\pi}{2}\right).
\]
For sufficiently large $T$ near this direction,
take $\hbar=\left(\frac{t_j}{T}\right)^{\frac{5}{4}}$ with this argument.
It lies in the Borel disk at $t_j$.
Since $q_1=0$, the factorial estimates give
$q(t_j,\hbar)=\tau(t_j)+O(\hbar^2)$.
The identification above then yields
\[
    u_{\mathrm{tri}}(T)
       =\sqrt{\frac{e^{-i\pi}}{6}T}+O(T^{-2}),
    \qquad
    -\frac{\pi}{5}<\arg T<\frac{7\pi}{5}.
\]
The square root is continued from
\eqref{eq:RH-representative-sector}.
A finite angular cover gives uniformity on
every closed subsector, and the higher-order
factorial estimates give the full asymptotic
expansion. This ends the proof of \autoref{thm:canonical-RHP}.
\end{proof}

A related exact WKB approach to the direct
monodromy problem is developed in
\cite{Iwaki-P1}. Here the identification follows
from the common Borel sum at the double turning
point, the resulting vanishing Stokes multipliers,
and uniqueness of the canonical RHP.

\begin{corollary}[The five Borel sums]
\label{cor:tritronquee-characterisation}
The five Borel-resummed power-series solutions
in the five time sectors are the five deformed
tritronqu\'ee solutions.
Their canonically normalised Stokes data have
two consecutive zero multipliers and three
multipliers equal to $i$.
\end{corollary}

\begin{proof}
The weighted change
\begin{equation}
\label{WeightedChange}
    (x,t,q,p,\hbar)\longmapsto
    \left(
       e^{\frac{2i\alpha}{5}}x,\,
       e^{\frac{4i\alpha}{5}}t,\,
       e^{\frac{2i\alpha}{5}}q,\,
       e^{\frac{3i\alpha}{5}}p,\,
       e^{i\alpha}\hbar
    \right)
\end{equation}
takes phase zero to phase $\alpha$ and preserves
the Painlev\'e system. In the original Lax gauge,
the corresponding change of the wave matrix is
$\Psi\mapsto e^{-\frac{i\alpha}{10}\sigma_3}\Psi$. For the standard equation, put
$\omega=e^{\frac{2\pi i}{5}}$. Its fivefold symmetry is
\[
    u(T)\longmapsto\omega^{2n}u\left(\omega^nT\right),
    \qquad
    v(T)\longmapsto\omega^{3n}v\left(\omega^nT\right),
    \qquad \forall\, n\in\mathbb Z.
\]
It cyclically permutes the five pairs of
consecutive vanishing Stokes multipliers.
Applying \autoref{thm:canonical-RHP} with the
corresponding branches and labels identifies
the five Borel sums with the five rotated
tritronqu\'ee solutions.
\end{proof}

\section{Conclusion and outlook}
\label{sec:conclusion}

In this article, we have constructed the RHP for the
Painlev\'e~I Lax system associated with the Borel-summed
perturbative solution
(\autoref{thm:RH-WKB-existence-uniqueness}).
The central step is the construction of an exact WKB wave
matrix in a full neighbourhood of the double turning point.
The auxiliary time equation first gives formal amplitude
coefficients that are holomorphic there, and uniform Borel
estimates on moving action strips then establish the
corresponding analytic statement
(\autoref{thm:no-jump-quadrivalent}).
With consistent column normalisations, each of the four
incident jump matrices is therefore the identity
(\autoref{prop:RH-nodal-vanishing}).
Comparison at spatial infinity and the cyclic relations
determine the remaining Stokes data.
The resulting normalised Riemann--Hilbert problem identifies
the sectorial Borel sums with the classical tritronqu\'ee
solutions after the standard Painlev\'e scaling in $\hbar$
(\autoref{thm:canonical-RHP} and
\autoref{cor:tritronquee-characterisation}). Two extensions deserve further study:
\begin{itemize}
    \item
    Several steps of the construction apply to more general
    isomonodromic systems.
    In particular, \autoref{app:matrix-formal-construction}
    formulates the formal matrix recursion for several
    deformation parameters under explicit flatness and
    spectral nondegeneracy assumptions.
    It is therefore natural to investigate whether the
    strategy extends to the other five Painlev\'e equations
    and to more general isomonodromic systems in
    $\mathfrak{gl}_2(\mathbb C)$, including the classes for
    which explicit Lax matrices and Hamiltonians are
    available in
    \cite{MarchalAlameddineP1Hierarchy2023,
          MarchalOrantinAlameddine2022}.
    A useful starting point is to identify a deformation
    direction whose leading matrix has distinct eigenvalues
    at a double turning point of the spatial system.
    Extending the analytic argument requires suitable
    families of unbounded action domains and uniform Borel
    estimates on those families.

    \item
    A second direction is the treatment of solutions
    containing exponential sectors and free transseries
    parameters.
    For Painlev\'e~I, this would allow a comparison with the
    two-parameter construction of \cite{Iwaki-P1} beyond the
    $0$-parameter case considered here.
    Such an extension requires an asymptotic and summation
    framework adapted to the exponential scales, with
    domains on which their relative growth is controlled.
    In particular, the uniform Gevrey-$1$ power-series
    asymptotics on a full Borel disk used in this paper
    cannot simply be retained in the presence of nonzero
    exponentially small corrections.
\end{itemize}

\appendix

\section{Formal conventions}
\label{250410111125}

In this paper, only ordinary formal series and finite sums of WKB expressions are
needed. For a domain $U$, $\mathcal O(U)[[\hbar]]$ denotes formal series whose coefficients are holomorphic on $U$. Such a series is a unit (i.e. invertible) locally at a point if its constant coefficient (the $\hbar^0$ coefficient) is nonzero there. It is differentiated coefficientwise: no convergence in $\hbar$
is asserted.

\medskip 

To make the basis assertion precise in \autoref{250109214318}, at a fixed basepoint take the group algebra of the additive group of holomorphic phase germs, with coefficients in holomorphic formal series. Its elements are finite sums $\underset{1\leq j\leq N}{\sum}\hat a_j e^{\frac{S_j}{\hbar}}$. The exponential
symbols satisfy $e^{\frac{S}{\hbar}}e^{\frac{T}{\hbar}}=e^{\frac{S+T}{\hbar}}$ and distinct phase germs label independent summands. The operators used in the Lax system act by
\begin{equation}\label{xteq}
 \hbar\partial_z(\hat a e^{\frac{S}{\hbar}})
       =(\hbar\partial_z\hat a+(\partial_zS)\hat a)e^{\frac{S}{\hbar}},
       \qquad z\in\{x,t\}.
\end{equation}
An element is constant precisely when it is annihilated by both
$\hbar\partial_x$ and $\hbar\partial_t$. The algebra of constants for these two operators is
\begin{equation}
 \mathcal T=\mathbb C[[\hbar]][e^{\frac{a}{\hbar}}:a\in\mathbb C],
\end{equation}
again with finite exponential support. Every element of $\mathcal T$ is annihilated by both operators, by \eqref{xteq}. Conversely, let $\hat F=\underset{j=1}{\overset{N}{\sum}}\hat a_j e^{\frac{S_j}{\hbar}}$ be annihilated by both operators, with distinct phases and
nonzero amplitudes. Independence of the exponential symbols gives, for each $j\in\llbracket 1,N\rrbracket$:
\begin{equation*}
    \hbar\partial_x\hat a_j+(\partial_xS_j)\hat a_j=0,
    \qquad
    \hbar\partial_t\hat a_j+(\partial_tS_j)\hat a_j=0.
\end{equation*}
If $\hbar^{m_j}a_{j,m_j}$ is the first nonzero term of $\hat a_j$, comparison at order $\hbar^{m_j}$ yields $(\partial_xS_j)a_{j,m_j}=(\partial_tS_j)a_{j,m_j}=0$.
Since the ring of holomorphic germs has no zero divisors, both derivatives of $S_j$ vanish, so $S_j$ is a complex constant. The remaining equations imply that every coefficient of $\hat a_j$ is also constant. Thus $\hat F\in\mathcal T$, which proves the asserted description of the common kernel.
A nonzero normalisation of a single WKB column is a transmonomial $\hbar^m\hat b e^{-\frac{\kappa}{\hbar}}$, with $m\in\mathbb Z_{\geq0}$, $\kappa\in\mathbb C$, and
$\hat b\in\mathbb C[[\hbar]]^\times$. It is invertible in $\mathcal T$ exactly when $m=0$. Comparisons involving negative powers of $\hbar$ are understood in the Laurent localisation. No unrestricted infinite exponential sums, topology on a transseries space, or summability of such sums is used in this paper.

\section{Proofs for the formal WKB wave matrix}
\label{250109110415}

\subsection{Proof of \autoref{241113190220} and \autoref{250716102028}}
\label{250408124924}

We prove \autoref{241113190220} and
\autoref{250716102028} without assuming any previously
constructed wave matrix. We work on a sufficiently small
polydisc around $(\tau_0,w_0)$ where $\tau w\neq0$,
fix $\varepsilon\in\{+1,-1\}$, and suppress the subscript
$\varepsilon$ during the proof when no ambiguity can arise.

\paragraph{Proof of existence in \autoref{241113190220}.}
The formal ratio $\hat r_\varepsilon$ has already been constructed by
\eqref{eq:formal-r-recursion}: it is the unique formal solution
of \eqref{eq:formal-time-Riccati} with leading term
$r_0=\frac{\varepsilon w}{2}$. 
We now prove that it also satisfies the spatial Riccati
equation, and then reconstruct the amplitude $\hat{a}_\varepsilon$.
Put
\begin{equation}
 K\coloneqq x^2+x\hat q+\hat q^2+\frac{t}{2},\qquad
 f(r)\coloneqq\frac{x}{2}+\hat q-2r^2,\qquad
 B(r)\coloneqq K-2\hat p r-4(x-\hat q)r^2,
\end{equation}
and define the spatial defect
$\hat R=\hbar\partial_x\hat r-B(\hat r)$. The time Riccati equation \eqref{eq:formal-time-Riccati} is thus
\begin{equation}\label{timeRicc}
    \hbar\partial_t \hat{r}=f(\hat{r})
\end{equation}

We shall denote $\partial_r B=B_r$ and $\partial_r f=f_r$, derivatives of $B$ and $f$ with  $(x,t)$ fixed. Similarly $\partial_x f=f_x$ and $\partial_tf=f_t$ denote derivatives where $r$ is kept fixed in the next calculation.  The formal Painlev\'e equations for $(\hat{q}, \hat{p})$ give
\begin{align*}
K_r=&0,
    \qquad
    f_r=-4r, \quad f_x=\frac{1}{2}, \qquad
    B_r=-2\hat p-8(x-\hat q)r,\\
        \hbar K_t=&(x+2\hat q)\hbar\partial_t\hat q+\frac{\hbar}{2}=(x+2\hat q)\hat p+\frac{\hbar}{2}, \qquad
    \hbar f_t=\hbar\partial_t\hat q=\hat p\\
\hbar B_t
   =&\frac{\hbar}{2}+(x+2\hat q)\hat p-2(6\hat q^2+t)r+4\hat p r^2,\qquad
 0=B  f_r+\hbar f_x-\hbar B_t-f B_r.
\end{align*}
The last identity follows by expanding the displayed polynomials. Differentiating the spatial defect $\hat{R}$ and using $\hbar\partial_t\hat r=f(\hat r)$ from \eqref{timeRicc} gives
\begin{align*}
    \hbar\partial_t\hat R
      &=\hbar^2\partial_x\partial_t\hat r
        -\hbar B_t-\hbar B_r\partial_t\hat r
      =\hbar f_x+\hbar f_r\partial_x\hat r
        -\hbar B_t-
        B_rf\\
      &=f_r\hat R
        +\bigl(f_rB+\hbar f_x-\hbar B_t-B_rf\bigr)\\
      &=-4\hat r\,\hat R.
\end{align*}
The third equality uses the definition  $\hbar\partial_x\hat r=\hat R+B(\hat r)$, and the last equality uses the polynomial identity above and $f_r=-4r$. Comparison of its constant coefficients gives $4r_0R_0=0$ so that $R_0=0$. Successive comparison then gives $R_n=0$ for all $n\in \mathbb{N}$ so that $\hat{R}=0$. Thus
\begin{equation}
\label{eq:formal-spatial-Riccati-B}
 \hbar\partial_x\hat r=B(\hat{r})=K-2\hat p\hat r-4(x-\hat q)\hat r^2.
\end{equation}
This is where flatness enters the construction. Its use requires
only the formal Painlev\'e equations and the already defined time
Riccati series. Define $\hat P,\hat T$ by 
\begin{equation}\label{defPT} \hat P\coloneqq \hat p+4(x-\hat q)\hat r,
 \qquad \hat T \coloneqq 2\hat r.
\end{equation}
The time Riccati equation \eqref{timeRicc} and the spatial Riccati equation \eqref{eq:formal-spatial-Riccati-B}, together with the formal Painlev\'e equations $\hbar\partial_t\hat q=\hat p$ and $\hbar\partial_t\hat p=6\hat q^2+t$, imply
\begin{equation*}
 \hbar(\partial_t\hat P-\partial_x\hat T)
   =6\hat q^2+t-4\hat p\hat r
        +4(x-\hat q)f(\hat r)-2B(\hat r)=0\,\, \Rightarrow\, \partial_t \hat{P}=\partial_x \hat{T}
\end{equation*}
Their leading terms are $P_0=4(x-\tau)r_0=\varepsilon y$ and $T_0=\varepsilon w$.
Since $\mathrm dZ=y\,\mathrm dx+w\,\mathrm dt$, this proves
that 
\begin{equation}
\label{eq:formal-amplitude-form_bis}
 \hat\beta_\varepsilon
    =\frac{\hat P_\varepsilon-\varepsilon y}{\hbar}\,\mathrm dx
      +\frac{\hat T_\varepsilon-\varepsilon w}{\hbar}\,\mathrm dt
\end{equation}
is  an ordinary formal closed one-form. On the polydisc it has a unique coefficientwise primitive vanishing at $b_0$. The formula 
\begin{equation}
\label{eq:formal-amplitude-primitive_bis}
 \hat{a}_\varepsilon(\tau,w;\hbar)
       \coloneqq \hat{b}_\varepsilon(\hbar)
              \exp\left(\int_{(\tau_0,w_0)}^{(\tau,w)}\hat\beta_\varepsilon\right).
\end{equation} 
therefore defines a formal invertible power series with the prescribed initial value $\hat{a}_\varepsilon(\tau_0,w_0;\hbar)=\hat b_\varepsilon(\hbar)$.

We now define
\begin{equation*}
    \hat\psi_\varepsilon
       \coloneqq
       \hbar^{m_\varepsilon}\hat a_\varepsilon
       e^{\frac{S_\varepsilon}{\hbar}},
    \qquad
    \hat{\mathbf v}_\varepsilon
       \coloneqq
       \begin{pmatrix}
           \hat\psi_\varepsilon\\
           \hat r_\varepsilon\hat\psi_\varepsilon
       \end{pmatrix},
\end{equation*}
where $\hat a_\varepsilon$ is given by
\eqref{eq:formal-amplitude-primitive_bis} and
$S_\varepsilon$ by \eqref{eq:formal-column-construction}.
We verify directly that this column satisfies both equations
of the formal Lax system \eqref{250408174515}. Indeed, by \eqref{eq:formal-amplitude-primitive_bis}, $\mathrm d\hat a_\varepsilon=\hat a_\varepsilon\hat\beta_\varepsilon$. Using \eqref{eq:formal-amplitude-form_bis}, we therefore obtain
\begin{equation*}
    \hbar\partial_x\hat a_\varepsilon
       =(\hat P_\varepsilon-\varepsilon y)\hat a_\varepsilon,
    \qquad
    \hbar\partial_t\hat a_\varepsilon
       =(\hat T_\varepsilon-\varepsilon w)\hat a_\varepsilon.
\end{equation*}
Moreover, the phase satisfies
$\partial_xS_\varepsilon=\varepsilon y$ and
$\partial_tS_\varepsilon=\varepsilon w$.
Since $\hbar$ is held fixed under these derivatives, the product rule gives
\begin{align*}
    \hbar\partial_x\hat\psi_\varepsilon
      &=
      \hbar^{m_\varepsilon}e^{\frac{S_\varepsilon}{\hbar}}
      \left(
          \hbar\partial_x\hat a_\varepsilon
          +\hat a_\varepsilon\partial_xS_\varepsilon
      \right)
      =\hat P_\varepsilon\hat\psi_\varepsilon,\\
    \hbar\partial_t\hat\psi_\varepsilon
      &=
      \hbar^{m_\varepsilon}e^{\frac{S_\varepsilon}{\hbar}}
      \left(
          \hbar\partial_t\hat a_\varepsilon
          +\hat a_\varepsilon\partial_tS_\varepsilon
      \right)
      =\hat T_\varepsilon\hat\psi_\varepsilon.
\end{align*}
By the definitions of $\hat P_\varepsilon$ and
$\hat T_\varepsilon$ given in \eqref{defPT}, these are precisely the first-row
equations:
\begin{equation*}
    \hbar\partial_x\hat\psi_\varepsilon
       =\hat p\hat\psi_\varepsilon
        +4(x-\hat q)\hat r_\varepsilon\hat\psi_\varepsilon,
    \qquad
    \hbar\partial_t\hat\psi_\varepsilon
       =2\hat r_\varepsilon\hat\psi_\varepsilon.
\end{equation*}
of the Lax system.  For the second component, the spatial Riccati equation
\eqref{eq:formal-spatial-Riccati-B} gives
\begin{align*}
    \hbar\partial_x
       (\hat r_\varepsilon\hat\psi_\varepsilon)
      &=
      \left(
          \hbar\partial_x\hat r_\varepsilon
          +\hat r_\varepsilon\hat P_\varepsilon
      \right)\hat\psi_\varepsilon
      =
      \left(
          K-2\hat p\hat r_\varepsilon
          -4(x-\hat q)\hat r_\varepsilon^2
          +\hat r_\varepsilon
             [\hat p+4(x-\hat q)\hat r_\varepsilon]
      \right)\hat\psi_\varepsilon\\
      &=K\hat\psi_\varepsilon
        -\hat p(\hat r_\varepsilon\hat\psi_\varepsilon).
\end{align*}
Similarly, the time Riccati equation \eqref{timeRicc} gives
\begin{equation*}
    \hbar\partial_t
       (\hat r_\varepsilon\hat\psi_\varepsilon)
      =
      \left(
          \hbar\partial_t\hat r_\varepsilon
          +\hat r_\varepsilon\hat T_\varepsilon
      \right)\hat\psi_\varepsilon
      =
      \left(
          \frac{x}{2}+\hat q
          -2\hat r_\varepsilon^2+2\hat r_\varepsilon^2
      \right)\hat\psi_\varepsilon
      =\left(\frac{x}{2}+\hat q\right)\hat\psi_\varepsilon.
\end{equation*}
These are the second-row equations of the Lax system. Thus
\begin{equation}
    \hbar\partial_x\hat{\mathbf v}_\varepsilon
       =\hat L\hat{\mathbf v}_\varepsilon,
    \qquad
    \hbar\partial_t\hat{\mathbf v}_\varepsilon
       =\hat A\hat{\mathbf v}_\varepsilon.
\end{equation}
Finally, the amplitude and phase normalisations give
\begin{equation*}
   \hat\psi_\varepsilon(b_0;\hbar)
      =\hbar^{m_\varepsilon}\hat b_\varepsilon(\hbar)
         e^{-\frac{\kappa_\varepsilon}{\hbar}}
      =\hat c_\varepsilon(\hbar). 
\end{equation*}
The amplitude has holomorphic coefficients and a nowhere vanishing leading coefficient by construction. This proves the existence assertion of \autoref{241113190220}. 

\paragraph{Proof of uniqueness in \autoref{241113190220}.}
Let us now prove uniqueness. Fix $\varepsilon$ and the prescribed
normalisation data $m_\varepsilon$, $\kappa_\varepsilon$,
and $\hat b_\varepsilon(\hbar)$. Suppose that two formal WKB columns
\begin{equation*}
   \hat{\mathbf v}^{(j)}_\varepsilon(\tau,w;\hbar)
      =\hbar^{m_\varepsilon}
        \hat a^{(j)}_\varepsilon(\tau,w;\hbar)
        \begin{pmatrix}1\\ \hat r^{(j)}_\varepsilon\end{pmatrix}
        e^{\frac{S_\varepsilon}{\hbar}},
    \qquad j\in\{1,2\}, 
\end{equation*}
satisfy the formal Lax system \eqref{250408174515} and the conditions of \autoref{241113190220}. Here both columns have the prescribed phase
\begin{equation*}
    S_\varepsilon
       =\varepsilon\bigl(Z-Z(\tau_0,w_0)\bigr)
          -\kappa_\varepsilon,
\end{equation*}
their ratios have the same leading coefficient $r_{\varepsilon,0}=\frac{\varepsilon w}{2}$, and their amplitudes satisfy
\begin{equation*}
    \hat a^{(1)}_\varepsilon(\tau_0,w_0;\hbar)
       =\hat a^{(2)}_\varepsilon(\tau_0,w_0;\hbar)
       =\hat b_\varepsilon(\hbar).
\end{equation*}
After restriction, we may work on a common connected neighbourhood of $(\tau_0,w_0)$.

First, the time equation of the Lax system implies that
both ratios satisfy \eqref{eq:formal-time-Riccati}.
Their coefficients therefore obey the same recursion
\eqref{eq:formal-r-recursion}.
We prove their equality by induction on the order in $\hbar$.
Their leading coefficients agree. If their coefficients
agree through order $n-1$, subtracting the two recursions
at order $n$ gives $4r_{\varepsilon,0}\left(
r^{(1)}_{\varepsilon,n}-r^{(2)}_{\varepsilon,n}\right)=0$. Since $4r_{\varepsilon,0}=2\varepsilon w$ is nowhere zero, the coefficients at order $n$ also agree.
Thus $\hat r^{(1)}_\varepsilon=\hat r^{(2)}_\varepsilon=\hat r_\varepsilon$. Both columns consequently determine the same series $\hat P_\varepsilon$, $\hat T_\varepsilon$ and the same one-form $\hat\beta_\varepsilon$ in \eqref{eq:formal-amplitude-form_bis}. Rewriting the first-row Lax equations using the common phase
gives
\begin{equation*}
      \mathrm d\hat a^{(j)}_\varepsilon
       =\hat a^{(j)}_\varepsilon\hat\beta_\varepsilon,
    \qquad j\in\{1,2\}.  
\end{equation*}
The amplitudes are invertible formal power series, so their quotient is well-defined and satisfies
\begin{equation*}
    \mathrm d\left(
       \frac{\hat a^{(1)}_\varepsilon}
            {\hat a^{(2)}_\varepsilon}
    \right)
      =
      \frac{
        \hat a^{(2)}_\varepsilon
             \mathrm d\hat a^{(1)}_\varepsilon
        -\hat a^{(1)}_\varepsilon
             \mathrm d\hat a^{(2)}_\varepsilon
      }{(\hat a^{(2)}_\varepsilon)^2}
      =0.
\end{equation*}
Every coefficient of this quotient is therefore constant on the connected neighbourhood. Its value at $(\tau_0,w_0)$ is $\frac{\hat b_\varepsilon}{\hat b_\varepsilon}=1$, hence $\hat a^{(1)}_\varepsilon=\hat a^{(2)}_\varepsilon$. The phases, ratios, amplitudes and prescribed powers of $\hbar$ all coincide. Consequently, $\hat{\mathbf v}^{(1)}_\varepsilon=\hat{\mathbf v}^{(2)}_\varepsilon$. This proves uniqueness and completes the proof of \autoref{241113190220}.

\paragraph{Proof of \autoref{250716102028}.}
Apply \autoref{241113190220} for each sign $\varepsilon$ and restrict the two columns to a common sufficiently small simply connected neighbourhood $U$ of $(\tau_0,w_0)$.
Since each column satisfies \eqref{250408174515}, the matrix $\hat\Psi=\left[\hat{\mathbf v}_+\ \hat{\mathbf v}_-\right]$ satisfies both matrix Lax equations. Its uniqueness for the prescribed ordered normalisations follows from the uniqueness of each column. To establish invertibility and compute the leading term, write
\begin{equation*}
      \hat F
      \coloneqq
      \begin{pmatrix}
        \hat a_+&\hat a_-\\
        \hat a_+\hat r_+&\hat a_-\hat r_-
      \end{pmatrix},
    \qquad
    \hat E
      =\hat F\,
        \operatorname{diag}
          \left(\hbar^{m_+},\hbar^{m_-}\right).  
\end{equation*}
Both $\hat F$ and $\hat E$ are formal power series with holomorphic matrix coefficients on $U$. We compute the leading coefficient of $\hat F$ to establish
its invertibility and the formula \eqref{eq:formal-leading-amplitude}. The formal Painlev\'e solution satisfies
\begin{equation*}
    q_0=\tau,\qquad q_1=0,\qquad
    p_0=0,\qquad
    p_1=\partial_t\tau=-\frac{1}{12\tau}.
\end{equation*}
At order $\hbar$, the recursion \eqref{eq:formal-r-recursion} gives
\begin{equation*}
    r_{\varepsilon,1}
       =-\frac{\partial_t r_{\varepsilon,0}}
                {4r_{\varepsilon,0}}
       =-\frac{\partial_t w}{4w}
       =\frac{1}{48\tau w^2},
\end{equation*}
where $\partial_t$ is taken at fixed $x$. Consequently, the definitions of
$\hat P_\varepsilon$ and $\hat T_\varepsilon$ give
\begin{align*}
    P_{\varepsilon,1}
      &=p_1+4(x-\tau)r_{\varepsilon,1}=-\frac{1}{12\tau}
        +\frac{x-\tau}{12\tau w^2}
       =-\frac{1}{4w^2},\\
    T_{\varepsilon,1}
      &=2r_{\varepsilon,1}
       =\frac{1}{24\tau w^2}.
\end{align*}
Here the simplification of $P_{\varepsilon,1}$ uses
$w^2=x+2\tau$. The constant coefficient of the amplitude one-form
\eqref{eq:formal-amplitude-form_bis} is therefore
\begin{equation*}
    \beta_{\varepsilon,0}
      =-\frac{1}{4w^2}\,\mathrm dx
        +\frac{1}{24\tau w^2}\,\mathrm dt
      =-\frac12\,\frac{\mathrm dw}{w}, \quad \text{since} \quad 
    \mathrm dw
       =\frac{1}{2w}\,\mathrm dx
        -\frac{1}{12\tau w}\,\mathrm dt.
\end{equation*}
Taking the constant coefficient of
\eqref{eq:formal-amplitude-primitive_bis} yields
\begin{equation*}
    a_{\varepsilon,0}(\tau,w)
       =b_{\varepsilon,0}
          \exp\left(
              \int_{(\tau_0,w_0)}^{(\tau,w)}
                    \beta_{\varepsilon,0}
          \right)
       =b_{\varepsilon,0}
          \left(\frac{w_0}{w}\right)^{\frac{1}{2}}.
\end{equation*}
After shrinking $U$ if necessary, the square root is holomorphic and is uniquely fixed by requiring its value at $(\tau_0,w_0)$ to be $1$. Since $r_{+,0}=\frac{w}{2}$ and $r_{-,0}=-\frac{w}{2}$, we obtain
\begin{equation*}
    F_0
      =\left(\frac{w_0}{w}\right)^{\frac{1}{2}}
        \begin{pmatrix}
          1&1\\
          \frac{w}{2}&-\frac{w}{2}
        \end{pmatrix}
        \operatorname{diag}(b_{+,0},b_{-,0}).
\end{equation*}
In particular,
\begin{equation*}
    \det F_0
       =\frac{w_0}{w}\,(-w)b_{+,0}b_{-,0}
       =-w_0b_{+,0}b_{-,0}\neq0.
\end{equation*}
Thus $\hat F$ has a formal inverse with holomorphic coefficients on $U$.
If $m_+=m_-=0$, then $\hat E=\hat F$.
This proves both the formal invertibility of $\hat E$
and the leading-term formula
\eqref{eq:formal-leading-amplitude}.
For arbitrary nonnegative integers $m_\pm$, we have
\begin{equation*}
    \hat E^{-1}
      =
      \operatorname{diag}
        \left(\hbar^{-m_+},\hbar^{-m_-}\right)
      \hat F^{-1},
\end{equation*}
which is a matrix of formal Laurent series with holomorphic
coefficients. The diagonal exponential factor in
\eqref{250421155248} is also invertible, so $\hat\Psi$
is a formal fundamental matrix in the corresponding
Laurent localisation.
Finally, suppose that $x_0\neq\tau_0$.
After shrinking $U$, the function $x-\tau$ is nowhere zero.
Hence $x-\hat q$ has a formal inverse with holomorphic
coefficients on $U$.
The elimination leading to \eqref{250408193221} and
\eqref{eq:scalar-time-equation} applies coefficientwise,
with $(q,p)$ replaced by $(\hat q,\hat p)$.
The first component of each column therefore satisfies
both the formal scalar quantum curve and its time equation.

Conversely, any time-compatible formal WKB scalar solution
with the prescribed phase and normalisation determines a
column by the formal version of
\eqref{eq:scalar-reconstruction}:
\begin{equation*}
    \hat{\mathbf v}_\varepsilon
      =
      \begin{pmatrix}
        \hat\psi_\varepsilon\\[2pt]
        \dfrac{\hbar\partial_x\hat\psi_\varepsilon
                   -\hat p\hat\psi_\varepsilon}
              {4(x-\hat q)}
      \end{pmatrix}.
\end{equation*}
The reconstruction argument shows that this column satisfies
the full formal Lax system. Its component ratio has leading
term
\begin{equation*}
    \frac{\partial_xS_\varepsilon}{4(x-\tau)}
       =\frac{\varepsilon y}{4(x-\tau)}
       =\frac{\varepsilon w}{2},
\end{equation*}
so it belongs to the class covered by
\autoref{241113190220}.
The uniqueness of the reconstructed column therefore
proves the asserted uniqueness of the scalar solution. This ends the proof of \autoref{250716102028}.

\subsection{Proof of \autoref{250109214318}}\label{AppendixB2}
Since $m_+=m_-=0$, the inverse of $\hat\Psi$ is the diagonal exponential inverse followed by the
ordinary formal inverse of $\hat E$. If $\hat v$ is any formal
solution in the algebra of Appendix~A, differentiation shows that
$\hat\Psi^{-1}\hat v$ is annihilated by both
$\hbar\partial_x$ and $\hbar\partial_t$. Its entries belong to
$\mathcal T$, by the constant-algebra calculation in Appendix~A.
This proves the basis assertion. The uniqueness of the ratio and
amplitude on each chosen eigenline proves the diagonal comparison
for WKB matrices. In particular, a change of a column's prescribed
normalisation multiplies that column by the same constant factor.

The recursion \eqref{eq:formal-r-recursion} gives rational functions
of $\tau,w$ with possible poles only at $\tau w=0$:
the $t$ derivative at fixed $x$ preserves this ring.
The coefficients of $\hat\beta$ are consequently holomorphic
closed one-forms on $\mathring{\mathsf\Sigma}$. They have primitives
on its universal cover. Their exponentials continue the amplitudes
there. The leading coefficient is a nonzero constant times
$w^{-\frac{1}{2}}$, and so cannot vanish. This proves continuation,
preservation of the $\hbar$-adic order, and preservation of formal
invertibility. The primitive $Z$ itself is already global.

\subsection{Factorial estimates}
\label{app:factorial-type-proof}
Fix $\varepsilon\in\{+1,-1\}$ and write $\hat r=\hat r_\varepsilon=\underset{n\geq 0}{\sum}r_n\hbar^n$, suppressing the sign in the notation. We use the classical Nagumo distance-weighted norms \cite[Section~2, p.~42]{Nagumo1941}. Let $(\tau_0,w_0)$ be a basepoint such that $w_0\tau_0\neq 0$ and let $b_0=(x_0,t_0)$ be the associated point in the $(x,t)$ space. Choose a polydisc
$D_R=\{|x-x_0|<R,\ |t-t_0|<R\}$ whose closure lies in a slightly
larger domain where $\tau w\neq0$. For a scalar function $F$ holomorphic on a neighbourhood of $\overline{D_R}$ and an integer $n\geq0$, put
\begin{equation}
 d(x,t)=R-\max\{|x-x_0|,|t-t_0|\},\qquad
 \mathcal N_n(F)=\sup_{(x,t)\in D_R}\left\{d(x,t)^n|F(x,t)|\right\}.
\end{equation}
These quantities are finite because $F$ is bounded on $\overline{D_R}$. The coefficient functions $q_n$ and $r_n$ considered below satisfy this holomorphy assumption by the choice of $D_R$.
The norms satisfy
\begin{equation}
\label{eq:Nagumo-basic-estimates}
 \mathcal N_{n+m}(FG)\leq\mathcal N_n(F)\mathcal N_m(G),
 \qquad
 \mathcal N_{n+1}(\partial_zF)
       \leq e(n+1)\mathcal N_n(F),\qquad z\in\{x,t\}.
\end{equation}
Indeed, the product estimate is immediate. For $n\geq1$, fix an arbitrary point $(x,t)\in D_R$, write $d=d(x,t)$, and set $\rho=\frac{d}{n+1}$. Since $\rho<d$, the closed coordinate disk $\{(\zeta,t):|\zeta-x|\leq\rho\}$ lies inside $D_R$. Indeed, $|\zeta-x_0|\leq |x-x_0|+\rho\leq R-d+\rho<R$, while $|t-t_0|<R$. Thus $\zeta\mapsto F(\zeta,t)$ is holomorphic on a neighbourhood of this closed disk, and Cauchy's estimate applies on its boundary circle. Throughout the closed coordinate disk the distance is at least $d-\rho=\frac{nd}{n+1}$, since the distance function is $1$-Lipschitz in either coordinate. This gives
$|\partial_xF(x,t)|\leq (n+1)\left(1+\frac{1}{n}\right)^n\mathcal N_n(F)d^{-n-1}$. Multiplying by $d^{n+1}$, using $\left(1+\frac{1}{n}\right)^n\leq e$, and taking the
supremum over $(x,t)\in D_R$ gives $\mathcal N_{n+1}(\partial_xF)\leq e(n+1)\mathcal N_n(F)$. For $n=0$, take $\rho=\frac{d}{2}$. Cauchy's estimate gives
$|\partial_xF(x,t)|\leq 2\frac{\mathcal N_0(F)}{d}$, hence $\mathcal N_1(\partial_xF)\leq2\mathcal N_0(F)$, which is stronger than the required bound since $2<e$. The estimate for $\partial_t F$ follows by the same argument with $x$ fixed.
By \eqref{eq:q-factorial-bound}, there exist constants $C,A>0$ such that
\begin{equation*}
       \sup_{\overline{D_R}}|q_n|\leq CA^n n!,
    \qquad \forall\, n\geq0. 
\end{equation*}
Since $d(x,t)\leq R$, this gives
\begin{equation*}\mathcal N_n(q_n)\leq CK^n n! \,\, \text{ with }\,\, K=AR.
    \end{equation*}
Define the nonnegative real numbers
\begin{equation}
    u_n\coloneqq\frac{\mathcal N_n(r_n)}{n!},
    \qquad \forall\, n\geq0,
    \qquad
    M\coloneqq\sup_{D_R}|(4r_0)^{-1}|<\infty.
\end{equation}
The finiteness of $M$ follows from
$4r_0=2\varepsilon w$ and the nonvanishing of $w$ on
$\overline{D_R}$. Recall the coefficient recursion
\eqref{eq:formal-r-recursion}, with the sign $\varepsilon$
suppressed:
\begin{equation*}
    r_n
      =\frac{1}{4r_0}
        \left(
            q_n-\partial_t r_{n-1}
            -2\sum_{j=1}^{n-1}r_jr_{n-j}
        \right),
    \qquad \forall\, n\geq1.
\end{equation*} 
Applying the product and derivative estimates
\eqref{eq:Nagumo-basic-estimates} gives
\[
    \mathcal N_n(r_n)
       \leq M\left(
           \mathcal N_n(q_n)
           +e\,n\,\mathcal N_{n-1}(r_{n-1})
           +2\sum_{j=1}^{n-1}
                \mathcal N_j(r_j)\mathcal N_{n-j}(r_{n-j})
       \right).
\]
Dividing by $n!$ and using
$\mathcal N_n(q_n)\leq CK^n n!$, we obtain
\begin{equation}
\begin{aligned}
    u_n
       &\leq MCK^n+Me\,u_{n-1}
          +2M\sum_{j=1}^{n-1}
              \frac{j!(n-j)!}{n!}\,u_j u_{n-j}\\
       &\leq MCK^n+Me\,u_{n-1}
          +2M\sum_{j=1}^{n-1}u_j u_{n-j},
          \qquad \forall\, n\geq1.
\end{aligned}
\end{equation}
Here we used
\[
    \frac{n(n-1)!}{n!}=1,
    \qquad
    \frac{j!(n-j)!}{n!}
       =\binom{n}{j}^{-1}\leq1.
\]
The sums are understood to be empty when $n=1$.
Consider the scalar equation
\begin{equation}
    \mathcal V(z)
       =\frac{MCKz}{1-Kz}
        +Me\,z\bigl(u_0+\mathcal V(z)\bigr)
        +2M\mathcal V(z)^2,
    \qquad \mathcal V(0)=0.
\end{equation}
To justify the existence of a holomorphic solution to this
scalar equation, introduce the function
\begin{equation*}
    \Phi(z,v)
       \coloneqq
       v-\frac{MCKz}{1-Kz}
        -Me\,z(u_0+v)-2Mv^2.
\end{equation*}
It is holomorphic near $(z,v)=(0,0)$ and satisfies
\begin{equation*}
    \Phi(0,0)=0,
    \qquad
    \partial_v\Phi(z,v)=1-Me\,z-4Mv,
    \qquad
    \partial_v\Phi(0,0)=1.
\end{equation*}
Thus the Jacobian with respect to the unknown $v$ is
invertible at $(0,0)$. The holomorphic implicit-function
theorem provides a unique holomorphic germ $\mathcal V$
at $z=0$ such that
\begin{equation*}
    \mathcal V(0)=0,
    \qquad
    \Phi\bigl(z,\mathcal V(z)\bigr)=0.
\end{equation*}
In particular, its Taylor series
\begin{equation*}
    \mathcal V(z)=\sum_{n\geq1}v_nz^n
\end{equation*}
converges in a neighbourhood of the origin. We now compare its coefficients with the numbers $(u_n)_{n\geq1}$. Expanding
\begin{equation*}
    \frac{MCKz}{1-Kz}
       =\sum_{n\geq1}MCK^n z^n
\end{equation*}
and comparing coefficients in the equation for
$\mathcal V$ gives $ v_1=MCK+Me\,u_0$ and
\begin{equation*}
    v_n
       =MCK^n+Me\,v_{n-1}
        +2M\sum_{j=1}^{n-1}v_jv_{n-j}, \,\qquad\, \forall \, n\geq 2.
\end{equation*}
Since $M,C,K>0$ and $u_0\geq0$, these identities imply
$v_n\geq0$ by induction.
The preceding inequality for $u_n$ gives $u_1\leq MCK+Me\,u_0=v_1$. For $n\geq2$, suppose that $u_j\leq v_j$ for every $1\leq j<n$. Then
\begin{equation*}
    u_n\leq MCK^n+Me\,u_{n-1}+2M\sum_{j=1}^{n-1}u_ju_{n-j}\leq MCK^n+Me\,v_{n-1}+2M\sum_{j=1}^{n-1}v_jv_{n-j}=v_n. 
\end{equation*}
Thus $0\leq u_n\leq v_n$ for every $n\geq1$.
Choose $\rho>0$ such that $\mathcal V$ is holomorphic
on a neighbourhood of the closed disk $\{|z|\leq\rho\}$,
and put
\begin{equation*}
    L_\rho\coloneqq
       \max_{|z|=\rho}|\mathcal V(z)|.
\end{equation*}
Cauchy's coefficient formula, applied to $\mathcal V$,
gives
\begin{equation*}
    v_n
       =\frac{1}{2\pi i}
          \int_{|z|=\rho}
             \frac{\mathcal V(z)}{z^{n+1}}\,\mathrm dz
    \quad \Rightarrow\,
    |v_n|\leq L_\rho\rho^{-n}.
\end{equation*}
Consequently, setting $C_1\coloneqq\max\{L_\rho,u_0\}$ and $K_1\coloneqq\rho^{-1}$, we obtain
\begin{equation*}
    u_n\leq C_1K_1^n,
    \qquad \forall\, n\geq0.
\end{equation*}
The inclusion of $u_0$ in the definition of $C_1$
ensures that this bound also holds for $n=0$.
Returning to the definition of $u_n$, we find
\begin{equation*}
    \mathcal N_n(r_n)
       =n!\,u_n
       \leq C_1K_1^n n!,\quad \ \forall \, n\geq 0.
\end{equation*}
For $0<r<R$, let $D_r=\{|x-x_0|<r,\ |t-t_0|<r\}$. Since $d(x,t)\geq R-r>0$ on $D_r$, we have
\begin{equation*}
    |r_n(x,t)|
       \leq\frac{\mathcal N_n(r_n)}{d(x,t)^n}
       \leq
       C_1\left(\frac{K_1}{R-r}\right)^n n!,
       \qquad \forall\, (x,t)\in D_r.
\end{equation*}
Taking the supremum gives
\begin{equation}
    \sup_{(x,t)\in D_r}|r_n(x,t)|
       \leq
       C_1\left(\frac{K_1}{R-r}\right)^n n!,
       \qquad \forall\, n\geq0.
\end{equation}
This proves the local factorial-type property of the
formal ratio $\hat r$.
The argument also applies near $x=\tau\neq0$ where $w^2=3\tau\neq0$, so
$(4r_0)^{-1}=(2\varepsilon w)^{-1}$ remains holomorphic
and bounded on a sufficiently small closed polydisc.

\medskip

For reference, factorial series on a compact domain are closed
under products and exponentials. A factorial series with a nonvanishing
constant coefficient also has a factorial inverse. One direct
proof divides each coefficient bound by $n!$ and majorises the
resulting products by ordinary convergent series, using
$j!(n-j)!\leq n!$. For an inverse, factor out its leading
coefficient and use a geometric majorant with zero constant term.
For an exponential, separate the constant term and use the
exponential of a convergent majorant. These arguments work for
the supremum norm on every smaller compact domain.
The series $\hat p=\hbar\partial_t\hat q$ has factorial bounds
by Cauchy estimates on a larger compact domain. Equations
\eqref{eq:formal-log-derivatives} and
\eqref{eq:formal-amplitude-form} now show that $\hat\beta$ has
factorial bounds. The division by $\hbar$ merely shifts a series
with zero constant coefficient; the factor $(n+1)$ in
$(n+1)!$ is absorbed into an exponential bound. Integrating its
coefficients along radial segments in a smaller polydisc preserves
the bounds, since their lengths are uniformly bounded.
The exponential and product estimates applied to
\eqref{eq:formal-amplitude-primitive} prove
\autoref{250708195116} locally. Multiplication by a fixed nonnegative power of $\hbar$ preserves factorial type. This proves the assertions for
both column prefactors and for $\hat E$ locally.

To extend the estimates, continue the amplitude along a
fixed path using finitely many overlapping coordinate
polydiscs. On each successive polydisc, its value at a
point of overlap is a factorial series, and integration
of $\hat\beta_\varepsilon$ followed by exponentiation
preserves this property. Thus the continued amplitude
is locally uniformly of factorial type at every point
of the universal cover.

Any compact subset $K$ of that cover admits a finite
cover by neighbourhoods carrying such estimates.
Taking the maxima of their bounding constants gives
constants $C_K,M_K$ satisfying \eqref{250708201455}.

\section{Uniform Borel Summation and Exact WKB Reconstruction}
\label{app:auxiliary-Riccati-estimates}

We prove the uniform strip-summation lemma, construct the
spatial action strips, and justify the reconstruction of
exact WKB amplitudes from the summed projective solutions.
The scaled Painlev\'e estimates and their application to
the moving family through the double turning point are
proved in \autoref{app:uniform-Morse-estimates}.

\subsection{Proof of \autoref{lem:uniform-strip-summation}}
\label{app:uniform-strip-proof}

We prove \autoref{lem:uniform-strip-summation} using the
Borel-plane construction developed for scalar Riccati
equations in
\cite[Section~5.1.1, Lemma~5.1, and Section~5.2]
{nikolaev_riccati_2023}.
For general nonlinear systems, see also
\cite[Sections~5.1--5.2]{nikolaev_2022_3}.
We give the proof for the polynomial systems considered
here, making explicit the estimates uniform on an entire
unbounded strip and locally uniform in the holomorphic
parameters.

\paragraph{Step 1. Formal existence and uniqueness.} First apply the change of variables
\[
    \widetilde z=e^{-i\alpha}z,\qquad
    \widetilde z_0=e^{-i\alpha}z_0,\qquad
    \widetilde\hbar=e^{-i\alpha}\hbar.
\]
Since $\hbar\partial_z
=\widetilde\hbar\partial_{\widetilde z}$, this reduces
the proof to $\alpha=0$ without changing $\Lambda$.
We omit the tildes for readability and restrict to $\alpha=0$. Fix a relatively compact open subset $U_0\Subset U$.
The hypotheses of the lemma apply with
$K=\overline{U_0}$, so all constants below can be chosen
uniformly for $k\in U_0$.
Choose $d'<d_1<d_2<d$ and write
$H_\rho=H_\rho(z_0)$ for the horizontal strips to simplify notations (throughout this appendix, all strips are centred at $z_0$
unless stated otherwise).
\medskip

We first construct the formal solution $\hat{f}$ with prescribed
constant term $f_0(k,z)\equiv0$.
The coefficient of $\hbar^0$ in the formal equation is
\[
    F_0(k,z,f_0(k,z))=0,
\]
which is satisfied by this choice because
$F_0(k,z,0)=0$.
Suppose that $(f_1,\ldots,f_{n-1})$ have been constructed.
Define
\[
    \mathcal R_n(k,z)
       \coloneqq
       [\hbar^n]\sum_{j=0}^{n}\hbar^j
       F_j\left(
           k,z,\sum_{\ell=1}^{n-1}f_\ell(k,z)\hbar^\ell
       \right),
\]
where $[\hbar^n]$ denotes extraction of the coefficient
of $\hbar^n$, and the inner sum is empty when $n=1$.
Since $f_0=0$ and $\partial_vF_0(k,z,0)=\Lambda$,
the coefficient of $\hbar^n$ in the formal differential
equation is $\partial_z f_{n-1} =\Lambda f_n+\mathcal R_n$.
The invertibility of $\Lambda$ therefore determines
$f_n=\Lambda^{-1}\bigl(\partial_z f_{n-1}-\mathcal R_n\bigr)$. This recursion constructs a unique (uniqueness follows from the fact that $\Lambda$ is invertible) formal solution
$\hat f(k,z,\hbar)=\underset{n=1}{\overset{\infty}{\sum}}f_n(k,z)\hbar^n$
 of \eqref{deff} with holomorphic coefficients and vanishing constant
term.  No exact solution is assumed to exist at this stage.

\paragraph{Step 2. Normalisation and uniform coefficient estimates.}

In particular, comparing the coefficient of $\hbar$ in the formal
differential equation \eqref{deff}, with $f_0=0$, gives
\[
    \Lambda f_1(k,z)+F_1(k,z,0)=0
    \qquad \Rightarrow \qquad
    f_1(k,z)=-\Lambda^{-1}F_1(k,z,0).
\]
Applying \eqref{eq:strip-input-remainders} with
$K=\overline{U_0}$, $B=\{0\}$ and $N=1,2$, and using
$F_0(k,z,0)=0$, gives
\[
    \left\|\frac{F(k,z,\hbar,0)}{\hbar}\right\|\leq CA,
    \qquad
    \left\|
        \frac{F(k,z,\hbar,0)}{\hbar}-F_1(k,z,0)
    \right\|\leq 2CA^2|\hbar|,
\]
uniformly for $(k,z)\in U_0\times H_d(z_0)$.
The triangle inequality therefore yields
$\|F_1(k,z,0)\|
       \leq CA+2CA^2|\hbar|$.
Since the left-hand side is independent of $\hbar$,
letting $\hbar\to0$ within the Borel disk gives
\[
    \sup_{(k,z)\in U_0\times H_d}
        \|F_1(k,z,0)\|\leq CA.
\]
Consequently, $f_1(k,z)=-\Lambda^{-1}F_1(k,z,0)$
is uniformly bounded on $U_0\times H_d$.
Cauchy's estimate then bounds $\partial_z f_1$
uniformly on $U_0\times H_{d_2}$. Indeed, for each fixed $k\in U_0$ and $z\in H_{d_2}$, the closed disk centred at $z$ with radius $d-d_2$ lies inside $H_d$ and Cauchy's derivative estimate therefore gives
\[
    \sup_{U_0\times H_{d_2}}
        \|\partial_z f_1\|
    \leq
    \frac{1}{d-d_2}
    \sup_{U_0\times H_d}\|f_1\|
    \leq
    \frac{\|\Lambda^{-1}\|CA}{d-d_2}.
\]
 The substitution $f=\hbar(f_1+g)$ gives
\begin{equation}
\label{eq:strip-normal-form}
 \hbar\partial_zg=\Lambda g+\hbar G(k,z,\hbar,g),\qquad g_0=0 \, \text{ and }\, g_n=f_{n+1} \,,\, \forall \, n\geq 1.
\end{equation}
More explicitly,
\begin{equation}
\label{eq:strip-normalised-rhs}
 G(k,z,\hbar,g)
 =\frac{F(k,z,\hbar,\hbar(f_1+g))-\hbar\Lambda g}{\hbar^2}
     -\partial_z f_1.
\end{equation}
Indeed, substituting $f=\hbar(f_1+g)$ into
\eqref{deff} gives $\hbar^2(\partial_z f_1+\partial_z g)
       =F(k,z,\hbar,\hbar(f_1+g))$.
Dividing by $\hbar$ and rearranging, we obtain
\begin{equation*}
    \hbar\partial_z g
       =\frac{F(k,z,\hbar,\hbar(f_1+g))}{\hbar}
          -\hbar\partial_z f_1
    =\Lambda g+\hbar\left(
          \frac{F(k,z,\hbar,\hbar(f_1+g))-\hbar\Lambda g}
               {\hbar^2}
          -\partial_z f_1
         \right).
\end{equation*}

The asymptotic expansion in $\hbar$ of the numerator
in \eqref{eq:strip-normalised-rhs}, with $g$ held fixed,
has vanishing coefficients of orders zero and one.
Indeed, its constant coefficient is $F_0(k,z,0)=0$,
and its coefficient of $\hbar$ is
\[
    \Lambda(f_1+g)+F_1(k,z,0)-\Lambda g
       =\Lambda f_1+F_1(k,z,0)=0.
\]
Since $F$ is polynomial in its last variable, $G$ is
polynomial in $g$, with coefficient functions jointly
holomorphic on $U_0\times H_{d_2}\times S_{0,R_0}$. We now justify the uniform Gevrey estimates.
Fix a bounded polydisc in $g$.
Since $f_1$ is uniformly bounded and $|\hbar|<R_0$,
the argument $v=\hbar(f_1+g)$ remains in a fixed bounded
subset of $\mathbb C^r$.
Cauchy estimates in $v$, applied to
\eqref{eq:strip-input-remainders}, give uniform
Gevrey-$1$ estimates for the finitely many polynomial
coefficients of $F$.
The substitution $v=\hbar(f_1+g)$ involves only finite
sums, multiplication by bounded functions, and
nonnegative powers of $\hbar$.
Consequently, the resulting numerator also admits
uniform Gevrey-$1$ remainder estimates on the fixed
polydisc in $g$, after possibly shrinking the Borel disk.
To divide this numerator by $\hbar^2$, use its remainder
estimate with index $N+2$ and divide by $|\hbar|^2$.
Its first two coefficients vanish, so no negative powers
of $\hbar$ remain. Moreover,
\[
    (N+2)!
       =(N+1)(N+2)N!
       \leq 2\cdot 3^N N!,
    \qquad N\geq0,
\]
so the factorial shift is absorbed into new Gevrey
constants.
Subtracting the uniformly bounded function
$\partial_z f_1$ preserves these estimates.
Finally, Cauchy estimates in $g$ on the fixed polydisc
give uniform remainder estimates for each polynomial
coefficient of $G$.

Thus the coefficient functions of $G$, viewed as a
polynomial in $g$, admit uniform Gevrey-$1$ asymptotic
expansions in $\hbar$ on
$U_0\times H_{d_2}\times S_{0,R_1}$ for some
$0<R_1\leq R_0$.
The Borel-disk parameter $R_1$ and the remainder constants are
independent of $(k,z)\in U_0\times H_{d_2}$,
in particular of $\operatorname{Re}z$. Indeed, the original remainder estimates, the bounds on
$f_1$ and $\partial_z f_1$, and the Cauchy estimates on
fixed polydiscs all hold uniformly on
$U_0\times H_{d_2}$, so the same Borel disk and constants
can be used throughout this product.

\paragraph{Step 3. Borel transforms of the coefficient functions.} Write $G(k,z,\hbar,g)=\underset{m\in\mathcal M}{\sum}a_m(k,z,\hbar)g^m$, where $\mathcal M\subset\mathbb N^r$ is finite and contains the zero multi-index,
$g^m=g_1^{m_1}\cdots g_r^{m_r}$, and each coefficient function $a_m$ takes values in $\mathbb C^r$. Let $\hat a_m(k,z,\hbar)=\underset{n=0}{\overset{\infty}{\sum}}a_{m,n}(k,z)\hbar^n$ denote the Gevrey-$1$ asymptotic expansion of $a_m$, and define
its Borel transform near $\xi=0$ by
\begin{equation*}A_m(k,z,\xi)
       \coloneqq
       \sum_{n=0}^{\infty}
          \frac{a_{m,n+1}(k,z)}{n!}\xi^n.
\end{equation*}
The factorial coefficient bounds imply that this series converges
uniformly near $\xi=0$. By the Nevanlinna--Sokal theorem, the
uniform factorial remainder estimates for $a_m$ on the Borel disk
$S_{0,R_1}$ ensure that this germ extends holomorphically to a tube
around the positive real ray, with at most exponential growth,
and that its Laplace transform reconstructs $a_m$. In other words,
after possibly decreasing $R_1$,
\begin{equation}
 a_m(k,z,\hbar)=a_{m,0}(k,z)
          +\int_0^\infty e^{-\frac{\xi}{\hbar}}A_m(k,z,\xi)\,\mathrm d\xi.
\end{equation}
The functions $A_m$ are jointly holomorphic on
$U_0\times H_{d_2}\times\Xi_{0,\delta_0}$ and satisfy
$|A_m(k,z,\xi)|\leq C e^{c|\xi|}$, uniformly for
$(k,z)\in U_0\times H_{d_2}$. Recall that $\Xi_{0,\delta_0}$ is the open tubular neighbourhood of radius $\delta_0>0$ around the nonnegative real axis. 
For the scalar Nevanlinna--Sokal theorem and its proof, see
\cite[p.~262, especially Eqs.~(5)--(7)]{Sokal1980}.
To match our Borel-transform convention, we apply that theorem to
$\hbar\longmapsto
    \frac{a_m(k,z,\hbar)-a_{m,0}(k,z)}{\hbar}.
$
Its uniform factorial remainder estimates follow from those of
$a_m$ by shifting the truncation order by one. For uniformity in $(k,z)$, regard this function as taking values
in the Banach space of bounded holomorphic maps
$U_0\times H_{d_2}\to\mathbb C^r$, equipped with the supremum norm.
The uniform remainder estimates give local boundedness
in this Banach norm. Cauchy's formula on small
$\hbar$-circles contained in the Borel disk then proves
holomorphy with values in this Banach space.
The absolutely convergent contour integrals and estimates in
Sokal's proof remain valid in this Banach space.
They therefore provide a common Borel tube and exponential growth
bound, uniformly for $(k,z)\in U_0\times H_{d_2}$.
Evaluation at $(k,z)$ gives the functions $A_m$ above, with joint
holomorphy following from locally uniform convergence of the
integral construction.
Since there are finitely many indices $m$, the tube and constants
can be chosen common to all coefficients.

\medskip

At this stage, we have constructed the Borel transforms $A_m$
of the known coefficient functions of $G$.
It remains to prove Borel summability of the unknown formal
solution $\hat g$.
We do this by constructing a holomorphic solution of its
Borel-transformed equation and identifying its Taylor series
with $\mathcal B\hat g$.

\paragraph{Step 4. Construction of the solution in the Borel plane.}
Initially, let 
\begin{equation*}
    \varphi(k,z,\xi)
       =(\mathcal B\hat g)(k,z,\xi)
       =\sum_{n=0}^{\infty}
          \frac{g_{n+1}(k,z)}{n!}\xi^n
\end{equation*}
denote a formal power series in $\xi$ with holomorphic
coefficients in $(k,z)$; no convergence is assumed at this stage. For functions in the Borel variable
write 
\begin{equation}\label{ConvolutionDef} (u*v)(\xi)=\int_0^\xi u(\xi-\eta)v(\eta)\,\mathrm d\eta
\end{equation},
integrating along the straight segment. For formal power series in $\xi$, convolution is understood
coefficientwise, using the identity
\[
    \frac{\xi^p}{p!}*\frac{\xi^q}{q!}
       =\frac{\xi^{p+q+1}}{(p+q+1)!},
    \qquad \forall\, p\geq 0, q\geq0.
\]
No convergence is needed for this formal operation. For formal series $\hat u$ and $\hat v$, our Borel-transform
convention gives
\[
    \mathcal B(\hat u\hat v)
       =u_0\,\mathcal B\hat v
        +v_0\,\mathcal B\hat u
        +(\mathcal B\hat u)*(\mathcal B\hat v).
\]
Since every component of $\hat g$ has zero constant coefficient,
for every $m\in\mathcal M$ with $|m|\geq1$ we have
\[
    \mathcal B(\hat g^m)=\varphi^{*m},
\]
where $\varphi^{*m}$ denotes the convolution of $m_j$ copies
of $\varphi_j$ for each $j\in \llbracket 1,\ldots,r\rrbracket$; factors with $m_j=0$
are omitted. Here $\varphi_j$ denotes the $j$th component
of $\varphi$. Writing $\hat a_m$ for the formal expansion of $a_m$ and
$A_m=\mathcal B\hat a_m$, the product formula therefore gives
\[
    \mathcal B(\hat a_m\hat g^m)
       =a_{m,0}\varphi^{*m}+A_m*\varphi^{*m},
       \qquad \forall \, |m|\geq1.
\]
For the zero multi-index, the contribution is simply $A_0$. In expressions involving $\hat g$, the coefficients of $G$
are understood through their formal expansions $\hat a_m$.
Thus the formal Borel transform of
$G(k,z,\hbar,\hat g)$ is
\begin{equation}
    \mathcal G(\varphi)
       =A_0+
        \sum_{\substack{m\in\mathcal M\\|m|\geq1}}
        \left(
            a_{m,0}\varphi^{*m}
            +A_m*\varphi^{*m}
        \right).
\end{equation}
This sum is finite because $\mathcal M$ is finite. We now apply the formal Borel transform to
\eqref{eq:strip-normal-form}. For any formal series
$\hat H=\underset{n\geq0}{\sum}H_n\hbar^n$, our convention gives
\[
    \mathcal B(\hbar\hat H)(\xi)
       =\sum_{n\geq0}\frac{H_n}{n!}\xi^n
       =H_0+\int_0^\xi
                    \mathcal B\hat H(\eta)\,\mathrm d\eta.
\]
At this stage, integration in $\xi$ means the termwise
primitive of a formal power series, vanishing at $\xi=0$. Since $g_0=0$ and the formal Borel transform commutes
with $\partial_z$, we obtain
\[
    \mathcal B(\hbar\partial_z\hat g)
       =\int_0^\xi
              \partial_z\varphi(k,z,\eta)\,\mathrm d\eta.
\]
Moreover, the constant coefficient of
$G(k,z,\hbar,\hat g)$ is $a_{0,0}(k,z)$.
Consequently, the formal Borel transform of
\eqref{eq:strip-normal-form} is
\begin{equation}
\label{eq:strip-integrated-Borel-equation}
    \int_0^\xi
        \partial_z\varphi(k,z,\eta)\,\mathrm d\eta
       =\Lambda\varphi(k,z,\xi)+a_{0,0}(k,z)
        +\int_0^\xi
            \mathcal G(\varphi)(k,z,\eta)\,\mathrm d\eta.
\end{equation}

For $j\in \llbracket 1,r\rrbracket$, write $\lambda_j\coloneqq\Lambda_{jj}$,
and denote by $\mathcal G_j(\varphi)$ and $a_{0,0,j}$
the $j$th components of $\mathcal G(\varphi)$ and
$a_{0,0}$, respectively.
Differentiating the preceding identity with respect to
$\xi$ gives
\begin{equation}\label{Boreldiffeq}
    \partial_z\varphi
       =\Lambda\partial_\xi\varphi+\mathcal G(\varphi),
\end{equation}
whereas evaluating it at $\xi=0$ gives (the two integrals vanish) $0=\Lambda\varphi(k,z,0)+a_{0,0}(k,z)$, i.e.
\begin{equation}\label{InitialConditionBoreldiffeq}
\varphi(k,z,0)=-\Lambda^{-1}a_{0,0}(k,z)  .
\end{equation}
Componentwise, these two identities become
\begin{equation}\label{Components}
    (\partial_z-\lambda_j\partial_\xi)\varphi_j
       =\mathcal G_j(\varphi),
    \qquad
    \varphi_j(k,z,0)
       =-\lambda_j^{-1}a_{0,0,j}(k,z),\quad \forall\, j\in \llbracket 1,r\rrbracket.
\end{equation}
Together, the differential equation \eqref{Boreldiffeq} and its initial
condition \eqref{InitialConditionBoreldiffeq} are equivalent to
\eqref{eq:strip-integrated-Borel-equation}.
\medskip

At this stage, we have derived in \eqref{Components} equations and initial conditions for the formal Borel transform of $\hat g$. We have not yet proved that this formal series converges or continues along the Borel ray. The next goal is thus to construct a holomorphic solution of this equation, with an exponential growth bound. Recall that a characteristic is a curve in the $(z,\xi)$-variables along which the differential operator on the left of \eqref{Components} becomes a total derivative. To integrate this equation, fix $j\in \llbracket 1,r\rrbracket$ and introduce the linear change of variables
\[
    \zeta=z+\frac{\xi}{\lambda_j},
    \qquad
    \eta=\xi,
    \qquad\text{equivalently}\qquad
    z=\zeta-\frac{\eta}{\lambda_j},
    \quad \xi=\eta.
\]
The parameter $k$ remains fixed. In these coordinates, $\partial_z-\lambda_j\partial_\xi=-\lambda_j\partial_\eta$. Define
\[
    \widetilde\varphi_j(k,\zeta,\eta)
       \coloneqq
       \varphi_j\!\left(
           k,\zeta-\frac{\eta}{\lambda_j},\eta
       \right).
\]
The differential equation and its initial condition \eqref{Components} become
\[
    \partial_\eta\widetilde\varphi_j(k,\zeta,\eta)
       =-\lambda_j^{-1}
          \mathcal G_j(\varphi)\!\left(
              k,\zeta-\frac{\eta}{\lambda_j},\eta
          \right),
    \qquad
    \widetilde\varphi_j(k,\zeta,0)
       =-\lambda_j^{-1}a_{0,0,j}(k,\zeta).
\]
The curves obtained by keeping $\zeta$ fixed are called
the characteristics of the differential operator.
Integrating with respect to $\eta$ at fixed $\zeta$,
and then substituting $\zeta=z+\frac{\xi}{\lambda_j}$, gives
\begin{equation}
\label{eq:strip-Borel-integral}
    \varphi_j(k,z,\xi)
       =-\lambda_j^{-1}
          a_{0,0,j}\!\left(k,z+\frac{\xi}{\lambda_j}\right)
        -\lambda_j^{-1}\int_0^\xi
          \mathcal G_j(\varphi)\!\left(
              k,z+\frac{\xi-\eta}{\lambda_j},\eta
          \right)\,\mathrm d\eta.
\end{equation}

Recall the intermediate strip half-widths $d'<d_1<d_2<d$ chosen at the beginning of the proof and $\delta_0$ introduced in $\Xi_{0,\delta_0}$. Choose $0<\delta<\min\{\delta_0,d_1-d'\}$ and work on
\begin{equation}
 \mathcal D_{\mathrm B}=\{(z,\xi):\xi\in\Xi_{0,\delta},\quad
          |\operatorname{Im}(z-z_0)|+|\operatorname{Im}\xi|<d_1\}.
\end{equation}
Note that $\mathcal D_{\mathrm B}$ provides a joint domain in which the available strip width for $z$ decreases as $\xi$ moves away from the real axis. 

\begin{figure}[htbp]
    \centering
    \includegraphics[width=0.40\textwidth]{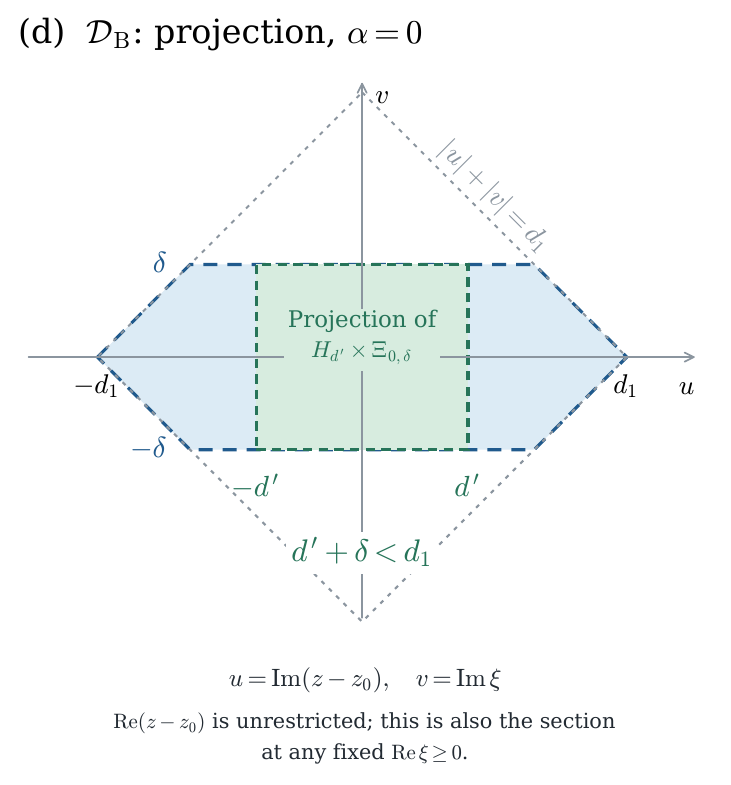}
    \caption{Projection of $\mathcal D_{\mathrm B}$ onto the
    imaginary coordinates $u=\operatorname{Im}(z-z_0)$ and
    $v=\operatorname{Im}\xi$, after rotation to $\alpha=0$.
    The real coordinate $\operatorname{Re}(z-z_0)$ is unrestricted,
    since the defining inequalities are invariant under real
    translations of $z$.
    The projected domain is
    $\{(u,v):|v|<\delta,\ |u|+|v|<d_1\}$.
    The green rectangle represents the projection of the
    smaller product $H_{d'}\times\Xi_{0,\delta}$, which is contained
    in $\mathcal D_{\mathrm B}$ because $d'+\delta<d_1$.
    For fixed $\operatorname{Re}\xi\geq0$ and fixed
    $\operatorname{Re}(z-z_0)$, this picture is also the exact
    section of the domain. For $\operatorname{Re}\xi<0$,
    membership in the rounded end of $\Xi_{0,\delta}$ additionally
    requires $(\operatorname{Re}\xi)^2+v^2<\delta^2$.
    All boundaries are excluded.}
    \label{fig:joint-Borel-domain}
\end{figure}

For $\eta=s\xi$ with $0\leq s\leq1$, using
$\lambda_j\in\{+1,-1\}$ and the triangle inequality, we obtain
\begin{align*}
    &\left|
        \operatorname{Im}\left(
            z+\frac{\xi-\eta}{\lambda_j}-z_0
        \right)
    \right|
    +|\operatorname{Im}\eta|=
    \left|
        \operatorname{Im}(z-z_0)
        +\frac{1-s}{\lambda_j}\operatorname{Im}\xi
    \right|
    +s|\operatorname{Im}\xi|\\
    &\qquad\leq
    |\operatorname{Im}(z-z_0)|
    +(1-s)|\operatorname{Im}\xi|
    +s|\operatorname{Im}\xi|=
    |\operatorname{Im}(z-z_0)|+|\operatorname{Im}\xi|.
\end{align*}
The tube $\Xi_{0,\delta}$ is star-shaped about $0$.
Hence, for $\eta=s\xi$ with $0\leq s\leq1$, both
$\eta$ and $\xi-\eta$ belong to $\Xi_{0,\delta}$.
Together with the preceding inequality, this proves that
the characteristic arguments $\left(z+\frac{\xi-\eta}{\lambda_j},\eta\right)$ 
remain in $\mathcal D_{\mathrm B}$ whenever $(z,\xi)\in\mathcal D_{\mathrm B}$. We must also check that convolution products only evaluate
their factors inside this domain. Indeed, the convolution
\[
    (u*v)(k,z,\xi)
       =\int_0^\xi
          u(k,z,\xi-\eta)\,v(k,z,\eta)\,\mathrm d\eta
\]
requires values at both $(z,\xi-\eta)$ and $(z,\eta)$.
For $\eta=s\xi$, we have
\begin{align*}
    |\operatorname{Im}(z-z_0)|+|\operatorname{Im}\eta|
       &=
       |\operatorname{Im}(z-z_0)|+s|\operatorname{Im}\xi|
       <d_1,\\
    |\operatorname{Im}(z-z_0)|+|\operatorname{Im}(\xi-\eta)|
       &=
       |\operatorname{Im}(z-z_0)|+(1-s)|\operatorname{Im}\xi|
       <d_1.
\end{align*}
Thus both arguments belong to $\mathcal D_{\mathrm B}$ whenever $(z,\xi)\in\mathcal D_{\mathrm B}$. Repeating this observation proves the same assertion for
all the iterated convolutions occurring in
$\mathcal G(\varphi)$. Finally, real translations of $z$ preserve
$\mathcal D_{\mathrm B}$, since $|\operatorname{Im}(z+\ell-z_0)|
       =|\operatorname{Im}(z-z_0)|$ for $\ell\in \mathbb{R}$. This matters because, for $\xi\geq0$ and
$\eta\in[0,\xi]$, the characteristic shift
$\frac{\xi-\eta}{\lambda_j}$ is real and can become arbitrarily
large as $\xi\to+\infty$. The full horizontal strip
allows these shifts without leaving the domain.
For complex $\xi$, their imaginary parts are controlled
by the preceding characteristic estimate.
For $a>0$, let $\mathcal X_a$ be the vector space of
jointly holomorphic maps
$U_0\times\mathcal D_{\mathrm B}\to\mathbb C^r$
for which the following weighted supremum norm is finite:
\begin{equation}
 \|\varphi\|_a
    \coloneqq\sup_{\substack{k\in U_0\\(z,\xi)\in\mathcal D_{\mathrm B}}}
       e^{-a|\xi|}(1+|\xi|^2)|\varphi(k,z,\xi)|.
\end{equation}
and we shall use the maximum norm for vectors. The norms $\|.\|_a$ are monotone with respect to the weight parameter: for $0<a\leq b$ and every $u$ with $\|u\|_a<\infty$,
\begin{equation}\label{monotonenorms}
    \|u\|_b
       =\sup_{\substack{k\in U_0\\
                        (z,\xi)\in\mathcal D_{\mathrm B}}}
          e^{-b|\xi|}(1+|\xi|^2)|u(k,z,\xi)|
       \leq\|u\|_a,
\end{equation}
because $e^{-b|\xi|}\leq e^{-a|\xi|}$.
Thus a bound established at one value of the weight parameter remains valid for every larger value. The space $\mathcal X_a$ is complete. Indeed, a Cauchy
sequence for this norm converges uniformly on compact
subsets to a holomorphic map. Passing to the limit
in the weighted Cauchy estimates shows that the limit
has finite norm and that convergence also holds
in $\|\cdot\|_a$. Scalar convolution \eqref{ConvolutionDef} satisfies
\begin{equation}
\label{eq:strip-convolution-bound}
 \|u*v\|_a\leq4\pi\|u\|_a\|v\|_a.
\end{equation}
Indeed, let $u$ and $v$ have finite weighted norms, and fix
$(k,z,\xi)\in U_0\times\mathcal D_{\mathrm B}$.
For $\xi=0$, the convolution vanishes. Otherwise, write
$\xi=re^{i\theta}$ with $r=|\xi|>0$ and parametrise the
integration segment by $\eta=se^{i\theta}$, $0\leq s\leq r$.
Then $|\mathrm d\eta|=\mathrm ds$ and $|\xi-\eta|=r-s$.
Using the definition of the weighted norms, we obtain
\begin{align*}
    &e^{-ar}(1+r^2)|(u*v)(k,z,\xi)|=
    e^{-ar}(1+r^2)
    \left|
        \int_0^\xi
            u(k,z,\xi-\eta)v(k,z,\eta)\,\mathrm d\eta
    \right|\\
    &\quad\leq
    e^{-ar}(1+r^2)
    \int_{[0,\xi]}
        |u(k,z,\xi-\eta)|\,|v(k,z,\eta)|\,|\mathrm d\eta|
    \leq
    \|u\|_a\|v\|_a(1+r^2)
    \int_0^r
        \frac{e^{-ar}e^{a(r-s)}e^{as}}
             {(1+(r-s)^2)(1+s^2)}\,\mathrm ds\\
    &\quad=
    2\|u\|_a\|v\|_a(1+r^2)
    \int_0^{\frac{r}{2}}
        \frac{\mathrm ds}
             {(1+(r-s)^2)(1+s^2)}
    \leq
    4\pi\|u\|_a\|v\|_a.
\end{align*}
Here the exponential factors cancel, and the factor $2$
comes from the symmetry of the integrand under
$s\mapsto r-s$. For the last inequality, we used,
for $0\leq s\leq \frac{r}{2}$,
\[
    \frac{1+r^2}{1+(r-s)^2}\leq4,
    \qquad
    (1+r^2)\int_0^{\frac{r}{2}}
        \frac{\mathrm ds}{(1+s^2)(1+(r-s)^2)}
       \leq4\int_0^\infty\frac{\mathrm ds}{1+s^2}
       =2\pi.
\]
Taking the supremum over
$(k,z,\xi)\in U_0\times\mathcal D_{\mathrm B}$
proves \eqref{eq:strip-convolution-bound}.
For $\lambda\in\{+1,-1\}$, define the characteristic
integration operator by
\[
    (\mathcal I_\lambda u)(k,z,\xi)
       \coloneqq
       \int_0^\xi
          u\!\left(
              k,z+\frac{\xi-\eta}{\lambda},\eta
          \right)\,\mathrm d\eta,
\]
where integration is along the straight segment $[0,\xi]$.
The domain invariance established above ensures that all
arguments of $u$ remain in $U_0\times\mathcal D_{\mathrm B}$ whenever $(z,\xi)\in \mathcal D_{\mathrm B}$. We claim that
\begin{equation}
\label{eq:strip-Volterra-bound}
 \|\mathcal I_\lambda u\|_a
       \leq\left(\frac2a+\frac4{a^3}\right)\|u\|_a.
\end{equation}
The domain invariance proved above ensures that $\left(z+\frac{\xi-\eta}{\lambda},\eta\right)\in\mathcal D_{\mathrm B}$, for all $\eta\in[0,\xi]$. Hence the definition of the weighted norm gives
$\left|
        u\!\left(k,z+\frac{\xi-\eta}{\lambda},\eta\right)
    \right|
       \leq \|u\|_a
             \frac{e^{a|\eta|}}{1+|\eta|^2}$.
For $\xi=0$, the characteristic integral vanishes. Otherwise, write $\xi=re^{i\theta}$ and parametrise $\eta=se^{i\theta}$, with $0\leq s\leq r$.
Since $|\mathrm d\eta|=\mathrm ds$, we obtain
\small{\begin{equation*}
e^{-ar}(1+r^2)|\mathcal I_\lambda u(k,z,\xi)|
    \quad\leq
    e^{-ar}(1+r^2)
    \int_0^r
        \left|
            u\!\left(
                k,z+\frac{(r-s)e^{i\theta}}{\lambda},
                se^{i\theta}
            \right)
        \right|\,\mathrm ds
    \quad\leq
    \|u\|_a\int_0^r
        e^{-a(r-s)}\frac{1+r^2}{1+s^2}\,\mathrm ds.
\end{equation*}}
\normalsize
To estimate the last integral, put $v=r-s\geq0$.
Since $r=s+v$, we have
\[
    1+r^2
       =1+(s+v)^2
       \leq1+2s^2+2v^2
       \leq2(1+s^2)(1+v^2).
\]
Consequently, changing variables from $s$ to $v=r-s$ gives
\begin{equation*}
    \int_0^r
        e^{-a(r-s)}\frac{1+r^2}{1+s^2}\,\mathrm ds
    \leq
    2\int_0^r e^{-av}(1+v^2)\,\mathrm dv
    \leq
    2\int_0^\infty e^{-av}(1+v^2)\,\mathrm dv
    =\frac2a+\frac4{a^3},
\end{equation*}
where we used $a>0$ and $\int_0^\infty e^{-av}\,\mathrm dv=\frac1a$ and
    $\int_0^\infty v^2e^{-av}\,\mathrm dv=\frac2{a^3}$. Taking the supremum over
$(k,z,\xi)\in U_0\times\mathcal D_{\mathrm B}$ proves \eqref{eq:strip-Volterra-bound}.
\medskip

We now use the preceding norm estimates to solve
\eqref{eq:strip-Borel-integral} by the Banach fixed-point
theorem. The convolution estimate controls the nonlinear
map $\mathcal G$, while the characteristic integration
estimate provides a factor tending to zero as $a\to+\infty$.

\medskip
\noindent\textbf{Claim.}
There exist constants $a>0$, $M>0$ and a jointly holomorphic
function $\varphi:U_0\times\mathcal D_{\mathrm B}
               \longrightarrow\mathbb C^r$ solving \eqref{eq:strip-Borel-integral}, such that
\begin{equation}\label{expoboundsphi}
    |\varphi(k,z,\xi)|
       \leq M\frac{e^{a|\xi|}}{1+|\xi|^2} \,\,, \forall \, (k,z,\xi)\in U_0\times\mathcal D_{\mathrm B}.
\end{equation}
Its Taylor series at $\xi=0$ is the formal Borel transform $\mathcal B\hat g(\xi)$:
\begin{equation}\label{TaylorPhi}\frac{\partial^n \varphi}{\partial \xi^n}(k,z,0)=g_{n+1}(k,z),
    \qquad \forall\, n\geq0,\quad (k,z)\in U_0\times H_{d_1}.
\end{equation}

\medskip
\noindent\emph{Proof of the claim.}
For $a\geq c+2$, we work in the Banach space
$\mathcal X_a$ defined above; the value of $a$
will be chosen sufficiently large below.
Define the known initial term by
\[
    \varphi^{\mathrm{in}}_j(k,z,\xi)
       \coloneqq
       -\lambda_j^{-1}
       a_{0,0,j}\!\left(k,z+\frac{\xi}{\lambda_j}\right),
       \qquad \forall\, j\in \llbracket 1,r\rrbracket.
\]
For $\Phi\in\mathcal X_a$, define the operator
$\mathcal T$ componentwise by
\[
    (\mathcal T\Phi)_j
       \coloneqq
       \varphi^{\mathrm{in}}_j
       -\lambda_j^{-1}
        \mathcal I_{\lambda_j}\bigl(\mathcal G_j(\Phi)\bigr).
\]
Thus \eqref{eq:strip-Borel-integral} is precisely the
fixed-point equation $\varphi=\mathcal T\varphi$ in $\mathcal X_a$. First, we obtain bounds independent of sufficiently large
$a$. Recall that $|A_m(k,z,\xi)|\leq Ce^{c|\xi|}$ and that
the coefficients $a_{m,0}(k,z)$ are uniformly bounded.
For $a\geq c+2$, we have
\[
    \|A_m\|_a
       \leq C\sup_{\rho\geq0}
                    (1+\rho^2)e^{-(a-c)\rho}
       \leq C\sup_{\rho\geq0}
                    (1+\rho^2)e^{-2\rho}
       \leq C.
\]
uniformly in $a$ and $m\in\mathcal M$. Here we used $a-c\geq2$ and
$1+\rho^2\leq e^{2\rho}$ for $\rho\geq0$. In particular, this bound is independent of $a\geq c+2$. Since $a_{0,0}$ is uniformly bounded on
$U_0\times H_{d_2}$, set
\[
    M_0\coloneqq
    \sup_{(k,\zeta)\in U_0\times H_{d_2}}
         |a_{0,0}(k,\zeta)|<\infty.
\]
For $(z,\xi)\in\mathcal D_{\mathrm B}$, the translated
arguments $z+\frac{\xi}{\lambda_j}$ belong to
$H_{d_1}\subset H_{d_2}$. Since $|\lambda_j|=1$,
the definition of $\varphi^{\mathrm{in}}$ therefore gives
$|\varphi^{\mathrm{in}}(k,z,\xi)|\leq M_0$.
Consequently, for $a\geq c+2\geq2$,
\[
    \|\varphi^{\mathrm{in}}\|_a
       \leq M_0\sup_{\rho\geq0}
                    (1+\rho^2)e^{-a\rho}
       \leq M_0,
\]
using $(1+\rho^2)e^{-2\rho}\leq1$.
Thus the same constant $M_0$ works for every $a\geq c+2$.
\medskip

Choose once and for all a radius $M>2M_0$, independently
of $a$. For every $a\geq c+2$, define the closed ball
\[
    \mathcal B_a(M)
       \coloneqq
       \left\{
           \Phi\in\mathcal X_a
           \,\middle|\,
           \|\Phi\|_a\leq M
       \right\}.
\]
This ball is nonempty and complete for the metric
induced by $\|\cdot\|_a$.
Choose a constant $C_\ast>0$, independent of $a\geq c+2$ (from \eqref{monotonenorms}), such that, for every $m\in\mathcal M$,
\[
    \sup_{(k,z)\in U_0\times H_{d_2}}
         |a_{m,0}(k,z)|\leq C_\ast,
    \qquad
    \|A_m\|_a\leq C_\ast.
\]
Such a common constant exists because $\mathcal M$ is finite. Fix $m\in\mathcal M$ with $p=|m|\geq1$, and write $\Phi^{*m}=\Phi_{i_1}*\cdots*\Phi_{i_p}$,
where each index $j$ occurs $m_j$ times.
For $\Phi\in\mathcal B_a(M)$, each component has norm
at most $M$. Repeated application of
\eqref{eq:strip-convolution-bound} therefore gives
\[
    \|\Phi^{*m}\|_a
       \leq(4\pi)^{p-1}
             \prod_{\ell=1}^p\|\Phi_{i_\ell}\|_a
       \leq(4\pi)^{p-1}M^p.
\]

For $(\Phi,\Psi)\in\mathcal B_a(M)\times \mathcal B_a(M)$, the telescoping identity
\[
    \Phi^{*m}-\Psi^{*m}
       =
       \sum_{\ell=1}^p
       \Phi_{i_1}*\cdots*\Phi_{i_{\ell-1}}
       *(\Phi_{i_\ell}-\Psi_{i_\ell})
       *\Psi_{i_{\ell+1}}*\cdots*\Psi_{i_p}
\]
holds, with empty blocks of factors omitted.
Each summand contains $p$ factors: one has norm at most
$\|\Phi-\Psi\|_a$, and the remaining $p-1$ factors
have norm at most $M$. Consequently,
\[
    \|\Phi^{*m}-\Psi^{*m}\|_a
       \leq
       p(4\pi)^{p-1}M^{p-1}\|\Phi-\Psi\|_a.
\]

Multiplication by the bounded coefficient $a_{m,0}$
has norm at most $C_\ast$, whereas convolution with
$A_m$ has norm at most $4\pi C_\ast$.
Thus
\begin{align*}
    \|a_{m,0}\Phi^{*m}+A_m*\Phi^{*m}\|_a
       &\leq C_\ast(1+4\pi)(4\pi)^{p-1}M^p,\\
    \|a_{m,0}(\Phi^{*m}-\Psi^{*m})
       +A_m*(\Phi^{*m}-\Psi^{*m})\|_a
       &\leq C_\ast(1+4\pi)
          p(4\pi)^{p-1}M^{p-1}\|\Phi-\Psi\|_a.
\end{align*}

Summing over the finitely many nonzero multi-indices
in $\mathcal M$, and including the term $A_0$ in
$\mathcal G(\Phi)$, we may therefore take
\begin{equation*}
    P_M
       \coloneqq
       C_\ast+
       C_\ast(1+4\pi)
       \sum_{\substack{m\in\mathcal M\\|m|\geq1}}
          (4\pi)^{|m|-1}M^{|m|},\quad
    L_M
       \coloneqq
       C_\ast(1+4\pi)
       \sum_{\substack{m\in\mathcal M\\|m|\geq1}}
          |m|(4\pi)^{|m|-1}M^{|m|-1}.
\end{equation*}
The term $A_0$ cancels in
$\mathcal G(\Phi)-\mathcal G(\Psi)$.
These constants are independent of $a\geq c+2$ and satisfy
\begin{equation}\label{Norms}
    \|\mathcal G(\Phi)\|_a\leq P_M,
    \qquad
    \|\mathcal G(\Phi)-\mathcal G(\Psi)\|_a
       \leq L_M\|\Phi-\Psi\|_a \,,\,\, \forall\, (\Phi,\Psi)\in\mathcal B_a(M) \times \mathcal B_a(M).
\end{equation}

Define $\kappa(a)\coloneqq\frac2a+\frac4{a^3}$. By \eqref{eq:strip-Volterra-bound}, \eqref{Norms} and since $|\lambda_j|=1$, we obtain
\begin{align*}
    \|\mathcal T\Phi\|_a
       &\leq M_0+\kappa(a)P_M,\\
    \|\mathcal T\Phi-\mathcal T\Psi\|_a
       &\leq\kappa(a)L_M\|\Phi-\Psi\|_a.
\end{align*}
Because $\kappa(a)\to0$ as $a\to+\infty$ and $(M_0,P_M,L_M)$ are independent of $a$, we can now fix $a\geq c+2$ sufficiently large that
\[
    M_0+\kappa(a)P_M\leq M,
    \qquad
    \kappa(a)L_M<1.
\]
The first inequality shows that $\mathcal T$ maps
$\mathcal B_a(M)$ into itself while the second shows that
it is a strict contraction there.
The Banach fixed-point theorem therefore supplies a
unique fixed point $\varphi$ in this ball.
It is jointly holomorphic, and $\|\varphi\|_a\leq M$
gives the asserted exponential bound.
\medskip

It remains to identify its Taylor series.
Differentiating the fixed-point equation along the
characteristics recovers the Borel differential equation
and its initial condition. Consequently, the formal series $\underset{n\geq 0}{\sum}
       \frac{\partial^n \varphi}{\partial \xi^n}(k,z,0)\,\hbar^{n+1}$
solves \eqref{eq:strip-normal-form} with zero constant
coefficient. Under the substitution $f=\hbar(f_1+g)$, this formal
series yields a formal solution of \eqref{deff}
with zero constant coefficient. By the uniqueness
proved in Step~1, the resulting series equals $\hat f$.
Consequently, the original series equals $\hat g$, and $\frac{\partial^n \varphi}{\partial \xi^n}(k,z,0)=g_{n+1}(k,z)$, for all $n\geq0$. This proves the claim.

\medskip

The proven claim will be used in the next step to justify the Laplace transform. Indeed, since
$H_{d'}\times\Xi_{0,\delta}\subset\mathcal D_{\mathrm B}$,
the formal Borel transform $\mathcal B\hat g$ admits
a jointly holomorphic continuation to
$U_0\times H_{d'}\times\Xi_{0,\delta}$ with a uniform
exponential bound. This is the function whose Laplace
transform will reconstruct $g$ in Step~5.

\paragraph{Step 5. Laplace reconstruction and uniform remainder estimates.}

Since $H_{d'}\times\Xi_{0,\delta}\subset\mathcal D_{\mathrm B}$,
the function $\varphi$ is jointly holomorphic on
$U_0\times H_{d'}\times\Xi_{0,\delta}$ and satisfies
$|\varphi(k,z,\xi)|\leq M e^{a|\xi|}$ there.
With $a>0$ fixed by the contraction argument, set $R\coloneqq\min\left\{R_1,\frac1{2a}\right\}$. Then $0<R\leq R_1$ and $R^{-1}\geq2a>a$.
For $(k,z,\hbar)\in U_0\times H_{d'}\times S_{0,R}$,
define $g$ by the Laplace transform of $\varphi$:
\[
    g(k,z,\hbar)
       \coloneqq
       \int_0^{+\infty}
          e^{-\frac{\xi}{\hbar}}\varphi(k,z,\xi)\,\mathrm d\xi.
\]
where $\varphi$ has been derived in the previous step from the contraction mapping theorem. The exponential bound on $\varphi$ implies that this integral
defines a jointly holomorphic function on
$U_0\times H_{d'}\times S_{0,R}$.
We verify that $g$ satisfies \eqref{eq:strip-normal-form}.
The exponential bounds and
$\operatorname{Re}(\hbar^{-1})>R^{-1}>a$ justify Fubini's
theorem in the Laplace transforms of primitives and
convolution products. Differentiation under the integral
with respect to $z$ is justified by Cauchy's estimate:
for real $\xi\geq0$, the function $\varphi$ is holomorphic
on the larger strip $H_{d_1}$ with the same exponential
bound, so $\partial_z\varphi$ has a uniform exponential
bound on $H_{d'}$.

Laplace transformation turns convolution into
multiplication. Using the reconstruction formula for
each coefficient $a_m$, we therefore obtain
\[
    \int_0^\infty e^{-\frac{\xi}{\hbar}}
        \mathcal G(\varphi)(k,z,\xi)\,\mathrm d\xi
       =G(k,z,\hbar,g(k,z,\hbar))-a_{0,0}(k,z).
\]
Moreover, Laplace transformation of a primitive
vanishing at $\xi=0$ introduces a factor $\hbar$.
Thus \eqref{eq:strip-integrated-Borel-equation} gives
\begin{equation*}
    \hbar\partial_z g
       =\Lambda g+\hbar a_{0,0}
         +\hbar\bigl(G(k,z,\hbar,g)-a_{0,0}\bigr)
       =\Lambda g+\hbar G(k,z,\hbar,g),
\end{equation*}
which is \eqref{eq:strip-normal-form}.
\medskip

We now establish the uniform factorial remainder estimates
on the whole Borel disk. On the smaller product $H_{d'}\times\Xi_{0,\delta}\subset\mathcal D_{\mathrm B}$, the fixed point $\varphi$ has a bound
$|\varphi(k,z,\xi)|\leq M e^{a|\xi|}$. Fix a radius $0<\rho<\delta$. For every point
$s\in\mathbb R_{\geq0}$ on the Borel integration ray,
the closed disk $|\xi-s|\leq\rho$ lies in
$\Xi_{0,\delta}$. On its boundary, $|\varphi(k,z,\xi)|\leq M e^{a(s+\rho)}$. Cauchy's estimate at the centre $\xi=s$ therefore gives
\begin{equation}
\label{eq:strip-Borel-derivative-bound}
 \left|\frac{\partial^N \varphi}{\partial \xi^N}(k,z,s)\right|
       \leq C A^N N!e^{as},\qquad \forall\, s\geq0,
\end{equation}
Thus \eqref{eq:strip-Borel-derivative-bound} holds with
$C=M e^{a\rho}$ and $A=\rho^{-1}$, independently of
$(k,z)\in U_0\times H_{d'}$, $s\geq0$, and $N\geq0$. 
For $N\geq1$, integration by parts in the
Laplace integral, using $g_n=\frac{\partial^{n-1} \varphi}{\partial \xi^{n-1}}(k,z,0)$,
gives
\begin{equation}
\label{eq:strip-exact-Laplace-remainder}
 g(k,z,\hbar)-\sum_{n=1}^{N-1}g_n(k,z)\hbar^n
   =\hbar^N\frac{\partial^{N-1} \varphi}{\partial \xi^{N-1}}(k,z,0)+\hbar^N\int_0^{+\infty} e^{-\frac{\xi}{\hbar}}
                   \frac{\partial^{N} \varphi}{\partial \xi^{N}}(k,z,\xi)\,\mathrm d\xi.
\end{equation}
Indeed, for every $\hbar\in S_{0,R}$ and every $s\geq0$, $|e^{-\frac{s}{\hbar}}|
       =e^{-s\operatorname{Re}(\hbar^{-1})}
       \leq e^{-sR^{-1}}$, because $\operatorname{Re}(\hbar^{-1})>R^{-1}$.
Combining this inequality with
\eqref{eq:strip-Borel-derivative-bound}, we obtain,
for every integer $\ell\geq0$,
\begin{equation*}
    \left|
        e^{-\frac{s}{\hbar}}\frac{\partial^{\ell} \varphi}{\partial \xi^{\ell}}(k,z,s)
    \right|
       \leq C A^\ell\ell!\,
          e^{-s\operatorname{Re}(\hbar^{-1})}e^{as}
       \leq C A^\ell\ell!\,
          e^{-(R^{-1}-a)s}
       \longrightarrow0
       \qquad\text{as }s\to+\infty.
\end{equation*}
Here $R^{-1}-a>0$. For each fixed $\ell\geq 0$, this convergence
is uniform in $(k,z,\hbar)\in
U_0\times H_{d'}\times S_{0,R}$.
Thus all boundary terms at infinity vanish in the
successive integrations by parts.
\medskip

We now estimate the two terms on the right-hand side of
\eqref{eq:strip-exact-Laplace-remainder}.
For every $N\geq1$, the triangle inequality gives
\begin{equation*}
    \left|
        g(k,z,\hbar)
        -\sum_{n=1}^{N-1}g_n(k,z)\hbar^n
    \right|
    \leq |\hbar|^N
        \left|\frac{\partial^{N-1} \varphi}{\partial \xi^{N-1}}(k,z,0)\right|
       +|\hbar|^N\int_0^\infty
          e^{-\xi\operatorname{Re}(\hbar^{-1})}
          \left|\partial_\xi^N\varphi(k,z,\xi)\right|
          \,\mathrm d\xi.
\end{equation*}
Applying \eqref{eq:strip-Borel-derivative-bound} at $s=0$
to the first term and at $s=\xi$ to the integrand yields
\begin{align*}
    \left|
        g(k,z,\hbar)
        -\sum_{n=1}^{N-1}g_n(k,z)\hbar^n
    \right|
    &\leq C|\hbar|^N
       \left[
          A^{N-1}(N-1)!
          +A^N N!\int_0^\infty
             e^{-(\operatorname{Re}(\hbar^{-1})-a)\xi}
             \,\mathrm d\xi
       \right]\\
    &\leq C|\hbar|^N
       \left[
          A^{N-1}(N-1)!
          +\frac{A^N N!}{R^{-1}-a}
       \right].
\end{align*}
For the last inequality, we used
$\operatorname{Re}(\hbar^{-1})-a>R^{-1}-a>0$, which gives
\[
    \int_0^\infty
        e^{-(\operatorname{Re}(\hbar^{-1})-a)\xi}\,\mathrm d\xi
       =\frac{1}{\operatorname{Re}(\hbar^{-1})-a}
       \leq\frac{1}{R^{-1}-a}.
\]

Set
\[
    A'\coloneqq\max\{1,A\},
    \qquad
    C'\coloneqq C\left(1+\frac{1}{R^{-1}-a}\right).
\]
For $N\geq1$, we have
$A^{N-1}(N-1)!\leq(A')^N N!$ and
$A^N N!\leq(A')^N N!$.
Consequently,
\[
    \sup_{(k,z)\in U_0\times H_{d'}}
    \left|
        g(k,z,\hbar)
        -\sum_{n=1}^{N-1}g_n(k,z)\hbar^n
    \right|
       \leq C'(A')^N N!|\hbar|^N.
\]
The constants $C'$ and $A'$ are independent of $N$,
$\hbar\in S_{0,R}$, and $(k,z)\in U_0\times H_{d'}$.
For $N=0$, the Laplace representation and
\eqref{eq:strip-Borel-derivative-bound} with derivative
order zero give directly
\[
    |g(k,z,\hbar)|
       \leq C\int_0^\infty
          e^{-(\operatorname{Re}(\hbar^{-1})-a)\xi}
          \,\mathrm d\xi
       \leq\frac{C}{R^{-1}-a}
       \leq C'.
\]
Since $g_0=0$, these estimates establish the uniform
Gevrey-$1$ asymptotic expansion
$g\sim\underset{n\geq0}{\sum}g_n\hbar^n$ on the whole Borel disk
$S_{0,R}$, uniformly for
$(k,z)\in U_0\times H_{d'}$.

\paragraph{Step 6. Return to the original equation and uniqueness.}

Define
\[
    f(k,z,\hbar)
       \coloneqq\hbar\bigl(f_1(k,z)+g(k,z,\hbar)\bigr).
\]
This function is jointly holomorphic on
$U_0\times H_{d'}\times S_{0,R}$.
Using \eqref{eq:strip-normal-form} and the definition
\eqref{eq:strip-normalised-rhs} of $G$, we obtain
\begin{equation*}
    \hbar\partial_z f
       =\hbar^2\partial_z f_1
         +\hbar^2\partial_z g
       =\hbar^2\partial_z f_1
         +\hbar\Lambda g+\hbar^2G(k,z,\hbar,g)
       =F(k,z,\hbar,f).
\end{equation*}
Thus $f$ solves the original equation \eqref{deff}.
We next transfer the uniform factorial remainder
estimates from $g$ to $f$. Since $f_0=g_0=0$ and
$f_n=g_{n-1}$ for $n\geq2$, for every $N\geq2$ we have
\[
    f-\sum_{n=0}^{N-1}f_n\hbar^n
       =\hbar\left(
          g-\sum_{j=1}^{N-2}g_j\hbar^j
         \right).
\]
The estimate from Step~5 with truncation order $N-1$
therefore gives
\[
    \sup_{(k,z)\in U_0\times H_{d'}}
    \left|
        f(k,z,\hbar)
        -\sum_{n=0}^{N-1}f_n(k,z)\hbar^n
    \right|
       \leq C'(A')^{N-1}(N-1)!|\hbar|^N.
\]
The cases $N=0$ and $N=1$ follow from the uniform
boundedness of $f_1$ and $g$, together with
$|f|\leq|\hbar|(|f_1|+|g|)$ and $|\hbar|<R$.
Since $A'\geq1$, increasing the bounding constant
if necessary gives uniform Gevrey-$1$ remainder
estimates for $f$ for every $N\geq0$ throughout
$U_0\times H_{d'}\times S_{0,R}$.

The Borel summability of $\hat f$ also follows directly
from that of $\hat g$. Indeed,
$\hat f=\hbar f_1+\hbar\hat g$ and $g_0=0$ imply
\[
    (\mathcal B\hat f)(k,z,\xi)
       =f_1(k,z)+\int_0^\xi
            \varphi(k,z,\eta)\,\mathrm d\eta.
\]
The right-hand side is jointly holomorphic on
$U_0\times H_{d'}\times\Xi_{0,\delta}$.
Since $f_1$ is uniformly bounded and
$|\varphi(k,z,\eta)|\leq M e^{a|\eta|}$, integration
along the straight segment gives
\[
    |(\mathcal B\hat f)(k,z,\xi)|
       \leq
       \sup_{U_0\times H_{d'}}|f_1|
       +\frac{M}{a}\bigl(e^{a|\xi|}-1\bigr).
\]
Thus this continuation has a uniform exponential bound.
Its Laplace transform is $\hbar f_1+\hbar g=f$,
so $f$ is the Borel sum of $\hat f$.

Finally, suppose that two holomorphic functions admit
$\hat f$ as their uniform Gevrey-$1$ asymptotic expansion
on such products. Restricting to a common smaller Borel
disk and fixing $(k,z)$, their difference has zero
Gevrey-$1$ asymptotic expansion. The Nevanlinna
uniqueness theorem therefore implies that this
difference vanishes. This proves uniqueness as a germ
among functions with the stated asymptotic expansion.
Undoing the rotation of $z$ and $\hbar$ made in
Step~1 restores the direction $\alpha$.
The resulting domains are
$U_0\times H_{d'}(z_0)\times S_{\alpha,R}$, and all
estimates remain uniform throughout the parameter
domain and the full strip.
This completes the proof of
\autoref{lem:uniform-strip-summation}.

\subsection{Proof of \autoref{lem:spatial-adapted-strips} }\label{AppendixProofLemma43}
\textbf{Step 1. Construction of a local lift.}
We first fix $t=t_0$ and put $\tau_0=\tau(t_0)$.
In the normalisation coordinate $w$, twice the action is
\[
    2\widetilde Z_{\tau_0}(w)
       =\frac85w^5-8\tau_0w^3,
    \qquad
    \partial_w\bigl(2\widetilde Z_{\tau_0}(w)\bigr)
       =8w^2(w^2-3\tau_0).
\]
Since $(x_0,t_0)\in\mathcal U_j$, the point
$x_0\in R_j(\tau_0)$ lies outside the Stokes graph.
In particular, it is neither of its vertices
$-2\tau_0$ and $\tau_0$. Thus $w_0\neq0$ and $w_0^2-3\tau_0=x_0-\tau_0\neq0$. The holomorphic inverse-function theorem therefore
gives a local inverse of $s=2\widetilde Z_{\tau_0}(w)$
near $(s_0,w_0)$.

Consider the straight line $\gamma(u)=s_0+e^{i\alpha}u, u \in\mathbb R$. Here $u$ is a real path parameter and the time variable $t$ remains fixed at $t_0$.
The local inverse gives a lift $W(u)$ for $u$ near
zero, satisfying
\[
    2\widetilde Z_{\tau_0}(W(u))=\gamma(u),
    \qquad W(0)=w_0.
\]
Continue this lift on its maximal open interval
$I\subset\mathbb R$ containing zero, subject to $\partial_w\bigl(2\widetilde Z_{\tau_0}\bigr)(W(u))
       \neq0$.
Local uniqueness of the inverse makes this continuation
unambiguous.

\noindent\textbf{Step 2. Continuation along the entire line.}
We prove that $I=\mathbb R$.
By construction, the continued lift satisfies
$2\widetilde Z_{\tau_0}(W(u))
       =s_0+e^{i\alpha}u,
    \qquad u\in I$.
Suppose, for example, that $I$ has a finite right endpoint $u_\ast$.
There is then a constant $M>0$ such that
\[
    \left|2\widetilde Z_{\tau_0}(W(u))\right|
       =|s_0+e^{i\alpha}u|\leq M
\]
for $u\in I$ sufficiently close to $u_\ast$. Choose $R_w>0$ such that $R_w^2\geq10|\tau_0|$.
For $|w|\geq R_w$, the triangle inequality gives
\[
    \left|2\widetilde Z_{\tau_0}(w)\right|
       \geq\frac85|w|^5-8|\tau_0||w|^3
       \geq\frac45|w|^5.
\]
Consequently,
\[
    |W(u)|
       \leq
       \max\left\{R_w,\left(\frac{5M}{4}\right)^{\frac{1}{5}}\right\}
\]
for $u$ sufficiently close to $u_\ast$.
Thus the lift remains bounded as $u\to u_\ast$.

Moreover, $W(u)$ has a limit $w_\ast$ as
$u\to u_\ast$.
Indeed, its accumulation set is nonempty, compact
and connected, and is contained in the finite set $\left\{
       w\in\mathbb C:
       2\widetilde Z_{\tau_0}(w)=\gamma(u_\ast)
    \right\}$.
It must therefore consist of a single point.
Here connectedness follows by expressing the
accumulation set as the intersection of the nested
compact connected sets obtained by closing the
images of terminal subintervals of $I$. If
$\partial_w(2\widetilde Z_{\tau_0})(w_\ast)\neq0$,
the inverse-function theorem continues the lift
beyond $u_\ast$, contradicting maximality.
Consequently, $w_\ast$ would have to be a critical
point of the action. We now show that this alternative is also impossible.
Since $2\widetilde Z_{\tau_0}(w_\ast)
       =s_0+e^{i\alpha}u_\ast$,
we have, along the lifted arc from $w_0$ to $w_\ast$,
\[
    2\left(
       \widetilde Z_{\tau_0}(W(u))
       -\widetilde Z_{\tau_0}(w_\ast)
      \right)
       =e^{i\alpha}(u-u_\ast).
\]
Taking the rotated imaginary part gives
\[
    \operatorname{Im}\left(
       e^{-i\alpha}
       \left(
          \widetilde Z_{\tau_0}(W(u))
          -\widetilde Z_{\tau_0}(w_\ast)
       \right)
    \right)=0.
\]
The lifted arc is regular before its endpoint and
approaches the critical point $w_\ast$.
It is therefore a level-set branch issuing from that
critical point, as in \eqref{eq:Stokes-curves}.
Its projection $X(u)=W(u)^2-2\tau_0$
is consequently a Stokes-curve arc containing
$X(0)=x_0$. This contradicts $x_0\in R_j(\tau_0)$, because
$R_j(\tau_0)$ is a component of the complement of
the Stokes graph.
Thus $I$ has no finite right endpoint.
The same argument excludes a finite left endpoint,
so the lift exists for every $u\in\mathbb R$.

\noindent\textbf{Step 3. Inverse branches near the two ends.}
We next construct inverse branches near the two ends
of the line, including their dependence on $t$.
Choose a small relatively compact neighbourhood
$V_1\Subset V$ of $t_0$.
In particular, $\tau(t)$ is bounded on $V_1$. Choose a branch of $s^{\frac{1}{5}}$ in a sector around each
of the two directions $e^{i\alpha}$ and $-e^{i\alpha}$.
For large $|s|$ in either sector, put
\[
    \eta=s^{-\frac{1}{5}},
    \qquad w=s^{\frac{1}{5}}v.
\]
The equation $2\widetilde Z_{\tau(t)}(w)=s$ becomes
\[
    \frac85v^5-8\tau(t)\eta^2v^3=1.
\]
At $\eta=0$, its five roots satisfy $v^5=\frac{5}{8}$.
They are simple, since the derivative of the left-hand
side with respect to $v$ is $8v^4\neq0$ there.
The holomorphic implicit-function theorem therefore
gives five branches, jointly holomorphic in $(\eta,t)$
near $\eta=0$ and $t=t_0$.
After shrinking $V_1$, their expansions are uniform
for $t\in V_1$ and have the form
\[
    v=c+O(\eta^2),
    \qquad c^5=\frac58.
\]
Indeed, the defining equation depends on $\eta$ only through
$\eta^2$. By uniqueness in the implicit-function
theorem, each branch is even in $\eta$, so its
linear coefficient vanishes. The corresponding inverse branches satisfy
\[
    w(s,t)
       =c\,s^{\frac{1}{5}}
          \left(1+O(|s|^{-\frac{2}{5}})\right),
\]
uniformly for $t\in V_1$.
The choice of $c$ and of the fifth root may differ
between the two ends.

For sufficiently large $|s|$, these five branches
exhaust the roots of
$2\widetilde Z_{\tau(t)}(w)=s$.
At $t=t_0$, the lift $W(u)$ consequently agrees with
one of them at each end of the line.
Continuity and separation of the five roots ensure
that the selected branch is the same along each
sufficiently distant tail. Write
\[
    s=s_0+e^{i\alpha}(u+iv),
    \qquad u,v\in\mathbb R.
\]
For a sufficiently large $T>0$, the two selected
branches are therefore holomorphic on the two
half-strips
\[
    \{u>T,\ |v|<1\},
    \qquad
    \{u<-T,\ |v|<1\},
\]
jointly with $t\in V_1$.
They are noncritical there, after increasing $T$
if necessary.

\noindent\textbf{Step 4. Extension across the finite middle part.}
For every $u\in[-T-1,T+1]$, the point $W(u)$ is
regular for the action.
Applying the holomorphic implicit-function theorem
to $2\widetilde Z_{\tau(t)}(w)-s=0$
at $(s,t,w)=(\gamma(u),t_0,W(u))$ gives a local inverse
branch jointly holomorphic in $(s,t)$. A finite collection of such neighbourhoods covers
the compact middle segment.
After shrinking their time neighbourhoods to a common
neighbourhood of $t_0$, the local inverses agree on
overlaps.
Indeed, at $t=t_0$ they continue the same lift $W$;
local uniqueness in the implicit-function theorem
then gives agreement near each common point, and
the identity theorem extends this agreement over
the overlaps. Consequently, there are $d_{\mathrm{mid}}>0$ and
an open neighbourhood $V_0\Subset V_1$ of $t_0$
on which these local inverses define a single
holomorphic branch for
\[
    |u|<T+1,
    \qquad |v|<d_{\mathrm{mid}},
    \qquad t\in V_0.
\]
Its overlap with each end branch contains a segment
of the original lifted line.
By the same local uniqueness argument, after
shrinking $V_0$ if necessary, the middle and end
branches agree on these overlaps. Choosing $0<d<\min(1,d_{\mathrm{mid}})$ therefore gives a jointly holomorphic inverse
$w(s,t)$ on the full product
$H_d(s_0)\times V_0$, satisfying
\[
    2\widetilde Z_{\tau(t)}(w(s,t))=s,
    \qquad w(s_0,t_0)=w_0.
\]
The construction keeps this inverse branch regular
throughout the product.
It does not require other inverse branches over
the same action values to be regular.

\noindent\textbf{Step 5. Projection, estimates and local uniformity.}
Define
\[
    x(s,t)=w(s,t)^2-2\tau(t),
    \qquad
    y(s,t)=2\bigl(x(s,t)-\tau(t)\bigr)w(s,t).
\]
These functions are jointly holomorphic, and
$x(s_0,t_0)=x_0$.
Furthermore,
\[
    \partial_w\bigl(2\widetilde Z_{\tau(t)}\bigr)(w)
       =8w^2(w^2-3\tau(t))
       =4wy.
\]
Regularity of the constructed inverse therefore gives
$wy\neq0$ throughout $H_d(s_0)\times V_0$.
Differentiating with respect to $s$, with $t$ fixed,
also gives $\partial_s x
       =2w\,\partial_s w
       =\frac{1}{2y}\neq0$.
Thus $x(s,t)$ is the required continuation of the
local spatial inverse selected by $w_0$. The expansions from Step~3 imply, uniformly for
$t\in V_0$ at either end, $|w(s,t)|\asymp |s|^{\frac{1}{5}}$. Since $\tau(t)$ is bounded, we also have
\[
    x(s,t)
       =w(s,t)^2\left(1+O(|s|^{-\frac{2}{5}})\right),
    \qquad
    y(s,t)
       =2w(s,t)^3\left(1+O(|s|^{-\frac{2}{5}})\right).
\]
These identities give the asserted estimates for
$x$ and $y$.

Finally, let the base point vary in a sufficiently
small neighbourhood of $(x_0,t_0)$ contained in
$\mathcal U_j$, and continue the chosen sheet
holomorphically there.
The corresponding action value $s_b=2Z_{\tau(t_b)}(x_b)$ is then close to $s_0$.
After shrinking this neighbourhood, we have
$t_b\in V_0$ and
\[
    \left|\operatorname{Im}
       \bigl(e^{-i\alpha}(s_b-s_0)\bigr)\right|
       <\frac d2,
    \qquad
    H_{\frac{d}{2}}(s_b)\subset H_d(s_0).
\]
Local uniqueness of the inverse ensures that the
already constructed branch takes the prescribed
lifted value at $(s_b,t_b)$.
Its restrictions to these smaller strips therefore
provide the construction for all such nearby base
points, with the same width $\frac{d}{2}$, time neighbourhood
$V_0$ and bounds at infinity.
This proves the stated local uniformity and ends the proof of \autoref{lem:spatial-adapted-strips}.

\subsection{Proof of \autoref{lem:projective-reconstruction}}
\label{app:projective-reconstruction-proof}
We use the notation and hypotheses of
\autoref{lem:projective-reconstruction}. In particular, $\Omega$ is
a polydisc centred at $b$, the formal columns are given by
\eqref{eq:formal-projective-columns}, and
$\Omega_0\Subset\Omega$ is a smaller polydisc centred at $b$.
The proof first establishes the summability operations used in the
reconstruction, then constructs the scalar amplitudes from a closed
one-form and verifies the full Lax system.

\begin{proof}
All Borel sums below are taken in direction $\alpha$. We first record
the summability operations needed in the proof. Locally uniform Borel
summation commutes with addition, multiplication, differentiation in
$x$ and $t$, and integration along finite parameter paths staying in a fixed
compact subset. It also commutes with scalar exponentiation and, when
the constant coefficient vanishes, division by $\hbar$.

For completeness, these statements follow directly in the Borel plane.
Multiplication of series with zero constant coefficient becomes
convolution. If the Borel transform of such a series is bounded by
$C e^{c|\xi|}$ on a tube, its $m$-fold convolution is bounded by
$ C^m e^{c|\xi|}\frac{|\xi|^{m-1}}{(m-1)!}$. Consequently the convolution series for the exponential converges
normally and has at most exponential growth. Constant coefficients
can be separated before applying this argument. Parameter derivatives
are controlled by Cauchy estimates on a slightly larger compact set;
finite path integrals preserve the same bounds. Finally, if
$\hat f=f_1\hbar+f_2\hbar^2+\cdots$, the Borel transform of
$\frac{\hat f}{\hbar}-f_1$ is $\partial_\xi\mathcal B\hat f$.
Cauchy estimates on a narrower tube control this derivative, and
integration by parts in the Laplace integral gives
$s_\alpha\left(\frac{\hat f}{\hbar}\right)=\frac{s_\alpha(\hat f)}{\hbar}$.
Inversion of a series with nonvanishing leading coefficient is
also permitted locally: after reducing to leading coefficient $1$,
the Borel transform of its inverse is given by a convergent
convolution geometric series on a smaller tube, with an increased
exponential bound. The same convolution estimate just displayed
proves this assertion. These arguments justify the interchanges below.

\paragraph{The two projective equations.}
For each sign, the formal ratio $\hat r_\varepsilon$
satisfies the time Riccati equation
\eqref{eq:formal-time-Riccati} and the spatial
Riccati equation \eqref{eq:formal-spatial-Riccati-B},
whose compatibility was established in
\autoref{250408124924}.
Applying Borel--Laplace summation to these identities,
using its compatibility with multiplication and
differentiation in $(x,t)$ proved above, shows that
$r_\varepsilon=s_\alpha(\hat r_\varepsilon)$ satisfies
\begin{align}
    \hbar\partial_t r_\varepsilon
       &=\frac{x+2q}{2}-2r_\varepsilon^2,
       \label{eq:auxiliary-Riccati}\\
    \hbar\partial_x r_\varepsilon
       &=x^2+xq+q^2+\frac t2
          -2pr_\varepsilon-4(x-q)r_\varepsilon^2.
       \label{eq:exact-projective-space}
\end{align}
Both identities hold on $\Omega_0\times S_{\alpha,R}$ after decreasing
$R$. In particular, no assumption that an exact wave matrix already
solves the $x$-equation has been used here.

\paragraph{A closed logarithmic one-form.}
Recall the formal closed one-form
$\hat\beta_\varepsilon$ in
\eqref{eq:formal-amplitude-form}, whose closedness
is established in
\autoref{prop:formal-amplitude-reconstruction}.
Set
\[
    P_\varepsilon
       =s_\alpha(\hat P_\varepsilon)
       =p+4(x-q)r_\varepsilon,
    \qquad
    T_\varepsilon
       =s_\alpha(\hat T_\varepsilon)
       =2r_\varepsilon,
\]
using \eqref{eq:formal-log-derivatives}.
The summability rules proved above, including division
by $\hbar$ when the constant coefficient vanishes, give
\begin{equation}
\label{eq:closed-logarithmic-amplitude-form}
    \beta_\varepsilon
       \coloneqq s_\alpha(\hat\beta_\varepsilon)
       =\frac{P_\varepsilon-\varepsilon y}{\hbar}
          \,\mathrm dx
        +\frac{T_\varepsilon-\varepsilon w}{\hbar}
          \,\mathrm dt.
\end{equation}
Here summation is applied to the coefficients of the
one-form. Since it commutes with differentiation in
$(x,t)$, we obtain
\[
    \mathrm d\beta_\varepsilon
       =s_\alpha(\mathrm d\hat\beta_\varepsilon)=0.
\] 
Since $\Omega_0$ is a polydisc, define
\begin{equation*}
    J_\varepsilon(x,t,\hbar)
       =\int_b^{(x,t)}\beta_\varepsilon,
    \qquad J_\varepsilon(b,\hbar)=0.
\end{equation*}
For each fixed $\hbar\in S_{\alpha,R}$,
$\beta_\varepsilon$ is a holomorphic closed one-form
on the polydisc $\Omega_0$, which is simply connected.
It therefore has a unique holomorphic primitive
vanishing at $b$. Consequently, the integral defining
$J_\varepsilon$ is independent of the path in $\Omega_0$;
in particular, it may be taken along the straight
segment from $b$ to $(x,t)$. It is holomorphic in $(x,t,\hbar)$ and is the
Borel sum of the corresponding formal primitive. Define
\begin{equation}
\label{eq:reconstructed-amplitude}
    a_\varepsilon
       =s_\alpha(\hat b_\varepsilon)\exp J_\varepsilon
    \qquad
    u_\varepsilon=a_\varepsilon e^{\frac{S_\varepsilon}{\hbar}},
    \qquad \widetilde u_\varepsilon=r_\varepsilon u_\varepsilon.
\end{equation}
Applying Borel--Laplace summation to the formal
amplitude formula \eqref{eq:formal-amplitude-primitive},
with basepoint $b_0=b$, and using compatibility with
finite path integration and exponentiation, gives
$a_\varepsilon=s_\alpha(\hat a_\varepsilon)$.

\paragraph{Verification of the full system.}
Equations \eqref{eq:reconstruction-phases} and
\eqref{eq:closed-logarithmic-amplitude-form} give
\begin{equation}
\label{eq:amplitude-transport}
    \hbar\partial_ta_\varepsilon
       =(2r_\varepsilon-\varepsilon w)a_\varepsilon,
    \qquad
    \hbar\partial_xa_\varepsilon
       =(p+4(x-q)r_\varepsilon-\varepsilon y)a_\varepsilon.
\end{equation}
Consequently
\begin{equation*}
    \hbar\partial_tu_\varepsilon=T_\varepsilon u_\varepsilon,
    \qquad
    \hbar\partial_xu_\varepsilon=P_\varepsilon u_\varepsilon.
\end{equation*}
Together with \eqref{eq:auxiliary-Riccati} and
\eqref{eq:exact-projective-space}, this yields
\begin{align*}
    \hbar\partial_tu_\varepsilon&=2\widetilde u_\varepsilon,&
    \hbar\partial_t\widetilde u_\varepsilon&=\left(\frac x2+q\right)u_\varepsilon,\\
    \hbar\partial_xu_\varepsilon&=pu_\varepsilon+4(x-q)\widetilde u_\varepsilon,&
    \hbar\partial_x\widetilde u_\varepsilon
       &=\left(x^2+xq+q^2+\frac t2\right)u_\varepsilon-p\widetilde u_\varepsilon.
\end{align*}
Thus the matrix with columns $(u_\varepsilon,\widetilde u_\varepsilon)^{\mathsf T}$
solves the complete Painlev\'e~I system \eqref{250408155610_bis}. By construction, $J_\varepsilon(b,\hbar)=0$.
Thus \eqref{eq:reconstructed-amplitude} gives $a_\varepsilon(b,\hbar)=s_\alpha(\hat b_\varepsilon)(\hbar)$. Moreover, \eqref{eq:reconstruction-phases} gives
$S_\varepsilon(b)=-\kappa_\varepsilon$.
Since the first component of the corresponding column
is $u_\varepsilon=a_\varepsilon e^{\frac{S_\varepsilon}{\hbar}}$
by \eqref{eq:reconstructed-amplitude}, we obtain $u_\varepsilon(b,\hbar)=s_\alpha(\hat b_\varepsilon)(\hbar)e^{-\frac{\kappa_\varepsilon}{\hbar}}$, which is the prescribed first-row normalisation. The amplitude determinant is $a_+a_-(r_--r_+)$. Its leading coefficient is $-w a_{0,+}a_{0,-}\neq0$.
Uniform asymptotics on a slightly larger polydisc therefore make the
matrix invertible on $\Omega_0\times S_{\alpha,R}$ when $R$ is
sufficiently small. Borel--Laplace summation with tubular continuation
gives \eqref{eq:exact-amplitude-Gevrey}.
Finally, uniqueness follows from the Nevanlinna uniqueness theorem:
an amplitude with zero uniform Gevrey-$1$ expansion on a Borel disk
vanishes identically. See also the uniqueness discussion following
\cite[Theorem~1]{nikolaev_2022_3}. Apply this entrywise to the
difference of two amplitudes with the prescribed phases and expansion.
\end{proof}

\section{Proof of \autoref{prop:vertex-projective-summability}: Projective summability at the double turning point}
\label{app:uniform-Morse-estimates}

This appendix gives the complete proof of
\autoref{prop:vertex-projective-summability}.
We first establish uniform estimates for the rescaled Painlev\'e
functions on a full action strip. We then pull back the Lax system
to the coordinates of \autoref{lem:Morse-family} and verify the
coefficient estimates required by \autoref{lem:projective-strip}.
Finally, we apply that lemma and return to the original variables.
Throughout, the summation direction $\alpha$, the time sector,
the branches, and the point $v_\ast$ are those fixed in
\autoref{prop:vertex-projective-summability}.

\subsection{Uniform Painlev\'e estimates on action strips}

We use the change of independent variables constructed in
\autoref{lem:Morse-family}. In the chosen branches, it is
given locally near the double turning point by
\[
    z=-\frac{48\sqrt3}{5}\tau(t)^{\frac52},
    \qquad
    \zeta=\tau(t)^{\frac54}
        f\!\left(\frac{x}{\tau(t)}\right),
\]
where $f$ is the Morse function introduced in the proof
of that lemma.
The first relation has a holomorphic inverse
$\tau=\tau(z)$ on $\mathbb H_\ast$.
Writing $F$ for the local inverse of $f$, the inverse
change of variables is
\[
    t(z)=-6\tau(z)^2,
    \qquad
    x(\zeta,z)=\tau(z)
        F\!\left(\zeta\tau(z)^{-\frac54}\right),
\]
as in \eqref{eq:explicit-Morse-family}.
Thus $\Phi(\zeta,z)=(x(\zeta,z),t(z))$ is the
parametrisation of \autoref{lem:Morse-family}.
It is locally biholomorphic near $(0,z_\ast)$
and therefore defines a local change of coordinates
between $(x,t)$ and $(\zeta,z)$.

In this subsection, we use only the time change
$t=t(z)$: the Painlev\'e functions $q(t(z),\hbar)$
and $p(t(z),\hbar)$, and hence their rescalings $Q$
and $P$ below, depend on $(z,\hbar)$ and are
independent of $\zeta$.
In the next subsection, we pull back the full Lax
system along $\Phi$, treating $\zeta$ as a parameter
and differentiating with respect to $z$ at fixed
$\zeta$.

\begin{lemma}[Uniform Painlev\'e estimates on action strips]
\label{lem:PI-uniform-action-strip}
For the Borel sums $(q,p)$ of \autoref{250408141643}, set locally
\begin{equation}
\label{eq:scaled-Q-P}
    Q(z,\hbar)=\frac{q(t(z),\hbar)}{\tau(z)},
    \qquad
    P(z,\hbar)=\frac{p(t(z),\hbar)}{2\sqrt3\,\tau(z)^{\frac{3}{2}}}.
\end{equation}
Let $0<d'<d$, with
$\overline{H_d(z_\ast)}\subset\mathbb H_\ast$.
There exist constants $C,K,R>0$ such that $Q$ and $P$ extend
jointly holomorphically to
$H_{d'}(z_\ast)\times S_{\alpha,R}$ and satisfy
\begin{equation}
\label{eq:uniform-scaled-PI-remainders}
\begin{aligned}
    \sup_{z\in H_{d'}(z_\ast)}
    \Bigg(
       \left|Q(z,\hbar)-\sum_{n=0}^{N-1}Q_n(z)\hbar^n\right|
       +\left|P(z,\hbar)-\sum_{n=0}^{N-1}P_n(z)\hbar^n\right|
    \Bigg)
    \leq C K^N N!|\hbar|^N
\end{aligned}
\end{equation}
for every integer $N\geq0$ and every $\hbar\in S_{\alpha,R}$.
Here $(Q_n,P_n)_{n\geq 0}$ are the rescaled formal Painlev\'e coefficients. They satisfy $Q_{2n+1}=P_{2n}=0$ for every $n\geq0$ and $Q_0=1$ and $P_1(z)=\frac{2}{5z}$.
These extensions agree with the original local Borel sums
wherever both are defined, after restriction to a common Borel disk.
\end{lemma}

\begin{proof}
The estimates of \autoref{250408141643}, together with
\eqref{eq:scaled-Q-P}, already give the required estimates
on every compact subset of $H_d(z_\ast)$, with constants
and a Borel disk that may depend on that compact subset.
The main point here is to obtain a common Borel disk and
uniform constants on the entire unbounded strip
$H_{d'}(z_\ast)$.
We establish this additional uniformity using the scaled
differential system for $(Q,P)$.

Recall that the strip was chosen with $0<d<
    \left|\operatorname{Im}(e^{-i\alpha}z_\ast)\right|$,
as in the setup of \autoref{lem:Morse-family}.
With $\delta_\ast
       =\left|\operatorname{Im}(e^{-i\alpha}z_\ast)\right|-d
       >0$,
we have, for every $z\in H_d(z_\ast)$,
\begin{equation}
    |z|
       \geq
       \left|\operatorname{Im}(e^{-i\alpha}z)\right|
       \geq
       \left|\operatorname{Im}(e^{-i\alpha}z_\ast)\right|
       -
       \left|\operatorname{Im}
          \bigl(e^{-i\alpha}(z-z_\ast)\bigr)\right|
       >\delta_\ast.
\end{equation}
Consequently, 
\begin{equation} 0\notin H_d(z_\ast) \,\, \text{ and } \,\, 
\sup_{z\in H_d(z_\ast)}
       \left|\frac1z\right|
       \leq\delta_\ast^{-1}.
\end{equation}
The coordinate identities are
\[
    \partial_tz=2\sqrt3\,\tau^{\frac{1}{2}},
    \qquad \partial_t\tau=-\frac1{12\tau},
    \qquad z=-\frac{48\sqrt3}{5}\tau^{\frac{5}{2}}.
\]
Substitute \eqref{eq:scaled-Q-P} into the Painlev\'e equations
$\hbar\partial_tq=p$ and $\hbar\partial_tp=6q^2+t$.
For the first equation, use
\[
    q_t=\tau_tQ+\tau z_tQ_z,
    \qquad \frac{\tau_t}{\tau z_t}=\frac2{5z}.
\]
For the second equation, use
\[
    p_t=2\sqrt3\,\tau^{\frac{3}{2}}z_tP_z
           +3\sqrt3\,\tau^{\frac{1}{2}}\tau_tP,
    \qquad 2\sqrt3\,\tau^{\frac{3}{2}}z_t=12\tau^2,
\]
together with $6q^2+t=6\tau^2(Q^2-1)$.
The resulting system is
\begin{equation}
\label{eq:scaled-Painleve-system}
    \hbar\partial_zQ=P-\frac{2\hbar}{5z}Q,
    \qquad
    \hbar\partial_zP=\frac{Q^2-1}{2}-\frac{3\hbar}{5z}P.
\end{equation}

Put $U=Q-1$. The coefficient of $\hbar^0$ in the
right-hand side for the column $(U,P)^{\mathsf T}$
is $(P,U+\frac{U^2}{2})^{\mathsf T}$.
It vanishes at $(U,P)=(0,0)$ and has Jacobian $J=\begin{pmatrix}0&1\\1&0\end{pmatrix}$ there. The constant change of unknowns
$\binom{U}{P}
       =\begin{pmatrix}1&1\\1&-1\end{pmatrix}
          \binom{u_+}{u_-}$
diagonalises $J$ to $\operatorname{diag}(1,-1)$.
The transformed right-hand side is polynomial in the
unknowns and affine in $\hbar$, so its expansion in
$\hbar$ has zero remainder after the terms of orders
zero and one. Its dependence on $z$ occurs only
through $1/z$.
Since $1/z$ is uniformly bounded on the entire strip
$H_d(z_\ast)$, the transformed right-hand side satisfies
the Gevrey assumptions of
\autoref{lem:uniform-strip-summation} uniformly for
$z\in H_d(z_\ast)$ and for the unknown vector
$v=(u_+,u_-)^{\mathsf T}$ in any fixed bounded subset
of $\mathbb C^2$. The constants may depend on that
bounded subset, but are independent of $z$ throughout
the full strip.
Applying \autoref{lem:uniform-strip-summation},
undoing the constant change of unknowns, and restoring
$Q=1+U$ gives jointly holomorphic Borel sums on
$H_{d'}(z_\ast)\times S_{\alpha,R}$ satisfying
\eqref{eq:uniform-scaled-PI-remainders}.

The formal solution of \eqref{eq:scaled-Painleve-system} with
$(Q_0,P_0)=(1,0)$ is unique, since $J$ is invertible.
The rescaling of the formal solution in \autoref{250408135328}
satisfies that system with these leading coefficients, so it is
the formal solution just summed.
Locally in $z$, Borel summation commutes with the holomorphic
change $t=t(z)$ and multiplication by the scaling factors in
\eqref{eq:scaled-Q-P}. Uniqueness of the Borel sum therefore
identifies the constructed functions with the rescaled exact
functions of \autoref{250408141643} on their common domains.
This proves both the extension assertion and the uniform estimate
for the selected Painlev\'e solution.
\end{proof}

\begin{remark}
Estimate \eqref{eq:uniform-scaled-PI-remainders} holds
uniformly on the full strip $H_{d'}(z_\ast)$.
No additional factor $|z|^{-N}$ is needed.
In particular, this argument does not require
$\tau(z)$ to remain in a closed subsector of
$\widetilde V$ as $z$ tends to either end of the strip.
\end{remark}

\subsection{The Lax system in the moving coordinates}

Use the functions $x(\zeta,z)$, $t(z)$ and $w(\zeta,z)$ of
\autoref{lem:Morse-family}, defined on
$\Delta_\rho\times H_d(z_\ast)$.
Since this product is simply connected and $w$ is nowhere zero,
choose a holomorphic square root of $w$ and use the eigenbasis
\eqref{eq:regular-eigenbasis}, namely
\[
    E=w^{-\frac{1}{2}}
       \begin{pmatrix}1&1\\\frac{w}{2}&-\frac{w}{2}\end{pmatrix},
    \qquad \det E=-1.
\]
For a local solution column $\mathbf v$ of the full Lax system
\eqref{250408155610}, set
\[
    \chi(\zeta,z,\hbar)
       =E(\zeta,z)^{-1}
          \mathbf v(x(\zeta,z),t(z),\hbar).
\]
At fixed $\zeta$, the chain rule gives
$\hbar\partial_z\mathbf v=(x_zL+t_zA)\mathbf v$.
Consequently,
\begin{equation}
\label{eq:Morse-pulled-connection}
    \hbar\partial_z\chi=M_{\mathrm v}\chi,
    \qquad
    M_{\mathrm v}
       =E^{-1}(x_zL+t_zA)E-\hbar E^{-1}\partial_zE.
\end{equation}
All matrices here are evaluated at $(x,t)=(x(\zeta,z),t(z))$.
Both the definition of $M_{\mathrm v}$ and the subsequent
coefficient estimates are independent of the choice of solution
column used to explain this change of variables. The same eigenbasis $E$ diagonalises the two leading Lax matrices:
\[
    E^{-1}L_0E=\operatorname{diag}(y,-y),
    \qquad E^{-1}A_0E=\operatorname{diag}(w,-w).
\]
Since $\mathrm dZ=y\,\mathrm dx+w\,\mathrm dt$, the action identity
\eqref{eq:Morse-action-coordinate} yields
\[
    x_zy+t_zw
       =\partial_z\bigl(Z_{\tau(z)}(x(\zeta,z))\bigr)
       =\frac12
\]
at fixed $\zeta$. Thus the leading coefficient of $M_{\mathrm v}$ is
$D=\operatorname{diag}\left(\frac{1}{2},-\frac{1}{2}\right)$, including at $\zeta=0$.

\begin{lemma}[Uniform coefficients on the Morse family]
\label{lem:Morse-uniform-coefficients}
For $0<\rho_1<\rho$ and $0<d_1<d$, there exist constants
$C_1,K_1,R_1>0$ such that $M_{\mathrm v}$ is jointly holomorphic on $\Delta_{\rho_1}\times H_{d_1}(z_\ast)
       \times S_{\alpha,R_1}$
and has a uniform Gevrey-$1$ expansion
$\hat M_{\mathrm v}=\underset{n\geq 0}{\sum}M_{\mathrm v,n}\hbar^n$
there, with $M_{\mathrm v,0}=D$.
More explicitly,
\[
    \sup_{\substack{|\zeta|<\rho_1\\z\in H_{d_1}(z_\ast)}}
    \left\|M_{\mathrm v}(\zeta,z,\hbar)
       -\sum_{n=0}^{N-1}M_{\mathrm v,n}(\zeta,z)\hbar^n\right\|
       \leq C_1(K_1)^N N!|\hbar|^N
\]
for every integer $N\geq0$ and every
$\hbar\in S_{\alpha,R_1}$.
\end{lemma}

\begin{proof}Let $f$ be the Morse function constructed in
\autoref{lem:Morse-family}, and let $F$ be its local
inverse, defined for $|\eta|<\eta_0$.
Put
\[
    s=s(\zeta,z)=\frac{x(\zeta,z)}{\tau(z)},
    \qquad
    v(s)\coloneqq s-\frac54\frac{f(s)}{f'(s)}.
\]
All derivatives with respect to $\tau$ or $z$ below
are taken with $\zeta$ fixed.
By \eqref{eq:explicit-Morse-family} and the bounds
$|\tau(z)|\geq\mu$ and
$\rho\mu^{-\frac54}<\frac{\eta_0}{2}$, we have
\[
    s(\zeta,z)\in K_s
       \coloneqq
       F\!\left(
          \left\{\eta\in\mathbb C:
                 |\eta|\leq\frac{\eta_0}{2}\right\}
       \right)
\]
throughout $\Delta_{\rho_1}\times H_{d_1}(z_\ast)$.
The set $K_s$ is compact and lies in the neighbourhood
of $1$ where $f$ and the chosen branch of $(s+2)^{\frac{1}{2}}$
are holomorphic, with $f'(s)\neq0$ and $s+2\neq0$.
Consequently, $v(s)$, $\frac{f(s)}{f'(s)}$ and
$(s+2)^{\pm \frac{1}{2}}$ are uniformly bounded on $K_s$,
independently of $(\zeta,z)$.
Differentiating
$\zeta=\tau^{\frac{5}{4}}f(s)$ gives
\begin{equation}
\label{eq:Morse-coordinate-derivatives}
    x_\tau=v(s),\qquad
    x_z=-\frac{v(s)}{24\sqrt3\,\tau^{\frac{3}{2}}},\qquad
    t_z=\frac1{2\sqrt3\,\tau^{\frac{1}{2}}},\qquad
    s_z=-\frac{f(s)}{2zf'(s)}.
\end{equation}
In particular, the functions of $s$ in these formulas are
uniformly bounded on the stated product.
Since $w=\tau^{\frac{1}{2}}(s+2)^{\frac{1}{2}}$, we also obtain
\begin{equation}
\label{eq:log-w-Morse}
    \frac{w_z}{w}
       =\frac1{5z}-\frac{f(s)}{4z(s+2)f'(s)}.
\end{equation}

Put $\Delta=q-\tau$.
Direct multiplication using \eqref{eq:regular-eigenbasis} gives
\begin{align}
    E^{-1}(L-L_0)E
       &=\begin{pmatrix}
          \frac{\Delta^2}{w} & p+\frac{\Delta(2x+q+3\tau)}{w}\\
          p-\frac{\Delta(2x+q+3\tau)}{w} & -\frac{\Delta^2}{w}
         \end{pmatrix}
         \eqqcolon D_x,
         \label{eq:delta-L-eigenbasis}\\
    E^{-1}(A-A_0)E
       &=\frac{\Delta}{w}
           \begin{pmatrix}1&1\\-1&-1\end{pmatrix}
         \eqqcolon D_t,\qquad
    E^{-1}E_z
       =-\frac{w_z}{2w}
           \begin{pmatrix}0&1\\1&0\end{pmatrix}.
         \label{eq:delta-A-eigenbasis}
\end{align}
Together with the leading coefficient computed above, this yields
\begin{equation}
\label{eq:Morse-remainder-formula}
    M_{\mathrm v}-D
       =x_zD_x+t_zD_t
          +\frac{\hbar w_z}{2w}
              \begin{pmatrix}0&1\\1&0\end{pmatrix}.
\end{equation}

To see why this expression is uniformly controlled even at the
strip ends, set $\theta(s)=(s+2)^{\frac{1}{2}}$ and substitute
$\Delta=\tau(Q-1)$, $p=2\sqrt3\,\tau^{\frac{3}{2}}P$,
$x=\tau s$ and $w=\tau^{\frac{1}{2}}\theta(s)$.
All powers of $\tau$ cancel:
\[
\begin{aligned}
    x_zD_x
       &=-\frac{v(s)}{24\sqrt3}
          \begin{pmatrix}
           \frac{(Q-1)^2}{\theta(s)}
             &2\sqrt3 P+\frac{(Q-1)(2s+Q+3)}{\theta(s)}\\
           2\sqrt3 P-\frac{(Q-1)(2s+Q+3)}{\theta(s)}
             &-\frac{(Q-1)^2}{\theta(s)}
          \end{pmatrix},\\
    t_zD_t
       &=\frac{Q-1}{2\sqrt3\,\theta(s)}
           \begin{pmatrix}1&1\\-1&-1\end{pmatrix}.
\end{aligned}
\]
The gauge term in \eqref{eq:Morse-remainder-formula} is $\frac{\hbar}{z}$
times a bounded holomorphic matrix, by \eqref{eq:log-w-Morse}.
Thus every entry of $M_{\mathrm v}-D$ is a polynomial in
$Q-1$, $P$ and $\frac{\hbar}{z}$, with bounded holomorphic coefficients
depending on $s$. In particular,
\[
    \|M_{\mathrm v}-D\|
       \leq C\left(|Q-1|+|P|+|Q-1|^2+\left|\frac{\hbar}{z}\right|\right).
\]
By \autoref{lem:PI-uniform-action-strip}, $Q$ and $P$ have
uniform Gevrey-$1$ expansions on $H_{d_1}(z_\ast)$ with a common
Borel disk. Finite products and bounded holomorphic multipliers
preserve these remainder estimates, after decreasing the disk
parameter if necessary. Since $\frac{1}{z}$ is uniformly bounded, the
displayed formulas prove the asserted estimates and joint
holomorphy. Their leading remainder coefficient vanishes because
$Q_0=1$ and $P_0=0$.
\end{proof}

\subsection{Proof of \autoref{prop:vertex-projective-summability}}

\begin{proof}
Choose $0<\rho_2<\rho_1<\rho$ and $0<d_2<d_1<d$.
By \autoref{lem:Morse-uniform-coefficients}, the pulled-back
system \eqref{eq:Morse-pulled-connection} has the form
\[
    \hbar\partial_z\chi=(D+B)\chi,
    \qquad B=M_{\mathrm v}-D,\qquad B_0=0,
\]
with the uniform coefficient estimates required by
\autoref{lem:projective-strip}.
We apply \autoref{lem:projective-strip} with $z$ as the
differential variable and $\zeta$ as the parameter.
More explicitly, the correspondence with the notation of \autoref{lem:projective-strip} is $(m,k,U,U_0,z_0,d,d',R_0)=(1,\zeta,\Delta_{\rho_1},\Delta_{\rho_2},z_\ast,d_1,d_2,R_1)$. The summation direction remains $\alpha$, and the matrix $B$ in \autoref{lem:projective-strip} is $M_{\mathrm v}-D$.
Indeed, \autoref{lem:Morse-uniform-coefficients}
gives its joint holomorphy and uniform Gevrey estimates on
$\Delta_{\rho_1}\times H_{d_1}(z_\ast)\times S_{\alpha,R_1}$,
with $B_0=0$.
Since $\Delta_{\rho_2}\Subset\Delta_{\rho_1}$ and
$d_2<d_1$, all hypotheses of
\autoref{lem:projective-strip} are satisfied.
\autoref{lem:projective-strip} supplies unique formal eigenbasis ratios $\hat g_\pm$
with zero constant coefficient and their Borel sums $g_\pm$,
jointly holomorphic on $\Delta_{\rho_2}\times H_{d_2}(z_\ast)
       \times S_{\alpha,R_2}$
for some $0<R_2\leq R_1$.
The Borel tubes, exponential bounds and Gevrey remainder
constants are uniform on this product, and $g_\pm=O(\hbar)$
uniformly there.

We identify these series with the directions of the formal
columns from \autoref{241113190220}.
Pulling back such a column by $\Phi$ and multiplying by $E^{-1}$
gives a formal column of \eqref{eq:Morse-pulled-connection}.
For the positive sign its leading direction is $(1,0)^{\mathsf T}$,
and for the negative sign it is $(0,1)^{\mathsf T}$.
Its corresponding eigenbasis ratio therefore has zero constant
coefficient and solves the same formal Riccati equation as
$\hat g_\varepsilon$.
The linearised leading coefficients of these two equations
are $-1$ and $1$, respectively, so coefficientwise recursion
determines their formal solutions uniquely.
It follows that these ratios are precisely $\hat g_\varepsilon$.
Returning to the original basis gives $\hat r_\varepsilon
       \bigl(x(\zeta,z),t(z);\hbar\bigr)
       =
       \varepsilon\frac{w(\zeta,z)}2
       \frac{1-\hat g_\varepsilon(\zeta,z;\hbar)}
            {1+\hat g_\varepsilon(\zeta,z;\hbar)}$,
as in \eqref{eq:projective-back-transform}.

After decreasing $R_2$, the uniform estimate $g_\pm=O(\hbar)$
ensures $|g_\pm|<\frac{1}{2}$ throughout the product.
Thus $1+g_\pm$ does not vanish, and we may define
\[
    \widetilde r_\varepsilon(\zeta,z,\hbar)
       =\varepsilon\frac{w(\zeta,z)}2\,
          \frac{1-g_\varepsilon(\zeta,z,\hbar)}
               {1+g_\varepsilon(\zeta,z,\hbar)}.
\]
Stability of Borel summation under products and reciprocals
of formal units shows that this is the locally uniform
Borel sum, with respect to $\hbar$, of
$\hat r_\varepsilon\bigl(x(\zeta,z),t(z);\hbar\bigr)$.
The factor $w$ is holomorphic and bounded on each compact
coordinate neighbourhood, which is sufficient for this local
assertion.

By \autoref{lem:Morse-family}, $\Phi$ is biholomorphic from a
neighbourhood $W$ of $(0,z_\ast)$ onto a full neighbourhood of
$v_\ast$. We may choose $W$ inside
$\Delta_{\rho_2}\times H_{d_2}(z_\ast)$.
Since $\Phi$ is independent of $\hbar$, composing with its inverse
commutes with the Borel transform and the Laplace integral.
Consequently, $r_\varepsilon(x,t,\hbar)
       =\widetilde r_\varepsilon
          \bigl(\Phi^{-1}(x,t),\hbar\bigr)$
is the locally uniform Borel sum of $\hat r_\varepsilon$ on
$\Phi(W)$.

Choose a polydisc $\Omega$ centred at $v_\ast$ with
$\overline\Omega\subset\Phi(W)$.
Local uniform summability and compactness of $\overline\Omega$
give a common tube $\Xi_{\alpha,\delta}$ and constants $C,c>0$
such that the continued Borel transforms $\varphi_\varepsilon$
are jointly holomorphic on $\Omega\times\Xi_{\alpha,\delta}$ and
satisfy $\sup_{(x,t)\in\Omega}
       |\varphi_\varepsilon(x,t,\xi)|
       \leq C e^{c|\xi|}$
for both signs.
Choose $0<R\leq R_2$ with $R^{-1}>c$.
The Borel--Laplace integral in the proposition then converges
normally on $\Omega\times S_{\alpha,R}$ and defines jointly
holomorphic functions there.
After a further decrease of $R$ if needed, their uniform
Gevrey-$1$ remainder estimates hold throughout $\Omega$, by
the product and reciprocal estimates used above on a slightly
larger neighbourhood of $\overline\Omega$.
Finally, $\Omega$ is a full neighbourhood of $v_\ast$;
it therefore crosses all four incident Stokes curves.
The same Borel transforms and Laplace integrals apply on all
sides of these curves, as asserted.
\end{proof}

\section{An Order-by-Order Matrix Construction of Formal Solutions}
\label{app:matrix-formal-construction}

This appendix gives the full matrix proof of
\autoref{241113190220} and \autoref{250716102028}, independently of
the Riccati reduction in \autoref{250109110415}. We retain the seven-step gauge
reduction, the detailed first-order example, and the coefficientwise
compatibility calculations. The last subsection records how the
matrix induction extends to several deformation variables.
The factorial estimates remain in \autoref{app:factorial-type-proof}.

\subsection{Existence and Uniqueness}
\label{Dold:250408124924}
\label{app:matrix-classical-gauge}

This subsection is devoted to the proof of \autoref{250716102028}, by an explicit order-by-order construction of the formal WKB columns $\hat v_\pm$. The column statement \autoref{241113190220} then follows by taking either column of the resulting wave matrix.
The proof is organised into three main stages.
The first stage --- summarised in \autoref{Dold:250109115117} --- is to make a suitable gauge transformation in order to reduce the construction of the transseries solution $\hat{\Psi}$ to the question of constructing an (ordinary) formal power series solution $\hat{\Phi}$ to a related differential system.
The second stage --- summarised in \autoref{Dold:250417183730} --- is to construct this solution $\hat{\Phi}$ recursively order-by-order in $\hbar$.
The spatial and time equations give two algebraic constraints on each off-diagonal coefficient and two differential constraints on each diagonal coefficient.
The third and final stage --- summarised in \autoref{Dold:250109110444} and \autoref{Dold:250123150709}, proved in \autoref{Dold:app:compatibility-identities} --- verifies the compatibility of these constraints. The lemmas use only the already constructed orders; their precise inductive hypotheses are stated below. 
The first stage in the proof of \autoref*{250716102028} is to reduce the construction to \autoref{Dold:250417183730}.
This stage is organised into seven steps.
In Step 1, we plug the \textit{matrix WKB ansatz} \eqref{250421155248} into the formal Painlevé I system \eqref{250408174515} to separate the problem into equations for $\hat{\mathit{E}}$ and equations for $\mathit{S}$.
In Step 2, we solve for the \textit{WKB exponent} $\mathit{S}$.
In Steps 3 and 4, we apply a gauge transformation to the system for $\hat{\mathit{E}}$ that diagonalises the leading-order part, which results in a related and simplified system for $\hat{\Phi}$ as summarised in \autoref{Dold:250109115117}.
In Steps 5 and 6, we solve for the leading-order part of $\hat{\Phi}$ and hence $\hat{\mathit{E}}$.
Finally, in Step 7 we state the existence and uniqueness result for $\hat{\Phi}$ (\autoref{Dold:250417183730}) and explain how it implies \autoref*{241113190220}.
Fix a sufficiently small polydisc $W$ about $(x_0,t_0)$ on which the chosen branches $\tau,w$ are holomorphic and nonzero. All coefficients below are constructed on this same polydisc.
For each sign, write the prescribed nonzero normalisation as in \autoref{241113190220}:
\begin{equation}
\label{Dold:250110171027}
    \hat c_\pm(\hbar)
        =
        \hbar^{m_\pm}\hat{b}_\pm(\hbar)e^{-\frac{\kappa_\pm}{\hbar}},
    \qquad
    m_\pm\in\mathbb{Z}_{\geq0}, \quad \text{where} \quad 
    \hat{b}_\pm(\hbar)
        =
        \sum_{k=0}^{\infty}b_{\pm,k}\hbar^k,
    \qquad
    b_{\pm,0}\neq0.
\end{equation}
Since multiplication by the formal constant $\hbar^{m_\pm}$ preserves the formal Painlevé~I system, it is enough first to construct the solutions normalised by
\begin{equation}
    \hat{b}_\pm(\hbar)e^{-\frac{\kappa_\pm}{\hbar}}.
\end{equation}
The solutions with the prescribed normalisations $\hat c_\pm(\hbar)$ will then be obtained by multiplying the corresponding columns by $\hbar^{m_\pm}$.

\paragraph*{Step 1. The matrix WKB ansatz.}
If we plug the ansatz \eqref{250421155248} into the formal Painlevé I system \eqref{250408174515} and multiply through by $e^{-\frac{\mathit{S}}{\hbar}}$ on the right, then we find the following new formal system for $\hat{\mathit{E}}$ and $\mathit{S}$:
\begin{equation}
\label{Dold:241112153348}
\begin{cases}
	\hbar \partial_x \hat{\mathit{E}} + \hat{\mathit{E}} \partial_x \mathit{S} = \hat{\mathit{L}} \hat{\mathit{E}}
\\
	\hbar \partial_t \hat{\mathit{E}} + \hat{\mathit{E}} \partial_t \mathit{S} = \hat{\mathit{A}} \hat{\mathit{E}}
.
\end{cases}
\end{equation}
We will first examine these equations at the zeroth order in $\hbar$ to determine $\mathit{S}$ and the leading-order $\mathit{E}_0$ of $\hat{\mathit{E}}$:
\begin{equation}
\label{Dold:241216203759}
\begin{cases}
	\mathit{E}_0 \partial_x \mathit{S} = \mathit{L}_0 \mathit{E}_0
\\	\mathit{E}_0 \partial_t \mathit{S} = \mathit{A}_0 \mathit{E}_0
.
\end{cases}
\end{equation}
Meanwhile, the formal zero-curvature equation implies $[\mathit{L}_0, \mathit{A}_0] = 0$ which means the matrices $\mathit{L}_0$ and $\mathit{A}_0$ have common eigenvectors.
So the most straightforward way to solve equations \eqref{Dold:241216203759} is to take $\mathit{E}_0$ to be a matrix of common eigenvectors of $\mathit{L}_0$ and $\mathit{A}_0$, which forces the diagonal matrix $\mathit{S}$ to satisfy the following differential equation:
\begin{equation}
\label{Dold:241113155001}
	\,\mathrm{d} \mathit{S} = \mathit{Y} \,\mathrm{d}{x} + \mathit{W} \,\mathrm{d}{t}
\qquad\text{i.e.}\qquad
    \partial_x \mathit{S} = \mathit{Y}
\quad\text{and}\quad
    \partial_t \mathit{S} = \mathit{W},
\end{equation}
where $\mathit{Y}$ and $\mathit{W}$ are the diagonal matrices of eigenvalues of $\mathit{L}_0$ and $\mathit{A}_0$.
Namely,
\begin{align}
\label{Dold:241112152831}
    \mathit{L}_0 \mathit{E}_0 &= \mathit{E}_0 \mathit{Y}
\quad\text{where}\quad
    \mathit{Y} = \mathit{Y} (x,t) \coloneqq \begin{pmatrix} y (x,t) &0 \\ 0& - y (x,t) \end{pmatrix}
,
\\
\label{Dold:241125133953}
    \mathit{A}_0 \mathit{E}_0 &= \mathit{E}_0 \mathit{W}
\quad\text{where}\quad
    \mathit{W} = \mathit{W} (x,t) \coloneqq \begin{pmatrix} w (x,t) &0 \\0 & - w (x,t) \end{pmatrix}
,
\end{align}
where $y$ and $w$ are the local holomorphic solutions (chosen as part of the hypothesis) of
\begin{equation}
\label{Dold:241216205410}
	y^2 = 4 (x - \tau)^2 (x + 2\tau)
\quad\text{and}\quad
	w^2 = x + 2\tau
\quad\text{such that}\quad
	y = 2 (x - \tau) w
.
\end{equation}
In particular, the first row of
$\mathit A_0\mathit E_0=\mathit E_0\mathit W$
implies that the components of
\begin{equation}\label{Dold:DefE0}
	\mathit{E}_0 = \begin{pmatrix} e_+ & e_- \\ e'_+ & e'_- \end{pmatrix}
\end{equation}
must satisfy the following relations (here the primes label the second components, rather than derivatives):
\begin{equation}
\label{Dold:241112160711}
	e'_\pm =  \pm \tfrac{1}{2} w e_\pm
.
\end{equation}
This relation holds throughout the chosen polydisc, including
at $x=\tau$, without division by $x-\tau$.
In turn, it implies the identities
\begin{equation}
\label{Dold:241112201954}
    \det\mathit{E}_0=-w e_+e_-,
    \qquad
    e_+e'_-+e'_+e_-=0.
\end{equation}
Conversely, if $\mathit{S}$ is diagonal and
$\mathit{E}_0$ is invertible, then the leading-order equations
imply that the columns of $\mathit{E}_0$ are common eigenvectors
of $\mathit{L}_0$ and $\mathit{A}_0$.

\paragraph*{Step 2. Solving for the WKB exponent.}
The differential equation \eqref{Dold:241113155001} is solvable because the diagonal matrices $\mathit{Y}$ and $\mathit{W}$ satisfy the identity
\begin{equation}
\label{Dold:250129191049}
	\partial_t \mathit{Y} = \partial_x \mathit{W}
.
\end{equation}
This identity can be verified directly by differentiating $y (x,t) = 2 (x-\tau)\sqrt{x+2\tau}$ and $w (x,t) = \sqrt{x + 2\tau}$, and using the identity $-12 \tau \partial_t \tau = 1$.
Consequently, for any pair of constants $\kappa_\pm \in \mathbb{C}$, there is a unique solution $\mathit{S}$ satisfying the initial condition
\begin{equation}
\label{Dold:250110164958}
	\mathit{S} (x_0, t_0) = \operatorname{diag} (-\kappa_+, -\kappa_-)
.
\end{equation}

\paragraph*{Step 3. The classical eigenbasis gauge.}
Next, suppose $\mathit{E}_0$ is any matrix of common eigenvectors of the leading-order matrices $\mathit{L}_0$ and $\mathit{A}_0$.
Let us use it to gauge transform the system \eqref{Dold:241112153348}; i.e., we make the following change of the unknown variable:
\begin{equation}
\label{Dold:250415205035}
	\hat{\mathit{E}} \to \hat{\Phi}
\qquad\text{given by}\qquad
	\hat{\mathit{E}} \coloneqq \mathit{E}_0 \hat{\Phi}
.
\end{equation}
In this gauge, the differential system becomes diagonal to leading-order in $\hbar$ and it can be written like so:
\begin{equation}
\label{Dold:250417165300}
\begin{cases}
	\hbar \partial_x \hat{\Phi} - [\mathit{Y}, \hat{\Phi}] 
		= \Big( \mathit{E}_0^{-1} \hat{\mathit{L}} \mathit{E}_0 - \hbar \mathit{E}_0^{-1} \partial_x \mathit{E}_0 - \mathit{Y} \Big) \hat{\Phi}
,
\\
	\hbar \partial_t \hat{\Phi} - [\mathit{W}, \hat{\Phi}]
		= \Big( \mathit{E}_0^{-1} \hat{\mathit{A}} \mathit{E}_0 - \hbar \mathit{E}_0^{-1} \partial_t \mathit{E}_0 - \mathit{W} \Big) \hat{\Phi}
.
\end{cases}
\end{equation}
Let us summarise the transformations we have made so far in the following Lemma.

\begin{lemma}
\label{Dold:250109115117}
Let $\mathit{E}_0$ be an invertible matrix of common eigenvectors of $\mathit{L}_0, \mathit{A}_0$, and let $\mathit{S}$ be any solution of the differential equation \eqref{Dold:241113155001}, both defined and holomorphic on some domain $W \subset \mathbb{C}^2_{xt}$.
Then any invertible formal matrix $\hat\Phi\in\operatorname{GL}_2(\mathcal O(W)[[\hbar]])$ solving \eqref{Dold:250417165300} determines a formal WKB fundamental matrix of \eqref{250408174515}, in the finite-exponential algebra of \autoref{250410111125}, by setting
\begin{equation}
\label{Dold:250418093048}
	\hat{\Psi} (x,t,\hbar) = \mathit{E}_0 (x,t) \hat{\Phi} (x,t,\hbar) e^{\frac{\mathit{S} (x,t)}{\hbar}}
.
\end{equation}
\end{lemma}

\begin{proof}
Substitute \eqref{Dold:250418093048} into both equations of the formal
Painlev\'e~I system.  Using
$\partial_x\mathit S=\mathit Y$, $\partial_t\mathit S=\mathit W$,
the two resulting equations are exactly \eqref{Dold:250417165300}.
Since $\mathit E_0$, $\hat\Phi$, and $e^{\frac{\mathit S}{\hbar}}$ are
invertible, the resulting matrix $\hat\Psi$ is fundamental.
\end{proof}

\paragraph*{Step 4. Simplifying notation.}
Our focus from now on is solving the differential system \eqref{Dold:250417165300} for $\hat{\Phi}$ which we do order-by-order in $\hbar$.
To this end, let us simplify the notation in \eqref{Dold:250417165300} to enhance readability.
However, instead of introducing new notation for the matrices in brackets, let us re-define the symbols ``$\hat{\mathit{L}}$'' and ``$\hat{\mathit{A}}$'' only for the remainder of this construction and its compatibility proofs. The matrices in the main text retain their original meaning.
Thus, we put
\begin{equation}
\label{Dold:250119122930}
\begin{aligned}
	\hat{\mathit{L}} = \hat{\mathit{L}}^{\scriptscriptstyle \text{new}} &\coloneqq \mathit{E}_0^{-1} \hat{\mathit{L}}^{\scriptscriptstyle \text{old}} \mathit{E}_0 - \hbar \mathit{E}_0^{-1} \partial_x \mathit{E}_0 - \mathit{Y}
,
\\
	\hat{\mathit{A}} = \hat{\mathit{A}}^{\scriptscriptstyle \text{new}} &\coloneqq \mathit{E}_0^{-1} \hat{\mathit{A}}^{\scriptscriptstyle \text{old}} \mathit{E}_0 - \hbar \mathit{E}_0^{-1} \partial_t \mathit{E}_0 - \mathit{W}
.
\end{aligned}
\end{equation}
Explicitly, the new $\hat{\mathit{L}}$ and $\hat{\mathit{A}}$ are power series in $\hbar$ with matrix-valued function coefficients related to those of $\hat{\mathit{L}}^{\scriptscriptstyle \text{old}}$ and $\hat{\mathit{A}}^{\scriptscriptstyle \text{old}}$ by
\begin{equation}
\label{Dold:250203162604}
\begin{gathered}
	\mathit{L}_0 = 0
,
\qquad	\mathit{L}_1 = \mathit{E}_0^{-1} \mathit{L}^{\scriptscriptstyle \text{old}}_1 \mathit{E}_0 - \mathit{E}_0^{-1} \partial_x \mathit{E}_0
,
\qquad	\mathit{L}_{k} = \mathit{E}_0^{-1} \mathit{L}^{\scriptscriptstyle \text{old}}_{k} \mathit{E}_0
;
\\
	\mathit{A}_0 = 0
,
\qquad	\mathit{A}_1 = \mathit{E}_0^{-1} \mathit{A}^{\scriptscriptstyle \text{old}}_1 \mathit{E}_0 - \mathit{E}_0^{-1} \partial_t \mathit{E}_0
,
\qquad	\mathit{A}_{k} = \mathit{E}_0^{-1} \mathit{A}^{\scriptscriptstyle \text{old}}_{k} \mathit{E}_0
,
\end{gathered}
\qquad (k \geq 2)
.
\end{equation}
In this notation, system \eqref{Dold:250417165300} reads
\begin{equation}
\label{Dold:250417173233}
\begin{cases}
	\hbar \partial_x \hat{\Phi} - [\mathit{Y}, \hat{\Phi}] 
		= \hat{\mathit{L}} \hat{\Phi}
,
\\
	\hbar \partial_t \hat{\Phi} - [\mathit{W}, \hat{\Phi}]
		= \hat{\mathit{A}} \hat{\Phi}
.
\end{cases}
\end{equation}
Similarly, the formal zero-curvature equation \eqref{250408171411} leads --- after some manipulation, conjugation by $\mathit{E}_0$, and repeated use of identities \eqref{Dold:241112152831} and \eqref{Dold:241125133953} --- to the identity
\begin{equation}
\label{Dold:250109114843}
	\Big( \hbar \partial_t \hat{\mathit{L}} - [\mathit{W}, \hat{\mathit{L}}] \Big)
	- \Big( \hbar \partial_x \hat{\mathit{A}} - [\mathit{Y}, \hat{\mathit{A}}] \Big)
	+ [\hat{\mathit{L}}, \hat{\mathit{A}}] = 0
.
\end{equation}

\paragraph*{Step 5. Solving for $\Phi_0$.}
Examining the system \eqref{Dold:250417173233} at the leading-order in $\hbar$, we find that $\Phi_0$ must satisfy the following pair of linear algebraic equations:
\begin{equation}
\label{Dold:250418101326}
\begin{cases}
	[\mathit{Y}, \Phi_0] = 0
,
\\
	[\mathit{W}, \Phi_0] = 0
.
\end{cases}
\end{equation}
Since $\mathit{W}=\operatorname{diag}(w,-w)$ has distinct eigenvalues on $\mathring{\mathsf{\Sigma}}$, the equation $[\mathit{W},\Phi_0]=0$ implies that $\Phi_0$ is diagonal.
But if $\Phi_0$ is diagonal, then $\mathit{E}_0 \Phi_0$ is also a matrix of common eigenvectors of $\mathit{L}^{\scriptscriptstyle \text{old}}_0, \mathit{A}^{\scriptscriptstyle \text{old}}_0$.
Thus, let us assume that $\Phi_0$ is the identity matrix $\mathit{I}$, or in other words let us search for a solution $\hat{\Phi}$ of \eqref{Dold:250417173233} in the following form:
\begin{equation}
\label{Dold:250418102713}
	\hat{\Phi} (x,t,\hbar) = \mathit{I} + \sum_{k=1}^\infty \Phi_k (x,t) \hbar^k
.
\end{equation}
Then this assumption puts a constraint on the eigenvector matrix $\mathit{E}_0$ which we investigate in the next step.

\paragraph*{Step 6. Solving for the eigenvector matrix.}
Examining the new system \eqref{Dold:250417173233} at the first order in $\hbar$ under the assumption that $\Phi_0 = \mathit{I}$, we find the following pair of linear algebraic equations for $\Phi_1$:
\begin{equation}
\label{Dold:241112183751}
\begin{cases}
	[\mathit{Y}, \Phi_1] = - \mathit{L}_1
,
\\
	[\mathit{W}, \Phi_1] = - \mathit{A}_1
.
\end{cases}
\end{equation}
Since $\mathit{Y}$ and $\mathit{W}$ are diagonal matrices, the commutators in these equations are strictly off-diagonal.
This has a number of important consequences.
Firstly, if a solution $\Phi_1$ of \eqref{Dold:241112183751} exists, then only its off-diagonal part $\operatorname{off} \Phi_1$ is determined uniquely.
Secondly, for \eqref{Dold:241112183751} to be well-posed, $\mathit{L}_1$ and $\mathit{A}_1$ must likewise be strictly off-diagonal; i.e.,
\begin{equation}
\label{Dold:250122133055}
	\operatorname{diag} \mathit{L}_1 = \operatorname{diag} \mathit{A}_1 = 0 
.
\end{equation}
Here, we are using the following notation for any $2 \!\times\! 2$ matrix $\mathit{M} = [ \mathit{M}_{ij} ]$:
\begin{equation}
	\operatorname{diag} \mathit{M} \coloneqq \begin{pmatrix} \mathit{M}_{11} & 0 \\ 0 & \mathit{M}_{22} \end{pmatrix}
\qquad\text{and}\qquad
	\operatorname{off} \mathit{M} \coloneqq \begin{pmatrix} 0 & \mathit{M}_{12} \\ \mathit{M}_{21} & 0 \end{pmatrix}
\end{equation}

Equation \eqref{Dold:250122133055}, in turn, imposes an additional constraint on the eigenvector matrix $\mathit{E}_0$.
Indeed, the first-order corrections $\mathit{L}^{\scriptscriptstyle \text{old}}_1$ and $\mathit{A}^{\scriptscriptstyle \text{old}}_1$ can be read off the expressions \eqref{241115131339}:
\begin{equation}
	\mathit{L}^{\scriptscriptstyle \text{old}}_1 (x,t) = - \frac{1}{12 \tau} \begin{pmatrix} 1 & 0 \\ 0 & -1 \end{pmatrix}
\qquad\text{and}\qquad
	\mathit{A}^{\scriptscriptstyle \text{old}}_1 (x,t) = \begin{pmatrix} 0 & 0 \\ 0 & 0 \end{pmatrix}
.
\end{equation}
Then, using the identities \eqref{Dold:241112160711} and \eqref{Dold:241112201954}, we find:
\begin{equation}
	\mathit{E}_0^{-1} \mathit{L}^{\scriptscriptstyle \text{old}}_1 \mathit{E}_0
		= - \frac{1}{6 \tau \det (\mathit{E}_0)}
		\begin{pmatrix}
			0 & e_- e'_-
		\\	- e_+ e'_+ & 0
		\end{pmatrix}
		= - \frac{1}{12 \tau}
		\begin{pmatrix}
			0 & \frac{e_-}{e_+}
		\\	\frac{e_+}{e_-} & 0
		\end{pmatrix}
.
\end{equation}
Thus, the diagonal parts of $\mathit{E}_0^{-1} \mathit{L}^{\scriptscriptstyle \text{old}}_1 \mathit{E}_0$ and $\mathit{E}_0^{-1} \mathit{A}^{\scriptscriptstyle \text{old}}_1 \mathit{E}_0$ are zero.
So, using \eqref{Dold:250203162604}, the diagonal parts of the first-order corrections $\mathit{L}_1$ and $\mathit{A}_1$ are
\begin{equation}
    \operatorname{diag}\mathit{L}_1
        =-\operatorname{diag}(\mathit{E}_0^{-1}\partial_x\mathit{E}_0)
\quad\text{and}\quad
    \operatorname{diag}\mathit{A}_1
        =-\operatorname{diag}(\mathit{E}_0^{-1}\partial_t\mathit{E}_0)
.
\end{equation}
Consequently, equation \eqref{Dold:250122133055} leads to the following differential equation constraints on the eigenvector matrix $\mathit{E}_0$:
\begin{equation}
	e'_- \,\mathrm{d}{e_+} = e_- \,\mathrm{d}{e'_+}
\qquad\text{and}\qquad
	e_+ \,\mathrm{d}{e'_-} = e'_+ \,\mathrm{d}{e_-}
.
\end{equation}
Using identities \eqref{Dold:241112160711}, we can further reduce these constraints to a pair of differential equations for each $e_\pm$ that read as follows:
\begin{equation}
\label{Dold:241112205842}
	\operatorname{dlog} e_\pm = - \tfrac{1}{2} \operatorname{dlog} w
.
\end{equation}
General solutions of these equations are constant multiples of $w^{-\frac{1}{2}}$, and any such solution completely determines the eigenvector matrix $\mathit{E}_0$ via the relations \eqref{Dold:241112160711}.
In particular, from now on we let $\mathit{E}_0$ be the unique eigenvector matrix satisfying the initial condition
\begin{equation}
\label{Dold:250417183827}
	e_\pm (x_0, t_0) = b_{\pm,0}
.
\end{equation}

\paragraph*{Step 7. Existence and uniqueness.}
Finally, we claim that the differential system for $\hat{\Phi}$ with the above choice of $\mathit{E}_0$ is compatible and therefore can be solved uniquely for a given initial condition.

\begin{lemma}
\label{Dold:250417183730}
\label{lem:D-matrix-recursion}
Choose any pair of formal power series 
\begin{equation}
\label{Dold:250417190714}
	\hat{d}_\pm (\hbar) = \sum_{k=0}^\infty d_k^\pm \hbar^k \in \mathbb{C}[[\hbar]]
\quad\text{with}\quad
	d_0^\pm = 1
.
\end{equation}
On the fixed polydisc $W$ chosen above, the system \eqref{Dold:250417173233} has a unique fundamental solution $\hat\Phi\in\operatorname{GL}_2(\mathcal O(W)[[\hbar]])$ with leading-order part $\Phi_0 = \mathit{I}$ that satisfies the initial condition
\begin{equation}
\label{Dold:250417191149}
	\operatorname{diag} \hat{\Phi} (x_0, t_0, \hbar) = \operatorname{diag} \big( \, \hat{d}_+ (\hbar), \hat{d}_- (\hbar) \big)
.
\end{equation}
Moreover, $\hat{\Phi}$ is constructed recursively order-by-order in $\hbar$ in the sense that, for every $k \geq 1$, the off-diagonal part $\operatorname{off} \Phi_k$ is completely specified by $\Phi_1, \ldots, \Phi_{k-1}$.
\end{lemma}

The proof is done in \autoref{Dold:ProoffLemmaB2}. Assuming this lemma for the moment, we can finish the proof of \autoref{250716102028}. The existence part follows from \autoref{Dold:250109115117}. To prove uniqueness, it remains to explain how to choose the initial data \eqref{Dold:250417190714} from the prescribed normalisations $\hat c_\pm$.
Let $\mathit{C} \coloneqq \mathit{E}_0 (x_0, t_0)$ and $d_k^{ij} \coloneqq [\Phi_k (x_0, t_0)]_{ij}$ for all $k$.
We show how to choose constants $d_k^{11}=d_k^+$ and $d_k^{22}=d_k^-$ for every $k$.
Then the initial condition of \autoref{241113190220} becomes
\begin{align*}
    \hat{b}_+(\hbar)e^{-\frac{\kappa_+}{\hbar}}
        = \hat{\psi}_+(x_0,t_0,\hbar)
        &= \big[\hat{\Psi}(x_0,t_0,\hbar)\big]_{11}
        = \big[\mathit{C}\hat{\Phi}(x_0,t_0,\hbar)\big]_{11}
          e^{-\frac{\kappa_+}{\hbar}}
,
\\
    \hat{b}_-(\hbar)e^{-\frac{\kappa_-}{\hbar}}
        = \hat{\psi}_-(x_0,t_0,\hbar)
        &= \big[\hat{\Psi}(x_0,t_0,\hbar)\big]_{12}
        = \big[\mathit{C}\hat{\Phi}(x_0,t_0,\hbar)\big]_{12}
          e^{-\frac{\kappa_-}{\hbar}}
.
\end{align*}
which, after cancelling the exponentials and comparing the power series coefficients, reduces to
\begin{equation}
\begin{gathered}
    b_{+,k}
        = \mathit{C}_{11}d_k^{11}+\mathit{C}_{12}d_k^{21},
\\
    b_{-,k}
        = \mathit{C}_{11}d_k^{12}+\mathit{C}_{12}d_k^{22}.
\end{gathered}
.
\end{equation}
Note that for $k = 0$, we have that $d_0^{11} = d_0^{22} = 1$ and $d_0^{12} = d_0^{21} = 0$.
By \autoref{Dold:250417183730}, for every $k \geq 1$, the constants $d_k^{21}$ and $d_k^{12}$ are completely determined by $d_1^{ij}, \ldots, d_{k-1}^{ij}$, which means $d_k^{11}$ and $d_k^{22}$ can be uniquely solved for in terms of $b_{\pm,k}$.
This proves the result when $m_\pm=0$. For the general prescribed normalisations
\[
 \hat c_\pm(\hbar)=\hbar^{m_\pm}\hat{b}_\pm(\hbar)e^{-\frac{\kappa_\pm}{\hbar}},
\]
multiply the two constructed columns by $\hbar^{m_+}$ and $\hbar^{m_-}$, respectively. Since these factors are formal constants, the resulting columns still satisfy the formal Painlevé~I system, and their first-row components have precisely the prescribed normalisations $\hat c_\pm(\hbar)$. Uniqueness follows
from the uniqueness in the case $m_\pm=0$.
This proves \autoref{250716102028}. The column statement \autoref{241113190220} follows by taking either column.

\subsection{Proof of  \autoref{Dold:250417183730}}\label{Dold:ProoffLemmaB2}
The proof proceeds by constructing $\hat{\Phi}$ order-by-order in $\hbar$ using induction on the order $k$.
The base case $k = 0$, for which $\Phi_0 = \mathit{I}$, has already been dealt with.
The remainder of the proof is organised into five steps.
Steps 1-3 illustrate the construction of $\Phi_1$, whilst Steps 4 and 5 describe the general inductive step for $\Phi_k$ for general $k \geq 1$.
Only Steps 4 and 5 are logically necessary.
However, we have included Steps 1-3 in order to clarify the general argument because Steps 4-5 rely on two Lemmas (\autoref{Dold:250109110444} and \autoref{Dold:250123150709}) whose lengthy proofs are relegated to \autoref{Dold:app:compatibility-identities}.

The induction is simultaneous. Suppose that
$\Phi_0,\ldots,\Phi_{k-1}$ have been constructed, that both
equations hold through order $\hbar^{k-1}$, and that their diagonal
parts at order $\hbar^k$ already hold. The last condition is imposed
by the preceding diagonal integration; for $k=1$ it is exactly
$\operatorname{diag}\mathit L_1=\operatorname{diag}\mathit A_1=0$.
First, \autoref{Dold:250109110444} determines
$\operatorname{off}\Phi_k$ consistently in both equations.
Both full equations then hold through order $\hbar^k$, independently
of the as-yet-undetermined diagonal part of $\Phi_k$.
Next, \autoref{Dold:250123150709} shows that the one-form prescribing
$\operatorname{diag}\Phi_k$ is closed. Its primitive with the
prescribed value at $(x_0,t_0)$ determines this diagonal part and
imposes the diagonal equations at order $\hbar^{k+1}$.
This is the invariant used in every subsequent step.

\paragraph*{Step 1. Solving for the off-diagonal part of $\Phi_1$.}
The system \eqref{Dold:250417173233} at first order in $\hbar$ reduces to the following pair of linear algebraic equations for $\Phi_1$:
\begin{equation}
\label{Dold:241112183751x}
\begin{cases}
	[\mathit{Y}, \Phi_1] = - \mathit{L}_1
,
\\
	[\mathit{W}, \Phi_1] = - \mathit{A}_1
.
\end{cases}
\end{equation}
Since $w\neq0$ on $\mathring{\mathsf{\Sigma}}$, the restriction of
$\operatorname{ad}_{\mathit W}$ to the off-diagonal matrices is
invertible.  Hence the second equation in \eqref{Dold:241112183751x}
uniquely determines the off-diagonal part of $\Phi_1$:
\begin{equation}
\label{Dold:250109171250}
    \operatorname{off}\Phi_1=-\tfrac12\mathit{W}^{-1}\mathit{A}_1
.
\end{equation}
The formal zero-curvature equation \eqref{Dold:250109114843} at first order in $\hbar$ gives
\begin{equation}
\label{Dold:250109165619}
    [\mathit{W},\mathit{L}_1]=[\mathit{Y},\mathit{A}_1]
.
\end{equation}
Since $\mathit{L}_1$ and $\mathit{A}_1$ are strictly off-diagonal and $\mathit{Y}=2(x-\tau)\mathit{W}=\frac{y}{w}\mathit{W}$, this identity implies that $\mathit{L}_1$ and $\mathit{A}_1$ are commuting matrices and 
\begin{equation}
\label{Dold:250109171618}
    \mathit{L}_1=2(x-\tau)\mathit{A}_1=\frac{y}{w}\mathit{A}_1
.
\end{equation}
Consequently, away from the divisor $x=\tau(t)$, one may equivalently
write $\operatorname{off}\Phi_1=-\tfrac12\mathit{Y}^{-1}\mathit{L}_1$. The apparent singularity of this latter expression along $x=\tau(t)$ is removable, and its holomorphic extension is given by
\eqref{Dold:250109171250}.

\paragraph*{Step 2. Solving for the diagonal part of $\Phi_1$.}
To determine the diagonal components of $\Phi_1$, we examine equations \eqref{Dold:250417173233} at the next order in $\hbar$, where we find the following equations:
\begin{equation}
\label{Dold:250109182351}
\begin{cases}
	[\mathit{Y},\Phi_2] = \partial_x \Phi_1 - \mathit{L}_1 \Phi_1 - \mathit{L}_2
,
\\
	[\mathit{W},\Phi_2] = \partial_t \Phi_1 - \mathit{A}_1 \Phi_1 - \mathit{A}_2
.
\end{cases}
\end{equation}
In these equations, $\Phi_2$ is unknown whilst $\Phi_1$ is only partially known.
The lefthand side is again strictly off-diagonal so the diagonal part of the righthand side must be zero, which yields two uncoupled first-order linear differential equations, one for each diagonal component of $\Phi_1$.
Namely,
\begin{equation}
\label{Dold:250109183536}
\begin{cases}
	\operatorname{diag} \partial_x \Phi_1
		= \operatorname{diag} \big( \mathit{L}_1 \Phi_1 + \mathit{L}_2 \big)
		= \mathit{L}_1 \operatorname{off} (\Phi_1) + \operatorname{diag} (\mathit{L}_2)
,
\\
	\operatorname{diag} \partial_t \Phi_1
		= \operatorname{diag} \big( \mathit{A}_1 \Phi_1 + \mathit{A}_2 \big)
		= \mathit{A}_1 \operatorname{off} (\Phi_1) + \operatorname{diag} (\mathit{A}_2)
.
\end{cases}
\end{equation}
Notice that the righthand side only depends on the off-diagonal components of $\Phi_1$ because $\mathit{L}_1$ and $\mathit{A}_1$ are strictly off-diagonal.
Assuming that these equations are compatible --- which we will verify in the next step and more generally in \autoref{Dold:250123150709} below --- then they can be solved by direct integration.
Namely, for any constant diagonal matrix $\mathit{B}_1 \coloneqq \operatorname{diag} (d^{+}_1, d^{-}_1)$, let $\mathit{U}_1 \coloneqq \operatorname{diag} (u^{+}_1, u^{-}_1)$ and $\mathit{V}_1 \coloneqq \operatorname{diag} (v^{+}_1, v^{-}_1)$ be the unique solutions of the following initial value problems:
\begin{equation}
\label{Dold:250109185404}
\begin{aligned}
	\partial_t \mathit{U}_1 (t) = \operatorname{diag} \big( \mathit{A}_1 \Phi_1 + \mathit{A}_2 \big) (x_0, t)
&\quad\text{such that}\quad
	\mathit{U}_1 (t_0) = \mathit{B}_1
,
\\
	\partial_x \mathit{V}_1 (x) = \operatorname{diag} \big( \mathit{L}_1 \Phi_1 + \mathit{L}_2 \big) (x, t_0)
&\quad\text{such that}\quad
	\mathit{V}_1 (x_0) = \mathit{B}_1
.
\end{aligned}
\end{equation}
Then the holomorphic function
\begin{equation}
\label{Dold:250109184046}
\begin{aligned}
	\operatorname{diag} \Phi_1 (x,t) 
		&= \int_{x_0}^x \operatorname{diag} \big( \mathit{L}_1 \Phi_1 + \mathit{L}_2 \big)(x',t) \,\mathrm{d}{x'}
		+ \mathit{U}_1 (t)
\\
		&= \int_{t_0}^t \operatorname{diag} \big( \mathit{A}_1 \Phi_1 + \mathit{A}_2 \big)(x,t') \,\mathrm{d}{t'}
		+ \mathit{V}_1 (x)
\end{aligned}
\end{equation}
is the unique solution of the system of equations \eqref{Dold:250109183536} for all $(x,t)$ near $(x_0,t_0)$ satisfying 
\begin{equation}
\label{Dold:250113113442}
	\operatorname{diag} \Phi_1 (x_0,t_0) = \mathit{B}_1
.
\end{equation}

\paragraph*{Step 3. Compatibility of the system for $\Phi_1$.}
To demonstrate that the solution \eqref{Dold:250109184046} is well-defined, we must check that equations \eqref{Dold:250109183536} form a compatible system, meaning the following identity is true:
\begin{equation}
\label{Dold:250119135222}
	\operatorname{diag} \partial_t \big( \mathit{L}_1 \Phi_1 + \mathit{L}_2 \big)
	= \operatorname{diag} \partial_x \big( \mathit{A}_1 \Phi_1 + \mathit{A}_2 \big)
.
\end{equation}
Thus, we need to show that the following quantity is strictly off-diagonal:
\begin{equation}
\label{Dold:250119135712}
	\big( \partial_t \mathit{L}_1 - \partial_x \mathit{A}_1 \big) \Phi_1
		+ \mathit{L}_1 \partial_t \Phi_1 - \mathit{A}_1 \partial_x \Phi_1
		+ ( \partial_t \mathit{L}_2 - \partial_x \mathit{A}_2 )
.
\end{equation}
This can be seen as follows.
Firstly, the formal zero-curvature equation \eqref{Dold:250109114843} at orders $\hbar^2$ and $\hbar^3$ yields the following two identities:
\begin{equation}
\label{Dold:250119141945}
\begin{aligned}	
	\partial_t \mathit{L}_1 - \partial_x \mathit{A}_1
		&= [\mathit{W}, \mathit{L}_2] - [\mathit{Y}, \mathit{A}_2]
,
\\
	\partial_t \mathit{L}_2 - \partial_x \mathit{A}_2
		&= \textcolor{gray}{[\mathit{W}, \mathit{L}_3] - [\mathit{Y}, \mathit{A}_3]} + [\mathit{A}_1, \mathit{L}_2] + [\mathit{A}_2, \mathit{L}_1]
.
\end{aligned}
\end{equation}
Note that $[\mathit{A}_1, \mathit{L}_1] = 0$ because $\mathit{A}_1$ and $\mathit{L}_1$ are scalar multiples of each other (cf. \eqref{Dold:250109171618}).
Notice also that the terms $[\mathit{W}, \mathit{L}_3]$ and $[\mathit{Y}, \mathit{A}_3]$ are strictly off-diagonal, so their contribution to \eqref{Dold:250119135712} is irrelevant and for this reason they have been greyed out.
Secondly, we claim that
\begin{equation}
\label{Dold:250123123741}
	\mathit{L}_1 \partial_t \Phi_1 - \mathit{A}_1 \partial_x \Phi_1
		= - \Phi_1 \big( \partial_t \mathit{L}_1 - \partial_x \mathit{A}_1 \big)
			\textcolor{gray}{- [\mathit{Y}, \Phi_1 \partial_t \Phi_1] + [\mathit{W}, \Phi_1 \partial_x \Phi_1]}
.
\end{equation}
Since the greyed-out commutators $[\mathit{Y}, \Phi_1 \partial_t \Phi_1]$ and $[\mathit{W}, \Phi_1 \partial_x \Phi_1]$ are strictly off-diagonal, they
vanish after applying $\operatorname{diag}$ and therefore do not contribute to the compatibility condition \eqref{Dold:250119135222}.
Using the identities $[\mathit{Y}, \Phi_1] = -\mathit{L}_1$ and $[\mathit{W}, \Phi_1] = -\mathit{A}_1$ derived in Step 1, as well as their derivatives $[\partial_t \mathit{Y}, \Phi_1] + [\mathit{Y}, \partial_t \Phi_1] = - \partial_t \mathit{L}_1$ and $[\partial_x \mathit{W}, \Phi_1] + [\mathit{W}, \partial_x \Phi_1] = - \partial_x \mathit{A}_1$, 
we find:
\begin{equation}
\label{Dold:250123122257}
\begin{aligned}
	\mathit{L}_1 \partial_t \Phi_1
		= - [\mathit{Y}, \Phi_1] \partial_t \Phi_1 ~
		&= \textcolor{gray}{- [\mathit{Y}, \Phi_1 \partial_t \Phi_1]} + \Phi_1 [\mathit{Y}, \partial_t \Phi_1]
\\
		&
		= \textcolor{gray}{- [\mathit{Y}, \Phi_1 \partial_t \Phi_1]} - \Phi_1 \partial_t \mathit{L}_1 - \Phi_1 [\partial_t \mathit{Y}, \Phi_1]
,
\\
	- \mathit{A}_1 \partial_x \Phi_1
		= + [\mathit{W}, \Phi_1] \partial_x \Phi_1
		&= \textcolor{gray}{[\mathit{W}, \Phi_1 \partial_x \Phi_1]} - \Phi_1 [\mathit{W}, \partial_x \Phi_1]
\\
		&
		= \textcolor{gray}{[\mathit{W}, \Phi_1 \partial_x \Phi_1]} + \Phi_1 \partial_x \mathit{A}_1 + \Phi_1 [\partial_x \mathit{W}, \Phi_1]
.
\end{aligned}
\end{equation}
Identity \eqref{Dold:250123123741} now follows from this because $\partial_t \mathit{Y} = \partial_x \mathit{W}$.
Putting \eqref{Dold:250119141945} and \eqref{Dold:250123123741} together, the expression \eqref{Dold:250119135712} becomes the following after omitting all strictly off-diagonal terms:
\begin{multline*}
	\Big[ \big(\partial_t \mathit{L}_1 - \partial_x \mathit{A}_1 \big), \Phi_1 \Big]
		+ [\mathit{A}_1, \mathit{L}_2] + [\mathit{A}_2, \mathit{L}_1]
= \Big[ [\mathit{W}, \mathit{L}_2], \Phi_1 \Big]
		- \Big[ [\mathit{Y}, \mathit{A}_2], \Phi_1 \Big]
		+ [\mathit{A}_1, \mathit{L}_2] + [\mathit{A}_2, \mathit{L}_1]
.
\end{multline*}
On the other hand, using the Jacobi identity as well as the identities $[\mathit{Y}, \Phi_1] = -\mathit{L}_1$ and $[\mathit{W}, \Phi_1] = -\mathit{A}_1$ again, we find:
\begin{align*}
	\Big[ [\mathit{W}, \mathit{L}_2], \Phi_1 \Big]
		&= \textcolor{gray}{+ \Big[ \mathit{W}, [\mathit{L}_2, \Phi_1] \Big]} + \Big[ \mathit{L}_2, [\Phi_1, \mathit{W}] \Big]
		= \textcolor{gray}{+ \Big[ \mathit{W}, [\mathit{L}_2, \Phi_1] \Big]} + [ \mathit{L}_2, \mathit{A}_1 ]
,
\\
	-\Big[ [\mathit{Y}, \mathit{A}_2], \Phi_1 \Big]
		&= \textcolor{gray}{- \Big[ \mathit{Y}, [\mathit{A}_2, \Phi_1] \Big]} - \Big[ \mathit{A}_2, [\Phi_1, \mathit{Y}] \Big]
		= \textcolor{gray}{- \Big[ \mathit{Y}, [\mathit{A}_2, \Phi_1] \Big]} - [ \mathit{A}_2, \mathit{L}_1 ]
.
\end{align*}
Consequently, the diagonal part of the expression \eqref{Dold:250119135712} is zero.

Having determined $\Phi_1$ completely, we return to \eqref{Dold:250109182351} to solve for the off-diagonal part of $\Phi_2$.
Then, by examining the next order in $\hbar$ of the system \eqref{Dold:250417173233}, the diagonal part of $\Phi_2$ is found by direct integration.
And so on, this process can be repeated at every order in $\hbar$ as described next.

\paragraph*{Step 4. Solving for the off-diagonal part of $\Phi_k$.}
Suppose now that $k \geq 1$ is arbitrary, and that $\Phi_1, \ldots, \Phi_{k-1}$ have been constructed.
Thus, $\Phi_i$ for each $i = 1, \ldots, k-1$ satisfies the system of equations \eqref{Dold:250417173233} at order $\hbar^i$:
\begin{equation}
\label{Dold:250123171817}
\begin{cases}
	\displaystyle
	[\mathit{Y}, \Phi_i] = \partial_x \Phi_{i-1} - \sum_{m+n = i} \mathit{L}_m \Phi_n
,
\\
	\displaystyle
	[\mathit{W}, \Phi_i] = \partial_t \Phi_{i-1} - \sum_{\smash{m+n = i}} \mathit{A}_m \Phi_n
.
\end{cases}
\end{equation}
The diagonal equations at order $\hbar^k$ also hold, by the induction invariant stated above. Thus both right-hand sides in \eqref{Dold:250109192430} have zero diagonal. We explain how to construct $\Phi_k$.
First, we solve for the off-diagonal part of $\Phi_k$ using our system of equations \eqref{Dold:250417173233} at order $k$ in $\hbar$.
It reduces to the following pair of linear algebraic equations for $\Phi_k$:
\begin{equation}
\label{Dold:250109192430}
\begin{cases}
	\displaystyle
	[\mathit{Y}, \Phi_k] = \partial_x \Phi_{k-1} - \sum_{i+j = k} \mathit{L}_i \Phi_j \eqqcolon \mathrm{RHS}_{\mathit{L},k}
,
\\
	\displaystyle
	[\mathit{W}, \Phi_k] = \partial_t \Phi_{k-1} - \sum_{\smash{i+j = k}} \mathit{A}_i \Phi_j \eqqcolon \mathrm{RHS}_{\mathit{A},k}
.
\end{cases}	
\end{equation}
Here and in the rest of this subsection, the sums are understood to run over all $0 \leq i,j \leq k$ satisfying the constraint $i + j = k$.
Note also that the righthand side does not involve $\Phi_k$ because $\mathit{L}_0 = \mathit{A}_0 = 0$.
Since $w\neq0$ on $\mathring{\mathsf{\Sigma}}$, the restriction of
$\operatorname{ad}_{\mathit W}$ to the off-diagonal matrices is
invertible.  Hence the second equation in \eqref{Dold:250109192430}
uniquely determines the off-diagonal part of $\Phi_k$.  Namely, for
every off-diagonal matrix $\mathit{M}$ one has
$[\mathit{W},\mathit{M}]=2\mathit{W}\mathit{M}$, and therefore
\begin{equation}
\label{Dold:250109191643}
    \operatorname{off}\Phi_k
        =
        \tfrac12\mathit{W}^{-1}
        \operatorname{off}\bigl(\mathrm{RHS}_{\mathit{A},k}\bigr)
.
\end{equation}
That this matrix also satisfies the first equation in \eqref{Dold:250109192430} is a consequence of the formal zero-curvature equation \eqref{Dold:250109114843}, as expressed in the following lemma.

\begin{lemma}
\label{Dold:250109110444}
Assume $\Phi_0=I$ and that both equations of \eqref{Dold:250417173233} hold through order $\hbar^{k-1}$. For the right-hand sides defined in \eqref{Dold:250109192430}, the following identity holds:
\begin{equation}
\label{Dold:250109122646}
	\Big[ \mathit{W}, \mathrm{RHS}_{\mathit{L},k} \Big]
	=	\Big[ \mathit{Y}, \mathrm{RHS}_{\mathit{A},k} \Big]
.
\end{equation}
In particular, since $\mathit{W}, \mathit{Y}$ are diagonal matrices satisfying $\mathit{Y} = 2 (x-\tau) \mathit{W}$, the off-diagonal entries of $\mathrm{RHS}_{\mathit{L},k}$ and $\mathrm{RHS}_{\mathit{A},k}$ are related by the following identity:
\begin{equation}
	\Big[ \mathrm{RHS}_{\mathit{L},k} \Big]_{ij}
		= 2 (x-\tau) \Big[ \mathrm{RHS}_{\mathit{A},k} \Big]_{ij}
		= \frac{y}{w} \Big[ \mathrm{RHS}_{\mathit{A},k} \Big]_{ij}
\qquad
	\text{for $i \neq j$}
.
\end{equation}
\end{lemma}

Away from the divisor $x=\tau(t)$, the same solution can equivalently be written as $\operatorname{off}\Phi_k=\tfrac12\mathit{Y}^{-1}\operatorname{off}\bigl(\mathrm{RHS}_{\mathit{L},k}\bigr)$. 
The expression involving $\mathit{W}^{-1}$ is the one that extends holomorphically across the double turning divisor.

\paragraph*{Step 5. Solving for the diagonal part of $\Phi_k$.}
To solve for the diagonal part $\operatorname{diag} \Phi_k$ of $\Phi_k$, we examine the diagonal part of the system of equations \eqref{Dold:250417173233} at order $k + 1$ in $\hbar$.
It yields the following pair of first-order linear differential equations:
\begin{equation}
\label{Dold:250109200635}
	\partial_x \operatorname{diag} \Phi_k
		= \!\! \sum_{m+n = k + 1} \operatorname{diag} \big( \mathit{L}_m \Phi_n \big)
\quad\text{and}\quad
	\partial_t \operatorname{diag} \Phi_k
		= \!\! \sum_{m+n = k + 1} \operatorname{diag} \big( \mathit{A}_m \Phi_n \big)
~.
\end{equation}
Notice that the righthand side in both of these equations does not involve the as-yet-undetermined matrix $\Phi_{k+1}$ because $\mathit{L}_0 = \mathit{A}_0 = 0$.
The righthand side also does not involve the components of $\operatorname{diag} \Phi_k$ because $\mathit{L}_1$ and $\mathit{A}_1$ are strictly off-diagonal.
Thus, these equations can be solved by direct integration, provided the system \eqref{Dold:250109200635} is compatible.
The compatibility of this system is a nontrivial consequence of the formal zero-curvature equation \eqref{Dold:250109114843}, summarised in the following Lemma whose proof is presented in \autoref{Dold:app:compatibility-identities}.

\begin{lemma}
\label{Dold:250123150709}
Assume that both equations of \eqref{Dold:250417173233} hold through order $\hbar^k$, after constructing $\operatorname{off}\Phi_k$ and assigning any holomorphic diagonal part to $\Phi_k$. Then the following identity holds, independently of that diagonal part:
\begin{equation}
\label{Dold:250123150851}
	\partial_t \sum_{i+j = k+1} \operatorname{diag} \big( \mathit{L}_{i} \Phi_j \big)
		= \partial_x \sum_{i+j = k+1} \operatorname{diag} \big( \mathit{A}_i \Phi_j \big)
.
\end{equation}
\end{lemma}

Integrating this closed diagonal one-form on the fixed polydisc $W$ with the prescribed diagonal initial value gives $\operatorname{diag}\Phi_k$ and completes the induction. All coefficients are holomorphic on $W$: the algebraic step divides only by $w$, and the differential step takes a holomorphic primitive on $W$. The same steps force every coefficient of any other solution with the same data, so uniqueness also follows. This completes the proof of \autoref{Dold:250417183730}, subject to the compatibility proofs below.

\subsection{Compatibility identities in the induction}
\label{Dold:app:compatibility-identities}

\subsubsection{Proof of \autoref{Dold:250109110444}}

In preparation, let us expand the formal zero-curvature equation \eqref{Dold:250109114843} in powers of $\hbar$: we obtain the following sequence of equations for each $k \geq 1$:
\begin{equation}
\label{Dold:250109121431}
	\Big( \partial_t \mathit{L}_{k-1} - [\mathit{W}, \mathit{L}_k] \Big)
	- \Big( \partial_x \mathit{A}_{k-1} - [\mathit{Y}, \mathit{A}_k] \Big)
	+ \sum_{i+j=k} [\mathit{L}_i, \mathit{A}_j] = 0
\qquad (k \geq 1)
.
\end{equation}
Here and throughout this subsection, the sum is understood to run over all $0 \leq i,j \leq k$ satisfying the constraint $i + j = k$.
We prove \autoref{Dold:250109110444} using only the equations through order $k-1$. The case $k=1$ was verified in Step~1 of \autoref{Dold:ProoffLemmaB2}.
Namely, we find from \eqref{Dold:250109192430} with $k = 1$ that $\mathrm{RHS}_{\mathit{L},1} = - \mathit{L}_1$ and $\mathrm{RHS}_{\mathit{A},1} = - \mathit{A}_1$.
Meanwhile, the compatibility equation \eqref{Dold:250109121431} at order $k = 1$ reduces to $- [\mathit{W}, \mathit{L}_1] + [\mathit{Y}, \mathit{A}_1] = 0$, which is nothing but the desired identity \eqref{Dold:250109122646} for $k = 1$.

Now, suppose $k \geq 2$ is arbitrary, and assume that $\Phi_1, \ldots, \Phi_{k-1}$ have been found such that
\begin{equation}
\label{Dold:250109122501}
	[\mathit{Y}, \Phi_\ell] = \partial_x \Phi_{\ell-1} - \sum_{i+j = \ell} \mathit{L}_i \Phi_j
\quad\text{and}\quad
	[\mathit{W}, \Phi_\ell] = \partial_t \Phi_{\ell-1} - \sum_{i+j = \ell} \mathit{A}_i \Phi_j
\end{equation}
for all $\ell \in\{1,\ldots,k-1\}$.
We must demonstrate that identity \eqref{Dold:250109122646} is true.
We start by calculating the term $[\mathit{W}, \partial_x \Phi_{k-1}]$ appearing on the left-hand side of \eqref{Dold:250109122646}.
Differentiating both sides of the identity \eqref{Dold:250109122501} for $[\mathit{W}, \Phi_{k-1}]$ 
with respect to $x$, we find:
\begin{equation}
	[\partial_x \mathit{W}, \Phi_{k-1}] + [\mathit{W}, \partial_x \Phi_{k-1}]
		= \partial_x \partial_t \Phi_{k-2} - \sum_{i+j = k - 1} 
		\Big( \partial_x \mathit{A}_i \Phi_j + \mathit{A}_i \partial_x \Phi_j \Big)
.
\end{equation}
Using the identity $\partial_x \mathit{W} = \partial_t \mathit{Y}$ which follows from \eqref{Dold:241113155001}, we obtain:
\begin{equation}
\label{Dold:250109124840}
	[\mathit{W}, \partial_x \Phi_{k-1}]
	= - [\partial_t \mathit{Y}, \Phi_{k-1}] + \partial_x \partial_t \Phi_{k-2} - \sum_{i+j = k - 1} 
		\Big( \partial_x \mathit{A}_i \Phi_j + \mathit{A}_i \partial_x \Phi_j \Big)
.
\end{equation}

At the same time, the identity $\partial_t [\mathit{Y}, \Phi_{k-1}] = [\partial_t \mathit{Y}, \Phi_{k-1}] + [\mathit{Y}, \partial_t \Phi_{k-1}]$, in conjunction with the expression \eqref{Dold:250109122501} for $[\mathit{Y}, \Phi_{k-1}]$, yields:
\begin{align*}
	- [\partial_t \mathit{Y}, \Phi_{k-1}]
		&= [\mathit{Y}, \partial_t \Phi_{k-1}] - \partial_t [\mathit{Y}, \Phi_{k-1}]
\\		&= [\mathit{Y}, \partial_t \Phi_{k-1}] - \partial_t \left( \partial_x \Phi_{k-2} 
			- \sum_{i+j=k-1} \mathit{L}_i \Phi_j \right)
\\		&= [\mathit{Y}, \partial_t \Phi_{k-1}] - \partial_t \partial_x \Phi_{k-2} 
			+ \sum_{i+j=k-1} 
				\Big( \partial_t \mathit{L}_i \Phi_j + \mathit{L}_i \partial_t \Phi_j \Big)
.
\end{align*}
Plugging this into \eqref{Dold:250109124840}, we see that the term $\partial_x \partial_t \Phi_{k-2}$ gets cancelled, whilst the terms in the sums can be combined to yield the following identity:
\begin{equation}
\label{Dold:250109131538}
	[\mathit{W}, \partial_x \Phi_{k-1}]
		= [\mathit{Y}, \partial_t \Phi_{k-1}] 
			+ \sum_{i+j=k-1} 
				\Bigg( \Big( \partial_t \mathit{L}_i - \partial_x \mathit{A}_i \Big) \Phi_j 
					+ \Big(\mathit{L}_i \partial_t \Phi_j - \mathit{A}_i \partial_x \Phi_j \Big)
				\Bigg)
.
\end{equation}

Next, we again use \eqref{Dold:250109122501} with $\ell = j+1$ to substitute for $\partial_x \Phi_j$ and $\partial_t \Phi_j$ in \eqref{Dold:250109131538}.
The term with $i=0$ is zero because $\mathit L_0=\mathit A_0=0$, so only $j+1\leq k-1$ is used; no equation for $\Phi_k$ is assumed. After rearranging the terms, we find:
\begin{equation}
\label{Dold:250109134117}
\begin{aligned}
	[\mathit{W}, \partial_x \Phi_{k-1}]
		= [\mathit{Y}, \partial_t \Phi_{k-1}] 
			&+ \sum_{i+j=k-1} 
				\Big( \partial_t \mathit{L}_i - \partial_x \mathit{A}_i \Big) \Phi_j
\\			&+ \sum_{i+j=k-1} 
				\Big( \mathit{L}_i [\mathit{W}, \Phi_{j+1}] - \mathit{A}_i [\mathit{Y}, \Phi_{j+1}] \Big)
\\			&+ \sum_{i+j=k-1} \sum_{m+n = j+1}
				\Big( \mathit{L}_i \mathit{A}_m 
					- \mathit{A}_i \mathit{L}_m \Big) \Phi_n
.
\end{aligned}
\end{equation}
Set the new first index equal to the old first index plus $1$ in the first sum, and set the new second index equal to the old second index plus $1$ in the second and third sums. Extending the resulting sums by zero boundary terms, and adopting the convention $\mathit{L}_{-1}=\mathit{A}_{-1}=0$, we obtain
\begin{equation}
\label{Dold:250109134138}
\begin{aligned}
	[\mathit{W}, \partial_x \Phi_{k-1}]
		= [\mathit{Y}, \partial_t \Phi_{k-1}] 
			&+ \sum_{i+j=k}
				\Big( \partial_t \mathit{L}_{i-1} - \partial_x \mathit{A}_{i-1} \Big) \Phi_j 
\\			&+ \sum_{i+j=k}
				\Big( \mathit{L}_i [\mathit{W}, \Phi_{j}] - \mathit{A}_i [\mathit{Y}, \Phi_{j}] \Big)
\\			&+ \sum_{i+j=k} \sum_{m+n = j}
				\Big( \mathit{L}_i \mathit{A}_m 
					- \mathit{A}_i \mathit{L}_m \Big) \Phi_n
.
\end{aligned}
\end{equation}
Both double sums range over triples of nonnegative integers with total $k$. Relabel $(i,m,n)$ as $(m,n,j)$ and interchange the two inner indices in the negative term. The double sum in the last line therefore becomes:
\begin{equation}
\label{Dold:250128122206}
	\sum_{i+j=k} \sum_{m+n = j} 
			\Big( \mathit{L}_i \mathit{A}_m 
				- \mathit{A}_i \mathit{L}_m \Big) \Phi_n
	= \sum_{i+j=k} \sum_{m+n = i}
			[\mathit{L}_m, \mathit{A}_n] \Phi_j
.
\end{equation}
Also, notice that
\begin{equation}
	\mathit{L}_i [\mathit{W}, \Phi_j]
		= [\mathit{W}, \mathit{L}_i \Phi_j] - [\mathit{W}, \mathit{L}_i] \Phi_j
\quad\text{and}\quad
	\mathit{A}_i [\mathit{Y}, \Phi_j]
		= [\mathit{Y}, \mathit{A}_i \Phi_j] - [\mathit{Y}, \mathit{A}_i] \Phi_j
.
\end{equation}
Altogether, identity \eqref{Dold:250109134138} simplifies as follows:
\begin{multline*}
	[\mathit{W}, \partial_x \Phi_{k-1}]
		= [\mathit{Y}, \partial_t \Phi_{k-1}]
			+ \sum_{i+j=k}
				\Big( [\mathit{W}, \mathit{L}_i \Phi_j] - [\mathit{Y}, \mathit{A}_i \Phi_j] \Big)
\\			+ \sum_{i+j=k}
				\left( 
					\Big( \partial_t \mathit{L}_{i-1} - [\mathit{W}, \mathit{L}_i] \Big)
					- \Big( \partial_x \mathit{A}_{i-1} - [\mathit{Y}, \mathit{A}_i] \Big)
					+ \sum_{m+n = i} [\mathit{L}_m, \mathit{A}_n]
				\right) \Phi_j
.
\end{multline*}
The parenthesised expression in the last sum is equal to zero thanks to the compatibility equation \eqref{Dold:250109121431}.
Consequently, we find:
\begin{align*}
	[\mathit{W}, \mathrm{RHS}_{\mathit{L},k}]
		= [\mathit{W}, \partial_x \Phi_{k-1}] - \sum_{i+j=k} [\mathit{W}, \mathit{L}_i \Phi_j]
		= [\mathit{Y}, \partial_t \Phi_{k-1}] - \sum_{i+j=k} [\mathit{Y}, \mathit{A}_i \Phi_j]
		= [\mathit{Y}, \mathrm{RHS}_{\mathit{A},k}]
\end{align*}
which is what we needed to show.

\subsubsection{Proof of \autoref{Dold:250123150709}}
The case $k=1$ was proved in Step~3 of \autoref{Dold:ProoffLemmaB2}. We therefore assume below that $k\geq2$. Let us assume by induction that we have built $(\Phi_j)_{j\leq k-1}$. Then \autoref{Dold:250109110444} defines the off-diagonal part of $\Phi_k$. Give $\Phi_k$ any holomorphic diagonal part, for example zero. By the induction invariant, the full equations through order $\hbar^k$ hold, and hence we have the following equations.
\begin{equation}
\label{Dold:KnownRelations2}
\begin{aligned}
 0&=\partial_x\Phi_j-[\mathit Y,\Phi_{j+1}]
       -\sum_{i=0}^{j}\mathit L_{j+1-i}\Phi_i,\\
 0&=\partial_t\Phi_j-[\mathit W,\Phi_{j+1}]
       -\sum_{i=0}^{j}\mathit A_{j+1-i}\Phi_i,
       \qquad \forall\, j\in\{1,\ldots,k-1\}.
\end{aligned}
\end{equation}
The compatibility equations read
\begin{equation} \label{Dold:KnownRelations}\forall\, n\geq 1\,:\, 0=\partial_t \mathit{L}_{n}-\partial_x \mathit{A}_{n}-\left[\mathit{W},\mathit{L}_{n+1}\right]+ \left[\mathit{Y},\mathit{A}_{n+1}\right]+\sum_{i=1}^{n}\left[\mathit{L}_i,\mathit{A}_{n+1-i}\right]
\end{equation}
In the upcoming computations, we shall need the following relations coming from direct computations for the first orders:
\begin{align} \label{Dold:AdditionalRelationsGeneral} 
\partial_t \mathit{Y}&=\partial_x \mathit{W}, \qquad
0=\left[\mathit{Y},\mathit{A}_1\right]-\left[\mathit{W},\mathit{L}_1\right], \qquad
\mathit{L}_1=\left[\Phi_1,\mathit{Y}\right] \nonumber \\
\mathit{A}_1&=\left[\Phi_1,\mathit{W}\right], \qquad
\partial_t \mathit{L}_1-\partial_x\mathit{A}_1=\left[\partial_t \Phi_1,\mathit{Y}\right]-\left[\partial_x \Phi_1,\mathit{W}\right]
\end{align}
We shall also make multiple uses of the following algebraic relations:
\begin{align}\label{Dold:JacobiId} \left[\left[A,B\right],C\right]&=-\left[\left[C,A\right],B\right]-\left[\left[B,C\right],A\right] \nonumber \\
A \left[B,C\right]&= \left[A,B\right]C +\left[B,AC\right]=\left[C,A\right]B+\left[AB,C\right] \nonumber \\
\left[A,B\right]C&=A\left[B,C\right]+\left[AC,B\right]=B\left[C,A\right]+\left[A,BC\right]
\end{align}
for any matrices $(A,B,C)\in \left(M_2(\mathbb{C})\right)^3$. Finally, we shall also denote $\left\{A\right\}_{i,j}$ the $(i,j)$ entry $A_{i,j}$ of any matrix $A$ (we use a curly bracket to avoid confusion with the usual bracket that shall also arise in some commutators). 
We now prove the following technical result.
\begin{lemma}
\label{Dold:TechnicalLemma2}
We have
\begin{equation}
    \left\{
        \mathit{L}_1\partial_t\Phi_k
        -
        \mathit{A}_1\partial_x\Phi_k
    \right\}_{i,i}
    =
    \left\{
        \sum_{m=0}^{k-1}
        \left(
            \mathit{L}_1\mathit{A}_{k+1-m}
            -
            \mathit{A}_1\mathit{L}_{k+1-m}
        \right)\Phi_m
    \right\}_{i,i}
\end{equation}
for every $i\in\{1,2\}$.
\end{lemma}

\begin{proof}
At this stage of the induction, the equations of order $\hbar^{k+1}$ have not yet been established in their diagonal part.  Nevertheless, their off-diagonal parts are available.
Indeed, the equations of order $\hbar^k$ are already satisfied after the construction of $\operatorname{off}\Phi_k$; they do not involve derivatives of $\operatorname{diag}\Phi_k$. Hence \autoref{Dold:250109110444}, applied with index $k+1$, shows that one can choose $\operatorname{off}\Phi_{k+1}$ so that
\begin{align*}
 \operatorname{off}(\partial_t\Phi_k)
   &=
   [\mathit W,\Phi_{k+1}]
   +
   \operatorname{off}\left(
      \sum_{m=0}^{k}\mathit A_{k+1-m}\Phi_m
   \right),\\
 \operatorname{off}(\partial_x\Phi_k)
   &=
   [\mathit Y,\Phi_{k+1}]
   +
   \operatorname{off}\left(
      \sum_{m=0}^{k}\mathit L_{k+1-m}\Phi_m
   \right).
\end{align*}
Since $\mathit L_1$ and $\mathit A_1$ are strictly off-diagonal, only these off-diagonal identities contribute after taking the diagonal part of
$\mathit L_1\partial_t\Phi_k-\mathit A_1\partial_x\Phi_k$.
Multiplying the first off-diagonal identity on the left by
$\mathit L_1$, multiplying the second one on the left by
$\mathit A_1$, and then taking the diagonal part, we obtain
\begin{align*}
\operatorname{diag}\left(
    \mathit L_1\partial_t\Phi_k
    -
    \mathit A_1\partial_x\Phi_k
\right) = &
\operatorname{diag}\left(
    \mathit L_1[\mathit W,\Phi_{k+1}]
    -
    \mathit A_1[\mathit Y,\Phi_{k+1}]
\right)
 +
\operatorname{diag}\left(
    \sum_{m=0}^{k}
    \left(
        \mathit L_1\mathit A_{k+1-m}
        -
        \mathit A_1\mathit L_{k+1-m}
    \right)\Phi_m
\right).
\end{align*}
Here we have used the fact that, since $\mathit L_1$ and
$\mathit A_1$ are strictly off-diagonal,
\[
    \operatorname{diag}(\mathit L_1 Z)
        =
    \operatorname{diag}(\mathit L_1\operatorname{off}Z),
    \qquad
    \operatorname{diag}(\mathit A_1 Z)
        =
    \operatorname{diag}(\mathit A_1\operatorname{off}Z)
\]
for every matrix $Z$.
By \eqref{Dold:250109171618},
\[
    \mathit Y=2(x-\tau)\mathit W,
    \qquad
    \mathit L_1=2(x-\tau)\mathit A_1,
\quad \text{and thus} \quad 
    \mathit L_1[\mathit W,\Phi_{k+1}]
        -
    \mathit A_1[\mathit Y,\Phi_{k+1}]
        =0.
\]
Moreover, the summand corresponding to $m=k$ vanishes because
\[
    \mathit L_1\mathit A_1-\mathit A_1\mathit L_1=0. \quad
 \Rightarrow \quad
\operatorname{diag}\left(
    \mathit L_1\partial_t\Phi_k
    -
    \mathit A_1\partial_x\Phi_k
\right)
=
\operatorname{diag}\left(
    \sum_{m=0}^{k-1}
    \left(
        \mathit L_1\mathit A_{k+1-m}
        -
        \mathit A_1\mathit L_{k+1-m}
    \right)\Phi_m
\right),
\]
which proves \autoref{Dold:TechnicalLemma2}.
\end{proof}
Let us now proceed to the proof of \autoref{Dold:250123150709} which is equivalent to proving
\begin{align} 
0&=\left\{\sum_{n=1}^{k+1}(\partial_t \mathit{L}_n-\partial_x \mathit{A}_n)\Phi_{k+1-n}\right\}_{i,i}+\left\{\sum_{j=1}^{k-1}\mathit{L}_{k+1-j}(\partial_t \Phi_j)- \sum_{j=1}^{k-1}\mathit{A}_{k+1-j}(\partial_x \Phi_j)\right\}_{i,i} \nonumber \\
&+\left\{\mathit{L}_1(\partial_t\Phi_k)-\mathit{A}_1(\partial_x\Phi_k)\right\}_{i,i} \,\,\,,\,\,\, \forall \,i\in \{1,2\}\end{align}
Let $i\in \{1,2\}$. We shall split the three contributions as follows:
\begin{align} S_1&:=\left\{\sum_{n=1}^{k+1}(\partial_t \mathit{L}_n-\partial_x \mathit{A}_n)\Phi_{k+1-n}\right\}_{i,i} \nonumber \\
S_2&:=\left\{\sum_{j=1}^{k-1}\mathit{L}_{k+1-j}(\partial_t \Phi_j)- \sum_{j=1}^{k-1}\mathit{A}_{k+1-j}(\partial_x \Phi_j)\right\}_{i,i} \nonumber \\
S_3&:=\left\{\mathit{L}_1(\partial_t\Phi_k)-\mathit{A}_1(\partial_x\Phi_k)\right\}_{i,i}
\end{align}
so that \autoref{Dold:250123150709} is equivalent to proving that $S_1+S_2+S_3=0$. Note that $S_3$ is already determined by \autoref{Dold:TechnicalLemma2}. We shall now compute $S_1$ from the compatibility equations \eqref{Dold:KnownRelations} multiplied by $\Phi_{k+1-n}$ on the right:
\begin{align}\label{Dold:BlueTerms}S_1
&\overset{\eqref{Dold:KnownRelations}}{=}\left\{\sum_{n=1}^{k+1}(\left[\mathit{W},\mathit{L}_{n+1}\right]-\left[\mathit{Y},\mathit{A}_{n+1}\right])\Phi_{k+1-n}-\sum_{n=1}^{k+1}\sum_{s=1}^{n}\left[\mathit{L}_s,\mathit{A}_{n+1-s}\right]\Phi_{k+1-n}\right\}_{i,i} \nonumber \\
&\overset{m=k+1-n}{=}\left\{\sum_{m=0}^k(\left[\mathit{W},\mathit{L}_{k+2-m}\right]-\left[\mathit{Y},\mathit{A}_{k+2-m}\right])\Phi_{m}-\sum_{m=0}^k\sum_{s=1}^{k+1-m}\left[\mathit{L}_s,\mathit{A}_{k+2-m-s}\right]\Phi_{m}\right\}_{i,i} \nonumber \\
&=\left\{\sum_{m=1}^k(\left[\mathit{W},\mathit{L}_{k+2-m}\right]-\left[\mathit{Y},\mathit{A}_{k+2-m}\right])\Phi_{m}-\sum_{m=0}^{k-1}\sum_{s=1}^{k+1-m}\left[\mathit{L}_s,\mathit{A}_{k+2-m-s}\right]\Phi_{m}\right\}_{i,i}
\end{align}
where we have used the fact that $\left[\mathit{L}_1,\mathit{A}_1\right]=0$ and $\left[\mathit{W},\mathit{L}_{k+2}\right]$, $\left[\mathit{Y},\mathit{A}_{k+2}\right]$ are off-diagonal in the last equality.
In order to obtain $S_2$, we shall multiply \eqref{Dold:KnownRelations2} on the left by $\mathit{L}_{k+1-j}$ or $\mathit{A}_{k+1-j}$. It provides:
{\small\begin{align} S_2
&=\Big\{\sum_{j=1}^{k-1}\mathit{L}_{k+1-j}\left[\mathit{W},\Phi_{j+1}\right]- \sum_{j=1}^{k-1}\mathit{A}_{k+1-j}\left[\mathit{Y},\Phi_{j+1}\right] \nonumber \\&
+\sum_{j=1}^{k-1}\sum_{m=0}^{j}\mathit{L}_{k+1-j} \mathit{A}_{j+1-m} \Phi_{m}- \sum_{j=1}^{k-1}\sum_{m=0}^{j} \mathit{A}_{k+1-j}\mathit{L}_{j+1-m} \Phi_{m} \Big\}_{i,i} \nonumber \\
&=\Big\{\sum_{j=1}^{k-1}\mathit{L}_{k+1-j}\left[\mathit{W},\Phi_{j+1}\right]- \sum_{j=1}^{k-1}\mathit{A}_{k+1-j}\left[\mathit{Y},\Phi_{j+1}\right]+\sum_{j=1}^{k-1}(\mathit{L}_{k+1-j} \mathit{A}_{j+1}-\mathit{A}_{k+1-j} \mathit{L}_{j+1}) \nonumber \\&
+\sum_{j=1}^{k-1}\sum_{m=1}^{j}\mathit{L}_{k+1-j} \mathit{A}_{j+1-m} \Phi_{m}- \sum_{j=1}^{k-1}\sum_{m=1}^{j} \mathit{A}_{k+1-j}\mathit{L}_{j+1-m} \Phi_{m} \Big\}_{i,i} \nonumber \\
&=\Big\{\sum_{j=1}^{k-1}\mathit{L}_{k+1-j}\left[\mathit{W},\Phi_{j+1}\right]- \sum_{j=1}^{k-1}\mathit{A}_{k+1-j}\left[\mathit{Y},\Phi_{j+1}\right]+\sum_{j=1}^{k-1}(\mathit{L}_{k+1-j} \mathit{A}_{j+1}-\mathit{A}_{k+1-j} \mathit{L}_{j+1}) \nonumber \\&
+\sum_{m=1}^{k-1}\sum_{j=m}^{k-1}(\mathit{L}_{k+1-j} \mathit{A}_{j+1-m}-\mathit{A}_{k+1-j}\mathit{L}_{j+1-m}) \Phi_{m} \Big\}_{i,i} \nonumber \\
&\overset{s=k+1-j}{\underset{s=j+1-m}{=}}\Big\{\sum_{j=1}^{k-1}\mathit{L}_{k+1-j}\left[\mathit{W},\Phi_{j+1}\right]- \sum_{j=1}^{k-1}\mathit{A}_{k+1-j}\left[\mathit{Y},\Phi_{j+1}\right]+\sum_{j=1}^{k-1}(\mathit{L}_{k+1-j} \mathit{A}_{j+1}-\mathit{A}_{k+1-j} \mathit{L}_{j+1}) \nonumber \\&
+\sum_{m=1}^{k-1}\left(\sum_{s=2}^{k+1-m}\mathit{L}_{s} \mathit{A}_{k+2-s-m} -\sum_{s=1}^{k-m}\mathit{A}_{k+2-s-m}\mathit{L}_{s}\right) \Phi_{m} \Big\}_{i,i} \nonumber \\
&=\Big\{\sum_{j=1}^{k-1}\mathit{L}_{k+1-j}\left[\mathit{W},\Phi_{j+1}\right]- \sum_{j=1}^{k-1}\mathit{A}_{k+1-j}\left[\mathit{Y},\Phi_{j+1}\right]+\sum_{s=2}^{k}\left[\mathit{L}_{s}, \mathit{A}_{k+2-s}\right] \nonumber \\&
+\sum_{m=1}^{k-1}\left(\sum_{s=1}^{k-m}\left[\mathit{L}_{s}, \mathit{A}_{k+2-s-m}\right] +\mathit{L}_{k+1-m}\mathit{A}_1-\mathit{L}_1\mathit{A}_{k+1-m} \right)  \Phi_{m} \Big\}_{i,i}
\end{align}}
But we have:$\mathit{L}_{k+1-j}\left[\mathit{W},\Phi_{j+1}\right]=\left[\mathit{L}_{k+1-j},\mathit{W}\right]\Phi_{j+1}+\left[\mathit{W},\mathit{L}_{k+1-j}\Phi_{j+1}\right] $
so that since the last term is off-diagonal:
\begin{align} \label{Dold:GreenTerms}S_2&
\overset{m=j+1}{=}\Big\{\sum_{m=2}^{k}(\left[\mathit{L}_{k+2-m},\mathit{W}\right]-\left[\mathit{A}_{k+2-m},\mathit{Y}\right])\Phi_{m}+\sum_{s=2}^{k}\left[\mathit{L}_{s}, \mathit{A}_{k+2-s}\right] \nonumber \\&
+\sum_{m=1}^{k-1}\left(\sum_{s=1}^{k-m}\left[\mathit{L}_{s}, \mathit{A}_{k+2-s-m}\right] +\mathit{L}_{k+1-m}\mathit{A}_1-\mathit{L}_1\mathit{A}_{k+1-m} \right)  \Phi_{m} \Big\}_{i,i}
\end{align}

We may now regroup \eqref{Dold:BlueTerms} and \eqref{Dold:GreenTerms}:
{\small\begin{align} \label{Dold:BlueAndGreen}S_1+S_2&=\Big\{\sum_{m=1}^k(\left[\mathit{W},\mathit{L}_{k+2-m}\right]-\left[\mathit{Y},\mathit{A}_{k+2-m}\right])\Phi_{m}-\sum_{m=0}^{k-1}\sum_{s=1}^{k+1-m}\left[\mathit{L}_s,\mathit{A}_{k+2-m-s}\right]\Phi_{m}\Big\}_{i,i} \nonumber \\&
+\Big\{\sum_{m=2}^{k}(\left[\mathit{L}_{k+2-m},\mathit{W}\right]-\left[\mathit{A}_{k+2-m},\mathit{Y}\right])\Phi_{m}+\sum_{s=2}^{k}\left[\mathit{L}_{s}, \mathit{A}_{k+2-s}\right] \nonumber \\&
+\sum_{m=1}^{k-1}\sum_{s=1}^{k-m}\left[\mathit{L}_{s}, \mathit{A}_{k+2-s-m}\right]\Phi_{m}+ \sum_{m=1}^{k-1}(\mathit{L}_{k+1-m}\mathit{A}_1-\mathit{L}_1\mathit{A}_{k+1-m})\Phi_m   \Big\}_{i,i} \nonumber \\
&=\Big\{(\left[\mathit{W},\mathit{L}_{k+1}\right]-\left[\mathit{Y},\mathit{A}_{k+1}\right])\Phi_1-\sum_{s=1}^{k+1}\left[\mathit{L}_s,\mathit{A}_{k+2-s}\right]+\sum_{s=2}^{k}\left[\mathit{L}_{s}, \mathit{A}_{k+2-s}\right] \nonumber \\&
-\sum_{m=1}^{k-1}\left[\mathit{L}_{k+1-m},\mathit{A}_{1}\right]\Phi_{m}+ \sum_{m=1}^{k-1}(\mathit{L}_{k+1-m}\mathit{A}_1-\mathit{L}_1\mathit{A}_{k+1-m})\Phi_m\Big\}_{i,i} \nonumber \\
&=\Big\{(\left[\mathit{W},\mathit{L}_{k+1}\right]-\left[\mathit{Y},\mathit{A}_{k+1}\right])\Phi_1-\left[\mathit{L}_1,\mathit{A}_{k+1}\right]-\left[\mathit{L}_{k+1},\mathit{A}_{1}\right] \nonumber \\&
+\sum_{m=1}^{k-1}\mathit{A}_1\mathit{L}_{k+1-m}\Phi_{m}- \sum_{m=1}^{k-1}\mathit{L}_1\mathit{A}_{k+1-m}\Phi_m\Big\}_{i,i} \nonumber  \\
&=\Big\{\mathit{L}_{k+1}\left[\Phi_1,\mathit{W}\right]-\mathit{A}_{k+1}\left[\Phi_1,\mathit{Y}\right]-\left[\mathit{L}_1,\mathit{A}_{k+1}\right]-\left[\mathit{L}_{k+1},\mathit{A}_{1}\right] \nonumber \\&
+\sum_{m=1}^{k-1}\left(\mathit{A}_1\mathit{L}_{k+1-m}-\mathit{L}_1\mathit{A}_{k+1-m}\right)\Phi_m\Big\}_{i,i} \nonumber \\
&=\Big\{\mathit{A}_1\mathit{L}_{k+1}-\mathit{L}_1\mathit{A}_{k+1}+\sum_{m=1}^{k-1}(\mathit{A}_1\mathit{L}_{k+1-m}-\mathit{L}_1\mathit{A}_{k+1-m})\Phi_m\Big\}_{i,i}
\end{align}}
Let us finally add \eqref{Dold:BlueAndGreen} and \autoref{Dold:TechnicalLemma2} to get
\begin{align} S_1+S_2+S_3&=\left\{\sum_{m=0}^{k-1}(\mathit{L}_1\mathit{A}_{k+1-m}-\mathit{A}_1\mathit{L}_{k+1-m} )\Phi_m\right\}_{i,i} \nonumber \\
&+\left\{\mathit{A}_1\mathit{L}_{k+1}-\mathit{L}_1\mathit{A}_{k+1}+\sum_{m=1}^{k-1}(\mathit{A}_1\mathit{L}_{k+1-m}-\mathit{L}_1\mathit{A}_{k+1-m})\Phi_m\right\}_{i,i} \nonumber \\
&=0
\end{align}
ending the proof of \autoref{Dold:250123150709}.

The construction also makes regularity and continuation explicit.
Since $\tau w\neq0$, the leading matrix $E_0$ is invertible,
and the only eigenvalue denominator in
\eqref{Dold:250109191643} is $w$. Thus all coefficients are
holomorphic on a full neighbourhood of either point above
$x=\tau\neq0$, even though $Y$ vanishes there.
Coefficientwise continuation along a path in
$\mathring{\mathsf\Sigma}$ uses the same algebraic recursion and
primitives of the same closed diagonal one-forms. Their uniqueness
on overlapping local charts gives the continuation asserted in
\autoref{250109214318}; choosing paths on the universal cover
removes the ambiguity of the primitives and square roots.

\subsection{Several variables and higher-rank systems}
\label{app:matrix-several-times}

The matrix induction separates the part specific to Painlev\'e~I
from the part that uses only flatness. We state the latter precisely.
Let $u=(u^1,\ldots,u^m)$ be coordinates on a polydisc, and suppose
that a classical eigenbasis gauge and an exact leading phase have
put a rank-$r$ system into the form
\begin{equation}
\label{eq:D-many-times-system}
 \hbar\partial_a\hat\Phi-[\Lambda_a,\hat\Phi]
       =\hat B_a\hat\Phi,\qquad
 \hat B_a=\sum_{n\geq1}B_{a,n}\hbar^n,
 \qquad 1\leq a\leq m.
\end{equation}
Assume that the $\Lambda_a$ are diagonal and holomorphic,
$\partial_a\Lambda_b=\partial_b\Lambda_a$, that one of them,
say $\Lambda_1$, has pairwise distinct eigenvalues throughout the
polydisc, and that
$\operatorname{diag}B_{a,1}=0$ for every $a$. The transformed
flatness equations are
\begin{equation}
\label{eq:D-many-times-flatness}
 \hbar(\partial_b\hat B_a-\partial_a\hat B_b)
 -[\Lambda_b,\hat B_a]+[\Lambda_a,\hat B_b]
 +[\hat B_a,\hat B_b]=0.
\end{equation}
Under these hypotheses, the same recursion gives a unique
$\hat\Phi=I_r+O(\hbar)$ for any prescribed diagonal initial
series $\operatorname{diag}\hat\Phi(b)=I_r+O(\hbar)$.
Here is a proof that does not use the rank-two proportionality
$Y=2(x-\tau)W$.
Define operators on formal matrices by
\begin{equation*}
 \mathcal D_a M=\hbar\partial_aM-[\Lambda_a,M]-\hat B_aM.
\end{equation*}
Expansion of their commutator, using
\eqref{eq:D-many-times-flatness} and the closedness of the leading
diagonal one-form, gives $[\mathcal D_a,\mathcal D_b]=0$.
For a truncated matrix $P$, put $R_a=\mathcal D_aP$; then
\begin{equation}
\label{eq:D-residual-identity}
 \mathcal D_bR_a=\mathcal D_aR_b.
\end{equation}
If $P=\underset{j=0}{\overset{k-1}{\sum}}\Phi_j\hbar^j$ already solves the equations
through order $k-1$, write $R_a=\hbar^kF_{a,k}+O(\hbar^{k+1})$.
The lowest order in \eqref{eq:D-residual-identity} yields
\begin{equation*}
 [\Lambda_b,F_{a,k}]=[\Lambda_a,F_{b,k}].
\end{equation*}
Together with the already imposed diagonal equations, this allows
us to set
\begin{equation*}
 (\operatorname{off}\Phi_k)_{ij}
       =\frac{(F_{1,k})_{ij}}{\lambda_{1,i}-\lambda_{1,j}},
 \qquad i\neq j,
\end{equation*}
where $\Lambda_a=\operatorname{diag}(\lambda_{a,1},\ldots,
\lambda_{a,r})$. It solves all order-$k$ equations at once.
To check the diagonal transport, temporarily give $\Phi_k$ zero
diagonal and put $P=\sum_{j=0}^{k}\Phi_j\hbar^j$. Write
$R_a=\hbar^{k+1}r_a+\hbar^{k+2}s_a+O(\hbar^{k+3})$.
The next two orders of \eqref{eq:D-residual-identity} imply
\begin{align*}
 [\Lambda_b,r_a]&=[\Lambda_a,r_b],\\
 \partial_b\operatorname{diag}r_a
       -\partial_a\operatorname{diag}r_b
   &=\operatorname{diag}(B_{b,1}r_a-B_{a,1}r_b).
\end{align*}
For $i\neq j$ set
$\rho_{a,ij}=\frac{\lambda_{a,i}-\lambda_{a,j}}{\lambda_{1,i}-\lambda_{1,j}}$.
The first of these identities and the order-one flatness equation
give
\begin{equation*}
 (r_a)_{ij}=\rho_{a,ij}(r_1)_{ij},\qquad
 (B_{a,1})_{ij}=\rho_{a,ij}(B_{1,1})_{ij}.
\end{equation*}
Since $\rho_{a,ji}=\rho_{a,ij}$ and all $B_{a,1}$ have zero
diagonal, each diagonal entry of
$B_{b,1}r_a-B_{a,1}r_b$ is zero: its summand of index $j\neq i$
has the factor
$\rho_{b,ij}\rho_{a,ji}-\rho_{a,ij}\rho_{b,ji}=0$.
Thus the one-form prescribing the diagonal part of $\Phi_k$ is
closed. Integrating it with its prescribed initial value closes
the induction exactly as in \autoref{lem:D-matrix-recursion}.
The assumption $\operatorname{diag}B_{a,1}=0$ is a local
normalisation, not an extra equation to impose on the original
connection. Before this normalisation, let the first coefficients
be $C_{a,1}$. Order-one flatness gives the same off-diagonal
ratios as above, hence $\operatorname{diag}[C_{a,1},C_{b,1}]=0$.
Taking the diagonal part of order-two flatness then gives
\begin{equation*}
 \partial_b\operatorname{diag}C_{a,1}
       =\partial_a\operatorname{diag}C_{b,1}.
\end{equation*}
The diagonal one-form has a holomorphic primitive $H$ on the
polydisc. Multiplying the classical eigenbasis on the right by
the diagonal matrix $e^H$ changes the first-order diagonal part
to $\operatorname{diag}C_{a,1}-\partial_aH=0$.
This is the general version of the choice of $e_\pm$ in
\autoref{app:matrix-classical-gauge}.

Thus the induction extends to other Painlev\'e systems and to
additional deformation variables whenever the local leading
matrices admit a holomorphic common eigenbasis, their eigenvalue
one-forms are closed, and a chosen direction has a simple leading
spectrum. These are the hypotheses to check in each application;
the proof makes no claim at a point where all available leading
directions lose that separation. The number of Darboux coordinates
enters through the coefficients of the flat connection, not through
the recursive argument.

{\footnotesize
\bibliography{RHP_references}
\bibliographystyle{plain}
}

\bigskip

\textsc{Mohamad Alameddine}\\
\small Section de Math\'{e}matiques, Universit\'{e} de Gen\`{e}ve%
\footnote{24 rue du G\'{e}n\'{e}ral Dufour, 1211 Gen\`{e}ve 4, Suisse.}

\smallskip

\textsc{Olivier Marchal}\\
\small Universit\'{e} Jean Monnet, Institut Camille Jordan, %
\small Institut Universitaire de France\footnote{Institut Camille Jordan, CNRS UMR 5208,
Les Forges 2, 20 Rue du Dr Annino, 42000 Saint-\'{E}tienne, France.}

\smallskip

\textsc{Nikita Nikolaev}\\
\small School of Mathematics, University of Birmingham%
\footnote{Watson Building, Edgbaston, Birmingham, B15 2TT, United Kingdom.}

\smallskip

\textsc{Nicolas Orantin}\\
\small Section de Math\'{e}matiques, Universit\'{e} de Gen\`{e}ve%
\footnote{24 rue du G\'{e}n\'{e}ral Dufour, 1211 Gen\`{e}ve 4, Suisse.}

\end{document}